\documentclass[12pt, a4paper, oneside, openright]{book}

\usepackage{vuwthesis} % sets up some local things, mostly the front page

\usepackage{palatino} % sets palatino as the default font

\usepackage{url} % for typesetting urls
\usepackage{bm}
\usepackage{amssymb, amsfonts, amsbsy}
\usepackage{geometry}
\usepackage{graphicx}
\usepackage{amssymb,amsmath,amsthm}
\usepackage{mathrsfs} %script fonts
\usepackage{epsfig,multicol}

\begin{document}

\def\d{{\mathrm{d}}}
\def\implies{\Rightarrow}
\newcommand{\scri}{\mathscr{I}}
\newcommand{\sun}{\ensuremath{\odot}}
%------------------------------------------------
\def\d{{\mathrm{d}}}
\def\J{{\mathscr{J}}}
\def\sech{{\mathrm{sech}}}
\def\T{{\mathcal{T}}}
\def\ep{\epsilon}
\def\k{\mathbf{k}}
\def\x{\mathbf{x}}
\def\v{\mathbf{v}}
\def\s{\mathbf{s}}
\def\e{\mathbf{e}}
\def\t{\mathbf{t}}
\def\n{\mathbf{n}}
\def\u{\mathbf{u}}
\def\w{\mathbf{w}}
\def\eg{{\it e.g.}}
\def\ie{{\it i.e.}}
\def\etc{{\it etc.}}
\def\sign{{\hbox{sign}}}
%-------------------------------------------------------------------------
\def\eof{\Box}
\def\supp{\hbox{supp}}

%-------------------------------------------------------------------------
\newenvironment{warning}{{\noindent\bf Warning: }}{\hfill $\eof$\break}

\frontmatter
% Book style knows about front matter
% Report style doesn't so you need to set roman numbering etc yourself :-(

%%%%%%%%%%%%%%%%%%%%%%%%%%%%%%%%%%%%%%%%%%%%%%%%%%%%%%%

\title{Aspects of General Relativity:\\ \Large\emph{ Pseudo-Finsler extensions, Quasi-normal frequencies and Multiplication of tensorial distributions}}
\author{Jozef Sk\'akala}

\subject{Mathematics}

\abstract{This thesis is based on three different projects, all of them are directly linked to the \emph{classical} general theory of relativity, but they might have consequences for quantum gravity as well.

The first chapter deals with pseudo-Finsler geometric extensions of the classical theory, these being ways of naturally representing high-energy Lorentz symmetry violations. In this chapter we prove a certain type of ``no-go'' result for significant number of theories. This seems to have important consequences for the question of whether some weaker formulation of Einstein's equivalence principle is sustainable, if (at least) certain types of Lorentz violations occur.

The second chapter deals with the problem of highly damped quasi-normal modes related to different types of black hole spacetimes. First, we apply to this problem the technique of approximation by analytically solvable potentials. We use the Schwarzschild black hole as a consistency check for our method and derive many new and interesting results for the Schwarzschild-de Sitter (S-dS) black hole. One of the most important results is the equivalence between having a rational ratio of horizon surface gravities and periodicity of quasi-normal modes. By analysing the complementary set of analytic results derived by the use of monodromy techniques we prove that all our theorems almost completely generalize to all the known analytic results. This relates to all the types of black holes for which quasi-normal mode results are currently known.

The third chapter is related to the topic of multiplication of tensorial distributions. We focus on an alternative approach to the ones presently known. The new approach is fully based on the Colombeau equivalence relation, but technically avoids the Colombeau algebra construction. The advantage of this approach is that it naturally generalizes the covariant derivative operator into the generalized tensor algebra. It also operates with much more general concept of piecewise smooth manifold, which is in our opinion natural to the language of distributions.
\enlargethispage{10pt}   }
% Books don't normally have abstracts, and this is a bit of a hack

% Uncomment the appropriate degree
\phd
%\mscthesisonly
%\mscwithhonours
%\mscbothparts
% \otherdegree{DEGREE OR DIPLOMA NAME}

%%%%%%%%%%%%%%%%%%%%%%%%%%%%%%%%%%%%%%%%%%%%%%%%%%%%%%%

\maketitle

\chapter*{Acknowledgments}\label{C:ack}

~~~~I'm deeply thankful to my supervisor Prof. Matt Visser for many
priceless discussions and a lot of support throughout my PhD years.
It was my great pleasure to spend more than three years in a
friendly environment created by him.

I'm also deeply thankful to the Victoria University of
Wellington for supporting me with a Vice-Chancellor's Strategic
Research Scholarship and also with travel grants. Part of my
expenses were covered by the Marsden grant, the credit for this goes
again to my supervisor.

Furthermore I would like to thank to the School of Mathematics, Statistics
and Operations Research for providing me with all the facilities and
service, to the department staff for being friendly and helpful
whenever needed.

I would like to thank to the students from the applied mathematics
group, particularly to C\'eline Catto\"en, Gabriel Abreu, Petarpa
Boonserm, Bethan Cropp, Nicole Walters, Jonathan Crook, Valentina
Baccetti and Kyle Tate, for strongly contributing to a friendly and
inspiring environment. It was my pleasure to exchange ideas with
them and also to get to know them personally throughout the years.

I would also like to thank to many other physicists and
mathematicians with which I had opportunity to exchange ideas,
especially to some of the people from my old Comenius University in
Bratislava, Slovakia.

Furthermore, I would like to thank to the thesis examiners for valuable suggestions that certainly improved the content of the thesis.

Last, but not least I would like to thank to all the people that are
close to me, in particular to my family, to my girlfriend and to all
my friends. Life is extremely easy with their love and support.

\tableofcontents

%%%%%%%%%%%%%%%%%%%%%%%%%%%%%%%%%%%%%%%%%%%%%%%%%%%%%%%

% book style knows about mainmatter
% if you are using report style you will have to rest page numbering etc.
\mainmatter

%%%%%%%%%%%%%%%%%%%%%%%%%%%%%%%%%%%%%%%%%%%%%%%%%%%%%%%

% individual chapters included here
\pagestyle{plain}
\chapter*{Introduction}\label{intro}
\addcontentsline{toc}{chapter}{Introduction}

This thesis is based on three different projects. All of these three projects lie somewhere on the boundary between the \emph{classical} general theory of relativity and the so far unknown quantum theory of gravity. All of them deal with the phenomena that can be fully understood by classical (non-quantum) language, but at the same time describe, or at least strongly indicate, high energy modifications of the classical theory, mostly due to quantum gravity.

Ever since Einstein's general theory of relativity was established as a highly successful theory of gravity, physicists have also realized its unavoidable limits. The first set of problems is due to the theory itself. As is proven by rigorous mathematical theorems, physically reasonable situations lead in general relativity to solutions containing singularities. The other set of problems is due to the fact that the theory is \emph{classical}. To be able to consistently describe the \emph{full} interaction between quantum fields and gravity, one naturally needs to go beyond the \emph{classical} language and somehow quantize gravity. These two sets of problems are, however, closely related. Singularities appear in the situations where we do \emph{not} expect the classical description to be relevant.

There are presently many ideas, suggestions and conjectures about the quantum theory of gravity. They are related to the many different approaches to the problem that have appeared over the last decades, some of them having quite different backgrounds. (These are approaches such as string theory, loop quantum gravity, non-commutative geometries, etc.) Despite the fact that the basic approaches are very different, there is a wide agreement about the results obtained within the semi-classical regime. (For example, the effect of Hawking radiation.) Surprisingly, the results of semi-classical approach already seem to give us considerable information about some of the objects of \emph{full} quantum gravity, such as black holes. These results (such as the relation between horizon area and black hole entropy) must be recovered by every quantum theory of gravity, that has ambitions to be correct.
   
First two chapters in my thesis present results obtained in close collaboration with my supervisor Prof. Matt Visser. They were published in various journals (see the list of publications at the end of the thesis). The last chapter is related to my own work and its significantly shortened version will be prepared for journal submission in the immediate future.

The first chapter of the thesis deals with some possible \emph{classical} geometric extensions of the general theory of relativity. Such geometric extensions are related to possible high-energy Lorentz symmetry violations, by many physicists being assumed to be one of the possible effects of the future quantum gravity theory.
To search for a link between some generalized geometry and possible high-energy Lorentz violations seems to be, with respect to the spirit of Einstein's equivalence principle, a very natural approach.

The second chapter deals with a semi-analytic approach to the highly damp\-ed quasi-normal modes of various classical black hole space-times. It is in fact a topic from classical general relativity, but it is also of particular interest of the quantum gravity community. This is because the asymptotic quasi-normal modes behaviour is suspected \cite{Hod, Maggiore} to be connected to the area spectrum of the quantum black holes.

The third chapter represents a topic from the field of mathematical physics closely related to the general theory of relativity. It offers new ideas about how to \emph{fully} generalize the language of differential geometry into the distributional framework. It is again a conceptual extension which has direct relevance for the classical theory, but it might not be unreasonable to assume that it can have important consequences for quantizing gravity as well. 

At the end of the thesis we summarize our results. This is followed by appendices containing some of the most common physics and mathematics results relevant for the calculations in this thesis.   
\pagestyle{headings}
\chapter{Pseudo-Finsler extensions to gravity}\label{C:Finsler}

\section{Introduction}

%add references

\paragraph{Possible low-energy manifold-like limits of quantum gravity}
One of the most significant problems in the recent history of
theoretical physics is the problem of quantization of gravity. Most
of the candidates for a quantum gravity theory suggest that the
description of spacetime is at the fundamental level far from the
traditional concept of manifold.

Despite this fact, one might be still interested in whether these
theories have a deeper manifold-like low-energy limit; one more subtle than the
ordinary pseudo-Riemannian geometry. If they have such limit, it
must be definitely an extension of pseudo-Riemannian geometry, but
in the same time it must be ``close to'' the highly successful
concept of pseudo-Riemannian geometry.

\paragraph{Ultra-high energy violations of Lorentz symmetry}
The theoretical physics community has recently exhibited increasing
interest in the possibility of ultra-high-energy violations of
Lorentz invariance \cite{Jacobson1, Jacobson2, Jacobson3, Jacobson4,
Jacobson5, SVW1, SVW2, SVW0}. Specifically, recent speculations
regarding Lorentz symmetry breaking and/or fundamental anisotropies
and/or multi-refringence arise separately in the many and various
approaches to quantum gravity.

Such phenomena arise in loop quantum gravity \cite{LQG2, LQG1},
string models \cite{string3, string1, string2}, and causal dynamical
triangulations \cite{causal1, causal2}, and are also part and parcel
of the ``analogue spacetime" programme \cite{analogue}, and of many
attempts at developing ``emergent gravity" \cite{emergence, Kostelecky11,
Kostelecky10}. Recently, the ultra-high energy breaking
of Lorentz invariance has been central to the Horava-Lifshitz models
\cite{Horava0, Horava1, Horava2, SVW1, SVW2, SVW0}. Of course not
all models of quantum gravity lead to high-energy Lorentz symmetry
breaking, and the comments below should be viewed as exploring one
particular class of interesting models.

\paragraph{The connection between Lorentz violations and geometry}
The extensions of pseudo-Riemannian geometry can typically modify dispersion
relations, so one can be interested in seeing if such modified
dispersions relations can be naturally embedded in some extension of
pseudo-Riemannian geometry. These extensions could be then naturally
viewed as a low energy manifold-like non-pseudo-Riemannian limit of
a given quantum gravity approach.

In the other direction, if we wish to follow the spirit of Einstein's
equivalence principle and develop a geometric spacetime framework
for representing Lorentz symmetry breaking, either due to spacetime
anisotropies or multi-refringence, then it certainly cannot be
standard pseudo-Riemannian geometry.

This strongly suggests that
carefully thought out extensions and modifications of
pseudo-Riemannian geometry might be of real interest to both the
general relativity and high-energy communities.

\paragraph{Why focus on the light cone structure?}
In particular, when attempting to generalize pseudo-Riemannian
geometry, the interplay between the ``signal cones" of a
multi-ref\-ring\-ent theory and the generalized spacetime geometry is an
issue of considerable interest:

\begin{itemize}

\item
In multi-refringent situations it is quite easy to unify all the signal cones in one single Fresnel
equation that simultaneously describes all polarization modes on an equal footing.

\item
In a standard manifold setting, where we retain the usual commutative coordinates, we shall see that it
is natural to demand that each polarization mode can be assigned a
specific geometric object. This object is in fact a Lorentzian
analogue of what mathematicians know as Finsler norm (see the next
section for details).

\item
In standard general relativity the (single, unique) signal cone
almost completely specifies the spacetime geometry --- one needs
only supplement the signal cone structure with one extra degree of
freedom at each point in spacetime, an overall conformal factor, in
order to completely specify the spacetime metric, and thereby
completely specify the geometry. This is ultimately due to the fact
that in standard pseudo-Riemannian geometry the scalar product is a
simple bi-linear operation. Unfortunately in the more general pseudo-Finsler geometry\footnote{For the details of what we mean by ``pseudo-Finsler geometry'' see the next section.}
life is more difficult, but one might still have a hope that
following the guideline of simplicity the light-cone behavior could
hold a crucial piece of information about the overall geometry.

\end{itemize}

\paragraph{Bi-refringent crystal and beyond}
Considerable insight into such Finsler-like models can be provided
by considering the ``analogue spacetime" programme, where analogue
models of curved spacetime emerge at some level from well understood
physical systems \cite{analogue}. In particular, the physics of
bi-axial bi-refringent crystals \cite{BW} provides a particularly
simple physical analogue model for the mathematical object
introduced some 155 years ago by Bernhard Riemann \cite{Riemann},
and now known as Finsler distance\footnote{It must be emphasized that, despite many misapprehensions to the
contrary, uni-axial birefringent crystals are relatively
uninteresting in this regard; they do not lead to Finsler 3-spaces,
but ``merely" yield bi-metric Riemannian 3-geometries.} (again see the next section for
more details). We shall soon see that this mathematical object can
reasonably easily be extended to a Lorentzian signature
pseudo-Finsler spacetime, with an appropriate pseudo-Finsler norm.

Note that we are not particularly interested in the properties of
bi-refringent crystals per se, we use them only as an exemplar of
Finsler 3-space and Finsler space-time, as a guidepost to more
complicated things that may happen in Finslerian extensions to
general relativity. As a result of this fact in the next step we
show that all our observations from the particular bi-refringent
crystal case hold in general bi-metric situations\footnote{For the difference between bi-refringence and bi-metricity see for example \cite{birefringence}.}. This is because our
ultimate goal is to be able to say something about the (presumed)
low-energy manifold-like limit of whatever quantum theory (or class
of quantum theories) is leading to a bi-metric / bi-refringent theory approximately
reproducing Einstein gravity.

\paragraph{The no-go result}

While the use of Finsler 3-spaces to describe crystal optics is reasonably
common knowledge within the community of mathematicians and physicists studying
Finsler spaces, it is very difficult to get a clear and concise explanation of
exactly what is going on when one generalizes to Lorentzian signature space-time.
In particular the fact that any relativistic formulation of Finsler space needs
to work in Lorentzian signature (- +++), instead of the Euclidean signature (+ +++)
more typically used by the mathematical community, leads to many technical subtleties
(and can sometimes completely invalidate naive conclusions). Moreover unlike the \emph{spacetime} Finsler norm, defining a \emph{spacetime} Finsler metric is fraught with technical problems. These problems seem to be fundamental and hold in arbitrary bi-metric situations.

So the basic things we assert are:

\begin{itemize}

\item
What is exceedingly difficult, and we shall argue is in fact outright impossible
within this framework, is to construct a unified and still simple
formalism that moves ``off-shell" (off the signal cones).

\item
This is a negative result, a ``no-go theorem", which we hope will focus attention on what can
and cannot be accomplished in any natural way when dealing with multi-refringent anisotropic Finsler-like
extensions to general relativistic spacetime.

\end{itemize}

To this end, our ``no-go" result indicates that the popular
assumption that anisotropies and multi-refingence are likely to
occur in ``quantum gravity" leads to significant difficulties for
the Einstein equivalence principle --- since even the loosest
interpretation of the Einstein equivalence principle would imply the
necessity of a coherent formalism for dealing with all signal comes,
and the spacetime geometry, in some unified manner. We conclude
that, despite the fact that spacetime anisotropies and
multi-refringence are very popularly assumed to be natural features
of ``quantum gravity", and while these features have a
straightforward ``on-shell" implementation in terms of a suitably
defined Fresnel equation, there is no natural way of extending them
``off-shell" and embedding them into a single over-arching spacetime
geometry.

But to remind the reader: if one steps outside of the usual manifold
picture, either by adopting non-commutative coordinates, or even
more abstract choices such as spin foams, causal dynamical
triangulations, or string-inspired models, then the issues addressed
in this chapter are moot --- our considerations are relevant only
insofar as one is interested in the first nontrivial deviations from
exact low-energy Lorentz invariance, and only relevant insofar as
these first nontrivial effects can be placed in a Finsler-like
setting.

\paragraph{The structure of this chapter}
This chapter begins with introducing the concept of Finsler
geometry. This is followed by proving our ``no-go'' result for the
particular example of the bi-refringent crystal analogue model.
After this we show that the result holds in arbitrary bi-metric situation. At
the end of this chapter we add a section where we explore the
general conditions, which any pseudo-Finslerian geometry must fulfil in order
to give bi-refringence. But as a consequence of our ``no-go''
result, such constructs represent only a ``complicated'' and
non-intuitive route for how to recover bi-refringence by
pseudo-Finslerian geometry.

\section{Basics of (pseudo-)Finsler geometry}

Mathematically, we define a Finsler function (Finsler norm, Finsler distance
function)~\cite{Bejancu, Finsler} to be a $\mathbb{C}$-valued function $F(x,v)$ on the tangent bundle to
a manifold, such that it is homogeneous of degree 1:
\begin{equation}
F(x, \kappa \, v) = \kappa \; F(x, v),~~~~~\kappa>0, ~~~~x\in M, ~~v\in T_{x}M.
\end{equation}
This then allows one to define a notion of distance on the manifold,
as the minimal value of the functional
\begin{equation}
S\left(x(t_i),x(t_f)\right) = \int_{t_i}^{t_f}  \left|F\left( x(t),  {\d
x(t) \over\d t} \right)\right| \; \d t,
\end{equation}
which is now guaranteed to be independent of the specific parameterization
$t$. 

By ``pseudo-Finsler geometry'' we mean Finsler geometry with Lorentzian signature. Now by Lorenzian signature of the general Finsler metric (see the equation \ref{Fmet}) we mean, that for any arbitrary vector taken as an argument of the metric we obtain matrix with $(-+++)$ signature. A pseudo-Riemannian norm is only a special case of a pseudo-Finsler
norm. For a pseudo-\-Riemannian manifold with metric $g_{ab}(x)$ one
would take
\begin{equation}\label{Riemann}
F(x, v) = \sqrt{ g_{ab}(x)\; v^a\,v^b},
\end{equation}
but for a general pseudo--Finslerian manifold the function $F(x, v)$
is arbitrary except for the 1-homogeneity constraint in $v$ and the metric signature constraint. Note that
in Euclidean signature (where $g_{ab}(x,v)$ is
taken to be for any arbitrary $v$ a positive definite matrix), the general Finsler function $F(x,v)$ is
typically smooth except at $v=0$. In Lorentzian signature however,
$F(x,v)$ is typically non-smooth for all null vectors --- so that
non-smoothness issues have grown to affect (and infect) the entire
null cone (signal cone). As we shall subsequently see below,
sometimes a suitable higher algebraic \emph{power}, $F^{2n}(x,v)$,
of the pseudo-Finsler norm is smooth.

To ensure smoothness of the (pseudo-)Finsler \emph{metric}, defined
below, it is enou\-gh to weaken the condition that $F(x,v)$ shall be
smooth and to demand only that the square $F^2(x,v)$ be smooth,
except possibly at $v=0$. It is standard to define the
(pseudo-)Finsler \emph{metric} as
\begin{equation}\label{Fmet}
g_{ab}(x,v) \equiv {1\over2} \; { \partial^2 [F^2(x,v)] \over \partial
v^a \; \partial v^b}
\end{equation}
which then satisfies the constraint that it is homogeneous of order zero
\begin{equation}
g_{ab}(x, \kappa\, v ) = g_{ab}(x, v), ~~~~~~\kappa>0.
\end{equation}
This can be viewed as a ``direction-dependent metric'', and is
clearly a significant generalization of the usual (pseudo-)Riemannian
case. One can immediately see that pseudo-Riemannian metric fulfills
this definition with respect to the pseudo-Riemannian norm
\eqref{Riemann}.

Almost all of the relevant mathematical literature has been
developed for the Euclidean signature case. Because of this assumption,
any mathematical result  that depends critically on the assumed
positive definite nature of the matrix of metric coefficients
\emph{cannot} be carried over into the physically interesting
pseudo-Finsler regime, at least not without an independent proof
that avoids the positive definite assumption. (Unfortunately it is
not uncommon to find significant mathematical errors in the
pseudo-Finsler physics literature due to neglect of this elementary
point.)   Basic references within the mathematical literature
include \cite{Bejancu, Finsler}.

The Legendre transformation between a vector tangent space at the
point $x$ and its dual: $V_{x}\to V_{x}^{*}$, is defined as

\begin{equation}
l_{b}(v)\equiv g_{ab}(x,v)v^{a}
\end{equation}

Then the dual (pseudo-)Finsler norm $F^{*}$ can be defined by the
condition:

\begin{equation}
F^{*}(l (v))\equiv F(v).
\end{equation}

The dual metric is again naturally obtained as:

\begin{equation}
g^{ab}(x,v)\equiv\frac{1}{2}\frac{\partial^{2}[F^{*2}(v,x)]}{\partial v_{a}
\partial v_{b}}.
\end{equation}

All this is a natural generalization from the (pseudo-)Riemannian
case. The construction of a full (pseudo-)Finsler geometry is in general
significantly more complicated than in the (pseudo-)Riemannian subcase.
But since the definition of objects like non-linear connection,
Finsler connection, (etc.), is not needed for the purpose of this
chapter, it will be omitted here and left for the specialized
literature (see for example \cite{Rund}).

\section{Analogue model: Birefringent crystal}

\subsection{Outline}

Purely for the purposes of developing a useful analogy, which we
shall use as a guide to the mathematics we wish to develop, we will
focus on the optical physics of bi-axial bi-refringent crystals.
After the basic definitions are presented, we will show how various
purely spatial 3-space Finsler structures arise. (Many purely
technical details, when not directly involved in the logic flow,
will be relegated to the appendix \ref{A:BW}.) We
again emphasize that uni-axial bi-refringent crystals, which are
what much of the technical literature and textbook presentations
typically focus on, are for our purposes rather uninteresting ---
uni-axial bi-refringent crystals ``merely'' lead to bi-metric
Riemannian space-times and are from a Finslerian perspective
``trivial''. Such crystals are only one
particular example demonstrating general difficulties with Finsler
representation of bi-metric theories. The generalization from this
example to any bi-metric case is presented in the following section.

We shall soon see that even in three-dimensional space there are at
least four logically distinct Finsler structures of interest: On the
tangent space each of the two photon polarizations leads, via study
of the group velocity, to two quite distinct Finsler spacetimes. On
the co-tangent space each of the two photon polarizations leads, via
study of the phase velocity, to two quite distinct co-Finsler
spacetimes. The inter-relations between these four structures is
considerably more subtle than one might naively expect.

Additionally, (apart from some purely technical difficulties along
the optical axes in bi-axial crystals), each of these four
3-dimensional spatial Finsler structures has a natural 4-dimensional
extension to a spacetime pseudo-Finsler structure. Beyond that,
there are reasonably natural ways of merging the two photon
polarizations into ``unified'' Finsler and co-Finsler norms, closely
related to the appropriate Fresnel equation, though the associated
Finsler metrics are considerably more problematic --- all these
mathematical constructions do come with a price --- and we shall be
careful to point out exactly where the technical difficulties
lie. Finally, using this well-understood physical system as a
template, we shall (in the spirit of analogue spacetime programme)
then ask what this might tell us about possible Finslerian
extensions to general relativity, and in particular to the subtle
relationship between bi-refringence and bi-metricity, (or more
generally, multi-refringence and multi-metricity).

Specifically, we have investigated the possibility of whether one can usefully and cleanly deal with both Finsler structure
(anisotropy) and multi-refringence simultaneously. That is, given
two (or more) ``signal cones'': Is it possible to naturally and
intuitively construct a ``unified" pseudo-Finsler spacetime such
that the pseudo-Finsler metric specifies null vectors on these
``signal cones'', but has no other zeros or singularities? Our
results are much less encouraging than we had originally hoped, and
lead to a ``no-go'' result.

\subsection{Space versus space-time: Interpretations of the Finsler and co-Finsler structures}
%----------------------------------------------------------------------------------------------------------

The key physics point in bi-axial bi-refringent crystal optics is
that the group velocities, and the phase velocities, are both
anisotropic and depend on direction in a rather complicated way
\cite{BW}. Technical details that would detract from the flow of the
text are relegated to appendix \ref{A:BW}.

%----------------------------------------------------------------------------------------------------------
\subsubsection{From group velocity to pseudo-Finsler norms}
%----------------------------------------------------------------------------------------------------------

We can summarize the situation by pointing out that the group velocity is given by
\begin{equation}
v_g^2(\n) = {\bar q_2(\n,\n) \pm \sqrt{ \bar q_2(\n,\n)^2 - \bar q_0(\n,\n)\;(\n \cdot \n) }\over \bar q_0(\n,\n)},
\end{equation}
where $\bar q_2(\n,\n)$ and $\bar q_0(\n,\n)$ are known quadratic functions of the direction $\n$ and are given as

\begin{equation}
\tilde q_{0}(\mathbf{n},\mathbf{n})=n^{2}_{x}v^{-2}_{y}v^{-2}_{z}+n^{2}_{y}v^{-2}_{x}v^{-2}_{z}+n^{2}_{z}v^{-2}_{x}v^{-2}_{y};
\end{equation}

\begin{equation}
\tilde q_{2}(\mathbf{n},\mathbf{n})=\frac{1}{2}(n^{2}_{x}(v^{-2}_{y}+v^{-2}_{z})+n^{2}_{y}(v^{-2}_{x}+v^{-2}_{z})+n^{2}_{z}(v^{-2}_{x}+v^{-2}_{y})).
\end{equation}

 The coefficients in these quadratic forms are explicit functions of the components of the ~$3~\times~3$~ permittivity tensor. (See appendix \ref{A:BW}.)
The function $v_g(\n)$ so defined is homogeneous of degree zero in the components of $\n$:
\begin{equation}
v_g( \kappa\, \n ) = v_g(\n) = v_g(\hat\n).
\end{equation}
The homogeneous degree zero property should remind one of the relevant feature exhibited by the Finsler metric. There is a natural connection between the concept of group velocity and the geometric objects on a tangent space (rather than a co-tangent space). This is given by the fact that group velocity describes how energy propagates. %add more - explanation

Let us now first define the quantities
\begin{eqnarray}
F_{3\pm}(\n) &=&  {||\n||\over v_g(\n)}
= \sqrt{\bar q_2(\n,\n) \mp \sqrt{ \bar q_2(\n,\n)^2 - \bar q_0(\n,\n)\;(\n \cdot \n) }},
\end{eqnarray}
or adopt the perhaps more transparent notation
\begin{eqnarray}
F_{3\pm}(\d\x) =  {||\d\x||\over v_g(\d\x)}
=
\sqrt{\bar q_2(\d\x,\d\x) \mp \sqrt{ \bar q_2(\d\x,\d\x)^2 - \bar q_0(\d\x,\d\x)\;(\d\x \cdot \d\x) }}~.~~~\label{3norm}
\end{eqnarray}

The quantity $1/v_{g}({\bf n})$ appearing in (\ref{3norm}) is in the literature often called ``slowness''. (\ref{3norm}) is by inspection a 3-dimensional (Riemannian) Finsler distance defined on \emph{space}, having all the correct homogeneity properties, $F_{3\pm}(\kappa\,\d\x) = |\kappa|\, F_{3\pm}(\d\x)$. Physically, the Finsler distance is in this situation the time taken for the wavepacket to travel a distance $\d\x$.

\bigskip

To now extend the construction given above to full (3+1) dimensional \emph{spacetime}, we first define a generic 4-vector
\begin{equation}
\d X = (\d t; \d\x),
\end{equation}
and then formally construct
\begin{equation}
F_{4\pm}(\d X) = \sqrt{ - (\d t)^2 + F_{3\pm}(\d\x)^2 }.
\end{equation}
That is
\begin{equation}
F_{4\pm}(\d X)
=
\sqrt{ - (\d t)^2 + {\d\x\cdot\d\x\over v_g(\d\x)^2} }.
\end{equation}
Even more explicitly, one may write
\begin{eqnarray}
F_{4\pm}(\d X)  ~~~~~~~~~~~~~~~~~~~~~~~~~~~~~~~~~~~~~~~~~~~~~~~~~~~~~~~~~~~~~~~~~~~~~~~~~~~~~~~~~~~~~~~~~~~~~~~~~~~~~~~~~~~~~~~~~~~~~~\nonumber\\
=\Big[ - (\d t)^2 +\bar q_2(\d\x,\d\x)
\mp
\sqrt{ \bar q_2(\d\x,\d\x)^2 - \bar q_0(\d\x,\d\x)\;(\d\x\cdot \d\x) }\Big]^{1/2}.
\qquad
\end{eqnarray}
The null cones (signal cones) of $F_{4\pm}(\d X)$ are defined by
\begin{equation}
F_{4\pm}(\d X) = 0 \qquad \Leftrightarrow \qquad ||\d\x|| = v_g(\d\x) \; \d t.
\end{equation}
So far this has given us a very natural \emph{pair} of (3+1)-dimensional pseudo--Finsler structures in terms of the ray velocities corresponding to the two photon polarizations.

\bigskip

\noindent
For future use, let us now formally define the quantity
\begin{eqnarray}
\d s^4 &=& \left\{ F_4(\d X) \right\}^4
\\
&=& \left\{ F_{4+}(\d X) \; F_{4-}(\d X) \right\}^2
\nonumber
\label{E:4d-dc-F}
\\
&=&
(\d t)^4 - 2 (\d t)^2 \; \bar q_2(\d\x,\d\x) + \bar q_0(\d\x,\d\x)\;(\d\x\cdot \d\x).
\nonumber
\end{eqnarray}
 This certainly provides an example of a specific and simple  $4^{th}$-root Finsler norm that can naturally and symmetrically be constructed from the two polarization mo\-des, and its properties (and defects) are certainly worth investigating. Physically the condition $\d s=0$ defines a double-sheeted conoid (a double-sheeted topological cone)   that is the union of the propagation cone of the individual photon polarizations. This Finsler norm defines the Finsler geometry naturally unifying the two original geometries.  It is also very close to quartic extension of the notion of distance that Bernhard Riemann speculated about in his inaugural lecture (see \cite{Riemann}).

However we shall soon see that when it comes to defining a Finsler spacetime \emph{metric} this construction nevertheless leads to a number of severe technical difficulties; difficulties that can be tracked back to the fact that we are working in non-Euclidean signature.

%----------------------------------------------------------------------------------------------------------
\subsubsection{From phase velocity to pseudo-co-Finsler norms}
%----------------------------------------------------------------------------------------------------------
In counterpoint, as a function of wave-vector the phase velocity is
\begin{equation}
v_p^2(\k) = {q_2(\k,\k) \pm\sqrt{ q_2(\k,\k)^2 - q_0(\k,\k) \; (\k\cdot \k) }\over (\k\cdot\k)},
\end{equation}
where the quadratics $q_2(\k,\k)$ and $q_0(\k,\k)$ are now given by equations~

\begin{equation}
 q_{0}(\mathbf{k},\mathbf{k})=k^{2}_{x}v^{2}_{y}v^{2}_{z}+k^{2}_{y}v^{2}_{x}v^{2}_{z}+k^{2}_{z}v^{2}_{x}v^{2}_{y};
\end{equation}

and

\begin{equation}
 q_{2}(\mathbf{k},\mathbf{k})=\frac{1}{2}\left(k^{2}_{x}(v^{2}_{y}+v^{2}_{z})+k^{2}_{y}(v^{2}_{x}+v^{2}_{z})+k^{2}_{z}(v^{2}_{x}+v^{2}_{y})\right).
\end{equation}

This expression is homogeneous of order zero in $\k$, so that
 \begin{equation}
v_p(\kappa  \, \k) = v_p(\k) = v_p(\hat \k).
\end{equation}
Again, we begin to see a hint of Finsler structure emerging.
Because $\k$ is a wave-vector it transforms in the same way as the gradient of the phase; thus $\k$ is most naturally thought of as living in the 3-dimensional space of co-tangents to physical 3-space. Let us now define a \emph{co-Finsler} structure on that co-tangent space by
\begin{eqnarray}
G_{3\pm}(\k) &=&  v_p(\k)\;  ||\k||
=
\sqrt{ q_2(\k,\k) \pm\sqrt{ q_2(\k,\k)^2 - q_0(\k,\k) \; (\k\cdot \k) } }.
\end{eqnarray}
We use the symbol $G$ rather than $F$ to emphasize that this is a co-Finsler structure, and note that this object satisfies the required homogeneity property
\begin{equation}
G_{3\pm}(\kappa \k) =  |\kappa|\,G_{3\pm}(\k).
\end{equation}
Now let us go for a (3+1) dimensional spacetime interpretation: Consider the 4-co-vector
\begin{equation}
k = \left(\omega; \k \right),
\end{equation}
and define
\begin{equation}
G_{4\pm}(k) = \sqrt{ - \omega^2 + G_{3\pm}(\k)^2 }.
\end{equation}
That is
\begin{equation}
G_{4\pm}(k) =  \sqrt{ - \omega^2 + v_p(\k)^2\; (\k \cdot \k) }.
\end{equation}
More explicitly
\begin{eqnarray}
G_{4\pm}(k)  &=&
 \Big[ - \omega^2 +    q_2(\k,\k)
 \pm \sqrt{q_2(\k,\k)^2 - q_0(\k,\k) \; (\k\cdot \k)}\Big]^{1/2} .
\end{eqnarray}
We again see that this object satisfies the required homogeneity property
\begin{equation}
G_{4\pm}(\kappa k) =  |\kappa|\,G_{4\pm}(k),
\end{equation}
so that this object is indeed suitable for interpretation as a co-Finsler structure. Furthermore the \emph{null co-vectors} of $G_4$ are defined by
\begin{equation}
G_{4\pm}(k) = 0 \qquad \Leftrightarrow \qquad \omega = v_p(\k) \; ||\k||,
\end{equation}
which is exactly the notion of \emph{dispersion relation} for allowed ``on mass shell'' wave-4-vectors that we are trying to capture. Thus $G_{4\pm}$ lives naturally on the co-tangent space to physical spacetime, and we can interpret it as a pseudo-co-Finsler structure.

We can again define a ``unified'' quantity
\begin{eqnarray}
G_{4}(k)^4 &=& \left\{ G_{4+}(k) \; G_{4-}(k) \right\}^2
\nonumber\\
&=&  \omega^4 - 2 \omega^2  \; q_2(\k,\k) + q_0(\k,\k) \; (\k\cdot \k).\quad
\label{E:4d-dc-G}
\end{eqnarray}
Physically, the condition $G_4(k)=0$ simultaneously encodes both dispersion relations for the two photon polarizations. It defines a double-sheeted conoid (a double-sheeted topological cone) that is the union of the dispersion relations of the individual photon polarizations. The vanishing of $G_{4}(k)$ can be viewed as a Fresnel equation, and can indeed be directly related to Fresnel's condition for the propagation of a mode of 4-wavenumber $k = (\omega; \k)$.
As is the case for $F_4(\d X)$,  we shall soon see that this construction (once one tries to extract a spacetime co-Finsler \emph{metric}) nevertheless leads to a number of severe technical difficulties; difficulties that can again be tracked back to the fact that we are now working in non-Euclidean signature.

%----------------------------------------------------------------------------------------------------------
\subsection{Technical issues and problems}
%----------------------------------------------------------------------------------------------------------

The situation as presented so far looks very pleasant and completely
under control --- and if what we had seen so far were all there was
to the matter, then the study of pseudo-Finsler space-times would be
very straightforward indeed --- but now let us indicate where
potential problems are hiding.
\begin{itemize}
\item
Note that up to this stage we have not established any direct
connection between the Finsler functions $F_{3\pm}(\n)$ and the
co-Finsler functions $G_{3\pm}(\k)$. Physically it is clear that they
must be very closely related, but (as we shall soon see)
establishing the precise connection is tricky.

\item
Furthermore, the transition from Finsler \emph{distance} to Finsler
\emph{metric} requires at least two derivatives.  Even in Euclidean
signature this places some smoothness constraints on the Finsler
distance, smoothness constraints that are nontrivial and not always
satisfied. 

\item Especially, there are problematic technical issues involving the 4-di\-men\-sion\-al \emph{spacetime} Finsler and co-Finsler metrics --- certain components of the metric are infinite, and this time the potential pathology is widespread. (In Lorentzian-like signature situations potential problems tend to infect the entire null cone.)

\end{itemize}

%----------------------------------------------------------------------------------------------------------
\subsubsection{The Finsler and co-Finsler 3-metrics}
%----------------------------------------------------------------------------------------------------------

The standard definition used to generate a Finsler \emph{metric}
from a Finsler \emph{distance} is to set:
\begin{equation}
g_{ij}(\n) =  {1\over2} {\partial^2 [F_{3\pm}(\n)^2]\over \partial
n^i \, \partial n^j},
\end{equation}
which in this particular case implies
\begin{equation}
g_{ij}(\n)  = {1\over2} {\partial^2 [\bar q_2(\n,\n) \mp \sqrt{ \bar
q_2(\n,\n)^2 - \bar q_0(\n,\n)\;(\n \cdot \n) }]\over \partial n^i
\, \partial n^j}.
\end{equation}
It is convenient to rewrite the quadratics as
\begin{equation}
\bar q_2(\n,\n) = [\bar q_2]_{ij}\; n^i\,n^j;
\end{equation}
\begin{equation}
\bar q_0(\n,\n) = [\bar q_0]_{ij}\; n^i\,n^j;
\end{equation}
since then we see
\begin{eqnarray}
g_{ij}(\n) &=&  [\bar q_2]_{ij} \mp \hbox{(discriminant
contributions)}.
\end{eqnarray}
Unfortunately we shall soon see that the contributions coming from
the discriminant are both messy, and in certain directions,
ill-defined. This is obvious from the fact that squares of both
Finsler functions are not even everywhere differentiable.

Similarly we can construct a Finsler co-metric:
\begin{equation}
h^{ij}(\k) =  {1\over2} {\partial^2 [G_{3\pm}(\k)^2]\over \partial
k_i \, \partial k_j},
\end{equation}
which specializes to
\begin{equation}
h^{ij}(\k) = {1\over2} {\partial^2 [q_2(\k,\k) \mp
\sqrt{q_2(\k,\k)^2 - \bar q_0(\k,\k)\;(\k \cdot \k) }]\over \partial
k_i \, \partial k_j}.
\end{equation}
It is again convenient to rewrite the quadratics as
\begin{equation}
q_2(\k,\k) = [q_2]^{ij}\; k_i\, k_j;
\end{equation}
\begin{equation}
q_0(\k,\k) = [q_0]^{ij}\; k_i\,k_j;
\end{equation}
since then we see
\begin{eqnarray}
h^{ij}(\k) &=&  [q_2]^{ij} \mp \hbox{(discriminant contributions)}.
\end{eqnarray}
Again we shall soon see that the contributions coming from the
discriminant are, in certain directions, problematic.

%----------------------------------------------------------------------------------------------------------
\subsubsection{Technical problems with the Finsler 3-metric}
%---------------------------------------------------------------------------------------------------------

Consider the (ray) discriminant
\begin{equation}
\bar D = \bar q_2(\n,\n)^2 - \bar q_0(\n,\n)\;(\n \cdot \n).
\end{equation}
There are three cases of immediate (mathematical) interest:

{\bf Isotropic:} If $v_x=v_y=v_z$ then $\bar D=0$; in this case the
two Finsler functions $F_\pm$ are equal to ech other. $F_3(\d\x)$
then describes an ordinary Riemannian geometry, and $F_4(\d X)$ an
ordinary pseudo--Riemannian geometry. This is the standard
situation, and is \emph{for our current purposes} physically uninteresting.

{\bf Uni-axial:} If one of the principal velocities is distinct from
the other two, then  we can without loss of generality set $v_x=v_y
= v_o$ and $v_z=v_e$.  The discriminant then factorizes into a
perfect square
\begin{equation}
\bar D = \left\{ {(v_o^2-v_e^2)(n_x^2+n_y^2)\over 2 v_o^2 v_e^2 }
\right\}^2.
\end{equation}
In this case it is immediately clear that both $F_{3\pm}(\d\x)$
reduce to simple quadratics, and so describe two ordinary Riemannian
geometries. Indeed
\begin{equation}
F_{3+}(\n) = {\n\cdot \n\over v_o^2};
\end{equation}
and
\begin{equation}
F_{3_-}(\n) = {n_x^2+n_y^2\over v_e^2} + {n_z^2\over v_o^2}.
\end{equation}
In the language of crystal optics $v_o$ and $v_e$ are the
``ordinary'' and ``extraordinary'' ray velocities of a uni-axial
birefringent crystal.  In geometrical language  the two photon
polarizations ``see'' distinct Riemannian 3-geometries
$F_{3\pm}(\d\x)$ and distinct pseudo-Riemannian 4-geometries
$F_{4\pm}(\d X)$ --- this situation is referred to as ``bi-metric''.
This situation is \emph{for our current purposes} physically uninteresting.

{\bf Bi-axial:} The full power of the Finsler approach is
\emph{only} needed for the bi-axial situation where the three
principal velocities are distinct. This is the \emph{only} situation
of real physical interest for us, as it is the only situation that
leads to a non-trivial Finsler metric. In this case we can without
loss of generality orient the axes so that $v_x>v_y>v_z$.  There are
now two distinct directions in the $x$--$z$ plane where the
discriminant vanishes --- these are the called the (ray) optical
axes. After some manipulations that we relegate to Appendix \ref{A:axes}, the
discriminant can be factorized as
\begin{equation}
\label{E:discriminant} \bar D = {(v_x^2-v_z^2)^2\over 4 v_x^4 v_z^4}
\times \left[  (\n\cdot \n) - (\bar\e_1\cdot\n)^2 \right] \left[
(\n\cdot \n) - (\bar\e_2\cdot\n)^2 \right],
\end{equation}
where the two distinct  (ray) optical axes are
\begin{equation}
\bar \e_{1,2} = \left( \pm { {v_y\over v_x} \sqrt{v_x^2-v_y^2\over
v_x^2-v_z^2}}; \;\; 0 \;\; ;
 { {v_y\over v_z} \sqrt{v_y^2-v_z^2\over v_x^2-v_y^2}}  \right).
\end{equation}
Note that $\bar \e_{1,2}$ are unit vectors (in the ordinary
Euclidean norm) so that the discriminant $\bar D$ vanishes for any
$\n\propto \bar\e_{1,2}$, and does not vanish anywhere else.  We can
thus  introduce projection operators $\bar P_1$ and $\bar P_2$ and
write
\begin{equation}
\bar P_1(\n,\n) =  (\n\cdot \n) - (\bar\e_1\cdot\n)^2;
\end{equation}
\begin{equation}
\bar P_2(\n,\n) =  (\n\cdot \n) - (\bar\e_2\cdot\n)^2.
\end{equation}
Combining this with our previous results:
\begin{eqnarray}
\left\{ F_{3\pm}(\n)\right\}^2 &=& \bar q_2(\n,\n) \mp
{(v_x^2-v_z^2)\over 2 v_x^2 v_z^2} \sqrt{ \bar P_1(\n,\n) \; \bar
P_2(\n,\n) }.\quad
\end{eqnarray}
If we now calculate the Finsler metric $[g_{3\pm}(\n)]_{ij}$ we
shall rapidly encounter technical difficulties due to the
discriminant term. To make this a little clearer, let us define
\begin{equation}
[\bar P_3(\n)]_{ij} = {\partial^2 \sqrt{ \vbox to 10.5pt{\null} \bar
P_1(\n,\n) \; \bar P_2(\n,\n) } \over \partial n^i \; \partial n^j},
\end{equation}
since then
\begin{equation}
[g_3(\n)]_{ij}  = [\bar q_2(\n)]_{ij} \mp   {(v_x^2-v_z^2)\over 2
v_x^2 v_z^2}  [\bar P_3(\n)]_{ij}.
\end{equation}
Temporarily suppressing the argument $\n$, we have
\begin{equation}
[\bar P_3]_{ij} =  {1\over2} \partial_i \left[ \; \partial_j  \bar
P_1 \;\sqrt{ \bar P_2\over  \bar P_1} +  \partial_j  \bar P_2
\;\sqrt{ \bar P_1\over  \bar P_2} \; \right].
\end{equation}
A brief computation now yields the rather formidable result
\begin{eqnarray}
[\bar P_3]_{ij} &=& {1\over2}  \left[ \; \partial_i \partial_j \bar
P_1 \; \sqrt{\bar P_2\over \bar P_1} +  \partial_i \partial_j \bar
P_2\; \sqrt{\bar P_1\over \bar P_2} \; \right] \nonumber
\\
&& + {1\over4} \left[  {\partial_i \bar P_1 \; \partial_j \bar P_2 +
\partial_i \bar P_2 \; \partial_j \bar P_1 \over \sqrt{\bar P_1 \bar
P_2}} -
\partial_i \bar P_1 \; \partial_j \bar P_1 \; {\bar P_2^{1/2}\over \bar  P_1^{3/2}}
-
\partial_i \bar P_2\;  \partial_j \bar P_2 \; {\bar P_1^{1/2}\over \bar P_2^{3/2}}
\right]\nonumber
\\
&=& {1\over2\sqrt{\bar P_1\bar P_2}}  \left[ \; \partial_i
\partial_j \bar P_1 \; \bar P_2 +  \partial_i \partial_j \bar P_2\;
\bar P_1 \; \right] \nonumber
\\
&& + {1\over4\sqrt{\bar P_1\bar P_2}} \left[  \partial_i \bar P_1 \;
\partial_j \bar P_2 + \partial_i \bar P_2 \; \partial_j \bar P_1 -
\partial_i \bar P_1 \; \partial_j \bar P_1 \; {\bar P_2\over \bar P_1}
-
\partial_i \bar P_2\;  \partial_j \bar P_2 \; {\bar P_1\over \bar P_2}
\right].~~~~\qquad
\end{eqnarray}
From this expression it is clear that along either optical axis, (as
long as the optical axes are distinct, which is automatic in any
bi-axial situation), \emph{some} of the components of $[\bar
P_3]_{ij}$, and therefore \emph{some} of the components of the
Finsler metric $[g_{3\pm}]_{ij} = [\bar q_2]_{ij} \pm
\hbox{(constant)} \times [\bar P_3]_{ij}$, will be \emph{infinite}.

To see this in an invariant way, let $\u$ and $\w$ be two 3-vectors
and consider
\begin{equation}
[\bar P_3(\n)](\u,\w)  =  [\bar P_3]_{ij} \; u^i\; w^j.
\end{equation}
After a brief computation:
\begin{eqnarray}
[\bar P_3(\n)](\u,\w) &=& {1\over2\sqrt{\bar P_1(\n,\n)\,\bar
P_2(\n,\n)}} \left[ \; \bar P_1(\u,\w) \; \bar P_2(\n,\n) +  \bar
P_2(\u,\w)\; \bar P_1(\n,\n) \; \right] \nonumber
\\
&&+ {1\over4\sqrt{\bar P_1(\n,\n)\,\bar P_2(\n,\n)}} \Bigg[  \bar
P_1(\n,\u) \; \bar P_2(\n,\w) + \bar P_2(\n,\u) \; \bar P_1(\n,\w)
\nonumber
\\
&& - \bar P_1(\n,\u) \; \bar P_1(\n,\w)  \; {\bar P_2(\n,\n)\over
\bar P_1(\n,\n)} -
 \bar P_2(\n,\u)\;  \bar P_2(\n,\w)  \; {\bar P_1(\n,\n)\over \bar P_2(\n,\n)}
\Bigg].~~~~~\qquad
\end{eqnarray}

\noindent This quantity will tend to infinity as $\n$ tends to
either optical axis provided:
\begin{itemize}

\item The optical axes are distinct.

(If the optical axes are coincident then $\bar P_1=\bar P_2$ and so
$\bar P_3$ degenerates to
\begin{equation}
[\bar P_3(\n)](\u,\w) \to \bar P_1(\u,\w) = \bar P_2(\u,\w).
\end{equation}
One recovers the [\emph{for our purposes} physically uninteresting] result
for a uni-axial crystal.)

\item One is not considering the special case $\u=\w=\n$.

(In this particular special case $\bar P_3$ degenerates to
\begin{equation}
[\bar P_3(\n)](\n,\n) \to \sqrt{ \bar P_1(\n,\n)\; \bar P_2(\n,\n)},
\end{equation}
which is well-behaved on either optical axis.)
\end{itemize}

\noindent In summary:
\begin{itemize}
\item
The spatial Finsler 3-metric is $[g_{3\pm}]_{ij}$ \emph{generically
ill-behaved on either optical axis}.
\item
This feature will also afflict the \emph{spacetime} pseudo-Finsler
4-metric $[g_{4\pm}]_{ab}$ defined by suitable derivatives of the
Finsler 4-norm $F_{4\pm}$.
\item
This particular feature is annoying, but seems only to be a
technical problem to do with the specifics of crystal optics, it
does not seem to us to be a critical obstruction the developing a
space-time version of Finsler geometry. Ultimately it arises from
the fact that the ``null conoid'' is given by a quartic; this leads
to two topological cones that for topological reasons always
intersect, this intersection defining the optical axes.
\item It is the technical problems associated with the (3+1) \emph{spacetime} Finsler metric, to be discussed below, which much more deeply concern us.
\end{itemize}

\subsubsection{Technical problems with the co-Finsler 3-metric}
%----------------------------------------------------------------------------------------------------------

The phase discriminant
\begin{equation}
D = q_2(\k,\k)^2 - q_0(\k,\k)\;(\k \cdot \k),
\end{equation}
arising from the Fresnel equation (and considerations of the phase velocity) exhibits features similar to those arising for the ray discriminant. There are three cases:

{\bf Isotropic:} If the crystal is isotropic, then $D=0$. (This again is \emph{for our purposes} physically uninteresting.)

{\bf Uni-axial:} If the crystal  is uni-axial, then $D$ is a perfect square
\begin{equation}
D = \left\{ {(v_o^2-v_e^2)(k_x^2+k_y^2)\over 2}   \right\}^2,
\end{equation}
and so the co-Finsler structures $G_{3\pm}$ are both Riemannian:
\begin{equation}
G_{3+}(\k) = {v_o^2 \; \k\cdot \n};
\end{equation}
\begin{equation}
G_{3_-}(\k) = { v_e^2\,(k_x^2+k_y^2)} + {v_o^2\, k_z^2}.
\end{equation}
 (This situation again is \emph{for our purposes} physically uninteresting.)

{\bf Bi-axial:} Only in the bi-axial case are  the co-Finsler structures $G_{3\pm}$ ``truly'' Finslerian.
There are now two distinct (phase) optical axes (wave-normal optical axes) along which the discriminant is zero, these optical axes being given by
\begin{equation}
\hat\e_{1,2} =
\left( \pm \sqrt{v_x^2-v_y^2\over v_x^2-v_z^2}; \;\;0\;\; ; \sqrt {v_y^2-v_z^2\over v_x^2-v_z^2} \right),
\end{equation}
in terms of which the phase discriminant also factorizes
\begin{equation}
 D =  {(v_x^2-v_z^2)^2\over 4}
\left[ (\k\cdot\k) - \left(\k \cdot \e_1 \right)^2 \right]  \left[ (\k\cdot \k) - \left(\k \cdot \e_2 \right)^2 \right].
\end{equation}
The co-Finsler norm is then (now using projection operators $P_1$ and $P_2$ based on the phase optical axes $\e_{1,2}$)
\begin{eqnarray}
\left\{ G_{3\pm}(\k)\right\}^2 &=& q_2(\k,\k) \mp
{{(v_x^2-v_z^2)\over2}}
\sqrt{
P_1(\k,\k) \; P_2(\k,\k)
}.
\end{eqnarray}
The co-Finsler \emph{metric} is defined in the usual way
\begin{equation}
[h_{3\pm}]^{ij}(\k) = {1\over2} {\partial^2 [G_{3\pm}(\k)^2]\over\partial k^i \; \partial k^j}.
\end{equation}
This now has the interesting ``feature'' that some of its components are infinite when evaluated on the (phase) optical axes. That is: The co-Finsler 3-metric is $[h_{3\pm}]^{ij}$ generically ill-behaved on either optical axis. This feature will also afflict the pseudo-co-Finsler 4-metric $[h_{4\pm}]_{ab}$ defined by suitable derivatives of the Finsler 4-norm $G_{4\pm}$.

%----------------------------------------------------------------------------------------------------------
\subsubsection{Technical problems with the (3+1) spacetime
interpretation}
%----------------------------------------------------------------------------------------------------------

The (3+1)-dimensional spacetime objects that give the best way how to merge two (3+1) pseudo-Finsler norms and pseudo-co-Finsler norms into one geometry are the quantities
\begin{equation}
F_4(\d X)^4 = \left\{ F_{4+} (\d X) \; F_{4-} (\d X) \right\}^2;
\end{equation}
and
\begin{equation}
G_{4}(k)^4 =\left\{ G_{4+}(k) \; G_{4-}(k) \right\}^2.
\end{equation}
as defined in equation (\ref{E:4d-dc-F}) and (\ref{E:4d-dc-G}).
This is tantamount to taking
\begin{equation}
F_4(\d X) = \sqrt{  F_{4+} (\d X) \; F_{4-} (\d X) };
\end{equation}
and
\begin{equation}
G_{4}(k) =\sqrt{ G_{4+}(k) \; G_{4-}(k) }.
\end{equation}
Now $F_4(\d X)$ and $G_4(k)$ are by construction perfectly well behaved Finsler and co-Finsler \emph{norms}, with the correct homogeneity properties --- and with the nice and concise physical interpretation that the vanishing of $F_4(\d X)$ defines a double-sheeted
``signal cone'' that includes both polarizations, while the vanishing of $G_4(k)$ defines a double-sheeted ``dispersion relation''   (``mass shell'') that includes both polarizations. (Thus $F_4(\d X)$ and $G_4(k)$ successfully unify the ``on-shell'' behaviour of the signal cones in a Fresnel-like manner.)

While this is not directly a ``problem" as such, the norms $F_4(\d X)$ and $G_4(k)$ do have the interesting ``feature" that they pick up non-trivial complex phases: Since $F_{4\pm}(\d X)^2$ is always real,  (positive inside the propagation cone, negative outside), it follows that  $F_{4\pm}(\d X)$ is either pure real or pure imaginary. But then, thanks to the additional square root in defining  $F_4(\d X)$, one has:
\begin{itemize}
\item  $F_4(\d X)$  is pure real inside both propagation cones.
\item  $F_4(\d X)$  is proportional to $\sqrt{i} = {(1+i)\over\sqrt2}$
between the two propagation cones.
\item  $F_4(\d X)$  is pure imaginary outside both propagation cones.
\end{itemize}
Similar comments apply to the co-Finsler norm $G_4(k)$.

A considerably more problematic point is this: In the usual Euclidean signature situation the Finsler norm is taken to be smooth everywhere except for the zero vector --- this is usually phrased mathematically as ``smooth on the slit tangent bundle''.  What we see here is that in a Lorentzian-like signature situation the Finsler norm cannot be smooth as one crosses the propagation cones --- what was in Euclidean signature a feature that only arose at the zero vector of each tangent space has in Lorentzian-like signature situation grown to affect (and infect) all null vectors. The Finsler norm is here at best ``smooth on the tangent bundle excluding the null cones''. (In a mono-refringent case the squared norm, $F_4(\d X)^2$, is smooth across the propagation cones, but in the bi-refringent case one has to go to the fourth power of the norm, $F_4(\d X)^4$, to get a smooth function.)

{\bf A ``no go'' result:}
Unfortunately, when attempting to bootstrap these two reasonably well-behaved  \emph{norms} to Finsler and co-Finsler \emph{metrics} one encounters additional and more significant complications. We have already seen that there are problems with the spatial 3-metrics $[g_{3\pm}]_{ij}(\n)$ and $[h_{3\pm}]^{ij}(\k)$
on the optical axes, problems which are inherited by the single-polarization spacetime (3+1)-metrics $[g_{4\pm}]_{ab}(\n)$ and $[h_{4\pm}]^{ab}(\k)$, again on the optical axes.

But now the spacetime (3+1)-metrics
\begin{equation}
[g_4]_{ab}(n) =  {1\over2} {\partial^2 [F_4(n)^2]\over\partial n^ a \; \partial n^b},
\end{equation}
and
\begin{equation}
[h_4]^{ab}(k) = {1\over2} {\partial^2 [ G_4(k)^2]\over\partial k_a \; \partial k_b},
\end{equation}
both have (at least some) infinite components --- $[g_4]_{ab}(n)$ has infinities on the entire signal cone, and $[h_4]^{ab}(k)$ has infinities on the entire mass shell. Since the argument is essentially the same for both cases, let us perform a single calculation:
\begin{eqnarray}
g_{ab} &=& {1\over2} \; \partial_a \partial_b \sqrt{ [F_+^2 \; F_-^2 ]}
\\
&=& {1\over4} \partial_a \left[ \partial_b [F_+^2] \; {F_-\over F_+}
+  \partial_b [F_-^2]\;  {F_+\over F_-} \right],
\end{eqnarray}
so that
\begin{eqnarray}
g_{ab}
&=& {1\over4}  \left[ \partial_a \partial_b [F_+^2] \; {F_-\over F_+} +
 \partial_b \partial_b [F_-^2]  \; {F_+\over  F_-} \right]
\nonumber
\\
&&+ {1\over2} \bigg[
{\partial_a F_+ \partial_b F_- + \partial_a F_- \partial_b F_+ }
-
\partial_a F_+ \partial_b  F_+ {F_-\over F_+}
-
\partial_a F_- \partial_b F_- {F_+\over F_-}
\bigg].~~~
\end{eqnarray}
That is, tidying up:
\begin{eqnarray}
g_{ab} &=& {1\over2} \left[ (g_+)_{ab} {F_2\over F_1}
+  (g_-)_{ab} {F_1\over F_2} \right]
\nonumber\\
&&
+ {1\over2} \bigg[  {\partial_a F_+ \partial_b F_-+ \partial_a F_- \partial_b F_+}
-
\partial_a F_+ \partial_b F_+ {F_-\over F_+}
-
\partial_a F_- \partial_b F_- {F_+\over F_-}
\bigg].~~~
\end{eqnarray}
The problem is that this ``unified''  metric $g_{ab}(n)$ has singularities on both of the signal cones.
The (relatively) good news is that the quantity $g_{ab}(n) \, n^a n^b = F^2(n)$, and so on either propagation cone $F\to 0$, so $F(n)$ itself has a well defined limit.
But now let $n^a$ be the vector the Finsler metric depends on, and let  $w^a$ be some \emph{other} vector.
Then
\begin{eqnarray}
 g_{ab}(n)  \, n^a w^b &=&  {1\over2} n^a w^b \partial_a \partial_b [F^2]
 \\
 &=& {1\over2} w^b \partial_b [F^2]
 \\
 &=&  {1\over2} w^b \partial_b \sqrt{F_+ F_-}
\\
&=& {1\over4} w^b \left[ \partial_b [F_+^2] \; {F_-\over F_+} +
\partial_b [F_-^2] \; {F_+\over F_-} \right]
\\
&=& {1\over2} \bigg\{ (\,[g_+]_{ab} \, n^a w^b) \; {F_-\over F_+}
+
 (\,[g_-]_{ab} \, n^a w^b]) \; {F_+\over F_-} \bigg\}.
\end{eqnarray}
The problem now is this: $g_+$ and $g_-$ have been carefully constructed to be individually well defined and finite (except at worst on the optical axes). But now as we go to propagation cone ``$+$'' we have
\begin{equation}
g_{ab}(n) \, n^a w^b \to {1\over2}  (\,[g_+]_{ab} \; n^a w^b) \;  {F_-\over 0}  = \infty,
\end{equation}
and as we go to the other propagation cone ``$-$'' we have
\begin{equation}
g_{ab}(n) \, n^a w^b \to {1\over2}  (\,[g_-]_{ab} \; n^a w^b) \;  {F_+\over 0}  = \infty.
\end{equation}
So at least some components of this ``unified'' Finsler metric $g_{ab}(n)$ are unavoidably singular on the propagation cones. Related (singular) phenomena have previously been encountered in multi-component BECs, where multiple phonon mo\-des can interact to produce Finslerian propagation cones \cite{lnp}.

Things are just as bad if we pick $u$ and $w$ to be \emph{two} vectors distinct from ``the direction we are looking in'', $n$. In that situation
\begin{eqnarray}
g_{ab}(n) \, u^a w^b &=& {1\over2} \Bigg[
g_+(u,w) \,{F_-\over F_+} + g_-(u,w) \,{F_+\over F_-}
\nonumber
\\
&&
+ {g_+(u,n) \, g_-(w,n) + g_+(w,n) \, g_-(u,n)\over F_+ \, F_-}
\nonumber
\\
&&
- g_+(u,n) \, g_+(w,n) \, {F_-\over F_+^3}
- g_-(u,n) \, g_-(w,n) \, {F_+\over F_-^3}
\Bigg].
\end{eqnarray}
Again, despite the fact that both $g_+$ and $g_-$ have been very carefully set
up to be regular on the propagation cones (except for the known, isolated, and tractable
problems on the optical axes), the ``unified'' metric $g_{ab}(n)$ is
unavoidably singular there --- unless, that is, you \emph{only}
choose to look in the $nn$ direction.

If we give up the condition of mathematical simplicity in the construction of the ``unified'' space-time (co-)Finsler norms, then $F_4$ and $G_4$ do not have to factorize into a product of powers of the Finsler functions for individual polarizations. We could then look for more complicated ways of building ``unified'' Finsler and co-Finsler structures (constrained mainly by giving the correct light propagation cones in the birefringent crystal), and we might be able to find an appropriate pseudo-Finsler geometry (which might also fulfill some additional reasonable physical conditions). The necessary conditions for such a geometry are relatively easy to formulate (and this is done in the last section), but because of the lack of any intuitive interpretation, and the corresponding lack of direct physical motivation, the  physical meaning of such an approach is highly doubtful \cite{thessalonika}.

\section{General bi-metric situations}

In the previous part one could observe that the arguments (given the
way they were constructed) might apply for a large set of  different
situations than the bi-refringent crystal. In fact they generally hold within the class of bi-metric theories.

Take arbitrary bi-metric theory. Such theories contain two distinct pseudo-Riemannian
metrics $g^{\pm}_{ab}$, so we can define two distinct ``elementary''
pseudo-\-Rie\-mann\-ian norms
\begin{equation}
F_{\pm}(x,v) = \sqrt{g^{\pm}_{ab}(x)v^{a}v^{b}}.
\end{equation}
Suppose one now wants a combined Finsler norm that simultaneously
encodes both signal cones --- then the natural thing is always to
take

\begin{equation}
F(x,v)=\sqrt{F_{+}(x,v)F_{-}(x,v)};~~~~~~~~~~g_{abcd}=g^{+}_{(ab}g^{-}_{cd)}.
\end{equation}

This construction for $F(x,v)$ is automatically  1-homogeneous in $v$. Generally the vanishing of $F(x,v)$ correctly
encodes the two signal cones. So this definition of $F(x,v)$ provides a perfectly good Finsler norm.

However only from the fact that the individual $F_{\pm}$ are not positive definite automatically follows that:

\begin{itemize}

\item
 $F(x,v)$ is proportional to $\sqrt{i} = \frac{1+i}{\sqrt{2}}$ between the two propagation cones, hence picks up a non-trivial phase.

\item
 As a result this Finsler norm is at best ``smooth on the tangent bundle excluding the null cones''.
\end{itemize}
For the ``unified'' metric, from the fundamental definition we see
\begin{eqnarray}
g_{ab}(x,v)=~~~~~~~~~~~~~~~~~~~~~~~~~~~~~~~~~~~~~~~~~~~~~~~~~~~~~~~~~~~~~~~~~~~~~~~~~~~~~~~~~~~~~~~~~~~~~~~~~~~~~~~\nonumber\\
\frac{1}{2}\left[g^{+}_{ab}\frac{F_{-}}{F_{+}}+g^{-}_{ab}\frac{F_{+}}{F_{-}}\right]+\frac{1}{2}\left[\partial_{a}F_{+}\partial_{b}F_{+}-\partial_{a}F_{+}\partial_{b}F_{-}\frac{F_{-}}{F_{+}}-\partial_{a}F_{-}\partial_{b}F_{-}\frac{F_{+}}{F_{-}}\right].
\end{eqnarray}
This ``unified'' and ``natural'' Finsler metric $g_{ab}(x,v)$ has
necessarily singularities on both of the signal cones.

 The quantity $g_{ab}(x,v)v^{a}v^{b} = F^{2}(x,v)$ is still being well defined, but if $v^{a}$ is the vector the Finsler metric depends on, and $w^{a}$ some other (arbitrary) vector, then
\begin{eqnarray}
g_{ab}(x,v)v^{a}w^{b}~~~~~~~~~~~~~~~~~~~~~~~~~~~~~~~~~~~~~~~~~~~~~~~~~~~~~~~~~~~~~~~~~~~~~~~~~~~~~~~\nonumber\\
=\frac{1}{2}v^{a}w^{b}\partial_{a}\partial_{b}\sqrt{F^{2}_{+}F^{2}_{-}}=\frac{1}{2}w^{b}\partial_{b}\sqrt{F_{+}F_{-}}~~~~~~~~~~~~~~~~~~~~~~~~~\nonumber\\
=\frac{1}{4}w^{b}\left[\partial_{b}[F^{2}_{+}]\frac{F_{-}}{F_{+}}+\partial_{b}[F^{2}_{-}]\frac{F_{+}}{F_{-}}\right]~~~~~~~~~~~~~~~~~~~~~~\nonumber\\
=\frac{1}{2}\left[(g^{+}_{ab}v^{a}w^{b})\frac{F_{-}}{F_{+}}+(g^{-}_{ab}v^{a}w^{b})\frac{F_{+}}{F_{-}}\right].
\end{eqnarray}
As we go to propagation cone ``+'', necessarily
\begin{equation}
g_{ab}(x,v)v^{a}w^{b}\to\frac{1}{2}(g^{+}_{ab}v^{a}w^{b})\frac{F_{-}}{0}=\infty,
\end{equation}
and as we go to the other propagation cone ``-'', necessarily
\begin{equation}
g_{ab}(x,v)v^{a}w^{b}\to\frac{1}{2}(g^{-}_{ab}v^{a}w^{b})\frac{F_{+}}{0}=\infty.
\end{equation}
So at least some components of this ``unified'' Finsler metric
$g_{ab}(x,v)$ are (completely generally) unavoidably singular on the
propagation cones. And again, things are just as bad if we pick $u$
and $w$ to be two vectors distinct from $v$.

In other words we see that all the problems which arose in the
particular case of bi-refringent bi-axial crystal are in fact also features of
arbitrary bi-metric situation.

\section{General construction}

\subsection{General constraints}
%-----------------------------------
In the previous sections we observed that an ``intuitive'' way of encoding bi-metricity into pseudo-Finsler geometry leads to
significant problems which seem inevitable and unavoidable. The other problematic
part which was omitted is whether such intuitive construction leads
to appropriate pseudo-Riemann\-ian limit for, (in some frame), slowly moving objects.
Here we discuss all the general constrains on a pseudo-Finsler
geometry recovering arbitrary bi-refrin\-gence.

What are the specific physical constraints given by bi-refringent
theories on our geometry? We consider it meaningful to impose the
following constraints:
\begin{itemize}
\item[a)] Locally there must exist a coordinate frame in which holds the following:
At any arbitrary point from the domain of these coordinates take
 within this frame the purely
time-oriented, $(v,\mathbf{0})$ vector. Then on some neighborhood of
this vector the pseudo-Finsler norm approaches the
pseudo-Riemannian norm. Also each
``constant proper time'' hypersurface, specified by $F^{2}(v)=-\tau_0^{2}$,~ is connected, and contains vectors $(v,\mathbf{0})$.
\item[b)] The pseudo-norm must break Lorentz invariance as encoded in the ``Fresnel'' equation
 giving the bi-refringence.

 \end{itemize}
Let us discuss these conditions a little bit: The first is just
requirement that we want to recover, in some frame, low-energy
physics as we know it (in a pseudo-Riemann form). It also means that if we have a massive particle, it should not
be able to have two distinct sets of four-velocities between which
it is not able to undergo a smooth transition. The last constraint it gives is that
any massive particle should be able to move arbitrarily slowly with
respect to the used ``preferred'' frame. The
second condition is trivial: it is the specific contribution of the
given bi-refringent model (unlike the first condition,
which can be considered as generic).

\subsection{How one proceeds for general bi-metric situations}

\subsubsection{Co-metric structure construction}
%-----------------------------------
Take the general bi-metric situation and define
\begin{equation}
G(x,k)^{4}=G_{+}(x,k)^{2} G_{-}(x,k)^{2}
\end{equation}
where
\begin{equation}
G_{\pm}(x,k)^{2}=g_{\pm}^{ab}(x)k_{a}k_{b}.
\end{equation}
This suggests a generic candidate for the pseudo-co-Finsler norm,
which can be written as $\tilde G(x,k)^{2}\propto G(x,k)^{4}$ or more
precisely $\tilde G(x,k)^{2}=M_{1}(x,k)\,G(x,k)^{4}$, where $M_{1}(x,k)$ is
an otherwise arbitrary $\mathbb{R}$-valued function fulfilling following constraints:
\begin{itemize}
\item[a)]
$\displaystyle \frac{\partial[M_{1}(x,k)G(x,k)^{4}]}{\partial k^{i}}$ is
a bijection from $V^*\to V$;
\item[b)]
$M_{1}(x,k)$ is in $k$ a homogeneous map of degree $-2$;
\item[c)]
 $M_{1}(x,k)$ is smooth on $V^*/\{0\}$;
\item[d)]
 $M_{1}(x,k)$ is inside both signal cones negative, outside both positive, between them nonzero;
\item[e)]
$M_{1}(x,k)\,G(x,k)^{4}$ must approximate the pseudo-Riemann norm
(squared) locally in some coordinates for $k$ close to
$(v,\mathbf{0})$ and such vectors must lie on some hypersurface
(mass-shell) given by: $M_{1}(x,k)\, G(x,k)^{4}=-m^{2}$, which must be
always connected.
\end{itemize}
These conditions follow trivially from the conditions we imposed on any physically meaningful Finslerian geometry.

\subsubsection{Metric structure construction}
%-----------------------------------
Take the given bi-metric situation and define
\begin{equation}
F(x,v)^{4}=F_{+}(x,v)^{2}F_{-}(x,v)^{2},
\end{equation}
with
\begin{equation}
F_{\pm}(x,v)^{2}=g_{ab}^{\pm}(x)v^{a}v^{b}.
\end{equation}
The generic Finsler pseudo-norm is in this case expressed exactly in
the same way as was the  co-Finsler pseudo-norm in the previous
case, by the function $\tilde F(x,v)^{2}=M_{2}(x,v)\; F(x,v)^{4}$, where
$M_{2}(x,v)$ is again an arbitrary function fulfilling exactly the
same conditions as the $M_{1}(x,k)$ function, we just have to exchange
$V^*$ for $V$.

\subsubsection{Interconnecting the metric structure with the co-metric structure}
Now the last step in putting constraints on geometric construction
is to relate the Finsler and co-Finsler structures. The items which
are as yet undetermined  are the functions $M_{1}(x,v)$, and
$M_{2}(x,k)$. So they have to fulfill the last condition, which is the
condition of forming the full united geometry. This condition is
given by the Legendre transform relation:
\begin{equation}
\left[\tilde G\left(x,\frac{\partial [\tilde F(x,v)^2]}{\partial
v^{i}}\right)\right]^2=\tilde F(x,v)^{2}.
\end{equation}
In our language, the function $M_{2}(x,v)$ must be connected to
$M_{1}(x,k)$ by:
\begin{eqnarray}
M_{1}\left(x,\frac{\partial [M_{2}(x,v)F(x,v)^{4}]}{\partial v^{i}}\right)
\; G\left(x,\frac{\partial [M_{2}(x,v)F(x,v)^{4}]}{\partial
v^{i}}\right)^{4}~~~~~~~~~~~~~~~~~~~~\nonumber\\
=M_{2}(x,v) \; F(x,v)^{4}. ~~~\label{E:compatibility}
\end{eqnarray}
Thus we see that if we have somehow  found an  appropriate
$M_{2}(x,v)$, then $M_{1}(x,k)$ will be uniquely determined\footnote{It holds both ways, also $M_{2}(x,k)$ is uniquely determined by $M_{1}(x,k)$.}.

To find solutions of such an abstract (although precisely
formulated) mathematical exercise is not an easy task. At this early
stage it is quite doubtful whether general solutions exist.

\section{Conclusions}
In conclusion: Our ``no-go'' theorem suggests that while
pseudo-Finsler spacetimes are certainly useful constructs, in
bi-refringent situations it does not appear possible to naturally and
intuitively construct a  ``unified'' pseudo-Finsler spacetime such
that the pseudo-Finsler metric is null on both ``signal cones'', but
has no other zeros or singularities --- it seems physically more
appropriate to think of physics as taking place in a single
topological manifold that carries two distinct pseudo-Finsler metrics,
one for each polarization mode.

This means that in the case of high-energy Lorentz violations it is
highly doubtful whether one has to follow the idea to try to
formulate the theory geometrically. There is no analogy between this
case and the very successful Minkowski geometric formulation of
special theory relativity. This leads to a considerable concern about how
the equivalence principle might be sustainable if we admit the high-energy Lorentz violations.

Physically this might suggest that high energy Lorentz violations -- if they actually occur in nature -- should be ``universal''. That is, all particles should see the \emph{same} Lorentz violation. In this case one could at least have a single ``signal cone'', avoid all the problems that were demonstrated in this chapter and so have a reasonable chance of satisfying the Einstein equivalence principle.

~

\chapter[Analytic results for highly damped QNMs]{Analytic results for highly damped quasi-normal modes}

\section{Introduction}
%---------------------------------------------------------------
\paragraph{Brief introduction into the topic}
Black hole quasi-normal modes are physically intuitive gravitational
perturbations of various types of black hole spacetimes. Take black hole spacetimes with spherical symmetry, where the metric is of the form:
\begin{equation}
g_{\mu\nu}=-f(r)dt^{2}+f(r)^{-1}dr^{2}+r^{2}d\Omega^{2}~.
\end{equation}
One
decomposes the general perturbation into a tensorial generalization
of spherical harmonics (for a detailed introduction into the topic see \cite{Nollert}):
\begin{equation}
h_{\mu\nu}=\sum_{l=0}^{\infty}\sum_{m=-l}^{l}\sum_{n=1}^{10}\Psi^{n}_{lm}(t,r)\left\{\left(Y^{n}_{lm}\right)_{\mu\nu}(\theta,\phi)\right\}~.
\end{equation}
They split into scalar,
vector and tensor perturbations. Moreover we can split them with
respect to parity into two sets: \emph{axial} and \emph{polar}
perturbations. Both follow the one-dimensional
equation:\footnote{This equation reminds us of the Klein-Gordon equation with a potential, describing a relativistic scalar particle scattering.}
\begin{equation}
\frac{\partial^{2}\Psi_{lm}(t,x)}{\partial t^{2}}-\frac{\partial^{2}\Psi_{lm}(t,x)}{\partial x^{2}}+V(x)\Psi_{lm}(t,x)=0.\label{RWZT}
\end{equation}
Here we transformed the coordinate $r$ to the tortoise coordinate ~$x\equiv\int\frac{dr}{f(r)}$.~ The $V(x)$ function is called the Regge-Wheeler potential (axial perturbations)
or Zerilli potential (polar perturbations). Both Regge-Wheeler and
Zerilli potentials (in general different) depend on the background
geometry, the spin of the perturbation, and the wave-mode number of the
perturbation.

If we are interested only in perturbations with the ``harmonic''
time dependence  ~~$\Psi_{lm}(t,x)=e^{i\omega t}\psi_{lm}(x)$,~ we obtain:

\begin{equation}
\frac{\partial^{2}\psi_{lm}(x)}{\partial x^{2}}-(V(x)-\omega^{2})\psi_{lm}(x)=0.\label{RWZ}
\end{equation}

Despite the fact that this equation \emph{formally} reminds
us of the Schr\"odinger equation for $L^{2}$ Hilbert space self-adjoint
operator eigenfunctions, it should be kept in mind that here the
whole situation is very different. As a result of this fact the
$\omega^{2}$ ``eigenvalues'' are in general non-real numbers.

The quasi-normal modes (QNMs) are solutions of \eqref{RWZ} with the
boundary conditions giving purely outgoing radiation:
\begin{equation}
\psi_{lm}(x)\to C_{\pm}e^{\pm i\omega x}~~~~~~~~x\to\mp\infty.
\end{equation}
These particular perturbations are of a general interest because:

\begin{itemize}
\item
 The boundary conditions are the physically \emph{intuitive} ones.
\item
 It can be proven \cite{Nollert} that after some time scale these perturbations become dominant within arbitrary black hole perturbation.
\item
 In the field of quantum gravity, there exists Hod's conjecture \cite{Hod}, and more recently Maggiore's \cite{Maggiore} conjecture, concerning the connection between the highly damped QNMs and the black hole area spectrum.
\item
The QNMs related to asymptotically anti-de Sitter black holes are interesting for people working with AdS/CFT correspondence.
\end{itemize}

The second point suggests QNMs describe the characteristic ``sound''
of black holes. They are in principle observable in black hole
oscillations and ring-down phenomena. As we see the modes are
characterized by the QNM frequencies $\omega$ (QNFs - quasinormal frequencies). It can be generally
proven that for the QNM frequencies we have:

\begin{itemize}
\item
the quasinormal frequencies $\omega$ form an infinite but countable set,
\item
the $\omega$-s are complex with $Im(\omega)>0$, and hence describe stable
perturbations,
\item
they have real parts symmetrically spaced with respect to the
imaginary axis.
\end{itemize}

The literature lists many techniques that have been used to calculate the QNM frequencies. The basic ones are:

\begin{itemize}
\item
WKB inspired approximations~\cite{Ferrari, Guinn2, Iyer, Konoplya, Guinn1};
\item
phase-amplitude methods~\cite{Andersson1, Andersson2};
\item
continued fraction approximations ~\cite{Leaver1, Leaver2, Leaver3};
\item
monodromy techniques~\cite{Shanka1, Shanka2, Motl1, Motl2};
\item
Born approximations~\cite{Choudhury, Medved1, Medved2, Padmanabhan};
\item
approximation by analytically solvable potentials.
\end{itemize}
Some of these techniques were used only to calculate the fundamental (least
damp\-ed) QNM frequencies (like the approximation by the analytically
solvable P\"oschl-Teller potential \cite{Ferrari}), other techniques were
used to estimate the asymptotic behavior of highly damped
frequencies as well.

\paragraph{The basic focus and results of our work}
The main focus of this work is the analytic results for
\emph{highly} damped QNM frequencies. These can be obtained by two
different methods: approximation by analytically solvable
potentials, and monodromy techniques. Unfortunately the first
method was previously used only to estimate the fundamental frequencies, and one
of the main contributions of this work is using this method to
also explore the highly damped QNMs. Unlike the method of
approximation by analytically solvable potentials, the monodromy
technique was used many times to understand the asymptotic QNM
behavior. There are striking similarities between our results
obtained by the analytically solvable potentials and the known
monodromy results. These allow us to analyze the monodromy results
in a new way bringing much deeper understanding about the behaviour
of the asymptotic QNM frequencies encoded in the monodromy formulae.

A problem of special interest is the following: for the highly damped frequencies the asymptotic behavior
\begin{equation}\label{gap}
\omega_{n}=\omega_{0}+in\cdot\mathrm{gap}+O(n^{-m}),~~~~~~~~~~~~m>0,
\end{equation}
(where ``gap'' denotes some real constant), was often observed \cite{Motl1, Motl2}. On the other hand
situations have been observed where this behavior seems to fail. The
analysis of the equations derived by both the analytically solvable
potentials and monodromy techniques gives us an indication to when
such behavior is to be expected. It also tells us how the gap
spacing is given.
%-----------------------------------------------------------------------------------------------------------------------------------------

\paragraph{The structure of this chapter}
In the first part of this chapter we analyze  the highly damped QNMs
by using the idea of approximation by analytically solvable potentials. We
verify our method by applying it to the Schwarzschild black hole, where
the results are known and widely accepted. After that we explore the
much less known Schwarzschild-de Sitter (S-dS) case. As a result of our
approach we will prove interesting theorems about the highly damped
QNM behavior for the S-dS black hole. We will also discuss our results in the context of black hole thermodynamics. In the second part we explore the
complementary set of analytic results obtained by monodromy
calculations (related to many different types of black holes) and
prove that all of the results follow the patterns discovered by our
approximation. This means we are able to generalize our theorems for
almost every analytic result presently known. This suggests the
behavior discovered is very generic also between different black
hole spacetimes.

\section{Approximation by analytically solvable potentials}

\subsection{Introduction}

As previously mentioned, one of the ways to derive analytic approximate expressions for the highly damped QNMs is to approximate the real Regge-Whee\-ler\-/Zerilli potential by analytically solvable potentials. In the past this method has been used to give a formula for the fundamental QNM frequencies. Ferrari and Mashoon \cite{Ferrari} used the P\"oschl-Teller (Eckart) potential
\begin{equation}
V(x)=V_{0}/\cosh^{2}(\alpha x)
\end{equation}
 to approximate the Regge-Wheeler potential at the peak (by fitting the $V_{0}$ and $\alpha$ parameters by the peak height and peak curvature). The QNM frequencies of the P\"oschl-Teller potential are given by the formula

\begin{equation}
\omega_{n}=\pm\sqrt{V_{0}-\frac{\alpha^{2}}{4}}-i\alpha\left(n+\frac{1}{2}\right).
\end{equation}

The widely accepted result (see for example \cite{Motl2}) for highly damped QNM frequencies of Schwarzschild black hole is
\begin{equation}
\omega_{n}=\pm\frac{\ln{3}}{2\pi}\kappa+i\kappa\left(n+\frac{1}{2}\right)+O\left(n^{-1/2}\right),\label{HighSchw}
\end{equation}
where $\kappa$ is the surface gravity at the black hole horizon. The interesting observation is the following:

\begin{itemize}
\item
the asymptotic formula (\ref{HighSchw}) can be obtained by the appropriate fitting of P\"oschl-Teller potential, by setting $\alpha=\kappa$,
\item
such P\"oschl-Teller potential qualitatively recovers the behavior of one tail\footnote{By the ``tail of the potential'' we mean the potential in one of the asymptotic regions, given as ~$|x|>>0$.~ By the ``peak of the potential'' we mean the potential in the region near its global maximum, typically near ~$x=0$.} of the Regge-Wheeler potential.
\end{itemize}
This fact seems not to be a coincidence as one can also turn the logic around. The wavepacket formed of highly damped QNMs close to the peak will quickly spread out from the peak region to the region of the tails of the potential. That suggests the tails are the most important factor in determining the highly damp\-ed QNMs. There is also a general observation that the wavelength (given by the $Re(\omega)$) is higher for highly damped modes, so this supports the view that these modes must be more sensitive to the asymptotic behavior of the potential. So the general expectation is that \emph{a good approximation to the tails of the Regge-Wheeler/Zerilli potentials should give a good qualitative estimate for the behavior of the highly damped quasinormal modes}.

To give this statement an exact meaning it is appropriate to exactly define what we mean by a good \emph{qualitative} vs. \emph{quantitative} asymptotic estimate. We say that sequence of QNFs $\omega_{1n}$ \emph{quantitatively} matches with the sequence $\omega_{2n}$ (for the asymptotic QNFs), if

\begin{equation}
\lim_{n\to\infty} (\omega_{1n}-\omega_{2n})=0.
\end{equation}
We say that sequence $\omega_{1n}$ matches \emph{qualitatively} with the sequence $\omega_{2n}$ if

\begin{equation}
\lim_{n\to\infty} \frac{|\omega_{1n}-\omega_{2n}|}{|\omega_{1n}|}=\lim_{n\to\infty} \frac{|\omega_{1n}-\omega_{2n}|}{|\omega_{2n}|}=0.
\end{equation}
So in the first case the ``error'' goes to zero, in the second case only the ``relative error'' (error compared to the result) goes to zero. One can observe that if $\omega_{1n}$ and $\omega_{2n}$ are both equispaced, then if they match \emph{quantitatively} they must be equal. On the other hand for \emph{qualitative} matching only the same gap spacing is required (they can have different $\omega_{0}$ modes). The second statement can be proven immediately from the definition
\begin{equation}\label{gap(.)}
\lim_{n\to\infty} \frac{|\omega_{1n}-\omega_{2n}|}{|\omega_{1,2n}|}=\frac{|\mathrm{gap}(\omega_{1n})-\mathrm{gap}(\omega_{2n})|}{|\mathrm{gap}(\omega_{1,2n})|},
\end{equation}
where $\mathrm{gap}(\omega_{1n})$, $\mathrm{gap}(\omega_{2n})$ denote the ``gap'' constants from \eqref{gap} related to sequences $\omega_{1n}$, ~$\omega_{2n}$.
~Then the condition of \emph{qualitative} matching implies
\begin{equation}
\mathrm{gap}(\omega_{1n})-\mathrm{gap}(\omega_{2n})=0,
\end{equation}
and nothing else.

 All this means the question whether the tails of the real potentials can be approximated by piecewise smooth analytically tractable potentials is of high interest. As we will see this is the case of 1 and 2 horizon situations.

\subsubsection{One horizon situations}
This is the situation in the Schwarzschild geometry. Since it is asymptotically flat, only one side of the potential is approaching the (black hole) horizon.

For the specific case of a Schwarzschild black hole the tortoise coordinate is given by
\begin{equation}
{dr\over dx} = 1-{2m\over r}; \qquad \qquad x(r) = r + 2m\ln\left[{r-2m\over2m}\right];
\end{equation}
and the Regge--Wheeler potential is
\begin{equation}
V(x(r)) = \left(1-{2m\over r}\right) \left[ {\ell(\ell+1)\over r^2} + {2m(1-s^2)\over r^3} \right].
\end{equation}
Here $s$ is the spin of the particle and $\ell$ is the angular momentum of the specific wave mode under consideration, with $\ell\geq s$.
As $x\to -\infty$ we have $r\to 2m$ and
\begin{equation}
V(x)\to  \exp\left( {x-2m\over 2m} \right) \; {\ell(\ell+1) +(1-s^2) \over (2m)^2}  = V_0 \;  \exp( 2 \kappa x ),
\end{equation}
where $\kappa$ is the black hole surface gravity. The Zerilli potential is more complicated, but leads to the same asymptotic behavior. This specific behaviour in terms of the surface gravity generalizes beyond the Schwarzschild black hole and for an arbitrary black hole in an asymptotically flat spacetime one has
\begin{equation}
V(x) \to \left\{ \begin{array}{ll}
\vphantom{\Big|}
 V_{0-} \; \exp( - 2\kappa  |x|  ),   & x \to - \infty;
 \\
 \vphantom{\Big|}
  V_{0+}  \; (2m)^2/ x^2,  & x \to + \infty.
\end{array} \right.
\end{equation}
For the highly damped modes this behavior suggests we shall fit the real potentials by a proper combination of P\"oschl-Teller and $1/x^{2}$ (inverted harmonic oscilator) potential.

\subsubsection{Two horizon situations}

If one turns to asymptotically de Sitter black holes (or more generally any two horizon system) the situation is different, but following the same patterns --- the Regge--Wheeler potential and the Zerilli potential have the asymptotic behaviour
\begin{equation}
V(x) \to \left\{ \begin{array}{ll}
\vphantom{\Big|}
 V_{0-} \; \exp( - 2\kappa_-  |x|  ),   &\qquad  x \to - \infty;
 \\
 \vphantom{\Big|}
  V_{0+}  \;  \exp( - 2\kappa_+  |x|  ),  &\qquad  x \to + \infty;
\end{array} \right.
\end{equation}
where the two surface gravities are now (in general) distinct. It is this situation that we will model using a \emph{piecewise}  Eckart potential (P\"oschl--Teller potential).  A considerable amount of analytic information can be extracted from this model, information which in the concluding discussion we shall attempt to relate back to ``realistic'' black hole physics.

\subsubsection{3 (and more) horizon situations}

There is no simple or practicable way of dealing with three-horizon situations using semi-analytic techniques.

%-----------------------------------------------------------------------------------------------------------------------------------------
\subsection{Schwarzschild black hole}

\subsubsection{Potential}

We work with the Regge-Wheeler/Zerilli equation

\begin{equation}
-\psi''(x)+(V(x)-\omega^{2})\psi(x)=0.
\end{equation}
The potential recovering the tails behavior of the given Regge-Wheeler/Zerilli potential is

\begin{equation}
V(x) = \left\{  \begin{array}{lcl}
{V_{0-} \; \sech^2(x/b_-)}  & \hbox{ for }  & x < 0;\\
\\
V_{0+}/(x+a)^{2}  & \hbox{ for }  & x > 0. \\
\end{array}\right.
\end{equation}
Here $b_{-}$ is defined by ~$b_{-}\equiv\kappa^{-1}$.~
The potential can be discontinuous at the origin and the infinite peak can be placed anywhere in the negative part of the real line, so $a$ can be taken to be an arbitrary \emph{positive} real number.

\subsubsection{Wavefunction}

We are interested in the solutions of the Regge-Wheeler/Zerilli equation (\ref{RWZ}) with the boundary constraints

\begin{equation}
\psi_{+}(x\to +\infty)\to C_{+}e^{-i\omega x};~~~~\psi_{-}(x\to -\infty)\to C_{-}e^{+i\omega x},
\end{equation}
defining the quasi-normal modes. The solutions are (for $V_{0+}\geq-1/4$)

\begin{equation}
\psi_{+}(x)=C_{+}\sqrt{\frac{x+a}{\omega}}\left[J_{\alpha_{+}}\left(\omega(x+a)\right)-e^{-i\alpha_+\pi}J_{-\alpha_{+}}\left(\omega(x+a)\right)\right],
\end{equation}
where $J_{\alpha_{+}}(x),~J_{-\alpha_{+}}(x)$ are Bessel functions with

\begin{equation}
\alpha_{+}\equiv\sqrt{1+4V_{0+}}
\end{equation}
and

\begin{equation}
\psi_{-} = C_{-}e^{i\omega x} \; {} _2F_1\left({1\over2}+\alpha_{-},{1\over2}-\alpha_{-},1+ib_- \omega, {1\over 1+ e^{- 2x/b_-}} \right),
\end{equation}
where ${}_2F_1(...)$ is the hypergeometric function with

\begin{equation}
\alpha_- \equiv \left\{  \begin{array}{lcl}
\sqrt{{1\over4} - V_0 b^2_- } & \hbox{ for }  & V_0 b^2_- < 1/4;\\
\\
i \sqrt{V_0 b^2_- - {1\over4}  } & \hbox{ for }  & V_0 b^2_- > 1/4.\\
\end{array}\right.
\end{equation}

\subsubsection{Junction condition}

We need to match the functions $\psi_{\pm}(x)$ and their derivatives at the origin. We can combine the two matching conditions in such a way, that one of them will give the equation for quasinormal modes, while the other will only relate the normalization constants $C_{\pm}$, hence will not be of pressing interest. The first equation we obtain by equating the $\psi'_{\pm}(0)/\psi_{\pm}(0)$ ratios. (In special situations where $\psi_\pm(0)$ might accidentally equal zero one might need to perform a special case analysis. The generic situation is $\psi_\pm(0)\neq 0$, and will prove sufficient for almost everything we need to calculate.)

For the $\psi'_{+}(0)/\psi_{+}(0)$ ratio it holds:

\begin{eqnarray}
\frac{\psi'_{+}(0)}{\psi_{+}(0)}=-\frac{1}{2a}+~~~~~~~~~~~~~~~~~~~~~~~~~~~~~~~~~~~~~~~~~~~~~~~~~~~~~~~~~~~~~~~~~~~~~~~~~~~~~~~~~~~~~~~~~~~~~~~~~~~~~~\nonumber\\
+\frac{\omega}{2}\cdot\frac{J_{(\alpha_{+}-1)}(a\omega)-J_{(\alpha_{+}+1)}(a\omega)-e^{-i\alpha_{+}\pi}\left[J_{(-\alpha_{+}-1)}(a\omega)-J_{(-\alpha_{+}+1)}(a\omega)\right]}{J_{\alpha_{+}}(a\omega)-e^{-i\alpha_{+}\pi}J_{-\alpha_{+}}(a\omega)}.~~~~\label{bessel}
\end{eqnarray}

This equation is obtained after using known differential identities for Bessel functions (see identity \ref{Besselderivative} in appendix \ref{A:Bessel}).

The key step in obtaining the ratio $\psi'_{-}(0)/\psi_{-}(0)$ is to calculate the logarithmic derivative. By choosing the variable $z = 1/(1+ e^{- 2x/b_-})$, note that $x=0$ maps into $z=1/2$. Then using the Leibnitz rule and the chain rule one has:
\begin{equation}
{\psi_-'(0)\over\psi_-(0)} =
 i \omega +    {1\over 2 b_-} \left. {\d~ \ln \left\{ _2F_1\left({1\over2}+\alpha_-,{1\over2}-\alpha_-,1+ib_- \omega,z\right) \right\} \over \d z}\right|_{z=1/2.}
\end{equation}
Invoking the differential identity (\ref{E:differential}) in appendix \ref{A:hyper}, we see
\begin{equation}
{\psi_-'(0)\over\psi_-(0)} =  i \omega \;
\left.
{_2 F_1\left({1\over2}+\alpha_-,{1\over2}-\alpha_-,ib_- \omega,z\right)
\over
_2F_1\left({1\over2}+\alpha_-,{1\over2}-\alpha_-,1+ib_- \omega,z\right)}
\right|_{z=1/2.}
\end{equation}

Now using Bailey's theorem (\ref{E:bailey}) to evaluate the hypergeometric functions at $z\to{1\over2}$ we have the exact result
\begin{equation}
{\psi_-'(0)\over\psi_-(0)}
=  {2\over b_-} \; {\Gamma({\alpha_-+i\omega b_-\over2} +{3\over4}) ~\Gamma({-\alpha_-+i\omega b_-\over2} +{3\over4}) \over
\Gamma({\alpha_-+i\omega b_- \over2} + {1\over4}) ~\Gamma({-\alpha_-+i\omega b_- \over2} + {1\over4})}.\label{gamma}
\end{equation}

\subsubsection{The asymptotic approximations}

The exact junction condition we wish to apply at $x=0$ is
 \begin{equation}
{\psi_+'(0)\over\psi_+(0)}  = {\psi_-'(0)\over\psi_-(0)}.
\end{equation}
We see that equating the right sides of \eqref{gamma} and \eqref{bessel} gives a messy equation and the presence of the Gamma and Bessel functions above makes this exact junction condition intractable. Fortunately, from the beginning we are interested only in the asymptotic quasi-normal modes. This means $Im(\omega)\gg 0$ and $Im(\omega)\gg Re(\omega)$. In such a case we can use asymptotic formulas for both the Bessel and Gamma functions.

For Bessel functions one takes the asymptotic expansion (\ref{Besselasymptotic}) from the appendix \ref{A:Bessel}. In the dominant (zero-th) order, for ~$Re(x)$~ bounded and ~$Im(x)\to\infty$,~ the functions ~$P(\alpha,x)$~ and ~$Q(\alpha,x)$,~ (see ~\ref{P}~ and ~\ref{Q}~ in appendix ~\ref{A:Bessel}),~ behave as  ~$P(\alpha,x)\to 1$,~ while ~$Q(\alpha,x)\to 0$.~ Hence for ~$x\approx Im(x)\gg 0$~ the expansion ~(\ref{Besselasymptotic})~ gives:

\begin{equation}
J_{\alpha}(x)\approx\sqrt{\frac{2}{\pi x}}\cos\left(x-\frac{\pi\alpha}{2}-\frac{\pi}{4}\right).
\end{equation}
By substituting this into \eqref{bessel}, after some additional calculations using trivial trigonometric relations we obtain a very simple asymptotic formula:

\begin{equation}
\frac{\psi'_{+}(0)}{\psi_{+}(0)}\approx -\frac{1}{2a}+Im(\omega)\approx Im(\omega).
\end{equation}

The case \eqref{gamma} is slightly more difficult. In this case consider the following: if $\omega$ has a large positive imaginary part, then the Gamma function arguments above tend towards the negative real axis, a region where the Gamma function has many poles. This is computationally inconvenient, and to obtain a more tractable result it is extremely useful to use the reflection formula (\ref{E:reflection}) of appendix \ref{A:gamma} to derive
\begin{eqnarray}
 {\Gamma({\alpha_-+i\omega b_-\over2} +{3\over4}) \over \Gamma({\alpha_-+i\omega b_- \over2} + {1\over4}) }
 &=&
 { \Gamma(1-{\alpha_-+i\omega b_- \over2} - {1\over4})~ \sin(\pi [{\alpha_-+i\omega b_- \over2} + {1\over4} ])
 \over
 \Gamma(1-{\alpha_-+i\omega b_-\over2} -{3\over4})~ \sin(\pi[{\alpha_-+i\omega b_-\over2} +{3\over4}])}
 \nonumber
 \\
&=&  { \Gamma(-{\alpha_-+i\omega b_- \over2} + {3\over4}) \;~ \sin(\pi [{\alpha_-+i\omega b_- \over2} + {1\over4} ])
 \over
 \Gamma(-{\alpha_-+i\omega b_-\over2} +{1\over4})~ \sin(\pi[{\alpha_-+i\omega b_-\over2} +{1\over4}]+ {\pi\over2})}
\nonumber
 \\
&=&  { \Gamma(-{\alpha_-+i\omega b_- \over2} + {3\over4})
 \over
 \Gamma(-{\alpha_-+i\omega b_-\over2} +{1\over4})}  \times \tan\left(\pi\left[{\alpha_-+i\omega b_-\over2} +{1\over4}\right]\right).
 \end{eqnarray}
This leads to the exact result
\begin{eqnarray}
{\psi_-'(0)\over\psi_-(0)} &=&  {2\over b_-}  \;
 { \Gamma({-\alpha_--i\omega b_- \over2} + {3\over4})~  \Gamma({\alpha_--i\omega b_- \over2} + {3\over4})
 \over
 \Gamma({-\alpha_--i\omega b_-\over2} +{1\over4})~  \Gamma({\alpha_--i\omega b_-\over2} +{1\over4})}
 \nonumber
\\
&&
 \times
   \; \tan\left(\pi\left[{\alpha_-+i\omega b_-\over2} +{1\over4}\right]\right)  \; \tan\left(\pi\left[{-\alpha_-+i\omega b_-\over2} +{1\over4}\right]\right).~~~~~
 \end{eqnarray}
If $\omega$ has a large positive imaginary part, then the Gamma function arguments above now tend towards the positive real axis, a region where the Gamma function is smoothly behaved --- all potential poles in the logarithmic derivative have been isolated in the trigonometric functions.
We can also use one of the trigonometric  identities (\ref{E:trig1}) of appendix \ref{A:trig} to rewrite this as
  \begin{equation*}
{\psi_-'(0)\over\psi_-(0)} =   {2\over b_-}  \;
 { \Gamma({-\alpha_--i\omega b_- \over2} + {3\over4})~  \Gamma({\alpha_--i\omega b_- \over2} + {3\over4})
 \over
 \Gamma({-\alpha_--i\omega b_-\over2} +{1\over4})~  \Gamma({\alpha_--i\omega b_-\over2} +{1\over4})}
 \times  {\cos(\pi\alpha_-) - \cos(\pi[i\omega b_-+1/2]) \over \cos(\pi\alpha_-) + \cos(\pi[i\omega b_-+1/2]) },
 \end{equation*}
 which we can rewrite (still an \emph{exact} result) as
   \begin{equation}
{\psi_-'(0)\over\psi_-(0)} =   {2\over b_-}  \;
 { \Gamma({-\alpha_--i\omega b_- \over2} + {3\over4})~  \Gamma({\alpha_--i\omega b_- \over2} + {3\over4})
 \over
 \Gamma({-\alpha_--i\omega b_-\over2} +{1\over4})~  \Gamma({\alpha_--i\omega b_-\over2} +{1\over4})}
 \times  {\cos(\pi\alpha_-) + \sin(i\pi\omega b_-) \over \cos(\pi\alpha_-) - \sin(i\pi\omega b_-) }.
 \end{equation}

 We have already seen how to eliminate the Bessel functions in the approximation that $Im(\omega)$ is very large and $Re(\omega)$ bounded. Now we have to do the same with Gamma functions. Fortunately this is possible as well. As long as we are primarily focussed on the highly damped QNFs ($Im(\omega)\to\infty$, ~$Re(\omega)$ bounded) we can employ the Stirling approximation in the form (\ref{E:stirling}) indicated in appendix \ref{A:gamma} to deduce
\begin{eqnarray}
 { \Gamma({-\alpha_--i\omega b_- \over2} + {3\over4})
 \over
  \Gamma({-\alpha_--i\omega b_-\over2} +{1\over4})  }
  &=&
  \sqrt{ {-\alpha_--i\omega b_-\over2} +{1\over4} } \times  \left[1+ O\left({1\over Im(\omega b_-)}\right)\right]
  \nonumber
  \\
  &=& \sqrt{ {Im(\omega) b_-\over 2} }  \times \left[1+ O\left({1\over Im(\omega b_-)}\right)\right].
  \end{eqnarray}
This allows us to \emph{approximate} the junction condition by

\begin{equation}
1=\frac{\cos(\pi\alpha_{-})+\sin(i\pi\omega b_{-})}{cos(\pi\alpha_{-})-\sin(i\pi\omega b_{-})}
\end{equation}
leading directly to

\begin{equation}
\sin(i\pi\omega b_{-})=0
\end{equation}
and hence

\begin{equation}
\omega=\frac{in}{b_{-}}=in\kappa.\label{SQNM}
\end{equation}

\subsubsection{Discussion}

This result means that the QNMs are purely imaginary and equispaced with the gap being given by the surface gravity at the black hole horizon. This result can be considered to \emph{qualitatively} (but not fully \emph{quantitatively}) match with the well known result (\ref{HighSchw}).

\subsection{Schwarzschild - de Sitter black hole}

%-----------------------------------------------------------------------------------------------------------------------------------------

%---------------------------------------------------------------------------------------------------------
%-----------------------------------------------------------------------------------------------------------------------------------------
\subsubsection{Potential}

The model  we are interested in investigating is
\begin{equation}
-\psi''(x) + (V(x)-\omega^{2}) \; \psi(x) = 0,
\end{equation}
with
\begin{equation}
V(x) = \left\{  \begin{array}{lcl}
{V_{0-} \; \sech^2(x/b_-)}  & \hbox{ for }  & x < 0;\\
\\
V_{0+} \; \sech^2(x/b_+)  & \hbox{ for }  & x > 0. \\
\end{array}\right.
\end{equation}
Here we again take ~$b_{\pm}\equiv\kappa_{\pm}^{-1}$.~ We will allow a discontinuity in the potential at $x=0$. The standard case that is usually dealt with is for
\begin{equation}
V_{0-} = V_{0+} = V_0;
\qquad\qquad
b_-=b_+ = b;
\qquad\qquad
V(x) = {V_{0} \over \cosh^2(x/b)}.
\end{equation}
A related model where $V_{0-} = V_{0+} = V_0$ but $b_+\neq b_-$ has been explored by Suneeta~\cite{Suneeta}, but our current model is more general, and we will take the analysis much further.

%-----------------------------------------------------------------------------------------------------------------------------------------
\subsubsection{Wavefunction}
%-----------------------------------------------------------------------------------------------------------------------------------------

We start again by imposing quasi-normal boundary conditions (outgoing radiation boundary conditions)
\begin{equation}
\psi_+(x\to+\infty) \to C_+e^{-i\omega x}; \qquad \psi_-(x\to-\infty) \to C_-e^{+i\omega x}.
\end{equation}
On each half line ($x<0$, and $x>0$) the exact wavefunction   (see especially page 405 of the article by Beyer~\cite{Beyer}) is now:
\begin{equation}
\psi_\pm(x) = C_\pm e^{\mp i\omega x} \; {} _2F_1\left({1\over2}+\alpha_\pm,{1\over2}-\alpha_\pm,1+ib_\pm \omega, {1\over 1+ e^{\pm 2x/b_\pm}} \right),
\end{equation}
where
\begin{equation}
\alpha_{\pm} \equiv \left\{  \begin{array}{lcl}
\sqrt{{1\over4} - V_0 b_\pm^2 } & \hbox{ for }  & V_0 b_\pm^2 < 1/4;\\
\\
i \sqrt{V_0 b_\pm^2 - {1\over4}  } & \hbox{ for }  & V_0 b_\pm^2 > 1/4.\\
\end{array}\right.
\end{equation}
The same as in previous section, ``all'' we need to do is to appropriately match these wavefunctions at the origin.

\subsubsection{Junction condition}

The junction condition is again

\begin{equation}
\frac{\psi'_{+}(0)}{\psi_{+}(0)}=\frac{\psi'_{-}(0)}{\psi_{-}(0)}
\end{equation}
and the same remarks related to zero values of $\psi_\pm(0)$ as in the previous section hold. The ratio $\psi'_{\pm}(0)/\psi_{\pm}(0)$ can be rewritten through Gamma functions as:

\begin{equation}
{\psi_\pm'(0)\over\psi_\pm(0)}
= \mp {2\over b_\pm} \; {\Gamma({\alpha_\pm+i\omega b_\pm\over2} +{3\over4}) ~\Gamma({-\alpha_\pm+i\omega b_\pm\over2} +{3\over4}) \over
\Gamma({\alpha_\pm+i\omega b_\pm \over2} + {1\over4}) ~\Gamma({-\alpha_\pm+i\omega b_\pm \over2} + {1\over4})}.
\end{equation}

%-----------------------------------------------------------------------------------------------------------------------------------------
\subsubsection{The asymptotic approximations}
%-----------------------------------------------------------------------------------------------------------------------------------------

Here both sides of the junction condition contain Gamma functions in the same way as the right side of the junction condition from the previous section. Hence all the analysis can be repeated (rewriting some of the Gamma functions through the basic Gamma functions relations and using the Stirling approximation for the asymptotic modes). Following exactly the same arguments one obtains the following \emph{approximation} to the junction condition:

\begin{equation}
 {\cos(\pi\alpha_+) + \sin(i\pi\omega b_+) \over \cos(\pi\alpha_+) - \sin(i\pi\omega b_+) } =
 -  {\cos(\pi\alpha_-) + \sin(i\pi\omega b_-) \over \cos(\pi\alpha_-) - \sin(i\pi\omega b_-) },
\end{equation}
which is accurate up to fractional corrections of order $O\left({1/ Im(\omega b_\pm)}\right)$.
This is now an \emph{approximate} ``quantization condition'' for calculating the QNFs being asymptotically increasingly accurate for the highly-damped modes.

%-----------------------------------------------------------------------------------------------------------------------------------------
\subsubsection{QNF condition}
%-----------------------------------------------------------------------------------------------------------------------------------------

The asymptotic QNF condition above can, by cross multiplication and the use of trigonometric identities, be rewritten in any one of the equivalent forms:
 \begin{equation}
 \label{E:qnf1}
\sin(-i\pi\omega b_+)\sin(-i\pi\omega b_-)  =  \cos(\pi\alpha_+)\cos(\pi\alpha_-);
 \end{equation}
  \begin{equation}
   \label{E:qnf2}
\sinh(\pi\omega b_+)\sinh(\pi\omega b_-)  =  - \cos(\pi\alpha_+)\cos(\pi\alpha_-);
 \end{equation}
\begin{equation}
 \label{E:qnf3}
\cos(-i\pi\omega[b_+ -b_-]) - \cos(-i\pi\omega[b_+ + b_-]) = 2  \;  \cos(\pi\alpha_+)\cos(\pi\alpha_-);
 \end{equation}
 \begin{equation}
 \label{E:qnf4}
\cosh(\pi\omega[b_+ -b_-]) - \cosh(\pi\omega[b_+ + b_-]) = 2  \;  \cos(\pi\alpha_+)\cos(\pi\alpha_-).
 \end{equation}
Which particular form one chooses to use is a matter of taste that depends on exactly what one is trying to establish.
It is sometimes useful to split the asymptotic QNF condition into real and imaginary parts. To do so note
\begin{eqnarray}
\cos(A+iB) &=& \cos(A)\cos(iB) - \sin(A)\sin(iB)
\nonumber\\
& =& \cos(A) \cosh(B) - i \sin(A) \sinh(B),
\end{eqnarray}
so that
\begin{eqnarray}
\cos(-i\pi\omega[b_+ -b_-]) &=& \cos(Im(\omega)\pi[b_+ -b_-]) \cosh(Re(\omega)\pi[b_+ -b_-])
\nonumber
\\
&&
 - i  \sin(Im(\omega)\pi|b_+ -b_-|) \sinh(Re(\omega)\pi|b_+ -b_-|).~~~~~~~
\end{eqnarray}
Therefore the asymptotic QNF condition implies \emph{both}
\begin{eqnarray}
 \label{E:qnf:real}
&&
\cos(Im(\omega)\pi[b_+ -b_-]) \cosh(Re(\omega)\pi[b_+ -b_-])
\nonumber
\\ &&
\quad = \cos(Im(\omega)\pi[b_+ +b_-]) \cosh(Re(\omega)\pi[b_+ +b_-])
+ 2  \;  \cos(\pi\alpha_+)\cos(\pi\alpha_-),~~~
\qquad
\end{eqnarray}
and
\begin{eqnarray}
 \label{E:qnf:imaginary}
&&\sin(Im(\omega)\pi|b_+ -b_-|) \sinh(Re(\omega)\pi|b_+ -b_-|)
\nonumber\\
&& \quad = \sin(Im(\omega)\pi[b_+ +b_-]) \sinh(Re(\omega)\pi[b_+ +b_-]).
\end{eqnarray}
We shall now seek to apply this QNF condition, in its many equivalent forms, to extract as much information as possible regarding the distribution of the QNFs.

%----------------------------------------------------------------------------------------------------------------------------------------
\subsubsection{Some general observations}
%----------------------------------------------------------------------------------------------------------------------------------------

We shall start with some general observations regarding the QNFs.
\begin{enumerate}
\item Note that the $\alpha_\pm$ are either pure real or pure imaginary.

\item Consequently $\cos(\pi\alpha_+)\cos(\pi\alpha_-)$ is always pure real $\in [-1, +\infty)$.

\item If  $\cos(\pi\alpha_+)\cos(\pi\alpha_-) > 0$, then there are no pure real QNFs.

\emph{Proof:} Consider equation (\ref{E:qnf2}) and note that under this condition the LHS is positive while the RHS is negative.

\item If  $\cos(\pi\alpha_+)\cos(\pi\alpha_-) < 0$, then there is a pure real QNF.

\emph{Proof:} Consider equation (\ref{E:qnf2}) and note that under this condition the RHS is positive. The LHS is positive and by continuity there will be a a real root  $\omega \in (0,\infty)$.

\item If  $ \cos(\pi\alpha_+)\cos(\pi\alpha_-) > 1$, then there are no pure imaginary QNFs.

\emph{Proof:} Consider equation (\ref{E:qnf1}) and note that under this condition the LHS $\leq 1$ while the RHS $>1$.

\item There are infinitely many pure imaginary solutions to these asymptotic QNF conditions provided $ \cos(\pi\alpha_+)\cos(\pi\alpha_-)\leq Q(b_+,b_-) \leq 1$; that is, whenever  $ \cos(\pi\alpha_+)\cos(\pi\alpha_-)$  is ``sufficiently far'' below $1$.

\emph{Proof:} Define
\begin{equation}
Q(b_+,b_-) =   \max_\omega\left\{  {\cos( |\omega|\pi[b_+ -b_-])   -  \cos( |\omega|\pi[b_+ +b_-])\over 2}  \right\} \leq 1.
\end{equation}
Then by inspection equation (\ref{E:qnf3}) will have an infinite number of pure imaginary solutions as long as
\begin{equation}
 \cos(\pi\alpha_+)\cos(\pi\alpha_-)\leq Q(b_+,b_-).
\end{equation}

\item For any purely imaginary $\omega$ there will be \emph{some} choice of $b_\pm$, $\alpha_\pm$ that makes this a solution of the asymptotic QNF condition.

\emph{Proof:} Consider the specific case
\begin{equation}
\omega =  {i (\alpha_+ + {1\over2})\over b_+} =  {i (\alpha_- + {1\over2})\over b_-},
\end{equation}
and note this satisfies the QNF condition but enforces only two constraints among the four unknowns $b_\pm$, $V_{0\pm}$.

\end{enumerate}

%----------------------------------------------------------------------------------------------------------------------------------------

\subsubsection{Rational ratios for the falloff}
%-----------------------------------------------------------------------------------

Suppose $b_+/b_-$ is rational, that is
\begin{equation}
{b_+\over b_-} = {p_+\over p_-}\in\mathbb{Q},
\end{equation}
and suppose we now define $b_*$ by
\begin{equation}
b_+= p_+ b_*; \qquad b_- = p_- b_*; \qquad b_* = \mathrm{hcf}(b_+,
b_-),
\end{equation}
where ``hcf'' means ``highest common factor''. ($p_+, p_-$ are
relatively prime.\footnote{Later in this chapter will the symbols
$p_i/p_j$ automatically mean a rational number given by the
relatively prime integers $p_i$, $p_j$.}) Then the asymptotic QNF
condition is given by
 \begin{equation}
 \sin(-i\omega \pi p_+ b_*)\sin(-i\omega \pi p_- b_*)  =  \cos(\pi\alpha_+)\cos(\pi\alpha_-).
 \end{equation}
If $\omega_*$ is \emph{any} specific solution of this equation, then
\begin{equation}
\omega_n = \omega_* + {i2n\over b_*} = \omega_* + i2n \;
\mathrm{lcm}\left({1\over b_+}, {1\over b_-}\right)
\end{equation}
will also be a solution. (Here ``lcm'' stands for ``least common multiplier''.) But are these the only solutions? Most
definitely not. For instance, consider (for rational $b_+/b_-$) the
set of all QNFs for which
\begin{equation}
Im(\omega) < {1\over b_*},
\end{equation}
and label them as
\begin{equation}
\omega_{0,a} \qquad a\in\{1,2,3\dots N\}.
\end{equation}
Then the set of all QNFs decomposes into a set of families
\begin{equation}
\omega_{n,a} = \omega_{0,a} + {in\over b_*}; \qquad a\in\{1,2,3\dots
N\}; \qquad n\in\{0,1,2,3\dots\};
\end{equation}
where $N$ is yet to be determined.
But for rational  $b_+/b_-$  we can rewrite the QNF condition as
\begin{equation}
\cos(-i\omega\pi b_*|p_+-p_-|) - \cos(-i\omega\pi b_*[p_+ +p_-])  =
2\cos\left(\pi\alpha_+\right) \cos\left(\pi\alpha_-\right).
\end{equation}
Now define $z=\exp(\omega\pi b_*)$,  then the QNF condition can be
rewritten as
\begin{equation}
z^{|p_+-p_-|} + z^{-|p_+-p_-|} - z^{[p_+ +p_-]} - z^{-[p_+ +p_-]} =
4\cos\left(\pi\alpha_+\right) \cos\left(\pi\alpha_-\right),
\end{equation}
or equivalently
\begin{eqnarray}
z^{2[p_+ +p_-]} - z^{|p_+-p_-|+[p_+ +p_-]}~~~~~~~~~~~~~~~~~~~~~~~~~~~~~~~~~~~~~~~~~~~~~~~~~~~~~~~~~~~~~~~~~~~~~~~~~~~~~\nonumber\\ -
4\cos\left(\pi\alpha_+\right) \cos\left(\pi\alpha_-\right) z^{[p_+
+p_-]} -  z^{-|p_+-p_-|+[p_+ +p_-]} + 1 = 0,
\end{eqnarray}
that is
\begin{equation}
z^{2[p_+ +p_-]} - z^{2p_+} -  z^{2p_-} -
4\cos\left(\pi\alpha_+\right) \cos\left(\pi\alpha_-\right) z^{[p_+
+p_-]} + 1 = 0.
\end{equation}
This is a polynomial of degree $N=2(p_++p_-)$, so it has exactly $N$
roots $z_a$. In terms of equi-spaced families of QNM the situation is a little bit more tricky. If $p_+\cdot p_-$ is odd then the gap spacing is only ${1\over b_*}$ and the generic number of different families is only $p_+ + p_-$, while if $p_+\cdot p_-$ is even, then the gap is given by ${2\over b_*}$ and the generic number of families is $2(p_++p_-)$. This can be summarized as:
\newtheorem{rational}[section]{Theorem}
\begin{rational}\label{rational}
Take ${b_+/b_-}={p_+/p_-}\in\mathbb{Q}$ and $b_*=\mathrm{hcf}(b_+,b_-)$. The QNFs
are, with the imaginary part of the logarithm lying in $[0,2\pi)$ given by
 \begin{equation}
 \omega_{n,a} =
  { \ln(z_a) \over \pi b_*} + {i2n\over gb_*}
 \qquad  a\in\left\{1,2,3\dots N\leq {2(p_++p_-)\over g}\right\} \qquad n\in\{0,1,2,3,\dots\},
 \end{equation}
with ~$g=2$~ for ~$p_+\cdot p_-$~ odd and ~$g=1$~ for ~$p_+\cdot
p_-$~ even.
\end{rational}
 So for rational $b_+/b_-$ all modes form equi-spaced families, all families have the same gap spacing and are characterized by distinct offsets $  { \ln(z_a) /( \pi b_*) } $.
 That is: \emph{Arbitrary rational ratios of ${b_+/ b_-} $ automatically  imply the}  $\omega_n =  \mathrm{offset} + i n \cdot\mathrm{gap}$ \emph{behaviour.}

%-----------------------------------------------------------------------------------
\subsubsection{Irrational ratios for the falloff}
%-----------------------------------------------------------------------------------

Now suppose $b_+/b_-$ is irrational, that is
\begin{equation}
b_*=  \mathrm{hcf}(b_+, b_-) = 0.
\end{equation}
Then all of the ``families'' considered above only have one element
\begin{equation}
\omega_{0,a}  \qquad a\in\{1,2,3\dots \infty\}.
\end{equation}
That is, there will be no ``pattern'' in the QNFs, and they will not
be regularly spaced. (Conversely, if there is a ``pattern'' then
$b_+/b_-$ is rational.) Stated more formally, it is possible to
derive a theorem as below.

\bigskip

\newtheorem{irational}[section]{Theorem}
\begin{irational}\label{irrational}
Suppose we have at least one family of
equi-spaced QNFs such that
\begin{equation}
\label{E:equal} \omega_n = \omega_0 + i n\cdot \mathrm{gap}~,
\end{equation}
then  $b_+/b_-$ is rational.
\end{irational}

\begin{proof}
If we have a family of QNFs of the form given
in equation (\ref{E:equal}) then we know that $\forall n \geq 0$
\begin{eqnarray}
&& \cos(-i\omega_0\pi |b_+-b_-| + nK\pi |b_+-b_-| ) -
\cos(-i\omega_0\pi [b_+ +b_-] + nK\pi |b_++b_-| ) \qquad
\nonumber\\
&& \qquad\qquad \qquad = \cos(-i\omega_0\pi |b_+-b_-| ) -
\cos(-i\omega_0\pi [b_+ +b_-]  ).
\end{eqnarray}
Let us write this in the form $\forall n \geq 0$
\begin{equation}
\cos(A+n J) - \cos(B+nL) = \cos(A)-\cos(B),
\end{equation}
and realize that this also implies
\begin{equation}
\cos(A+[n+1] J) - \cos(B+[n+1]L) = \cos(A)-\cos(B),
\end{equation}
and
\begin{equation}
\cos(A+[n+2] J) - \cos(B+[n+2]L) = \cos(A)-\cos(B).
\end{equation}
Now appeal to the trigonometric identity (based on equation
(\ref{E:trig3}))
\begin{equation}
\cos(A+[n+2] J) + \cos(A+n J) = 2 \cos(J) \cos(A+[n+1] J),
\end{equation}
to deduce
\begin{equation}
\cos(J) \cos(A+[n+1] J) - \cos(L) \cos(B+[n+1]L) = \cos(A)-\cos(B).
\end{equation}
That is, $\forall n\geq 0$ we have \emph{both}
\begin{equation}
\cos(A+[n+1] J) - \cos(B+[n+1]L) = \cos(A)-\cos(B),
\end{equation}
\emph{and}
\begin{equation}
\cos(J) \cos(A+[n+1] J) - \cos(L) \cos(B+[n+1]L) = \cos(A)-\cos(B).
\end{equation}
The first of these equations asserts that all the points
\begin{equation}
\left( \vphantom{\Big|} \cos(A+[n+1] J) , \; \cos(B+[n+1]L)  \right)
\end{equation}
lie on the straight line of slope 1 that passes through the point
$(0, \cos B-\cos A)$. The second of these equations asserts that all
the points
\begin{equation}
\left( \vphantom{\Big|}  \cos(A+[n+1] J) , \; \cos(B+[n+1]L)
\right)
\end{equation}
\emph{also} lie on the straight line of slope $\cos(J)/\cos(L)$ that
passes through the point $(0, [\cos B-\cos A]/\cos L)$. We then
argue as follows:
\begin{itemize}
\item If $\cos J \neq \cos L$ then these two lines are not parallel and so meet only at a single point, let's call it $(\cos A_*, \cos B_*)$, whence we deduce
\begin{equation}
\cos(A+[n+1] J) = \cos A_* ; \qquad \cos(B+[n+1]L) = \cos B_*.
\end{equation}
But then both $J$ and $L$ must be multiples of $2\pi$, and so $\cos
J = 1=  \cos L$ contrary to hypothesis.
\item
If  $\cos J = \cos L \neq 1$ then we have both
\begin{equation}
\cos(A+[n+1] J) - \cos(B+[n+1]L) = \cos(A)-\cos(B),
\end{equation}
and
\begin{equation}
\cos(J)\left[  \cos(A+[n+1] J) -  \cos(B+[n+1]L) \right] =
\cos(A)-\cos(B).
\end{equation}
but these are two parallel lines, both of slope 1, that never
intersect unless $\cos(J)=1$. Thus  $\cos J = 1=  \cos L$ contrary
to hypothesis.
\item
We therefore conclude that both $J$ and $L$ must be multiples of
$2\pi$ so that $\cos J = 1=  \cos L$ (in which case the QNF
condition is certainly satisfied).
\end{itemize}
But now
\begin{equation}
{|b_+-b_-|\over b_++b_-} = {J\over L} \; \in\; Q,
\end{equation}
and therefore
\begin{equation}
{b_+\over b_-}  \;\in\; Q.
\end{equation}
That is: \emph{Rational ratios of ${b_+/ b_-} $ are implied by the}
$\omega_n =  \mathrm{offset} + i n \; \mathrm{gap}$
\emph{behaviour}.

\end{proof}

\subsubsection{Explicit examples}
%-----------------------------------------------------------------------------------

\begin{itemize}

\item If $b_+=b_- = b_*$, but we do not necessarily demand $\alpha_+=\alpha_-$, then the asymptotic QNFs are exactly calculable and are given by
\begin{equation}
\omega_n  =  {i\cos^{-1}\left\{1 - 2  \;
\cos(\pi\alpha_+)\cos(\pi\alpha_-)\right\}\over2\pi b_*} + {in\over
b_*}.
\end{equation}

\emph{Proof:} The asymptotic QNF condition reduces to
\begin{equation}
1 - \cos(-i2\pi\omega b_*) = 2  \;
\cos(\pi\alpha_+)\cos(\pi\alpha_-),
 \end{equation}
whence
\begin{equation}
 \cos(-i2\pi\omega b_*) =  1 - 2  \;  \cos(\pi\alpha_+)\cos(\pi\alpha_-).
 \end{equation}
This is easily solved to yield
\begin{equation}
 -i2\pi\omega_n b_* =  \cos^{-1}\left\{1 - 2  \;  \cos(\pi\alpha_+)\cos(\pi\alpha_-)\right\} + 2n\pi,
 \end{equation}
 whence
\begin{equation}
\omega_n  =  {i\cos^{-1}\left\{1 - 2  \;
\cos(\pi\alpha_+)\cos(\pi\alpha_-)\right\}\over2\pi b_*} + {in\over
b_*}.
 \end{equation}

\emph{Comment:} These QNFs are pure imaginary for
$\cos(\pi\alpha_+)\cos(\pi\alpha_-)\leq 1$, and off-axis complex for
$\cos(\pi\alpha_+)\cos(\pi\alpha_-) > 1$. We can always, for
convenience, choose to define $Re(\cos^{-1}(x)) \in [0,2\pi)$; then because
$\cos^{-1}(\cdot)$ is double valued we see
\begin{equation}
\theta\in[0,\pi] \hbox{ and } \cos\theta=x \hbox{ implies
}\cos^{-1}(x) = \{\theta, \pi-\theta\},
\end{equation}
\begin{equation}
\theta\in[\pi,2\pi) \hbox{ and } \cos\theta=x \hbox{ implies }
\cos^{-1}(x) = \{\theta, 3\pi-\theta\} .
\end{equation}
With this notation
\begin{equation}
\omega_n  =    \omega_0 + {in\over b_*}; \qquad   \omega_0 =
{i\cos^{-1}\left\{1 - 2  \;
\cos(\pi\alpha_+)\cos(\pi\alpha_-)\right\}\over2\pi b_*};
 \end{equation}
 with $0\leq Im(\omega_0) < 1/b_*$ and $n\in\{0,1,2,3,\dots\}$. Note that because of the double valued nature of  $\cos^{-1}(\cdot)$ there are actually two branches of QNFs hiding in this notation --- which we will need if we wish to regain the known standard result when we  specialize to $\alpha_-=\alpha_+$. (We shall subsequently generalize this specific result, but it is explicit enough and compact enough to make it worthwhile presenting it in full. Furthermore we shall need this as input to our perturbative analysis.)

\item If $b_+=3b_-$, that is $b_+={3\over2} b_*$ and $b_-={1\over2} b_*$, but we do not necessarily demand $\alpha_+=\alpha_-$, then the asymptotic QNFs are calculable and are given by
\begin{equation}
\omega_n  =  {i\over \pi b_*} \cos^{-1}\left(  {1 \pm \sqrt{9 - 16
\;  \cos(\pi\alpha_+)\cos(\pi\alpha_-))} \over 4} \right) +
{2in\over b_*}.
\end{equation}

\emph{Proof:}  To see this note that in this situation $|b_+-b_-| =
b_* = (b_++b_-)/2$. Therefore the QNF condition reduces to
\begin{equation}
 \cos(-i\pi\omega b_*) - \cos(-i2\pi\omega b_*) = 2  \;  \cos(\pi\alpha_+)\cos(\pi\alpha_-),
\end{equation}
implying
\begin{equation}
 \cos(-i\pi\omega b_*) - 2\cos^2(-i\pi\omega b_*) + 1 = 2  \;  \cos(\pi\alpha_+)\cos(\pi\alpha_-).
\end{equation}
That is
\begin{equation}
2\cos^2(-i\pi\omega b_*) -  \cos(-i\pi\omega b_*)  - 1 +  2  \;
\cos(\pi\alpha_+)\cos(\pi\alpha_-) = 0,
\end{equation}
whence
\begin{equation}
\cos(-i\pi\omega b_*)  = {1 \pm \sqrt{ 1 + 8(1-2  \;
\cos(\pi\alpha_+)\cos(\pi\alpha_-))} \over 4},
\end{equation}
so that
\begin{equation}
\cos(-i\pi\omega b_*)  = {1 \pm \sqrt{9 - 16  \;
\cos(\pi\alpha_+)\cos(\pi\alpha_-))} \over 4},
\end{equation}
implying
\begin{equation}
-i\pi\omega_n b_* = \cos^{-1}\left(  {1 \pm \sqrt{9 - 16  \;
\cos(\pi\alpha_+)\cos(\pi\alpha_-))} \over 4} \right) + n 2\pi.
\end{equation}
Finally
\begin{equation}
\omega_n  =  {i\over \pi b_*} \cos^{-1}\left(  {1 \pm \sqrt{9 - 16
\;  \cos(\pi\alpha_+)\cos(\pi\alpha_-))} \over 4} \right) +
{2in\over b_*}.
\end{equation}
\emph{Comment:} This gives us another specific example of asymptotic
off-axis complex QNF's --- now with $b_+\neq b_-$. Note that because
of the $\pm$ and the double-valued nature of $\cos^{-1}(\cdot)$
there are actually 4 branches of QNFs hiding in this notation.

\item This particular trick can certainly be extended to the cubic and quartic polynomials, for which general solutions exist.
\begin{itemize}
\item
The quadratic corresponds to
\begin{equation}
{b_+-b_-\over b_++b_-} = 2; \qquad b_+ = 3 \;b_-.
\end{equation}
\item
The cubic corresponds to
\begin{equation}
{b_+-b_-\over b_++b_-} = 3; \qquad b_+ = 2 \; b_-.
\end{equation}
\item
The quartic corresponds to
\begin{equation}
{b_+-b_-\over b_++b_-} = 4; \qquad b_+ = {5\over3}\; b_-.
\end{equation}
\end{itemize}
\end{itemize}

\subsubsection{Some approximate results}
%----------------------------------------------------------------------------------------------------------------------------------------

A number of approximate results can be extracted by looking at special regions of parameter space.

%-----------------------------------------------------------------------------------
\paragraph{Case $b_- \approx b_+$}
%-----------------------------------------------------------------------------------

Suppose {$b_- \approx b_+$. Then the quantity $-i (b_+-b_-) \omega$ is slowly varying over the range where $-i (b_+ + b_-) \omega$ changes by $2\pi$.
Let $\omega_*$ be any solution of the approximate QNF condition, and define $b_*= (b_++b_-)/2$. Then for nearby frequencies we are trying to (approximately) solve
 \begin{equation}
\cos(-i\omega_*\pi[b_+ -b_-]) - \cos(-i\omega\pi[b_+ + b_-]) =  2 \cos(\pi\alpha_+)\cos(\pi\alpha_-),
 \end{equation}
 that is
 \begin{equation}
\cos(-i\omega_*\pi[b_+ -b_-]) - \cos(-i\omega2\pi b_*) =  2 \cos(\pi\alpha_+)\cos(\pi\alpha_-),
 \end{equation}
and the solutions of this are approximately
\begin{equation}
\omega_n \approx \omega_* + {in\over b_*}   \qquad    \hbox{valid for} \qquad |n| \ll {b_++b_-\over |b_+-b_-|}.
\end{equation}
Thus approximate result will subsequently be incorporated into a more general perturbative result to be discussed below.

%-----------------------------------------------------------------------------------
\paragraph{Case $b_- \ll b_+$}
%-----------------------------------------------------------------------------------

Now suppose {$b_- \ll b_+$. Then the quantity  $-i b_- \omega$ is slowly varying over the range where $-i b_+ \omega$ changes by $2\pi$.
Let $\omega_*$ be any solution of the approximate QNF condition, then for nearby frequencies we are trying to (approximately) solve
 \begin{equation}
 \sin(-i\omega \pi b_+)\sin(-i\omega_* \pi b_-)  =  \cos(\pi\alpha_+)\cos(\pi\alpha_-),
 \end{equation}
and the solutions of this are approximately
\begin{equation}
\omega_n \approx \omega_* + {2in\over b_+}   \qquad    \hbox{valid for} \qquad |n| \ll {b_+\over b_-}.
\end{equation}

%-----------------------------------------------------------------------------------
\paragraph{Case $\alpha_- \approx 1/2$}
%-----------------------------------------------------------------------------------

This corresponds to
\begin{equation}
V_{0-} b_-^2 \approx 0,
\end{equation}
in which case the QNF condition becomes
 \begin{equation}\label{sepEq}
 \sin(-i\omega \pi b_+)\sin(-i\omega \pi b_-)  \approx 0.
 \end{equation}
Therefore one obtains either (the physically relevant condition)
 \begin{equation}
 -i\omega b_+=n  \qquad \implies \qquad \omega = {in\over b_+},
 \end{equation}
 or (the physically uninteresting situation)
 \begin{equation}
- i\omega b_-=n  \qquad \implies \qquad \omega = {in\over b_-}.
 \end{equation}
\emph{Note:} If you go to the limit $\alpha_-=1/2$ by setting $V_{0-}=0$, one sees on physical grounds that $b_-$ is irrelevant, so it cannot contribute to the physical QNF. Alternatively if you hold  $V_{0-}\neq0$ but drive $b_-\to 0$, then these QNF's are driven to infinity --- and so decouple from the physics. Either way, the only physically interesting QNFs are $\omega = {in/ b_+}$. We explore these limits more fully below.

%-----------------------------------------------------------------------------------
\subsubsection{Some special cases}
%-----------------------------------------------------------------------------------

A number of special cases can now be analyzed in detail to give us an overall feel for the general situation.

%-----------------------------------------------------------------------------------
\paragraph{Case $\alpha_- = 0 = \alpha_+$}
%-----------------------------------------------------------------------------------

This corresponds to
\begin{equation}
V_{0-} b_-^2 = {1\over 4} = V_{0+} b_+^2,
\end{equation}
in which case the QNF condition becomes
 \begin{equation}
 \sin(i\omega \pi b_+)\sin(i\omega \pi b_-)  = 1.
 \end{equation}
Let us look for pure imaginary QNFs. (We do not  claim that these are the only QNFs.)  This  implies that we must \emph{simultaneously} satisfy \emph{both}
  \begin{equation}
 \sin(-i\omega \pi b_+)= \sin(-i\omega \pi b_-)  = 1,
 \end{equation}
 or \emph{both}
  \begin{equation}
 \sin(-i\omega \pi b_+)= \sin(-i\omega \pi b_-)  = -1.
 \end{equation}
 That is  \emph{both}
\begin{equation}
-i\omega b_+ = 2n_++{1\over2}; \qquad -i\omega b_- = 2n_-+{1\over2},
\end{equation}
or  \emph{both}
\begin{equation}
-i\omega b_+ = 2n_+-{1\over2}; \qquad -i\omega b_- = 2n_--{1\over2}.
\end{equation}
Therefore either
\begin{equation}
{b_+\over b_-} = {2n_++{1\over2}\over 2n_-+{1\over2}}  = {4n_++1\over 4n_-+1}, \qquad \hbox{or} \qquad
{b_+\over b_-} = {2n_+-{1\over2}\over 2n_--{1\over2}} = {4n_+-1\over 4n_--1}.
\end{equation}
In either case we need $b_+/b_-$ to be rational, so that $b_+=p_+ b_*$ and $b_- = p_- b_*$. This special case is thus evidence that there is something very special about the  situation where $b_+/b_-$ is rational, as we know from the theorem (\ref{rational}). Then either
\begin{equation}
-i\omega = {2n_++{1\over2}\over p_+ b_*} ; \qquad -i\omega ={ 2n_-+{1\over2} \over p_- b_*};
\end{equation}
or
\begin{equation}
-i\omega = {2n_+-{1\over2} \over p_+ b_*} ; \qquad -i\omega  = {2n_--{1\over2} \over p_- b_*}.
\end{equation}
Now write
\begin{equation}
n_+ = m_+ + n p_+; \qquad n_- = m_- + n p_-;
\end{equation}
with $m_+< p_+$ and $m_-<p_-$. (While $n\in\{0,1,2,3,\dots\}$.) Then in the first case
\begin{equation}
\omega = i\left\{ {2m_++{1\over2} \over p_+ b_*}+ {n\over b_*}\right\}  =  i \left\{ {2m_-+{1\over2} \over p_- b_*} + {n\over b_*}\right\}
=  \omega_* +  {in\over b_*},
\end{equation}
while in the second case
\begin{equation}
\omega = i\left\{ {2m_+-{1\over2} \over p_+ b_*}+ {n\over b_*}\right\}  =  i \left\{ {2m_--{1\over2} \over p_- b_*} + {n\over b_*}\right\}
=  \omega_* +  {in\over b_*}.
\end{equation}
As we already know all this is only a consequence of our general result (\ref{rational}).

%-----------------------------------------------------------------------------------
\paragraph{Case $V_{0-}=0$}
%-----------------------------------------------------------------------------------

We can best analyze this situation by working directly with the exact wavefunction.
If  $V_{0-}=0$ then $\alpha_-=1/2$ and
\begin{equation}
\psi_-(0) = 1; \qquad
\psi'_-(0) =  + i \omega; \qquad
{\psi_-'(0)\over\psi_-(0)}  = + i \omega.
\end{equation}
The exact QNF boundary condition is then
\begin{equation}
i \omega = - {2\over b_+} \;  {\Gamma({\alpha_++i\omega b_+\over2}
+{3\over4})~ \Gamma({-\alpha_++i\omega b_+\over2} +{3\over4}) \over
\Gamma({\alpha_++i\omega b_+ \over2} + {1\over4})
~\Gamma({-\alpha_++i\omega b_+ \over2} + {1\over4})}.
\end{equation}
But this we can rewrite as
\begin{equation}
\omega
=   {2i \over b_+} \;
 { \Gamma({-\alpha_+-i\omega b_+ \over2} + {3\over4})  ~\Gamma({\alpha_+-i\omega b_+ \over2} + {3\over4})
 \over
 \Gamma({-\alpha_+-i\omega b_+\over2} +{1\over4})  ~\Gamma({\alpha_+-i\omega b_+\over2} +{1\over4})}
 \times
  {\cos(\pi\alpha_+) - \sin(-i\pi\omega b_+) \over \cos(\pi\alpha_+) + \sin(-i\pi\omega b_+) }.
\end{equation}
This certainly has pure imaginary roots. If we write $\omega = i |\omega|$
then asymptotically ($|\omega|\to\infty$) this becomes
\begin{equation}
1
=
  {\cos(\pi\alpha_+) - \sin(\pi|\omega| b_+) \over \cos(\pi\alpha_+) + \sin(\pi|\omega| b_+) },
\end{equation}
implying
\begin{equation}
\sin(\pi|\omega| b_+) = 0; \qquad \implies \qquad \pi|\omega| b_+ = n \pi; \qquad \implies \qquad  \omega = {i n\over b_+}.
\end{equation}
This agrees with our previous  calculation for $\alpha_-\approx 1/2$ and as expected gives the Schwarzschild result (\ref{SQNM}).

%-----------------------------------------------------------------------------------
\paragraph{Case $b_{-}\to 0$}
%-----------------------------------------------------------------------------------

This is best dealt with by using a Taylor expansion to show that
\begin{equation}
{\psi_-'(0)\over\psi_-(0)}  = + i \omega + V_{0-} b_- + O(b_-^2) .
\end{equation}
That is
\begin{equation}
\lim_{b_-\to 0} \; {\psi_-'(0)\over\psi_-(0)}  = + i \omega.
\end{equation}
The analysis then follows that for the case $V_{0-}=0$ above, and furthermore agrees with our previous  calculation for $\alpha_-\approx 1/2$.

 %-----------------------------------------------------------------------------------
\paragraph{Case $b_{-}\to\infty$}
%-----------------------------------------------------------------------------------

This is best dealt with by using the Stirling approximation together with a Taylor expansion to show that
\begin{equation}
{\psi_-'(0)\over\psi_-(0)}  =   i \sqrt{ \omega^2 - V_{0-}} + O(1/b_-^2).
\end{equation}
That is
\begin{equation}
\lim_{b_-\to \infty} \; {\psi_-'(0)\over\psi_-(0)}  = i \sqrt{\omega^2 - V_{0-}}.
\end{equation}
The exact QNF boundary condition is then
\begin{equation}
 i \sqrt{\omega^2 - V_{0-}}
= - {2\over b_+} \;  {\Gamma({\alpha_++i\omega b_+\over2}
+{3\over4}) ~\Gamma({-\alpha_++i\omega b_+\over2} +{3\over4}) \over
\Gamma({\alpha_++i\omega b_+ \over2} + {1\over4})
~\Gamma({-\alpha_++i\omega b_+ \over2} + {1\over4})}.
\end{equation}
But this we can rewrite as
\begin{eqnarray}
\sqrt{\omega^2 - V_{0-}}~~~~~~~~~~~~~~~~~~~~~~~~~~~~~~~~~~~~~~~~~~~~~~~~~~~~~~~~~~~~~~~~~~~~~~~~~~~~~~~~~~~~~~~~~~~~~~~~~~~~~~~~~\nonumber\\
=   {2i \over b_+} \;
 { \Gamma({-\alpha_+-i\omega b_+ \over2} + {3\over4})  ~\Gamma({\alpha_+-i\omega b_+ \over2} + {3\over4})
 \over
 \Gamma({-\alpha_+-i\omega b_+\over2} +{1\over4})  ~\Gamma({\alpha_+-i\omega b_+\over2} +{1\over4})}
 \times
  {\cos(\pi\alpha_+) - \sin(-i\pi\omega b_+) \over \cos(\pi\alpha_+) + \sin(-i\pi\omega b_+) }.~~~
\end{eqnarray}
If we write $\omega = i |\omega|$
then asymptotically, ($|\omega|\to\infty$, with $V_{0-}$ held fixed, implying that $V_{0-}$ effectively decouples from the calculation), this becomes
\begin{equation}
1
=
  {\cos(\pi\alpha_+) - \sin(\pi|\omega| b_+) \over \cos(\pi\alpha_+) + \sin(\pi|\omega| b_+) },
\end{equation}
implying
\begin{equation}
\sin(\pi|\omega| b_+) = 0; \qquad \implies \qquad \pi|\omega| b_+ = n \pi; \qquad \implies \qquad  \omega = {i n\over b_+}.
\end{equation}
The importance of this observation is that it indicates that for ``one sided'' potentials it is only the side for which the potential has exponential falloff that contributes to the ``gap''. This again confirms (\ref{SQNM}), since this is the case when cosmological horizon surface gravity goes to 0 and we get the result appropriate to a Schwarzschild black hole.

%-----------------------------------------------------------------------------------

%-----------------------------------------------------------------------------------

%----------------------------------------------------------------------------------------------------------------------------------------
\subsubsection{Systematic first-order perturbation
theory}
%----------------------------------------------------------------------------------------------------------------------------------------

Sometimes it is worthwhile to adopt a perturbative approach and to estimate shifts in the QNFs from some idealized pattern.
Define
\begin{equation}
b ={ b_++b_-\over2};  \qquad  \Delta = |b_+-b_-|;
\end{equation}
and rewrite the asymptotic QNF condition as
\begin{equation}
 \cos(-i\pi\omega\Delta) - \cos(-i2\pi\omega b) = 2  \;  \cos(\pi\alpha_+)\cos(\pi\alpha_-),
 \end{equation}
where we are implicitly holding $\alpha_\pm$ fixed.
When $\Delta=0$ we have previously seen that  the QNF are explicitly calculable with
\begin{equation}
\hat \omega_n  =  {i\cos^{-1}\left\{1 - 2  \;  \cos(\pi\alpha_+)\cos(\pi\alpha_-)\right\}\over2\pi b} + {in\over b}.
 \end{equation}
Can we now obtain an approximate formula for the  the QNF's when $\Delta\neq 0$?  It is a good strategy to define the dimensionless parameter $\epsilon$ by
\begin{equation}
\Delta = 2\; \epsilon \; b,
\end{equation}
and to set
\begin{equation}
\omega= \hat \omega + \delta\omega; \qquad   \delta\omega = O(\epsilon);
\end{equation}
so that the asymptotic QNF condition becomes
\begin{equation}
\cos(-i2\pi[\hat \omega +\delta\omega]\epsilon b) - \cos(-i2\pi[\hat \omega+\delta\omega] b) = 2  \;  \cos(\pi\alpha_+)\cos(\pi\alpha_-).
\end{equation}
Then to first order in $\epsilon$
\begin{equation}
\cos(-i2\pi\hat\omega \epsilon b) - \cos(-i2\pi[\hat\omega+\delta\omega] b) = 2  \;  \cos(\pi\alpha_+)\cos(\pi\alpha_-),
\end{equation}
where implicitly this approximation requires $\epsilon |\delta\omega|  b \ll 1$.
Subject to this condition we have
\begin{equation}
\cos(-i2\pi[\hat \omega+\delta\omega] b) = \cos(-i2\hat \pi\omega \epsilon b) - 2  \;  \cos(\pi\alpha_+)\cos(\pi\alpha_-),
\end{equation}
whence
\begin{equation}
-i2\pi[\hat \omega_n+\delta\omega_n] b = \cos^{-1} \left\{ \cos(-i2\pi\hat\omega_n \epsilon b) - 2  \;  \cos(\pi\alpha_+)\cos(\pi\alpha_-) \right\} + 2\pi n,
\end{equation}
so that
\begin{equation}
\hat \omega_n+\delta\omega_n  = i {\cos^{-1} \left\{ \cos(-i2\pi\hat \omega_n \epsilon b) - 2  \;  \cos(\pi\alpha_+)\cos(\pi\alpha_-) \right\}\over 2\pi b} + {in\over b}.
\end{equation}
But we know that the unperturbed QNFs satisfy $\hat\omega_n = \hat\omega_0 +{in/b}$, so we can also write this as
\begin{equation}
\delta\omega_n  = i {\cos^{-1} \left\{ \cos(-i2\pi\hat \omega_n \epsilon b) - 2  \;  \cos(\pi\alpha_+)\cos(\pi\alpha_-) \right\}\over 2\pi b} - \hat\omega_0.
\end{equation}
Using the definition of $\hat\omega_n$ this can now be cast in the form
\begin{equation}
\delta\omega_n  = i {\cos^{-1} \left\{ \cos(-i2\pi\hat \omega_n \epsilon b)  +  \cos(-i2\pi\hat \omega_n  b)  - 1 \right\}\over 2\pi b} - \hat\omega_0,
\end{equation}
or the slightly more suggestive
\begin{equation}
\delta\omega_n  = i {\cos^{-1} \left\{ \cos(-i2\pi\hat \omega_n b)  +  \cos(-i2\pi\hat \omega_n  \epsilon b)  - 1 \right\}\over 2\pi b} - \hat\omega_0,
\end{equation}
which can even be simplified to
\begin{equation}
\delta\omega_n  = i \; {\cos^{-1} \left\{ \cos(-i2\pi\hat \omega_0 b)  +  \cos(-i2\pi\hat \omega_n  \epsilon b)  - 1 \right\}\over 2\pi b} - \hat\omega_0.
\end{equation}
Note that this manifestly has the correct limit as $\epsilon\to0$.
Note that we have \emph{not} asserted or required that $\hat\omega_n \, \epsilon \, b \ll 1$, in fact when $n \gg 1/\epsilon$ this is typically \emph{not} true.
(Consequently $\cos(-i2\pi\hat\omega_n \epsilon b) $ is relatively unconstrained.)
Note furthermore that $Im(\delta\omega_n) \leq 1/b$.

%----------------------------------------------------------------------------------------------------------------------------------------
\subsubsection{Discussion}
%----------------------------------------------------------------------------------------------------------------------------------------

The key lesson to be learned from our semi-analytic model for the QNFs is that the commonly occurring ~$\omega_n =  \mathrm{offset} + i ~\mathrm{gap}\cdot n$ ~behaviour  is \emph{common but not universal}. Specifically, in our semi-analytic model the key point is whether or not the ratio ${b_+/ b_-} $ is a rational number.

\paragraph{Relation to the real physical Regge-Wheeler / Zerilli potentials.}
The very interesting question is the following: \emph{If the physically relevant results are by this approximate potentials recovered only qualitatively (so only the correct gap structure is recovered), does also the multi-family splitting transfer to the physically relevant case?}
This is because all the different families have the same gap structure, so the multi-family splitting might be only uninteresting artefact of our approximate model. The answer to this question comes with the following theorem:

\newtheorem{discontin}{Theorem}[section]
\begin{discontin}
Take for ~$b_{+}/b_{-}=p_{+}/p_{-}\in\mathbb{Q}$~ the gap structure as $\mathrm{gap}=1/b_{*}$, where $b_{*}= b_{\pm}/p_{\pm}$ and assume that the number of equi-spaced families is bounded with respect to $b_{+}, b_{-}$. Then the function $\omega(b_{+},b_{-})$ is discontinuous at every $\tilde b_{+},\tilde b_{-}$; ~~$\tilde b_{+}/\tilde b_{-}\in\mathbb{Q}$~ for infinite number of QNM frequencies.
\end{discontin}

\begin{proof}
Take any $\tilde b_{+}/\tilde b_{-}=p_{+}/p_{-}$ and arbitrary monotonically growing sequence of primes $P_{l}$. Take a sequence

\begin{equation}
(b_{l+}, b_{l-})\equiv \left(\frac{P_{l}}{P_{l}-1}\tilde b_{+},\frac{P_{l}-2}{P_{l}-1}\tilde b_{-}\right).
\end{equation}
It is obvious that

\begin{equation}
\lim_{l\to\infty}(b_{l+},b_{l-})=(\tilde b_{+},\tilde b_{-}).
\end{equation}
But then

\begin{equation}
b_{l*}=\frac{\tilde b_{*}~\mathrm{gcd}(P_{l}-2,~p_{+})}{P_{l}-1}.
\end{equation}
Here ``gcd'' means ``greatest common divisor'' and that means $\mathrm{gcd}(P_{l}-2,p_{+})$ is from the definition integer. Also it clearly holds that $1\leq \mathrm{gcd}(P_{l}-2,p_{+})\leq p_{+}$, hence it is upper and lower bounded and $\tilde b_{*}=\tilde b_{\pm}/p_{\pm}$. But then

\begin{equation}
\lim_{l\to\infty} \mathrm{gap}(b_{l+},b_{l-})=\lim_{l\to\infty} \frac{P_{l}-1}{\tilde b_{*}~\mathrm{gcd}(P_{l}-2,~p_{+})}=\infty.
\end{equation}

But if the number of families is bounded with respect to $b_{+},b_{-}$ then only finite number of modes can be obtained at $\tilde b_{+},\tilde b_{-}$ as a limit

\begin{equation}
\lim_{l\to\infty} ~\omega_{N(l)}(b_{l+},b_{l-})=\omega_{N}(\tilde b_{+},\tilde b_{-}).
\end{equation}
This proves the theorem.

\end{proof}

But this just means that if we do not want to end up with extremely strongly discontinuous function $\omega(b_{+},b_{-})$, we have to accept the fact that \emph{ the given $1/b_{*}$ gap dependence automatically implies unbounded number of QNM families (with respect to $b_{+}, b_{-}$), as it was in the case of our analytically solvable potential.}

We also suspect that it might be possible to generalize the model potential even further --- the ``art'' would lie in picking a piecewise potential that is still analytically solvable (at least for the highly damped modes) but which might be closer in spirit to the Regge--Wheeler (Zerilli) potential that is the key physical motivation for the current work. (Of course if we temporarily forget the black hole motivation, it may already be of some mathematical and physical interest that we have a nontrivial extension of the Eckart potential that is asymptotically exactly solvable --- one could in principle loop back to Eckart's original article and start asking questions about tunnelling probabilities for electrons encountering such piecewise Eckart barriers.)

\paragraph{Relation to the black hole thermodynamics.}
The last point we would like to briefly discuss is the relation of our results to the conjectured connection between the highly damped QNMs and the black hole thermodynamics. There are two basic conjectures: The first conjecture, due to Hod \cite{Hod}, gives some strong arguments supporting the idea that there is a connection between the real part of asymptotic QNMs and the quantum black hole area spacing. The basic principle underlying this conjecture is Bohr's correspondence principle. The conjecture was also used in the context of Loop Quantum Gravity by Dreyer \cite{Dreyer}. The second conjecture, such that it modifies Hod's original proposal, is due to Maggiore \cite{Maggiore}. It solves some controversies of Hod's conjecture, (see \cite{KonoplyaReview}), and gives the relation between black hole mass spectra and highly damped QNMs as:
\begin{equation}
\Delta M\approx\hbar\cdot\Delta\sqrt{Re(\omega_{n})^{2}+Im(\omega_{n})^{2}}\approx \hbar\cdot\Delta Im(\omega_{n}),~~~~~~n>>0.
\end{equation}
Here we naturally assume that $Re(\omega_{n})$ is bounded and $Im(\omega_{n})$ is growing to infinity. But it means, that if we are interested in mass quantum $\Delta_{min} M$:
\begin{equation}
\Delta_{min} M\approx \hbar\cdot\left(Im(\omega_{n})-Im(\omega_{n-1})\right),~~~~n>>0.
\end{equation}
The area quantum for the Schwarzschild black hole one obtains from the formula:
\begin{equation}
\Delta_{min} A=32\pi M~\Delta_{min} M=8\pi l_{p}^{2}=const.~,
\end{equation}
(where $l_{p}$ denotes the Planck length). In our case, if we
naively extrapolate Maggiore's conjecture to S-dS spacetime and the
ratio of surface gravities is rational, we obtain black hole mass
quanta as:\footnote{As the reader might have noted: we claim that
the expression for the gap in the spacing of the QNMs in the
asymptotic formula, (the one derived in this section), is an exact
result, ``unharmed'' by the fact that we used only approximate
potential.}
\begin{equation}
\Delta_{min} M\approx \hbar\cdot\mathrm{gap}(\omega_{n})=\hbar\cdot\frac{2}{g}~\mathrm{lcm}(\kappa_{+},\kappa_{-}),~~~~n>>0.
\end{equation}
Here all the symbols are defined as in the proof of the theorem \ref{rational} and in the equation \eqref{gap(.)}. Now there is a fascinating result by Choudhury and Padmanabhan \cite{SdSthermo}, that clearly ``fits'' very well into all these ideas. It says that in the S-dS spacetime there exists a coordinate system with globally defined temperature, (hence some kind of thermodynamic equilibrium), if and only if the ratio of surface gravities is rational. (For a very nice discussion of what might be the physical meaning of rational ratios of surface gravities and the role of highly damped QNMs as a potential source of information see again \cite{SdSthermo}.) The global temperature is given as \cite{SdSthermo}:
\begin{equation}
T_{sds}=\frac{\mathrm{hcf}(\kappa_{+},\kappa_{-})}{2\pi}.
\end{equation}
Then the gap spacing is, in our case:
\begin{equation}
\mathrm{gap}(\omega_{n})=\frac{2}{g}~\mathrm{lcm}(\kappa_{+},\kappa_{-})=\frac{2}{g}~p_{+}p_{-}~\mathrm{hcf}(\kappa_{+},\kappa_{-})=\frac{4\pi p_{+}p_{-}T_{sds}}{g}.
\end{equation}
This is a different result from the Schwarzschild spacetime, where it is
\begin{equation}
\mathrm{gap}(\omega_{n})=2\pi T_{s}~.
\end{equation}
But the fact that something special happens with both the S-dS spacetime thermodynamics and highly damped QNMs when the surface gravities have rational ratio is very interesting. Especially because it allows us to relate the constant gap in the QNM spacing and global spacetime temperature, and strikingly they both exist under the same condition. Moreover, in terms of Maggiore's conjecture, the existence of equispaced families is very interesting, because the gap in the QNM spacing multiplied by Planck constant is simply the quantum of black hole mass.
These considerations suggest that something very interesting is happening, but the topic needs clearly further exploration. What is written here has much more a character of ambiguous indications than of a well founded ideas, but it gives very exciting suggestions for future work. Even more because, as we will show in the next section, the link between rational ratio of surface gravities and equispaced families of highly damped QNMs is very generic for the multi-horizon black hole spacetimes.

\subsection{Conclusions}

The arguments that highly damped modes qualitatively depend on the tails of the potential were confirmed for Schwarzschild black hole by deriving the behavior of (\ref{HighSchw}). That was the only straight computational test that could have been done. After this the method was used to obtain new results for Schwarzschild-de Sitter black holes. One of the nice features of the semi-analytic model related to S-dS black holes is that a quite surprising amount of semi-analytic information can be extracted,  in terms of general qualitative results,  approximate results,  perturbative results, and reasonably explicit computations. Our results also might have important implications for the black hole thermodynamics.

\section{Monodromy results}

\subsection{Introduction}

In monodromy approaches one works with an analytic continuation into the
complex radial plane. They are part of a wider approach, the phase integral method (for details see \cite{Andersson3}). One has to choose branch cuts (as a part of analytic continuation), identify the singular points of the solutions in the complex plane, locate the Stokes and Anti-Stokes lines and calculate the monodromy around the singularities \cite{Andersson2, Andersson3, Berti, Cardoso, Choudhury, Green, Shanka1, KR, Lopez-Ortega, Motl1,
Motl2, Musiri, Natario, Shu-Shen}.

While many technical details differ, both
between the semi-analytic and monodromy approaches, and often among
various authors seeking to apply the monodromy technique, there is
widespread agreement that not only the semi-analy\-tic approximation, but
also the monodromy approaches lead to QNF master equations of the general
form:

\begin{equation}\label{master}
\sum^{N}_{A=1}C_{A}\exp\left(\sum^{H}_{i=1}\frac{Z_{Ai}\pi\omega}{\kappa_{i}}\right)=0.
\end{equation}
Here $\kappa_{i}$ is the surface gravity of the $i$-th horizon, $H$ is
the number of horizons, the matrix $Z_{Ai}$ always has rational
entries (and quite often is integer-valued). The physics contained
in the master equation is invariant under substitutions of the form
$Z_{Ai}\to Z_{Ai} + (1, . . . , 1)^{T}_{A}h_{i}$, where the $h_{i}$
are arbitrary rational numbers. Either $\sum_{i}Z_{Ai} = 0$, or it
can without loss of generality be made zero. Furthermore $N$ is some
reasonably small positive integer. (In fact $N\leq 2H + 1$ in all
situations we have encountered, and typically $N > H$.) The $C_{A}$
are a collection of coefficients that are often but not always
integers, though in all known cases they are at least real. Finally
in almost all known cases the rectangular $N\times H$ matrix $Z_{Ai}$ has
rank $H$, and the QNF master equation is almost always irreducible
(that is, non-factorizable). We shall first demonstrate that all
known master equations (whether based on semi-analytic or monodromy
techniques) can be cast into this form. Then we will generalize the
results for rational/irrational surface gravities ratios from the
special cases obtained within the approximations by the analytically
solvable potentials to this more general case \eqref{master}.

\subsection{Particular results}

\subsubsection{Survey of the monodromy results}
%--------------------------------------------------------------------------------------------------------------------------------------------

\paragraph{One horizon:}
%------------------------------
In the one-horizon situation there is general agreement that the
relevant master equation is
\begin{equation}
\exp\left({\pi\omega\over\kappa}\right) + 1 + 2\cos(\pi j) = 0.
\end{equation}
Unfortunately there is distressingly little agreement over the
precise status of the parameter $j$. References~\cite{Andersson2, Berti, Motl1, Motl2,
Musiri, Natario} assert that this is the spin of the
perturbation under consideration, but with some disagreement as to
whether this applies to all spins and all dimensions.  In contrast
in reference~\cite{Shanka1} a particular model for the spacetime metric
is adopted, and in terms of the parameters describing this model
these authors take
\begin{equation}
j = {qd\over2}-1.
\end{equation}
Here, (and also later in this section), the symbol $d$ denotes spacetime dimension, and $q$ is a parameter related to the power with which the general $d$-dimensional black hole metric coefficients
\begin{equation}
-f(r)dt^{2}+\frac{dr^{2}}{g(r)}+\rho^{2}(r)d\Omega_{d}^{2}
\end{equation}
behave close to the singularity, (see \cite{Shanka1}).
Reference~\cite{Lopez-Ortega} asserts that for spin 1 perturbations
\begin{equation}
j = {2(d-3)\over d-2}.
\end{equation}
Be this as it may, there is universal agreement on the \emph{form}
of the QNF master condition, and it  is automatically of the the
form of equation (\ref{master}), with $H=1$ and $N=2$ terms. The
vector $C_A$ and the matrix $Z_{Ai}$ are
\begin{equation}
C_{A}= \left[
\begin{array}{c}
+1  \\ 1+2\cos(\pi j)
\end{array}
\right]; \qquad Z_{A1}= \left[
\begin{array}{r}
+1  \\ 0
\end{array}
\right].
\end{equation}
By multiplying through by $\exp(-\pi\omega/(2\kappa))$ we can
re-cast the QNF condition as
\begin{equation}
\exp\left({\pi\omega\over2\kappa}\right) + \{1 + 2\cos(\pi j)\}
\exp\left(-{\pi\omega\over2\kappa}\right)  = 0.
\end{equation}
This now corresponds to
\begin{equation}
C_{A}= \left[
\begin{array}{c}
+1  \\ 1+2\cos(\pi j)
\end{array}
\right]; \qquad Z_{A1}= \left[
\begin{array}{r}
+1/2  \\ -1/2
\end{array}
\right],
\end{equation}
and in this form we have $\sum_A Z_{Ai}=0$. (This is one of rather
few cases where it is convenient to take the $Z_{Ai}$ to be
rational-valued rather than integer-valued.)

\paragraph{Two horizons:}
%-------------------------------
For two-horizon situations the analysis is slightly different for
Schwarzschild--de~Sitter spacetimes (Kottler spacetimes) versus
Reissner--Nord\-str\"om spacetimes.
\begin{itemize}
%%%%%%%%%%%%%%%%%%%%%%%%%%%%%%%%%%%%%%%%%%%%%%%%%%%%
\item
For {\bf Schwarzschild--de~Sitter} spacetimes there is general
agreement that the relevant master equation for the QNFs is
%%%%%%%%%%%%%%%%%%%%%%%%%%%%%%%%%%%%%%%%%%%%%%%%%%%%
\begin{equation}
\label{E:SdS} \{1 + 2\cos(\pi j)
\}\cosh\left({\pi\omega\over\kappa_+} +
{\pi\omega\over\kappa_-}\right) +
\cosh\left({\pi\omega\over\kappa_+}- {\pi\omega\over\kappa_-}\right)
=  0.
\end{equation}
We shall again adopt conventions such that $\kappa_\pm$ are both
positive. Again, there is unfortunately distressingly little
agreement over the precise status of the parameter $j$.
References~\cite{Andersson2, Berti, Motl1, Motl2, Musiri, Natario} assert
that this is the spin of the perturbation under consideration, but
with some disagreement as to whether this applies to all spins and
all dimensions.  In contrast in reference~\cite{Shanka2} a particular
model for the spacetime metric is again adopted, and in terms of the
parameters describing this model they take
\begin{equation}
j = {qd\over2}-1.
\end{equation}
One still has to perform a number of trigonometric transformations
to turn the quoted result of reference~\cite{Shanka2} for $d\neq 5$
\begin{equation}
\tanh\left({\pi\omega\over\kappa_+}\right) \tanh\left(
{\pi\omega\over\kappa_-}\right) = {2\over\tan^2(\pi j/2)-1},
\end{equation}
into the equivalent form (\ref{E:SdS}) above. For $d=5$ the authors
of~\cite{Shu-Shen} assert the equivalent of
\begin{equation}
\label{E:SdS5} \{1 + 2\cos(\pi j)
\}\sinh\left({\pi\omega\over\kappa_+} +
{\pi\omega\over\kappa_-}\right) +
\sinh\left({\pi\omega\over\kappa_+}- {\pi\omega\over\kappa_-}\right)
=  0.
\end{equation}
Reference~\cite{Lopez-Ortega} again asserts that for spin 1
perturbations
\begin{equation}
j = {2(d-3)\over d-2}.
\end{equation}
Be this as it may, there is again universal agreement on the
\emph{form} of the QNF master condition, and converting hyperbolic
functions into exponentials, it can be transformed into  the form of
equation (\ref{master}), with $H=2$ and $N=4$ terms. The vector
$C_A$ and matrix $Z_{Ai}$ are
\begin{equation}
C_{A}= \left[
\begin{array}{c}
1+2\cos(\pi j) \\
+1 \\
+1\\
1+2\cos(\pi j) \\
\end{array}
\right]; \qquad Z_{Ai}= \left[
\begin{array}{rr}
+1  & +1 \\
+1 & -1\\
-1 & +1\\
-1 & -1 \\
\end{array}
\right].
\end{equation}
Note that we explicitly have $\sum_{A}Z_{Ai}=0$. There are two exceptional
cases:
\begin{itemize}
\item
If $j = 2m+1$ with $m\in Z$ then
\begin{equation}
C_{A}= \left[
\begin{array}{c}
-1\\
+1 \\
+1\\
-1\\
\end{array}
\right]; \qquad Z_{Ai}= \left[
\begin{array}{rr}
+1  & +1 \\
+1 & -1\\
-1 & +1\\
-1 & -1 \\
\end{array}
\right].
\end{equation}
In this situation the QNF master equation factorizes
\begin{equation}\label{E:SdSOne}
\sinh\left({\pi\omega\over\kappa_+} \right) \sinh\left(
{\pi\omega\over\kappa_-}\right) = 0.
\end{equation}
This appears to be the physically relevant case for spin 1
particles. The relevant QNF spectrum is that of equation
(\ref{sepEq}).

\item
If $\cos(\pi j) = - {1\over2}$, which does not appear to be a
physically relevant situation but serves to illustrate potential
mathematical pathologies, then
\begin{equation}
C_{A}= \left[
\begin{array}{r}
0\\
+1 \\
+1\\
0\\
\end{array}
\right]; \qquad Z_{Ai}= \left[
\begin{array}{rr}
+1  & +1 \\
+1 & -1\\
-1 & +1\\
-1 & -1 \\
\end{array}
\right].
\end{equation}
But in this situation the top row and bottom row do not contribute
to the QNF master equation and one might as well delete them. That
is, one might as well write
\begin{equation}
C_{A}= \left[
\begin{array}{r}
+1 \\
+1\\
\end{array}
\right]; \qquad Z_{Ai}= \left[
\begin{array}{rr}
+1 & -1\\
-1 & +1\\
\end{array}
\right].
\end{equation}
This is a situation (albeit unphysical) where the matrix $Z_{Ai}$
does not have maximal rank. The QNF master equation degenerates to
\begin{equation}
\sinh\left({\pi\omega\over\kappa_+}- {\pi\omega\over\kappa_-}\right)
=  0.
\end{equation}
In this situation the QNF spectrum is
\begin{equation}
\omega_n = { i n \kappa_+ \kappa_- \over |\kappa_+ - \kappa_-|},
\end{equation}
with no restriction on the relative values of $\kappa_\pm$. This
situation is however clearly non-generic (and outright unphysical).

\end{itemize}

%%%%%%%%%%%%%%%%%%%%%%%%%%%%%%%%%%%%%%%%%%%%%%%%%%%%
\item
For {\bf Reissner--Nordstr\"om} spacetime one has~\cite{Motl2}
%%%%%%%%%%%%%%%%%%%%%%%%%%%%%%%%%%%%%%%%%%%%%%%%%%%%
\begin{equation}
\label{E:RN} \exp\left({2\pi\omega\over\kappa_+}\right) +2 \{1 +
\cos(\pi j) \}\exp\left(- {2\pi\omega\over\kappa_-}\right) + \{1 +
2\cos(\pi j) \}=  0,
\end{equation}
where $\kappa_+$ is the surface gravity of the outer horizon and
$\kappa_-$ is the surface gravity of the inner horizon. There is
again some disagreement on the status of the parameter $j$.
Reference~\cite{Motl2} now takes $j={1\over3}$ for spin 0, and
$j={5\over3}$ for spins 1 and 2 (in any dimension).
Reference~\cite{Natario} asserts that for general dimension
\begin{equation}
j=\frac{d-3}{2d-5} \qquad \hbox{for spin 0, 2, and} \qquad
j=\frac{3d-7}{2d-5}\qquad \hbox{ for spin 1}.
\end{equation}
Be this as it may, there is universal agreement on the \emph{form}
of the QNF master condition, and it  is automatically of the form of
equation (\ref{master}), with $H=2$ and $N=3$ terms. The vector
$C_A$ and matrix $Z_{Ai}$ are
\begin{equation}
C_{A}= \left[
\begin{array}{c}
+1   \\
2\{1+\cos(\pi j)\} \\
1+2\cos(\pi j)\\
\end{array}
\right]; \qquad Z_{Ai}= \left[
\begin{array}{rr}
+2  & 0 \\
0 & -2\\
0 & 0\\
\end{array}
\right].
\end{equation}
If we multiply through by a suitable factor then we can write the
QNF condition in the equivalent form
\begin{eqnarray}
\label{E:RN3} &&\exp\left({2\pi\omega\over3\kappa_+} +
{\pi\omega\over3\kappa_-}\right) +2 \{1 + \cos(\pi j) \}\exp\left(
- {\pi\omega\over3\kappa_+}- {2\pi\omega\over3\kappa_-}\right)
\nonumber
\\
&& \qquad + \{1 + 2\cos(\pi j) \}  \exp\left( -
{\pi\omega\over3\kappa_+} + {\pi\omega\over3\kappa_-} \right)=  0,
\end{eqnarray}
This corresponds to
\begin{equation}
C_{A}= \left[
\begin{array}{c}
+1   \\
2\{1+\cos(\pi j)\} \\
1+2\cos(\pi j)\\
\end{array}
\right]; \qquad Z_{Ai}= \left[
\begin{array}{rr}
+2/3  & 1/3 \\
-1/3 & -2/3\\
-1/3& 1/3\\
\end{array}
\right].
\end{equation}
In this form we now explicitly have $\sum_A Z_{Ai}=0$. (This is one
of rather few cases where it is convenient to take the $Z_{Ai}$ to
be rational-valued rather than integer-valued.) Returning to the
original form in equation (\ref{E:RN}), there are two exceptional
cases:
\begin{itemize}
\item
If $j = 2m+1$ with $m\in Z$, (this does not appear to be a
physically relevant situation but again this serves to illustrate
the possible mathematical pathologies one might encounter), then
\begin{equation}
C_{A}= \left[
\begin{array}{r}
+1\\
0 \\
-1\\
\end{array}
\right]; \qquad Z_{Ai}= \left[
\begin{array}{rr}
+2  & 0 \\
0 & -2\\
0 & 0\\
\end{array}
\right].
\end{equation}
But then (without loss of information) one might as well eliminate
the second row, to obtain
\begin{equation}
C_{A}= \left[
\begin{array}{r}
+1\\
-1\\
\end{array}
\right]; \qquad Z_{Ai}= \left[
\begin{array}{rr}
+2  & 0 \\
0 & 0\\
\end{array}
\right].
\end{equation}
Furthermore, since $\kappa_-$ now decouples,  we might as well
eliminate the second column, to obtain
\begin{equation}
C_{A}= \left[
\begin{array}{r}
+1\\
-1\\
\end{array}
\right]; \qquad Z_{Ai}= \left[
\begin{array}{rr}
2  \\
0 \\
\end{array}
\right].
\end{equation}
The QNF master equation then specializes to
\begin{equation}
\label{E:RN-sp1} \exp\left({2\pi\omega\over\kappa_+}\right) -1 =  0.
\end{equation}

\item
If $\cos(\pi j) = - {1\over2}$, which does not appear to be a
physically relevant situation but serves to illustrate potential
mathematical pathologies, then
\begin{equation}
C_{A}= \left[
\begin{array}{r}
+1 \\
+1\\
0\\
\end{array}
\right]; \qquad \qquad Z_{Ai}= \left[
\begin{array}{rr}
+2  & 0 \\
0 & -2\\
0 & 0\\
\end{array}
\right].
\end{equation}
But in this situation the  bottom row does not contribute to the QNF
master equation and one might as well delete it. That is, one might
as well write
\begin{equation}
C_{A}= \left[
\begin{array}{r}
+1 \\
+1\\
\end{array}
\right]; \qquad \qquad Z_{Ai}= \left[
\begin{array}{rr}
+2  & 0 \\
0 & -2\\
\end{array}
\right].
\end{equation}
We can rearrange the terms in the master equation to have the QNF
master equation specialize to
\begin{equation}
\label{E:RN-sp2} \exp\left({2\pi\omega\over\kappa_+} +
{2\pi\omega\over\kappa_-}\right) +1 =  0.
\end{equation}
This corresponds to
\begin{equation}
C_{A}= \left[
\begin{array}{r}
+1 \\
+1\\
\end{array}
\right]; \qquad \qquad Z_{Ai}= \left[
\begin{array}{rr}
+2  & +2 \\
0 & 0 \\
\end{array}
\right].
\end{equation}
Note that in this exceptional case $Z_{Ai}$ is not of maximal rank.
In this situation the QNF spectrum is
\begin{equation}
\omega_n = {(2 n+1)i \;\kappa_+ \kappa_- \over \kappa_+ + \kappa_-},
\end{equation}
with no restriction on the relative values of $\kappa_\pm$. This
situation is however clearly non-generic (and outright unphysical).

\end{itemize}

\end{itemize}

\paragraph{Three horizons:}
%----------------------------------
For three horizons the natural example to consider is that of
Reissner--Nordstr\"om--de~Sitter spacetime. References~\cite{Natario,
Shu-Shen} agree that (for $d\neq 5$)
\begin{eqnarray}
\label{E:RNdS} &&\cosh\left({\pi\omega\over\kappa_+}-
{\pi\omega\over\kappa_{C}}\right) + \{1 + \cos(\pi j) \}
\cosh\left({\pi\omega\over\kappa_+}+ {\pi\omega\over\kappa_{C}}\right)
\nonumber\\
&& \qquad\qquad + 2 \{1 + \cos(\pi j) \}
\cosh\left({2\pi\omega\over\kappa_-}+{\pi\omega\over\kappa_+}+
{\pi\omega\over\kappa_{C}}\right) =  0.
\end{eqnarray}
Here $\kappa_\pm$ refer to the inner and outer horizons of the
central Riessner--Nordstr\"om black hole, while $\kappa_{C}$ is now
the surface gravity of the cosmological horizon. All these surface
gravities are taken positive.  In contrast for $d=5$ one has
\begin{eqnarray}
\label{E:RNdS5} &&\sinh\left({\pi\omega\over\kappa_+}-
{\pi\omega\over\kappa_{C}}\right) + \{1 + \cos(\pi j) \}
\sinh\left({\pi\omega\over\kappa_+}+ {\pi\omega\over\kappa_{C}}\right)
\nonumber\\
&& \qquad\qquad + 2 \{1 + \cos(\pi j) \}
\sinh\left({2\pi\omega\over\kappa_-}+{\pi\omega\over\kappa_+}+
{\pi\omega\over\kappa_{C}}\right) =  0,
\end{eqnarray}
Again
\begin{equation}
j=\frac{d-3}{2d-5} \qquad \hbox{for spin 0, 2, and} \qquad
j=\frac{3d-7}{2d-5}\qquad \hbox{ for spin 1}.
\end{equation}
There is universal agreement on the \emph{form} of the QNF master
condition, and converting hyperbolic functions into exponentials, it
can be transformed into  the form of equation (\ref{master}), with
$H=3$ and $N=6$ terms. The vector $C_A$ and matrix $Z_{Ai}$ are
\begin{equation}
C_{A}= \left[
\begin{array}{c}
+1\\
1+\cos(\pi j)  \\
2\{1+\cos(\pi j)\} \\
\pm 2\{1+\cos(\pi j)\} \\
\pm\{1+\cos(\pi j) \}\\
\pm 1\\
\end{array}
\right]; \qquad Z_{Ai}= \left[
\begin{array}{rrr}
+1  & -1 & 0 \\
+1 & +1 & 0\\
+1 & +1 & +2\\
-1 & -1 & -2\\
-1 & -1 & 0\\
-1 & +1 & 0\\
\end{array}
\right].
\end{equation}
Generically, $Z_{Ai}$ has maximal rank $H=3$. Note that we
explicitly have $\sum_A Z_{Ai}=0$.

The only exceptional case is $\cos(\pi j)= -1$ in which case
\begin{equation}
C_{A}= \left[
\begin{array}{c}
+1\\
0  \\
0\\
0 \\
0 \\
\pm 1\\
\end{array}
\right]; \qquad Z_{Ai}= \left[
\begin{array}{rrr}
+1  & -1 & 0 \\
+1 & +1 & 0\\
+1 & +1 & +2\\
-1 & -1 & -2\\
-1 & -1 & 0\\
-1 & +1 & 0\\
\end{array}
\right].
\end{equation}
But then the $2^{nd}$ to $5^{th}$ rows decouple and may as well be
removed, yielding
\begin{equation}
C_{A}= \left[
\begin{array}{c}
+1\\
\pm 1\\
\end{array}
\right]; \qquad Z_{Ai}= \left[
\begin{array}{rrr}
+1  & -1 & 0 \\
-1 & +1 & 0\\
\end{array}
\right].
\end{equation}
The $3^{rd}$ column, corresponding to $\kappa_-$, now decouples and
may as well be removed, yielding
\begin{equation}
C_{A}= \left[
\begin{array}{c}
+1\\
\pm 1\\
\end{array}
\right]; \qquad Z_{Ai}= \left[
\begin{array}{rr}
+1  & -1  \\
-1 & +1 \\
\end{array}
\right].
\end{equation}
Note that in this exceptional case $Z_{Ai}$ is not of maximal rank,
and the QNF master equation degenerates to
\begin{equation}
\cosh\left({\pi\omega\over\kappa_+}- {\pi\omega\over\kappa_{C}}\right)
= 0, \qquad \hbox{or} \qquad \sinh\left({\pi\omega\over\kappa_+}-
{\pi\omega\over\kappa_{C}}\right) = 0,
\end{equation}
respectively. In this situation the QNF spectrum is
\begin{equation}
\omega_n = {(2n+1) i \; \kappa_+ \kappa_{C} \over 2 |\kappa_+ -
\kappa_{C}|}, \quad \hbox{or} \qquad \omega_n = { i n \kappa_+
\kappa_{C} \over |\kappa_+ - \kappa_{C}|},
\end{equation}
respectively, with no restriction on the relative values of
$\kappa_\pm$. This situation is however clearly non-generic (and
outright unphysical).

\subsubsection{Rewriting the analytically solvable potentials results into our general form}

\paragraph{One horizon}
For highly damped QNFs the master equation is derived in previous section and also in
references~\cite{PRD, Granada}, in a form equivalent to
\begin{equation}
\sinh\left({\pi\omega\over\kappa}\right)=0.
\end{equation}

This means

\begin{equation}
C_{A}= \left[
\begin{array}{c}
1 \\
-1
\end{array} \right],
\end{equation}

and

\begin{equation}
Z_{Ai}= \left[
\begin{array}{c}
1 \\
-1
\end{array} \right].
\end{equation}

\paragraph{Two horizons}
For highly damped QNFs the master equation is derived in previous section and also in
references~\cite{PRD, JHEP1, Granada} in a form equivalent to
\begin{equation}
\cosh\left({\pi\omega\over\kappa_+} +
{\pi\omega\over\kappa_-}\right) -
\cosh\left({\pi\omega\over\kappa_+}- {\pi\omega\over\kappa_-}\right)
+ 2  \;  \cos(\pi\alpha_+)\cos(\pi\alpha_-) = 0.
\end{equation}
This means

\begin{equation}
C_{A}= \left[
\begin{array}{c}
1 \\
1 \\
-1 \\
-1 \\
2\cos(\pi\alpha_+)\cos(\pi\alpha_-)
\end{array} \right],
\end{equation}

and

\begin{equation}
Z_{Ai}= \left[
\begin{array}{rr}
1 & 1 \\
-1 & -1 \\
1 & -1 \\
-1 & 1 \\
0 & 0
\end{array} \right].
\end{equation}
Note that $\alpha_\pm = \sqrt{{1\over4} - {V_{0\pm}\over\kappa_\pm^2}}$ and hence $\alpha_\pm = {1\over2}$ is a physically degenerate case
corresponding either to $V_{0\pm}=0$ (in which case the
corresponding $\kappa_\pm$ is physically and mathematically
meaningless), or $\kappa_\pm = \infty$, (in which case the QNF
master equation is vacuous).  Either of these situations is
unphysical so one must have  $\alpha_\pm \neq {1\over2}$.  This QNF
condition above is irreducible (non-factorizable) unless $\alpha_\pm
= m+{1\over2}$ with $m\in Z$. This occurs when
\begin{equation}
V_{0\pm} = - m(m+1) \; \kappa_\pm^2,
\end{equation}
and in this exceptional situation the QNF master equation factorizes
to
\begin{equation}
\sinh\left({\pi\omega\over\kappa_+} \right) \sinh\left(
{\pi\omega\over\kappa_-}\right) = 0,
\end{equation}
and hence takes the same shape as \eqref{E:SdSOne}.

\subsection{Rational ratios of surface gravities}\label{ratsec}

\def\m{{\tilde m}}

\subsubsection{From master equation to
polynomial}\label{S:polynomial}
%-------------------------------------------------------------------------------------------------------------------------------------------

First let us suppose that the ratios $R_{ij} = \kappa_i/\kappa_j$
are all rational numbers. This is not as significant a constraint as
one might initially think. In particular, since the rationals are
dense in the reals one can always with arbitrarily high accuracy
make an approximation to this effect.  Furthermore since floating
point numbers are essentially a subset of the rationals, all
numerical investigations implicitly make such an assumption, and all
numerical experiments should be interpreted with this point kept
firmly in mind.

Provided that the ratios $R_{ij} = \kappa_i/\kappa_j$ are all
rational numbers, it follows that there is a constant $\kappa_*$ and
a collection of relatively prime integers $m_i$ such that
\begin{equation}
\kappa_i = {\kappa_*\over m_i}.
\end{equation}
The QNF master equation then becomes
\begin{equation}
\sum_{A=1}^N   C_A \; \exp\left(  \sum_{i=1}^H  Z_{Ai}  \; m_i \;
{\pi \omega\over \kappa_*} \right) = 0, \label{E:master2-1}
\end{equation}
Now define $z=\exp(\pi\omega/\kappa_*)$,  and define a new set of
integers $\m_A  =  \sum_{i=1}^H  Z_{Ai}  \; m_i $. (There is no
guarantee or requirement that the  $\m_A $ be relatively prime, and
some of the special cases we had to consider in the previous section and in
reference~\cite{JHEP1} ultimately depend on this observation.)
Then
\begin{equation}
\sum_{A=1}^N   C_A \; z^{\m_A} = 0. \label{E:master2-2}
\end{equation}
This is (at present) a Laurent polynomial, as some exponents may be
(and typically are) negative. Multiplying through by $z^{-\tilde
m_\mathrm{min}}$ converts this to a regular polynomial with a
nonzero constant $z^0$ term and with degree
\begin{equation}
D = \tilde m_\mathrm{max} - \tilde m_\mathrm{min}.
\end{equation}
If we write $\bar m_A = \m_A - \tilde m_\mathrm{min}$ then the
relevant regular polynomial is
\begin{equation}
\sum_{A=1}^N   C_A \; z^{\bar m_A} = 0. \label{E:master2-3}
\end{equation}
Note that the polynomial is typically ``sparse'' --- the number of
terms $N$ is small (typically $N\leq 2H+1$) but the degree $D$ can
easily be arbitrarily large. There are at most $D$ distinct roots
for the polynomial $z_a$, and the \emph{general} solution of the QNF
condition is
\begin{equation}
\omega_{a,n} = {\kappa_*\ln(z_a)\over\pi} + {2in \kappa_*}; \qquad
a\in\{1,...,D\}; \qquad n\in\{0,1,2,3,\dots\}.
\end{equation}
If the $\bar m_A$ are not relatively prime, define a degeneracy
factor $g= \mathrm{hcf}\{\bar m_A\}$. Then the roots will fall into
$D/g$ classes where the $g$ degenerate members of each class differ
only by the various $g$-th roots of unity. In this situation we can
somewhat simplify the above QNF spectrum to yield
\begin{equation}\label{finformula}
\omega_{a,n} = {\kappa_*\ln(z_a)\over\pi} + {2in \kappa_*\over g};
\qquad a\in\{1,...,D/g\}; \qquad n\in\{0,1,2,3,\dots\}.
\end{equation}
We again emphasize that behaviour of this sort certainly does occur
in practice. There is no guarantee or requirement that the  $\bar
m_A $ be relatively prime, and some of the special cases we had to
consider in the previous section and in reference~\cite{JHEP1} ultimately depend on this
observation.

So let us summarize what we just proved in a theorem:
\newtheorem{rational2}[section]{Theorem}
\begin{rational2}\label{rational2}
Take quasi-normal frequencies to be given by the equation \eqref{E:master2-1} and, take the ratio of arbitrary pairs of surface gravities to be rational. Then the quasi-normal frequencies are given by the formula \eqref{finformula}, where $z_{a}$ are, (in general), $D/g$ solutions of the equation \eqref{E:master2-3}.
\end{rational2}

There is a (slightly) weaker condition that also leads to polynomial
master equations and the associated families of QNFs. Suppose that
we know that the ratios
\begin{equation}
R_{AB} = {\sum_{i=1}^H Z_{Ai}/\kappa_i \over  \sum_{i=1}^H
Z_{Bi}/\kappa_i } \;\; \in Q
\end{equation}
are always rational numbers. Then it follows that there is a set of
integers $\hat m_A$ such that
\begin{equation}
\sum_{i=1}^H {Z_{Ai}\over\kappa_i} = {\hat m_A\over \bar \kappa_*},
\end{equation}
where the $\hat m_A$ are all relatively prime. (Note $\bar \kappa_*$
does not have to equal $\kappa_*$). This is actually a (slightly)
weaker condition than $R_{ij} = \kappa_i/\kappa_j$ being rational,
since it is only if $Z_{Ai}$ is of rank $H$ that one can derive
$R_{ij} \in Q$ from $R_{AB}\in Q$. Assuming  $R_{AB}\in Q$ the QNF
master equation becomes
\begin{equation}
\sum_{A=1}^N   C_A \; \exp\left(  \hat m_A \; {\pi \omega\over \bar
\kappa_*} \right) = 0. \label{E:master2-1b}
\end{equation}
This can now be converted into a polynomial in exactly the same
manner as previously, leading to families of QNFs as above. Provided
both $R_{AB}$ and $R_{ij}$ are rational we can identify $\bar
\kappa_* = \kappa_*/g$.

%--------------------------------------------------------------------------------------------------------------------------------------------
\subsubsection{Factorizability}\label{S:factor}
%---------------------------------------------------------------------------------------------------------------------------------------------

Now it is mathematically conceivable that in certain circumstances
the master equation might factorize into a product over two disjoint
sets of horizons
\begin{equation}
\left[\sum_{A=1}^{N_1}   C_{1A} \; \exp\left(  \sum_{i=1}^{H_1}
{Z_{1Ai} \; \pi \omega\over \kappa_{1i}} \right)\right]\;
\left[\sum_{A=1}^{N_2}   C_{2A} \; \exp\left(  \sum_{i=1}^{H_2}
{Z_{2Ai} \; \pi \omega\over \kappa_{2i}} \right)\right] =  0.
\label{E:master3}
\end{equation}
Physically one might in fact expect this if the horizons indexed by
$i\in\{1,\dots,H_1\}$ are very remote (in physical distance) from
the other horizons indexed by $i\in\{1,\dots,H_2\}$.  If such a
factorization were to occur then the QNFs would fall into two
completely disjoint classes,  being independently and disjointly
determined by  these two classes of  horizon.

\subsection{Irrational ratios of surface gravities}\label{iratsec}

%---------------------------------------------------------------------------------------------------------------------------------------------
\def\gap{{\;\mathrm{gap}}}
\def\B{{\mathcal{B}}}

We now wish to work ``backwards'' to see if the existence of a
family of equi-spaced QNFs can lead to constraints on the ratios
$R_{ij} = \kappa_i/\kappa_j$. Such an analysis has already been
performed for the specific class of QNF master equations arising
from semi-analytic techniques, and we now intend to generalize the
argument to the generic class of QNF master equations presented in
equation (\ref{master}). Let us therefore assume the existence of
at least one ``family'' of QNFs of the form:
\begin{equation}
\omega_{n} = \omega_0 + in \gap; \qquad n\in\{0,1,2,3,\dots\}.
\label{E:family2}
\end{equation}
Then we are asserting
\begin{equation}
\sum_{A=1}^N   C_A \; \exp\left(  \sum_{i=1}^H  {Z_{Ai} \; \pi
(\omega_0  + in \gap) \over \kappa_i} \right) = 0;  \qquad
n\in\{0,1,2,3,\dots\}. \label{E:master4}
\end{equation}
That is
\begin{eqnarray}
\sum_{A=1}^N   \left\{ C_A \; \exp\left(  \sum_{i=1}^H  {Z_{Ai} \;
\pi \omega_0 \over \kappa_i} \right) \right\} \exp\left(  i n \pi
\gap \sum_{i=1}^H  {Z_{Ai} \over \kappa_i} \right)= 0;  \qquad\\
\bigskip\nonumber\\ \bigskip
n\in\{0,1,2,3,\dots\}.~~~~~~~~~~~~~~~~~~~~~~~~~~~~~~~~~~~~~~~~~~~~~~~~~~~~~~~~~~~~~~~~~~~~~~~~~\nonumber \label{E:master4-2}
\end{eqnarray}
We can rewrite this as
\begin{equation}
\sum_{A=1}^N  D_A \exp\left(  2\pi i n J_A \right)= 0;  \qquad
n\in\{0,1,2,3,\dots\}. \label{E:xxx}
\end{equation}
A priori, there is no particular reason to expect either the $D_A$
or the $J_A$ to be real.

%--------------------------
\subsubsection{Case 1}
%--------------------------
One specific solution to the above collection of constraints is
\begin{equation}
\sum_{A=1}^N  D_A  = 0;  \qquad  \exp(2\pi i J_A)  =  r.
\label{E:master4-sp1}
\end{equation}
Furthermore, \emph{as long as no proper subset of the $D_A$'s sums
to zero}, we assert that this is the only solution. To see this let
us define
\begin{equation}
\lambda_A =  \exp(2\pi i J_A);   \qquad M_{AB} = (\lambda_A)^{B-1};
\qquad A,B \in \{1,2,3,\dots,N\}.
\end{equation}
Then $M_{AB}$ is a square $N\times N$ Vandermonde matrix, and then
equation (\ref{E:xxx}) implies
\begin{equation}
\sum_{A=1}^N  D_A  M_{AB} = 0,
\end{equation}
whence $\det(M_{AB}) =0$. But from the known form of the Vandermonde
determinant we have
\begin{equation}
\det(M_{AB}) = \prod_{A>B} (\lambda_A - \lambda_B) = 0,
\end{equation}
implying that at least two of the $\lambda_A$ are equal. Without
loss of generality we can shuffle the $\lambda_A$'s so that the two
which are guaranteed to be equal are $\lambda_1$ and $\lambda_2$.
Then equation  (\ref{E:xxx}) implies
\begin{equation}
(D_1+D_2) \lambda_1^{B-1} +    \sum_{A=3}^N  D_A   (\lambda_A)^{B-1}
= 0;  \qquad B\in\{1,2,3,\dots, N-1\}.
\end{equation}
But by hypothesis $D_1+D_2\neq 0$, so this equation can be rewritten
in terms of a non-trivial reduced $(N-1)\times(N-1)$ Vandermonde
matrix, whose determinant must again be zero, so that two more of
the $\lambda_A$'s must be equal.  Proceeding in this way one reduces
the size of the Vandermonde matrix by unity at each step and finally
has
\begin{equation}
\lambda_A = r,
\end{equation}
as asserted. We then see
\begin{equation}
J_A = -i \,{\ln(r)\over2\pi} + m_A; \qquad  m_A \in Z.
\end{equation}
Expressed directly in terms of the surface gravities this yields
\begin{equation}
\sum_{i=1}^H Z_{Ai} {\gap\over \kappa_i} = -i \,{\ln(r)\over\pi} + 2
m_{A}; \qquad  m_{A} \in Z.
\end{equation}
By assumption, we have asserted the existence of at least one
solution to these constraint equations. (Otherwise the family we
used to start this discussion would not exist.) We have seen that we
can choose to present the master equation in such a manner that
$\sum_{A=1}^N Z_{Ai}=0$. But then
\begin{equation}
0  = -i \,{\ln(r)\over\pi} N + 2 \sum_{A=1}^N m_{A}; \qquad  m_{A}
\in Z.
\end{equation}
This implies that
\begin{equation}
 i \,{\ln(r)\over\pi} = 2 \; {\sum_{A=1}^N m_{A}\over N}  = q  \in Q.
\end{equation}
That is, there is a rational number $q$ such that
\begin{equation}
\sum_{i=1}^H Z_{Ai} {\gap\over \kappa_i} = q + 2 m_{A}; \qquad  q\in
Q; \qquad m_{A} \in Z.
\end{equation}
This is already enough to imply that the ratios $R_{AB}$ are
rational. If in addition $Z_{Ai}$ is of rank $H$ then, (either using
standard row-echelon reduction of the augmented matrix, or invoking
the Moore--Penrose pseudo-inverse and noting that the Moore-Penrose
pseudo-inverse of an integer valued matrix has rational elements),
we see that for each horizon the ratio $(\mathrm{gap})/\kappa_i$
must be a rational number, and consequently the ratios $R_{ij} =
\kappa_i/\kappa_j$ must all be rational numbers.

That is:
\newtheorem{irational2}[section]{Theorem}
\begin{irational2}\label{irational2}
Take quasi-normal frequencies to be given by the equation \eqref{E:master2-1}, and suppose that we have a family of QNFs as described by equation \eqref{E:family2}. If no proper subset of the $D_{A}$'s from \eqref{E:master2-1} sums to 0, and the matrix $Z_{Ai}$ from the same equation is of rank $H$, then the ratio of arbitrary pair of surface gravities must be rational.
\end{irational2}

%---------------------------------
\subsubsection{Case 2}
%---------------------------------
More generally, \emph{if some proper subset of the $D_A$'s sums to
zero}, subdivide the $N$ terms $A\in\{1,2,3,\dots, N\}$ into a cover
of disjoint irreducible proper subsets $\B_a$ such that
\begin{equation}
\sum_{A\in\B_a}^N  D_A  = 0. \label{E:master4-sp2}
\end{equation}
Then the solutions of equation (\ref{E:xxx}) are uniquely of the
from
\begin{equation}
\exp(2\pi i J_{A\in \B_a})  = \lambda_{A\in\B_a} = r_a.
\label{E:zzz}
\end{equation}
It is trivial to see that under the stated conditions this is a
solution of equation~(\ref{E:xxx}), the only technically difficult
step is to verify that these are the only solutions.  One again
proceeds by iteratively using the Vandermonde matrix $M_{AB} =
(\lambda_A)^{B-1}$ and considering its determinant. Instead of
showing that \emph{all} of the $\lambda_A$'s equal each other, we
now at various stages of the reduction process use the condition
$\sum_{A\in\B_a}  D_A  = 0$ to completely decouple the corresponding
$\lambda_{A\in\B_a} = r_a$ from the remaining
$\lambda_{A\not\in\B_a}$. Proceeding in this way we finally obtain
equation (\ref{E:zzz}) as claimed.

We then see
\begin{equation}
J_{A\in\B_a} = -i \,{\ln(r_a)\over2\pi} + m_{A\in\B_a}; \qquad
m_{A\in B_a} \in Z.
\end{equation}
Expressed directly in terms of the surface gravities this yields
\begin{equation}
\sum_{i=1}^H Z_{A\in\B_a} {\gap\over \kappa_i} = -i
\,{\ln(r_a)\over\pi} + 2 m_{A\in\B_a}; \qquad  m_{A\in B_a} \in Z.
\end{equation}
With the obvious notation of $a(A)$ denoting the index of  the
particular disjoint set $\B_a$ that $A$ belongs to, we can write
this as
\begin{equation}
\sum_{i=1}^H Z_{Ai} {\gap\over \kappa_i} = -i
\,{\ln\{r_{a(A)}\}\over\pi} + 2 m_{A}; \qquad  m_{A} \in Z.
\end{equation}
This result now is somewhat more subtle to analyze. Let $A$ and $B$
both belong to a particular set $\B_a$. Then
\begin{equation}
\sum_{i=1}^H \{Z_{Ai}- Z_{Bi}\} {\gap\over \kappa_i} =  + 2 \{ m_{A}
- m_B\} ; \qquad  m_{A}, m_B \in Z; \qquad A,B \in\B_a. \label{E:11}
\end{equation}
That is
\begin{equation}
\gap = {2\{m_{A} - m_B\} \over  \displaystyle \sum_{i=1}^H
{\{Z_{Ai}- Z_{Bi}\} \over\kappa_i}};  \qquad \gap \in R; \qquad A,B
\in\B_a;
\end{equation}
so we see that the gap is real. (Furthermore, the gap is seen to be
a sort of ``integer-weighted harmonic average'' of the $\kappa_i$.)
But reality then implies that $r_a = e^{i\phi_a}$ so that
\begin{equation}
\sum_{i=1}^H Z_{Ai} {\gap\over \kappa_i} = {\phi_{a(A)}\over\pi} + 2
m_{A}; \qquad  m_{A} \in Z.
\end{equation}
By using $\sum_A Z_{Ai}=0$ we see that
\begin{equation}
{\bar\phi\over\pi} = \sum_a {\phi_a |\B_a|\over N \pi} \in Q,
\end{equation}
so that
\begin{equation}
\sum_{i=1}^H Z_{Ai} {\gap\over \kappa_i} = {\phi_{a(A)}-\bar\phi
\over\pi} + 2 (m_{A}-\bar m); \qquad  m_{A} \in Z.
\end{equation}
Unfortunately in the general case there is little more than can be
said and one has to resort to special case-by-case analyses. One
last point we can make is that even though in this situation the
$R_{ij}=\kappa_i/\kappa_j$ are sometimes irrational we can make the
weaker statement that
\begin{equation}
{ \displaystyle \sum_{i=1}^H {\{Z_{Ai}- Z_{Bi}\} \over\kappa_i}
\over  \displaystyle \sum_{i=1}^H {\{Z_{Ci}- Z_{Di}\} \over\kappa_i}
} \;\;\;\in Q; \qquad A,B,C,D \in\B_a;
\end{equation}
That is, certain weighted averages of the surface gravities are
guaranteed to be rational. If we wish to analyze whether rational
ratios of  $R_{ij}=\kappa_i/\kappa_j$ are implied in each of the
particular cases of interest, we need to:
\begin{itemize}
\item[a)] Check if there exists some $\omega_{0}$ giving non-trivial subsets $\B_{a}$, leading to $\eqref{E:master4-sp2}$.
\item[b)] Analyze the sets of equations $\eqref{E:11}$ implied by such an $\omega_{0}$.
\end{itemize}
By proceeding in this way we are able to prove that periodicity of
the QNFs implies rational ratios for the surface gravities in the
following physically interesting cases:
\begin{itemize}
\item[a)] For $j=2m$ in equation $\eqref{E:SdS}$.
\item[b)] For equation $\eqref{E:RN}$ when  $j\neq 2m+1$ and $\cos(\pi j)\neq -\frac{1}{2}$.
\item[c)] For equation $\eqref{E:RNdS}$ when $j=2m$.
\end{itemize}

\subsection{Analysis of particular cases}

%------------------------------------------------------------------------

~~~~~Now explore the familiar cases, which can serve also as
particular examples described by this theorem. In the first part of
each particular example we explore whether periodicity implies rational
ratios of the surface gravities (within this
particular case). If the rational ratios are implied, we can use the results of section \ref{ratsec} to determine the families and the gap structure. Despite having the analysis from section \ref{ratsec}, we will
derive (for each case) the gap structure also by analysing the equations (\ref{E:11}). This serves as:

\begin{itemize}
\item
a consistency check,
\item
to
bring more understanding in how the ideas used in the section \ref{iratsec} work.
\end{itemize}
This derivation
is made in the second part of each case analysis (to be exact,
it is made for illustrative reasons for one set splitting only).

%------------------------------------------------------------------------
\subsubsection{Some common notation}

We will use the following notation:
\begin{equation}
f_{A}\equiv\sum_{i=1}^{H}\frac{Z_{Ai}}{\kappa_{i}},
\end{equation}
but if there is some index $\tilde A$, such that $f_{\tilde
A}=-f_{A}$, we rename it to $\tilde A=-A$. So always
$f_{-A}=-f_{A}$. We will also use $m_{A,B}\equiv m_{A}-m_{B}$ and
for the equation \eqref{E:11} of the form
\begin{equation}
\mathrm{gap}(\omega_{n})(f_{A}-f_{B})=2m_{A,B}
\end{equation}
we use the symbol $E_{A,B}$. Furthermore let us add one conceptual
explanation: We say that equations $E_{A,B}$ and $E_{\tilde A,\tilde
B}$ are linearly independent if there do not exist such integers
$m_{A,B}$, $m_{\tilde A,\tilde B}$, that the equations $E_{A,B}$ and
$E_{\tilde A,\tilde B}$ will be linearly dependent in the usual
sense of the word. If $E_{A,B}$ and $E_{\tilde A,\tilde B}$ are not
linearly independent, we say they are linearly dependent.

\subsubsection{2 Horizons case, S-dS black hole by monodromy calculations:\\ (1) spin 1 perturbations, (2) spin
0 and 2 perturbations}
%------------------------------------------------------------------------

Take the first case, which is formula derived by the use of monodromy techniques for the S-dS black hole:
\begin{equation}
\cosh{\left(\frac{\pi\omega}{\kappa_{-}}-\frac{\pi\omega}{\kappa_{+}}\right)}+[1+2\cos(\pi
j)]
\cosh{\left(\frac{\pi\omega}{\kappa_{-}}+\frac{\pi\omega}{\kappa_{+}}\right)}=0.
\end{equation}
Here ~$j=0$~ for spin ~0,~2 ~and ~$j=1$~ for spin ~1.

\bigskip

This becomes for spin ~1~ perturbation:

\begin{equation}
e^{(\frac{\pi\omega}{\kappa_{-}} -
\frac{\pi\omega}{\kappa_{+}})}+e^{-(\frac{\pi\omega}{\kappa_{-}} -
\frac{\pi\omega}{\kappa_{+}})}-e^{(\frac{\pi\omega}{\kappa_{-}}+\frac{\pi\omega}{\kappa_{+}})}-e^{-(\frac{\pi\omega}{\kappa_{-}}
+ \frac{\pi\omega}{\kappa_{+}})}=0,\label{s1p}
\end{equation}
and for spin ~0 and~2 perturbations:

\begin{equation}
e^{(\frac{\pi\omega}{\kappa_{-}} -
\frac{\pi\omega}{\kappa_{+}})}+e^{-(\frac{\pi\omega}{\kappa_{-}} -
\frac{\pi\omega}{\kappa_{+}})}+3e^{(\frac{\pi\omega}{\kappa_{-}}+\frac{\pi\omega}{\kappa_{+}})}+3e^{-(\frac{\pi\omega}{\kappa_{-}}
+ \frac{\pi\omega}{\kappa_{+}})}=0.\label{s2p}
\end{equation}
In the case of spin 1 perturbation we already know that ratios of
surface gravities might be completely arbitrary and we still get periodic
solutions.

\paragraph{The question of surface gravities rational ratios}

The matrix $Z_{Ai}$ has rank 2 ($=H$) hence if some subset of $D_{A}$ does not sum to 0, the rational ratio of surface gravities is implied.

Let us have a look what happens here: The functions $f_{i}$ are
given as

\[f_{1}=\frac{1}{\kappa_{-}}-\frac{1}{\kappa_{+}},\]

\[f_{2}=\frac{1}{\kappa_{-}}+\frac{1}{\kappa_{+}},\]
and we have also ~$f_{-1}$ ~and $f_{-2}$. If the coefficients
$D_{A}$ split into two sets each having two elements (which is the
only way how they can be non-trivially split in this case), there
are two equations and three possibilities how one can split them:

\begin{itemize}
\item
$\{ D_{1}, D_{-1}\}$ and $\{ D_{2}, D_{-2}\}$,

\item
$\{ D_{1}, D_{2}\}$ and $\{ D_{-1}, D_{-2}\}$,

\item
$\{ D_{1}, D_{-2}\}$ and $\{ D_{2}, D_{-1}\}$.

\end{itemize}
In terms of equations this leads to the following combinations:

\begin{itemize}

\item
$E_{1,-1},~E_{2,-2}$  ~~~gives~ 2 linearly independent equations,

\item
$E_{1,2},~E_{-1,-2}$  ~~~gives~ only 1 linearly independent equation,

\item
$E_{1,-2},~E_{2,-1}$  ~~~gives~ only 1 linearly independent equation.

\end{itemize}

Because we have two surface gravities and two linearly independent equations the first combination of equations leads to the
condition that the ratio of surface gravities must be rational. That means
one needs to explore only the second and the third combination of equations. Here we have to check the step $a)$ from the end of the previous section. If there exists ~$\omega_{0}$~ giving us $\B_{a}$ sets leading to the second, or the third combination of equations, it
must fulfill the following conditions:

\begin{itemize}

\item

In the case of the second combination ($E_{1,2},~E_{-1,-2}$) it must fulfil the equations

\begin{equation}
\frac{\pi\omega_{0}}{\kappa_{+}}=i2\pi m_{1,2}+\ln\left|\frac{C_{1}}{C_{2}}\right|,
\end{equation}
\begin{equation}
-\frac{\pi\omega_{0}}{\kappa_{+}}=i2\pi m_{-1,-2}+\ln\left|\frac{C_{-1}}{C_{-2}}\right|.
\end{equation}

\item
In the case of the third combination ($E_{1,-2},~E_{-1,2}$) it must fulfil the equations

\begin{equation}
\frac{\pi\omega_{0}}{\kappa_{-}}=i2\pi m_{1,-2}+\ln\left|\frac{C_{1}}{C_{-2}}\right|,
\end{equation}
\begin{equation}
-\frac{\pi\omega_{0}}{\kappa_{-}}=i2\pi m_{-1,2}+\ln\left|\frac{C_{-1}}{C_{2}}\right|.
\end{equation}

\end{itemize}
But for the second combination this means:
\begin{equation}
\ln\left|\frac{C_{1}}{C_{2}}\right|=-\ln\left|\frac{C_{-1}}{C_{-2}}\right|,\label{combination2}
\end{equation}
and for the third combination this means:
\begin{equation}
\ln\left|\frac{C_{1}}{C_{-2}}\right|=-\ln\left|\frac{C_{-1}}{C_{2}}\right|.\label{combination3}
\end{equation}
From (\ref{s1p}) we see that for spin 1 perturbation it holds ~
\begin{equation}
|C_{1}|=|C_{2}|=|C_{-1}|=|C_{-2}|=1.\label{cs1p}
\end{equation}
As a result of (\ref{cs1p}) the conditions (\ref{combination2}) and (\ref{combination3}) are trivially fulfilled, since all the
logarithms are ~0.~ But for spin 0 and 2 ~$|C_{1}|=|C_{-1}|=1$~ and
~$|C_{2}|=|C_{-2}|=3$, which means that in each case we get
~$\ln(\frac{1}{3})$.~ This means there is no way how to fulfil (\ref{combination2}), or (\ref{combination3}), hence split the
coefficients in such way that we do \emph{not} get the surface gravities
rational ratio condition. This means that \emph{for spin 0 and spin 2 periodicity implies the rational ratio of surface
gravities.}

For the spin 1 perturbation we already showed we can
find explicit $\omega_{0}$-s leading to two different families of QNMs, each family related to different horizon surface gravity ($in\kappa_\pm$). This means that in the case of spin 1 perturbation rational ratios are \emph{not} implied.
Note also that each of the two families is just a result of different splitting of coefficients $D_A$.
%---------------

\paragraph{``Gap'' derivation by using our approach (spin 1 perturbation)}

~~~~As previously noted, the general solutions split into two
families, one related to one surface gravity ($\omega=i n
\kappa_{-}$), the other to another ($\omega=i n \kappa_{+}$). Now
take one solution from the set $\{ in\kappa_{-}\}$ (for example
$\omega=i\kappa_-$) and substitute it in \eqref{s1p} to get the
coefficients $D_{A}$. We see that the two set splitting of $D_{A}$
is:
\begin{eqnarray}
\left\{ D_{1}=-e^{-
\frac{i\pi\kappa_{-}}{\kappa_{+}}}, ~D_{-2}=e^{-
\frac{i\pi\kappa_{-}}{\kappa_{+}}}\right\}, ~~~\left\{ D_{-1}=-e^{\frac{i\pi\kappa_{-}}{\kappa_{+}}}, ~D_{2}=e^{ \frac{i\pi\kappa_{-}}{\kappa_{+}}}\right\}.
\end{eqnarray}
 (We can easily observe that the sum of those couples of coefficients
within each set is 0).
Now, since $E_{1,-2}$ and $E_{-1,2}$ are linearly dependent, the second equation can be taken only as a definition of $m_{-1,2}$. The only independent equation is then:

\[\mathrm{gap}(\omega_{n})=m_{1,-2}\kappa_{-}.\]
To obtain from this equation the basic gap structure, we have to choose the integer $m_{1,-2}$ to be such, that we obtain the ``narrowest'' gap. This is obtained by $m_{1,-2}=1$, confirming the known result.

If we start with the second set of solutions related to the cosmological horizon surface gravity, we can proceed completely analogically and verify also the second result. The gap will be in such case given by the cosmological horizon surface gravity.

\paragraph{``Gap'' derivation by using our approach (spin 0 and 2 perturbation)}

Here we proved we have 2 linearly independent equations (the rest of the equations are only definitions of $m_{A,B}$ terms):

\begin{itemize}

\item
$\mathrm{gap}(\omega_{n})=-m_{1,2}\kappa_{+}$,

\item
$\mathrm{gap}(\omega_{n})=-m_{1,-2}\kappa_{-}$.

\end{itemize}
Now, they immediately lead to the condition
$\frac{\kappa_{-}}{\kappa_{+}}=\frac{m_{1,2}}{m_{1,-2}}$, hence surface gravities ratio being
rational.

To explore the gap structure one has to explore all the other independent conditions as well (here there are three independent conditions obtained by equating 4 terms). This is because definitions of $m_{A,B}$ might put some additional constrains on the gap function. (The constraints come from the fact that the $m_{A,B}$ have to always be integers. This will be seen in the analysis of the following case.) But in this case the third independent condition can be given as
\begin{equation}
\mathrm{gap}(\omega_{n})\left(\frac{1}{\kappa_{-}}-\frac{1}{\kappa_{+}}\right)=m_{1,-1}
\end{equation}
and is consistent with taking the ``narrowest'' gap as
\[
\mathrm{gap}(\omega_{n})=\kappa_{*}=\kappa_{-} m_{1,-2}=\kappa_{-}p_-=\kappa_{+}
m_{1,2}=\kappa_{+}p_+=\frac{p_-}{b_-}=\frac{p_+}{b_+}=\frac{1}{b_{*}}.
\]
(Here $b_{*}$ is the highest common divisor of $b_-,b_+$ with
respect to integers.) This means: $-m_{1,2}=p_+$ and $-m_{1,-2}=p_-$.

\subsubsection{2 Horizons, S-dS black hole by analytically solvable potentials}
%-----------------------------------------------------------------------

The equation we obtained is the following:
\begin{equation}\label{analytic}
\cosh{\left(\frac{\pi\omega}{\kappa_{-}} -
\frac{\pi\omega}{\kappa_{+}}\right)}-\cosh{\left(\frac{\pi\omega}{\kappa_{-}}+\frac{\pi\omega}{\kappa_{+}}\right)}=2\cos(\pi\alpha_+)\cos(\pi\alpha_-).
\end{equation}

\paragraph{The question of surface gravities rational ratios}

~~~~In the previous section we already proved that in this case the rational ratios are implied
by the periodicity. We can prove the same by analyzing all the
possibilities how to split coefficients into different sets
~$\{D_{A}\}$.~ It is an alternative proof to the one in the previous
section. It is definitely a more complicated proof, but its advantage
is that it is part of more general approach, which generates proofs
for all the cases where the implication holds.

Let us show the ``ugly'' way of proving it by the new general
method:
The matrix $Z_{Ai}$ is of rank 2 ($=H$) and hence if the $D_{A}$ coefficients do not split into nontrivial subsets, the rational ratios of surface gravities is proven.
Now let us explore what ways of splitting the coefficients one can obtain and what they generally mean. We can rewrite the equation \eqref{analytic} to get the following functions:\\

$f_{1}=\frac{3}{\kappa_{-}}+\frac{1}{\kappa_{+}}$,
~~~~~$f_{2}=\frac{1}{\kappa_{-}}+\frac{3}{\kappa_{+}}$,
~~~~~$f_{3}=3(\frac{1}{\kappa_{-}}+\frac{1}{\kappa_{+}})$,\\

$f_{4}=\frac{1}{\kappa_{-}}+\frac{1}{\kappa_{+}}$,
~~~~~$f_{5}=2\left(\frac{1}{\kappa_{-}}-\frac{1}{\kappa_{+}}\right)$.

\bigskip

Let us explore what happens when we split the $D_{A}$ coefficients
into two sets, one having three and the other two elements (these
are the only possible nontrivial ways of splitting the
coefficients). This gives for each splitting 3 equations. One can
analyze all the combinations of equations in the following way: Pick
$f_{\tilde A}$ belonging to a coefficient in the three element set
and take all combinations of such equations $E_{A,B}$ (related to
all the possible ways of splitting the coefficients in which
$D_{\tilde A}$ is in the set of three elements), that they do
\emph{not} involve the given $f_{\tilde A}$. If these combinations
give 2 independent equations we are finished, and do not need to
explore the 3-rd equation involving ~$f_{\tilde A}$. If there is
only one independent equation from the given couple, we need to
explore if there is a way how to add an equation involving
~$f_{\tilde A}$~ ($E_{\tilde A,B}$), by keeping the number of
independent equations to be still ~1.  The $I/D$ letters in the
table mean linearly independent/dependent:

\bigskip

\begin{tabular}{|c|c|c|c|c|c|c|}
\hline $f_{\tilde A}$ & combination 1 & I/D & combination 2 & I/D & combination 3 & I/D \\
\hline $f_{1}$ & $E_{2,3},~E_{4,5}$ & I & $E_{2,4},~E_{3,5}$ & I & $E_{2,5},~E_{3,4}$ & I \\
\hline $f_{2}$ & $E_{1,3},~E_{4,5}$ & I & $E_{1,4},~E_{3,5}$ & I & $E_{1,5},~E_{3,4}$ & I \\
\hline $f_{3}$ & $E_{1,2},~E_{4,5}$ & I & $E_{1,4},~E_{2,5}$ & I & $E_{1,5},~E_{2,4}$ & I \\
\hline $f_{4}$ & $E_{1,2},~E_{3,5}$ & I & $E_{1,3},~E_{2,5}$ & I & $E_{1,5},~E_{2,3}$ & I \\
\hline $f_{5}$ & $E_{1,2},~E_{3,4}$ & I & $E_{1,3},~E_{2,4}$ & D & $E_{1,4},~E_{2,3}$ & D \\
\hline
\end{tabular}

\bigskip

We see that only the last two cases give linearly dependent
equations, but it is very easy to check that in both of these cases it holds, that if we add arbitrary third equation of the type ~$E_{5,B}$~ (hence
involving $f_{5}$ ), the number of linearly
independent equations grows to ~2. This is a result of the fact that in
~$E_{1,3},~E_{2,4}$~ there is only ~$\kappa_{+}$~ present and in
~$E_{1,4},~E_{2,3}$~ there is only $\kappa_{-}$ present.
This proves that here \emph{the rational ratios are implied by the
periodicity.}
%------------------------------------------------------------------------

\paragraph{The ``gap'' derivation by using our approach}

If there are no non-trivial ways of splitting the coefficients, then
we get from our result \eqref{E:11} 4 independent conditions
(although, in the sense defined, only 2 linearly independent
equations) and they are:
\begin{itemize}

\item [1)]
~~~$\mathrm{gap}(\omega_{n})\left(\frac{1}{\kappa_{-}} - \frac{1}{\kappa_{+}}\right)=m_{1,2}$~,

\item [2)]
~~~$\mathrm{gap}(\omega_{n})=-m_{1,3}\kappa_{+}$~,

\item [3)]
~~~$\mathrm{gap}(\omega_{n})=m_{1,4}\kappa_{-}$~,

\item [4)]
~~~$\mathrm{gap}(\omega_{n})\left(\frac{1}{\kappa_{-}} - \frac{1}{\kappa_{+}}\right)=2m_{1,5}$~.

\end{itemize}

Now any two from the first three equations\footnote{The same holds
for any two of the last three equations.} give surface gravity
rational ratio condition. Choose for this purpose, for example,
equations 2) and 3). (Notice here that in order to get the surface
gravities ratio positive, $m_{1,3},m_{1,4}$ must have opposite
signs.) Equation 1) is then uninteresting since it is only defining
$m_{1,2}$, without constraining the possible gap function. Now the
whole problem is encoded in the equation 4) and particularly in the
fact that there is 2 on the right side of the equation. In fact it
leads to ~$m_{1,4}+m_{1,3}=2m_{1,5}$,~ and that means one can fulfil
the equation 4) only if both $m_{1,4}, m_{1,3}$ are odd or both are
even. This means that the definition of $m_{1,5}$ constrains the gap
function. Now realize that only the ratio of $m_{1,3}, m_{1,4}$ is
determined by the ratio of the surface gravities $p_-/p_+$, so there
is still freedom to multiply both, $p_+$ and $p_-$, by the same
arbitrary scaling integer. We have to determine the integer to
fulfill the equation 4) and simultaneously to give the ``narrowest''
gap. But then, if the ratio
\[
\frac{\kappa_{-}}{\kappa_{+}}=\frac{b_+}{b_-}=\frac{p_+}{p_-}
\]
is in its most reduced form given by one odd and one even number ($p_+\cdot p_-$ is even), we
have to take both $m_{1,4}, m_{1,3}$ even by:
$m_{1,4}=2 p_-$ and $-m_{1,3}=2 p_+$. Then the gap will be
$\mathrm{gap}(\omega_{n})=2\kappa_{*}=\frac{2}{b_{*}}$ ~ $\left(\kappa_{*}=\kappa_{-} p_-=\kappa_{+}
p_+=\frac{1}{b_{*}}\right)$.
In the case both $p_-,p_+$ are odd numbers (hence $p_-\cdot p_+$ is odd), the 4)-th equation
is automatically fulfilled and to get the ``narrowest'' gap one
chooses the scaling integer to be $\pm 1$, hence $m_{1,4}=p_-$ and
$m_{1,3}=-p_+$. Then the gap is obtained as ~$\mathrm{gap}(\omega_{n})=\kappa_*=\frac{1}{b_{*}}$.

\subsubsection{2 Horizons, R-N black hole}
%------------------------------------------------------------------------

Take as another example the equation \eqref{E:RN} with $j=0$. We can rewrite it as:
\begin{equation}
e^{\pi\omega(\frac{1}{\kappa_{+}}+\frac{1}{\kappa_{-}})}+3e^{-\pi\omega(\frac{1}{\kappa_{+}}+\frac{1}{\kappa_{-}})}+3e^{\pi\omega(\frac{1}{\kappa_{-}}-\frac{1}{\kappa_{+}})}=0.
\end{equation}
There cannot be any non-trivial subset of coefficients $D_{A}$ summing up to 0. (This is because there are
three non-zero coefficients, so it cannot happen that two
of them will give ~0~ by the summation). That means we have always two independent conditions:

\begin{itemize}

\item[1)]
 $\mathrm{gap}(\omega_{n})=m_{1,3}\kappa_{+}$~,

\item[2)]
 $\mathrm{gap}(\omega_{n})=-m_{2,3}\kappa_{-}$~.

\end{itemize}
But $1)$ and $2)$ are also linearly independent equations, so this immediately implies the same results, as in the previous case.
These are:
\begin{itemize}
\item
the rational ratio of surface gravities
given as $-m_{1,3}/m_{2,3}$ (note again that to get the ratio
positive, $m_{1,3}$ $m_{2,3}$ have to be chosen with the opposite
signs),
\item
the ``narrowest'' gap given by $m_{1,3}=p_+$ and $m_{2,3}=-p_-$, hence
\[
\mathrm{gap}(\omega_{n})=\kappa_{*}=\kappa_{-}p_-=\kappa_{+}p_+=\frac{1}{b_{*}}~.
\]
\end{itemize}
So in this case \emph{the rational ratios are implied by the periodicity.}

%------------------------------------------------------------------------
\subsubsection{3 Horizons, Monodromy results: R-N-dS black hole}
%------------------------------------------------------------------------

Take the equation \eqref{E:RNdS} with
any $j$, for which the coefficients are non-zero (trivial example is
standard ~$j=0$~).

\paragraph{The question of surface gravities rational ratios}

Now this is a case with six functions $f_{A}$:

\medskip

$f_{1}=\frac{1}{\kappa_{+}}-\frac{1}{\kappa_{C}}$,
~~~$f_{2}=\frac{1}{\kappa_{+}}+\frac{1}{\kappa_{C}}$,
~~~$f_{3}=\frac{2}{\kappa_{-}}+\frac{1}{\kappa_{+}}-\frac{1}{\kappa_{C}}$~~
and ~~$f_{-1}, ~f_{-2}, ~f_{-3}$.\medskip\\
Here we have 3 surface gravities and the rank of $Z_{Ai}$ matrix is 3, so if there is no nontrivial subset of $D_{A}$ coefficients summing to 0, the surface gravities rational ratios are implied.

Let us explore the nontrivial ways of splitting the coefficients.
The splitting of coefficients giving the lowest number of
independent conditions is the splitting into three sets, each having
2 elements. It gives 3 equations (and we need three linearly
dependent). There are three basic ways how to split the
coefficients:

\begin{itemize}
 \item
The first splitting is giving $E_{1,-1}, ~E_{2,-2}, ~E_{3,-3}$, but these are necessarily 3 linearly independent equations as ~$f_{1}, ~f_{2}, ~f_{3}$ are linearly independent set of functions.

\item
The other type of splitting is $E_{A_{1},A_{2}}, ~E_{-A_{2},A_{3}},
~E_{-A_{3},-A_{1}}$, ~$A_{i}=\pm 1,\pm 2,\pm 3$, ~$|A_{i}|\neq
|A_{j}|$ ~if~ $i\neq j$. But all these ways of splitting the
coefficients give 3 linearly independent equations as well.

\item
The third basic way how to split coefficients into three two element
sets gives ~$E_{A_{1},-A_{1}}, ~E_{A_{2},A_{3}},
~E_{-A_{2},-A_{3}}$, ~$A_{i}$~ being defined as before.  This
splitting gives only 2 linearly independent equations.
\end{itemize}
That means in the last case we have to take the second step.
Fortunately here we can use for different $C_{A}$ coefficients
multiplying different $\cosh(\cdot)$ terms the same argument as in
the case of S-dS and spin ~0 and ~2 perturbations. This argument
excludes the possibility of such ways of splitting the coefficients.

Any splitting into two sets, one having 2 and the other 4 elements is
just more constrained version of some ``three set each having two
elements'' splitting. This means the arguments provided in the
previous case transfer automatically to this case.

The last, deeply nontrivial splitting is when $D_{A}$ coefficients
split into two sets, each having three elements. In such case we
have 4 independent conditions. There are two possible basic ways of
splitting the equations:

\begin{itemize}

\item
$E_{A_{1},-A_{1}}, ~E_{A_{1},A_{2}}, ~E_{-A_{2},A_{3}}, ~E_{A_{3},-A_{3}}$

\item
$E_{A_{1},A_{2}}, ~E_{A_{2},A_{3}}, ~E_{-A_{1},-A_{2}}, ~E_{-A_{2},-A_{3}}$

\end{itemize}
Now the first splitting leads to three independent equations, but
the second splitting only to two independent equations. So in the
second splitting case one has to proceed to the next step: One can
observe that if there exists an ~$\omega_{0}$~  giving the second
splitting, then the following two equations have to be fulfilled.
Take:
\begin{itemize}
\item
$z\equiv e^{\pi\omega_{0}}$,
\item
$\alpha\equiv\frac{1}{\kappa_{+}}-\frac{1}{\kappa_{C}}$,
\item
$\beta\equiv\pm\left(\frac{1}{\kappa_{+}}+\frac{1}{\kappa_{C}}\right)$,
\item
$\delta\equiv\pm\left(\frac{2}{\kappa_{-}}+\frac{1}{\kappa_{+}}+\frac{1}{\kappa_{C}}\right)$
~and
\item
$K\equiv 1+\cos(\pi j)$.
\end{itemize}
Then the following holds:
\begin{equation}
z^{\alpha}+Kz^{\beta}+2Kz^{\delta}=0,\label{first}
\end{equation}
\begin{equation}
z^{-\alpha}+Kz^{-\beta}+2Kz^{-\delta}=0.\label{second}
\end{equation}
But then by substituting to (\ref{second})
~$z^{-\delta}=-\frac{K}{z^{\alpha}+Kz^{\beta}}$~ we will obtain

\[\frac{1-3K^{2}}{K}+z^{\alpha-\beta}+z^{-(\alpha-\beta)}=0,\]
and then we get the result:
\[ z^{\alpha-\beta}=-\frac{1-3K^{2}}{2K}\pm\frac{1}{2}\sqrt{\left(\frac{1-3K^{2}}{K}\right)^{2}-4}.\]
For ~$K>\frac{1}{\sqrt{3}}$~ we have ~$z^{\alpha-\beta}>0$. ~We see
that if we take ~$j=\frac{n}{2}$, the only nontrivial cases (~$K\neq
0$~) are ~$K=1,~2$ and fulfil this condition. But the equation
(\ref{first}) is in fact:
\begin{equation}
z^{\alpha-\beta}+K+2Kz^{(\alpha-\beta)\frac{\delta-\beta}{\alpha-\beta}}=0.\label{first2}
\end{equation}
Define by $u\equiv z^{\alpha-\beta}$. Then by solving (\ref{first2}) we obtain

\[\frac{\delta-\beta}{\alpha-\beta}=\log_{u}\left(-\frac{u+K}{2K}\right).\]
But since the logarithm argument is for ~$K>\frac{1}{\sqrt{3}}$~ negative (and $u$ is in such case positive), it means that
$\frac{\delta-\beta}{\alpha-\beta}$ must be nonreal. But considering
how $\alpha,\beta,\delta$ are defined, this cannot happen for
surface gravities being real numbers. Hence it is proven that there
is no $\omega_{0}$ giving the splitting considered.
This means we proved that for (\ref{E:RNdS}) \emph{the periodicity implies rational ratios of surface gravities.}

\paragraph{The ``gap'' analysis by using our approach}

 If we consider again only such $\omega_{0}$-s, for which there is no nontrivial splitting of $D_{A}$ coefficients, we obtain the following 5 independent conditions:
\begin{itemize}

\item
$\mathrm{gap}(\omega_{n})\frac{1}{\kappa_{-}}=m_{2,3}$,

\item
$\mathrm{gap}(\omega_{n})\frac{1}{\kappa_{+}}=m_{1,-2}$,

\item
$\mathrm{gap}(\omega_{n})\frac{1}{\kappa_{C}}=-m_{1,2}$,

\item $\mathrm{gap}(\omega_{n})(\frac{1}{\kappa_{+}}-\frac{1}{\kappa_{C}})=m_{1,-1}$,

\item
$\mathrm{gap}(\omega_{n})\frac{1}{\kappa_{-}}=m_{-2,-3}$.

\end{itemize}
Now note: the last two equations are uninteresting, since they
are only definitions of $m_{1,-1}$ and $m_{-2,-3}$ without putting any constraint on the gap function. On the other hand the first
three equations tell us that the ratio between arbitrary two surface gravities is a rational number.
Then to get the ``narrowest'' gap, consider the following:
\[
\frac{\kappa_{-}}{\kappa_{+}}=\frac{p_+}{p_{1-}}=\frac{b_+}{b_-},
\qquad\qquad
\frac{\kappa_{-}}{\kappa_{C}}=\frac{p_{C}}{p_{2-}}=\frac{b_{C}}{b_-}.
\]
Further
~$\kappa_{*}=\frac{1}{b_{*}}=\kappa_{-}p_{1-}=\kappa_{+}p_+$~ and
~$\kappa'_{*}=\frac{1}{b'_{*}}=\kappa_{-}p_{2-}=\kappa_{C}p_{C}$.
~Then the gap must be given by
~$\mathrm{gap}(\omega_{n})=\mathrm{lcm}\{\kappa_{*},\kappa'_{*}\}=l_1
\kappa_{*}=l_2 \kappa'_{*}$, ($l_1, l_2\in\mathbb{Z}$). ~That means
the $m$ integers must be taken as: ~$m_{2,3}=l_1\cdot p_{1-}$~,
~$m_{1,-2}=l_1\cdot p_+$ ~and ~$-m_{1,2}=l_2\cdot p_{C}$.

\subsection{Conclusions}

It seems to be very hard (if not impossible) to find a more
constrained form of the monodromy results than the formulation
\eqref{master}. Since in \eqref{master} it is not a general result (but seems to be
generic enough) that the rational ratios of surface gravities are
implied by periodicity (see S-dS spin 1 perturbations as
counter-example), we developed a general method how to prove the
implication for every case in which it holds. The necessary step was
to prove the theorem at the beginning of the second section. By
using our method we proved that the implication (periodicity
$\rightarrow$ rational ratios) holds in every monodromy case (from
the cases described before), apart of the one given counter-example.

\bigskip

%----------------------------------------------------------------------------------------------------------------------------------------
%\clearpage
%----------------------------------------------------------------------------------------------------------------------------------------

%-----------------------------------------------------------------------------------------------------------------------------------------

\chapter{Multiplication of tensorial distributions}

\section{Introduction}
%%%%%%%%%%%%%%%%%%%%%%%%%%%%%%%%%%%%%%%%%%%%%%%%%%%

This chapter is devoted to a topic from the field of mathematical physics which is closely related to the general theory of relativity. It offers possible significant conceptual extensions of general relativity at short distances/high energies, and gives another arena in which one can conceptually/physically modify the classical theory.
But the possible meaning of these ideas is much wider than just the general theory of relativity. It is related to the general questions of how one defines the theory of distributions for any \emph{geometrically} formulated physics describing \emph{interactions}.

\paragraph{The main reasons why we ``bother''}
Let us start by giving the basic reasons why one should work with the language of distributions rather than with the old language of functions:

\begin{itemize}
\item
First there are deep physical reasons for working with distributions
rather than with smooth tensor fields. We think distributions are
more than just a convenient tool for doing computations in those
cases in which one cannot use standard differential geometry. We
consider them to be mathematical objects which much more accurately
express what one actually measures in physics experiments, more so
than when we compare them to the old language of smooth functions.
The reason is that the question: ``What is the `amount' of physical
quantity contained in an open set?'' is in our view a much more
reasonable physical question, (reasonable from the point of view of
what we can ask the experimentalists to measure), than the question:
``What is the value of quantity at a given point?''. But ``point
values'' as ``recovered'' by delta distributions do seem to give a
precise and reasonable meaning to the last question. There is also a
strong intuition that the ``amount'' of a physical quantity in the
open set $\Omega_{1}\cup \Omega_{2}$, where $\Omega_{1},\Omega_{2}$
are disjoint is the sum of the ``amounts'' of that quantity in
$\Omega_{1}$ and $\Omega_{2}$. That means it is more appropriate to
speak about distributions rather than general smooth mappings from
functions to the real numbers (the mappings should also be linear).

\item
The second reason is that many physical applications suggest the need for a much
richer language than the language of smooth tensor fields. Actually when we look for physically interesting solutions it might be always a matter of importance to have a much larger class of objects available than the class of smooth tensor fields.

\item
The third reason (which is a bit more speculative) is the relation of
the language defining the multiplication of distributions to quantum field theories (but specifically to quantum gravity). Note that the problems requireing distributions, that means problems going
beyond the language of classical differential geometry, might be related to physics on small scales. At the same time understanding some operations with distributions, specifically their product has a large \emph{formal} impact on the foundations of quantum field theory, particularly on the
problem with interacting fields. (See, for example, \cite{QFT}.) As a result of this it can have significant consequences for quantum gravity as well.
\end{itemize}

\paragraph{The intuition behind our ideas}
Considerations about language intuitiveness lead us to an interesting conclusion: the language of distributions (being connected with our intuition)~
strongly suggests that the properties of classical
tensor fields should not depend on the sets having (in every chart) Lebesgue measure 0. If we follow the idea that these measure zero sets  do not have any impact on physics, we should naturally expect that we will be able to
generalize our language from a smooth manifold into a piecewise
smooth manifold (which will bring higher symmetry to our conceptual network).
The first traces of piecewise smooth coordinate transformations are
already seen, for example, in~\cite{Penrose}.

\paragraph{The current situation is ``strange''}
Now it is worth noting how strange the current situation of the theory of distributions is: we have a useful and
meaningful language of distributions, which can be geometrically
generalized, but this language works only for linear physics. But linear physics
is only a starting point (or at best a rough approximation) for
describing real physical interactions, and hence nonlinear physics. So one
naturally expects that the ``physical'' language of distributions will be a result of some
mathematical language defining their product. Moreover,
at the same time we want this language to contain the old language of differential geometry (as a special case), as in the case of linear
theories. It is quite obvious that the nonlinear generalization of the geometric distributional theory and the construction of generalized differential geometry are just two routes to the same mathematical theory. The natural feeling that such theory might exist is the main motivation for this work. The practical need of this language is obvious as well, as we
see in the numerous applications \cite{Kerr, GerTra, QFT, Schwarzschild, geodesic, Reis-Nor, waves, Vickers}. (Although this is not the main
motivation of our work.) But is it necessary that such more general mathematical language exists as a full well defined theory? No, not at all. The potentially successful uses of distributions that go behind the Schwarz original theory might be only ``ad hoc'' from the fundamental reasons. Take the classical physics. The success of such uses of distributions here might not be a result of some more general language than smooth tensor calculus being a classical limit of the more fundamental physics. There might be only hidden specific reasons in the more fundamental physics why such ``ad hoc'' calculations in some of the particular cases work. But it is certainly very interesting and important to explore and answer the question whether: (i) such a full mathematical langauge exists, and (ii) to which extent it is useful to the physics community. For the first question this work suggests a positive answer, to answer the second question much more work has to be done.

\paragraph{So what did we achieve?}
The main motivation for this work is the development of
a language of distributional tensors, strongly connected with physical intuition. (This also means it has to be based on the concept of a piecewise smooth manifold.) This language must contain all the basic tensorial operations in a generalized way, enabling us to understand
the results the community has already achieved, and also the problems attached to them. It is worth to stress that some of this motivation results from a shift in view regarding the foundations of classical physics (so it is given by ``deeper'' philosophical reasons), but it can have also an impact on practical physical questions. The scale of this impact has still to be explored. We claim that the goals defined by our motivation (as described at the beginning) are achieved in this work. Particularly, we have generalized all the basic concepts from smooth tensor field calculus (including the fundamental concept of the covariant derivative) in two basic directions:

\begin{itemize}
\item
The first generalization goes in the direction of the class of objects that, (in every chart on the piecewise smooth manifold), are indirectly related to sets of piecewise continuous functions. This class we call $D'^{m}_{n E A}(M)$. For a detailed understanding see section \ref{GTFsection}.
\item
The second generalization is a generalization to the class of objects naturally connected with a smooth manifold belonging to our piecewise smooth manifold (in the sense that the smooth atlas of the smooth manifold is a subatlas of our piecewise smooth atlas). This is a good analogy to the generalization known from the classical distribution theory. The class of such objects we call $D'^{m}_{n (\mathcal{\mathcal{\mathcal{\mathcal{S}}}} o)}(M)$. For detailed understanding see again the section \ref{GTFsection}.
\end{itemize}
We view our calculus as the most natural and straightforward
construction achieving these two particular goals. The fact that
such a natural construction seems to exist supports our faith in the
practical meaning of the mathematical language here developed.

The last goal of this chapter is to suggest much more ambitious, natural generalizations, which are unfortunately at present only in the form of conjectures. Later in the text we provide the reader with such conjectures.

\paragraph{The structure of this chapter}
The structure of this chapter is the following:
In the first part we want to present the
current state of the Colombeau algebra theory and its geometric
formulations. We want to indicate where its
weaknesses are.
This part is followed by several technical sections, in which we define our theory and prove the basic theorems. First we define the basic concepts underlying our theory. After that we define the concept of generalized tensor fields, their important subclasses and basic operations on the generalized tensor fields (like tensor product). This is followed by the definition of the basic concept of our theory, the concept of equivalence between two generalized tensor fields. The last technical part deals with the definition of the covariant derivative operator and formulation of the initial value problem in our theory.
All these technical parts are followed by explanatory sections, where we discuss our results and show how our theory relates to the practical results already achieved (as described in the first part of the chapter).

\section{Overview of the present state of the theory}

\subsection{The theory of Schwartz distributions}

Around the middle of the 20-th century Laurent Schwartz found a mathematically rigorous way
for extending the language of physics from the
language of smooth functions into the language of distributions. Physicists such as Heaviside and Dirac had already given good physics reasons for believing that such a mathematical structure might exist.

The classical formulations of the distribution theory were directly connected with
$\mathbb{R}^{n}$ and were non-geometric. Distributions in such a formulation are typically understood as continuous, linear maps
from compactly supported smooth functions (on $\mathbb{R}^{n}$) to
real numbers. (The class of such functions is typically denoted by $D(\mathbb{R}^{n})$. The dual to such space, which is what distributions are, is typically denoted by $D'(\mathbb{R}^{n})$.) Here the word ``continuous'' refers to the following topology on
the given space of compactly supported smooth functions:
\emph{The sequence of smooth compactly supported functions $f_{l}(x_{i})$ converges
for $l\to\infty$ to a smooth compactly supported function $f(x_{i})$ iff an arbitrary
degree derivative with respect to arbitrary variables $\frac{\partial^{n}f_{l}(x_{i})}{\partial x_{1}^{m_{1}}...\partial x_{k}^{m_{k}}}$ $\left(\sum_{i=1}^{k}m_{k}=n\right)$
converges uniformly on each compact $\mathbb{R}^{n}$ subset to $\frac{\partial^{n}f(x_{i})}{\partial x_{1}^{m_{1}}...\partial x_{k}^{m_{k}}}$.}

The alternative classical way of formulating the theory of distributions is to extend the space of test objects to be such that they still follow appropriate fall-off properties. The properties can be summarized as:

$f(x_{i})$ belongs to such space if it is smooth and
\begin{eqnarray}
(\forall\alpha,~\forall\beta)~~~~~~ \lim_{x_{1},\dots x_{k}\to\infty}\left(x_{1}^{n_{1}}\dots x_{k}^{n_{k}}\frac{\partial^{\alpha}f(x_{i})}{\partial x_{1}^{m_{1}}...\partial x_{k}^{m_{k}}}\right)=0,\\
\sum_{i=1}^{k}n_{i}=\beta, ~~~ \sum_{j=1}^{k}m_{j}=\alpha.~~~~~~~~~~~~~~~~~~~~~~~~~~~~\nonumber
\end{eqnarray}
This is a topological space with topology given by a set of semi-norms $P_{\alpha,\beta}$ defined\footnote{We admit that the notation $P_{\alpha,\beta}$ might be somewhat misleading, since the $\alpha,\beta$ values do not specify the semi-norm in a unique way.} as
\begin{eqnarray}
P_{\alpha,\beta}=\sup_{x_{i}}\left|x_{1}^{n_{1}}\dots x_{k}^{n_{k}}\frac{\partial^{\alpha}f(x_{i})}{\partial x_{1}^{m_{1}}...\partial x_{k}^{m_{k}}}\right|,\\
\sum_{i=1}^{k}n_{i}=\beta, ~~~ \sum_{j=1}^{k}m_{j}=\alpha.~~~~~~~~~~~~~~~~~~~\nonumber
\end{eqnarray}
Naturally, the space of continuous and linear maps on such a space is more restricted as in the first, more common formulation. Objects belonging to such duals are called ``tempered distributions''. The advantage of this more restricted version is that Fourier transform is a well defined mapping on the space of tempered distributions (see \cite{Neumann}).

If we refer to the more common, first formulation, the space of distributions accommodates the linear
space of smooth functions by the mapping:

\begin{equation}\label{map}
f(x_{i})\to\int_{\mathbb{R}^{n}}f(x_{i})\cdots d^{n}x.
\end{equation}

Here $f(x_{i})$ is a smooth function on $\mathbb{R}^{n}$ and
$\int_{\mathbb{R}^{n}}f(x)\cdots dx$ is a distribution defined as the
mapping:

\begin{equation}
\Psi(x_{i})\to\int_{\mathbb{R}^{n}}f(x_{i})\Psi(x_{i})~d^{n}x, ~~~~\Psi(x_{i})\in D(\mathbb{R}^{n}).
\end{equation}

Moreover, by use of this mapping one can map into the space of
distributions any Lebesgue integrable function (injectively up to a
function the absolute value of which has Lebesgue integral 0) . The
distributional objects defined by the images of the map \eqref{map}
are called regular distributions. The map \eqref{map} always
preserves the linear structure, hence the space of Lebesgue integrable functions is a linear
subspace of the space of distributions. But the space of distributions
is a larger space than the space of regular distributions. This can be
easily demonstrated by defining the delta distribution

\begin{equation}
\delta(\Psi)\equiv\Psi(0)
\end{equation}
and showing that such mapping
cannot be obtained by a regular distribution. Note particularly that delta
distribution had an immediate use in physics in the description of
point-like sources. (Its intuitive use in the work of Paul A.M. Dirac before the theory of Schwartz distributions was found, was one of the main physics reasons why people searched for such language extension.)

These considerations show that the space of distributions is \emph{significantly} lar\-ger than the space
of smooth functions. The space of distributions is classically taken to be a
topological space with the weak ($\sigma-$) topology. In this
topology the space of regular distributions given by smooth
functions is a dense set.

Moreover one can continuously extend the derivative operator from the space
of $C^{1}(\mathbb{R}^{n})$ functions to the space of distributions by using the definition:

\begin{equation}
T_{,~x^{i}}(\Psi)\equiv -T(\Psi_{,~x^{i}}),
\end{equation}
(where $T$ denotes a distribution). This means that $C^{\infty}(\mathbb{R}^{n})$ functions form not only a
linear subspace in the space of distributions, but also a differential
linear subspace. It looks like there stands ``almost'' nothing in the
way of fully and satisfactorily extending the language of
$C^{1}(\mathbb{R}^{n})$ functions to the language of Schwartz
distributions. Unfortunately there is still one remaining trivial
operation and this is the operation of multiplication. That means we
need to obtain some distributional algebra having as subalgebra the algebra of
Lebesgue integrable functions factorized by functions the absolute
value of which has Lebesgue interal 0.
Unfortunately, shortly after Laurent Schwartz formulated the theory
of distributions he proved the following ``no-go'' result \cite{Schwarz}:

The requirement of constructing an algebra that fulfills the following three conditions is inconsistent:

\begin{itemize}
\item[a)] the space of distributions is linearly embedded into the
algebra,
\item[b)] there exists a linear derivative operator, which fulfills
the Leibniz rule and reduces on the space of distributions to the
distributional derivative,
\item[c)] there exists a natural number $k$ such, that our algebra has as subalgebra the algebra of $C^{k}(\mathbb{R}^{n})$
functions.
\end{itemize}
This is called the Schwartz impossibility theorem. There is a nice
example showing where the problem is hidden: Take the Heaviside
distribution $H$. Suppose that $a)$ and $b)$ hold and we multiply
the functions/distributions in the usual way. Then since ~$H^{m}=H$ the following must hold:

\begin{equation}
H'=(H^{m})'=\delta=mH^{m-1}\delta~.
\end{equation}
But this actually implies that $\delta=0$, which is nonsense.

The closest one can get to fulfill the conditions $a)-c)$ from the
Schwarz impossibility theorem is the Colombeau algebra (as defined in the next section) where
conditions $a)- c)$ are fulfilled with the exception that in the $c)$
condition $k$ is taken to be infinite. This means that only the smooth
functions form a subalgebra of the Colombeau algebra. This is obviously not
satisfactory, since we know (and need to recover) rules for multiplying
multiply much larger classes of functions than only smooth functions. In the case of Colombeau algebras this problem is ``resolved'' by
the equivalence relation, as we will see in the following
section.

\subsection{The standard $\mathbb{R}^{n}$ theory of Colombeau
algebras}

\subsubsection{The \emph{special} Colombeau algebra and the embedding of distributions}
The so called special Colombeau algebra is on
$\mathbb{R}^{n}$ defined as:
\begin{equation}
\mathcal{G}(\mathbb{R}^{n})=
\mathcal{E}_{M}(\mathbb{R}^{n})/\mathcal{N}(\mathbb{R}^{n}).
\end{equation}
Here $\mathcal{E}_{M}(\mathbb{R}^{n})$ (moderate functions) is
defined as the algebra of functions:
\begin{equation}
\mathbb{R}^{n} \times(0,1]\to\mathbb{R}
\end{equation}
 that are smooth on $\mathbb{R}^{n}$ (this is usually called $\mathcal{E} (\mathbb{R}^{n})$), and for any compact subset $K$ of ~$\mathbb{R}^{n}$ (for which we will henceforth use the notation ~$K\subset\subset\mathbb{R}^{n}$) it  holds that:
\begin{equation}
\forall\alpha\in\mathbb{N}^{n}_{0},~\exists p\in\mathbb{N} ~\hbox{
such that}\footnote{In this definition, (and also later in the text), the symbol $\mathbb{N}^{n}_{0}$ denotes a sequence of $n$ members formed of natural numbers (with 0 included).}  ~\sup_{x\in K}|D^{\alpha}f_{\epsilon}(x_{i})|\leq
O(\epsilon^{-p}) ~\hbox{ as } ~\epsilon\to
 0.
\end{equation}
The $\mathcal{N}(\mathbb{R}^{n})$ (negligible functions) are
functions from $\mathcal{E}(\mathbb{R}^{n})$ where for any
$K\subset\subset\mathbb{R}^{n}$ it holds that:
\begin{equation}
\forall\alpha\in\mathbb{N}^{n}_{0},~\forall p\in\mathbb{N} ~\hbox{
we have } ~\sup_{x_{i}\in K}|D^{\alpha}f_{\epsilon}(x_{i})|\leq
O(\epsilon^{p}) ~\hbox{ as } ~\epsilon\to 0.
\end{equation}
 The first definition tells us that moderate functions are those whose partial derivatives of arbitrary degree (with respect to variables $x_{i}$) do not diverge faster then any arbitrary negative power of $\epsilon$, as $\epsilon\to 0$. Negligible functions are those moderate functions whose partial derivatives of arbitrary degree go to zero faster than any positive power of $\epsilon$, as $\epsilon\to 0$. This simple formulation can be straightforwardly generalized
into general manifolds just by substituting the concept of Lie derivative for the ``naive'' derivative
used before.

It
can be shown, by using convolution with an arbitrary smoothing
kernel (or mollifier), that we can embed a distribution into the Colombeau algebra. By a
smoothing kernel we mean, in the widest sense a compactly supported, smooth
function $\rho_{\epsilon}(x_{i})$, with $\epsilon\in (0,1]$, such that:
\begin{itemize}
\item
$\supp(\rho_{\epsilon})\to\{0\}$ ~~for ~~($\epsilon\to 0$),
\item
$\int_{\mathbb{R}^{n}}\rho_{\epsilon}(x_{i})~d^{n}x\to 1$~~~ for ~~($\epsilon\to 0$),
\item
$\forall\eta >0$ ~~$\exists C$,~~
$\forall\epsilon\in (0,\eta)$~~~ $\sup_{x_{i}}
|\rho_{\epsilon}(x_{i})|<C$.
\end{itemize}
This most generic embedding approach is mentioned for example in
\cite{Penrose}~(in some sense also in \cite{geodesic}).

More ``restricted'' embeddings to $\mathcal{G}(\mathbb{R}^{n})$ are
also commonly used. We can choose for instance a subclass of
mollifiers called ~$A^{0}(\mathbb{R}^{n})$, which are smooth
functions from ~$D(\mathbb{R}^{n})$ (smooth, compactly supported)
and (i.e. \cite{Gsponer}) such that

\begin{equation}
\forall\epsilon~~\hbox{holds}~~\int_{\mathbb{R}^{n}}\rho_{\epsilon}(x_{i})d^{n}x =
 1.
\end{equation}
Their dependence on $\epsilon$ is given\footnote{Later in the text will the notation $\rho_{\epsilon}(x_{i}), \psi_{\epsilon}(x_{i})$ (etc.) automatically mean dependence on the variable $\epsilon$ as in \eqref{epsilon}.} as
\begin{equation}\label{epsilon}
\rho_{\epsilon}(x_{i})\equiv \frac{1}{\epsilon^{n}}\rho\left(\frac{x_{i}}{\epsilon}\right)~.
\end{equation}

Sometimes the class is even more restricted. To
obtain such a formulation, we shall define classes ~$A^{m}(\mathbb{R}^{n})$ as classes
of smooth, compactly supported functions, such that
\begin{equation}
 \int_{\mathbb{R}^{n}}
x^{i}_{1}\cdots x^{j}_{l}~\phi(x_{i})d^{n}x=\delta_{0k}~~ \hbox{ for }~~
i+\cdots+j=k\leq m.
\end{equation}
Clearly $A^{m+1}(\mathbb{R}^{n})\subset
A^{m}(\mathbb{R}^{n})$. Then the most restricted class of mollifiers
is taken to be the class $A^{\infty}(\mathbb{R}^{n})$. This approach
is taken in the references
\cite{Schwarzschild, Oberguggenberger, Reis-Nor, waves, Vickers}.

Even in the case of the more restricted class of mollifiers the embeddings are generally non-canonical \cite{Gsponer, Vickers}. The exception are smooth distributions, where the difference between two embeddings related to two different mollifiers is always a negligible function.

\subsubsection{The \emph{full} Colombeau algebra and the embedding of distributions}

What is usually considered
to be the canonical formulation of Colombeau algebras in
$\mathbb{R}^{n}$ is the following: The theory is formulated in terms
of functions
\begin{equation}
\mathbb{R}^{n} \times A^{0}(\mathbb{R}^{n})\to\mathbb{R} ~~~\hbox{ (call
them $F$) }.
\end{equation}
The Colombeau algebra is defined in such way that it is a factor
algebra of moderate functions over negligible functions, where:

\begin{itemize}
\item
Moderate functions are functions from $F$ that satisfy:
\begin{eqnarray}
\forall m\in \mathbb{N}^{n}_{0}, ~\forall K\subset\subset\mathbb{R}^{n} ~~\exists N\in\mathbb{N} ~\hbox{ such that
if } ~\phi\in A^{N}(\mathbb{R}^{n}), \hbox{ there are } \alpha, \rho
>0, \nonumber\\  \hbox{ such that }   \sup_{x_{i}\in K}~\left
|D^{m}F(\phi_{\epsilon},x_{i})\right|\leq
\alpha\epsilon^{-N} ~\hbox{ if }~ 0<\epsilon<\rho.~~~~~~~~~~~~~~~~
\end{eqnarray}
\item
Negligible functions are functions from $F$ that satisfy:
\begin{eqnarray}
  \forall m\in\mathbb{N}, ~\forall K\subset\subset\mathbb{R}^{n},~~\forall p\in\mathbb{N} ~~\exists q\in\mathbb{N} ~\hbox { such that if } ~~\phi\in A^{q}(\mathbb{R}^{n}), ~\exists\alpha,\rho>0, \nonumber\\ \hbox { we have } \sup_{x_{i}\in K}\left|D^{m}F(\phi_{\epsilon},x_{i})\right|\leq\alpha\epsilon^{p} ~\hbox { if }~ 0<\epsilon<\rho.~~~~~~~~~~~~~~~~~~~~~~~~
\end{eqnarray}

\end{itemize}

Then ordinary distributions automatically define such functions by
the convolution (\cite{multi, Damyanov, geodesic, Vickers}
etc.):\footnote{Here the ``$B_{x}$'' notation means that $x$ is the
variable removed by applying the distribution.}
\begin{equation}
 B\to
B_{x}\left[\frac{1}{\epsilon^{n}}\;
\phi\left(\frac{y_{i}-x_{i}}{\epsilon}\right)\right], ~~\phi\in
A^{0}(\mathbb{R}^{n}).
\end{equation}

\subsubsection{Important common feature of both formulations}

All of these formulations have two important consequences. Given that $C$ denotes the embedding mapping:
\begin{itemize}
\item
Smooth functions ($C^{\infty}(\mathbb{R}^{n})$) form a
subalgebra of the Colombeau algebra
($C(f)C(g)=C(f \cdot g)$~ for $f$, $g$ being
smooth distributions).
\item
Distributions form a differential linear subspace of Colombeau algebra
(this means
for instance that $C(f')=C'(f)$~).
\end{itemize}

\subsubsection{The relation of equivalence in the \emph{special} Colombeau algebra}

We can formulate a relation of equivalence between an element of the
special Colombeau algebra $f_{\epsilon}(x_{i})$ and a distribution $T$. We call them equivalent,
if for any ~$\phi\in D(\mathbb{R}^{n})$,~ we have
\begin{equation}
\lim_{\epsilon\to
0}\int_{\mathbb{R}^{n}} f_{\epsilon}(x_{i})\phi(x_{i})d^{n}x= T(\phi).
\end{equation}
 Then two elements of
Colombeau algebra $f_{\epsilon}(x_{i})$, $g_{\epsilon}(x_{i})$ are equivalent,
if for any ~$\phi(x_{i})\in D(\mathbb{R}^{n})$
\begin{equation}
 \lim_{\epsilon\to
0}\int_{\mathbb{R}^{n}} \left(f_{\epsilon}(x_{i})-g_{\epsilon}(x_{i})\right)\phi(x_{i})~d^{n}x=0.
\end{equation}
For the choice of
$A^{\infty}(\mathbb{R}^{n})$ mollifiers the following relations are respected by
the equivalence:
\begin{itemize}
\item
It respects multiplication of distribution by a
smooth distribution \cite{Vickers} in the sense that: ~~$C(f\cdot g)\approx
C(f)\cdot C(g)$,~ where ~$f$ is a smooth distribution and $g\in D'(\mathbb{R}^{n})$.~
\item
It
respects (in the same sense) multiplication of piecewise continuous
functions (we mean here regular distributions given by
piecewise continuous functions) \cite{multi}.
\item
If $g$ is a distribution and $f\approx g$,
then for arbitrary natural number $n$~ it holds ~$D^{n}f\approx D^{n}g$ ~\cite{Damyanov}.
\item
If $f$ is
equivalent to distribution $g$, and if $h$ is a smooth
distribution, then $f\cdot h$ is equivalent to $g\cdot h$ ~\cite{Damyanov}.
\end{itemize}

\subsubsection{The relation of equivalence in the \emph{full} Colombeau algebra}

In the canonical formulation the
equivalence relation is again formulated either between an element of the Colombeau
algebra and a distribution, or analogously between two elements of
the Colombeau algebra:
If there  $\exists m$, such that for any
$\phi\in A^{m}(\mathbb{R}^{n})$, and for any
~$\Psi(x_{i})\in D(\mathbb{R}^{n})$, it holds that
\begin{equation}
\lim_{\epsilon\to 0}\int_{\mathbb{R}^{n}}
(f(\phi_{\epsilon},x_{i})-g(\phi_{\epsilon},x_{i}))\Psi(x_{i})~d^{n}x=0,
\end{equation}
then we say that $f$ and $g$ are equivalent ($f\approx g$).
For the canonical
embedding and differentiation we have the same commutation relations as in the
non-canonical case. It can be also
proven that for $f_{1}\dots f_{n}$ being regular distributions given
by piecewise continuous functions it follows that
\begin{equation}
C(f_{1})...C(f_{n})\approx C(f_{1}...f_{n}),
\end{equation}
 and for $f$ being
arbitrary distribution and $g$ smooth distribution it holds that
\begin{equation}
C(f)\cdot C(g)\approx C(f\cdot g).
\end{equation}

\subsubsection{How this relates to some older Colombeau papers}

In
older Colombeau papers \cite{multi, Amulti} all these concepts are formulated (equivalently)
as the relations between elements of a Colombeau algebra taken as a
subalgebra of the ~$C^{\infty}(D(\mathbb{R}^{n}))$ algebra. The definitions of
moderate and negligible elements are almost exactly the same as in
the canonical formulation, the only difference is that their domain is taken here to be the class
~$D(\mathbb{R}^{n})\times\mathbb{R}^{n}$~ (being a larger domain than $A^{o}(\mathbb{R}^{n})\times\mathbb{R}^{n}$). The canonical formulation is related to the elements of the class $C^{\infty}(D(\mathbb{R}^{n}))$ through their convolution with the objects from the class
$D(\mathbb{R}^{n})$.  It is easy to see that you can
formulate all the previous relations as relations between elements of the ~$C^{\infty}(D(\mathbb{R}^{n}))$
subalgebra (with pointwise multiplication), containing also distributions.

\subsection{Distributions in the geometric approach}

This part is devoted to review the distributional theory in the geometric framework. How to define arbitrary rank tensorial distribution
on arbitrary manifolds by avoiding reference to preferred charts? Usually we mean by a
distribution representing an $(m,n)$ tensor field an element from the
dual to the space of objects given by the tensor product of $(m,n)$ tensor fields and smooth
compactly supported $k$-form fields
 (on $k$ dimensional space). That means for example a regular distributional $(m,n)$ tensor field $B^{\mu...}_{\nu...}$ is introduced as a map
\begin{equation}
T\otimes\omega\to\int B^{\nu...\beta}_{\mu...\alpha}~T^{\mu...\alpha}_{\nu...\beta}~\omega ~.
\end{equation}
This is very much the same as to say that the test space are smooth
compactly supported tensor densities
$T^{\mu....\alpha}_{\nu...\beta}$
\cite{Garfinkle, GerTra}. The topology taken on this space is the
usual topology of uniform convergence for arbitrary derivatives
related to arbitrary charts (so the convergence from $\mathbb{R}^{n}$ theory should be valid in
all charts). The derivative operator acting on this space is typically Lie derivative. (Lie derivative along a smooth vector field $\xi$ we denote $L_{\xi}$.) This does make sense, since:
\begin{itemize}
 \item
To use derivatives of distributions we
automatically need derivatives along vector fields.
 \item
Lie derivative preserves $p$-forms.
 \item
In case of Lie derivatives, we do not need to apply any additional geometric structure (such as connection in the case of covariant derivative).
\end{itemize}

There is an equivalent formulation to
\cite{GerTra}, given by \cite{Vickers}, which
takes the space of tensorial distributions to be ~~$D'(M)\otimes T^{m}_{n}(M)$.~~ Here $D'(M)$ is the
dual to the space of smooth, compactly supported $k$-form fields ($k$
dimensional space). Or in other words, it is the space of sections with distributional coefficients. In \cite{geometry} the authors
generalize the whole construction by taking the tensorial
distributions to be the dual to the space of compactly supported
sections of the bundle ~$E^{*}\otimes ~V\!ol^{1-q}$. Here
~$V\!ol^{1-q}$~ is a space of ~($1-q$)-densities and ~$E^{*}$~ is a
dual to a tensor bundle ~$E$~ (hence to the dual belong objects
given as ~$E\otimes ~V\!ol^{q}$). In all these formulations Lie
derivative along a smooth vector field represents the differential
operator\footnote{We can mention also another classical formulation
of distributional form fields, which comes from the old book of
deRham \cite{deRham}~(it uses the expression ``current''). It
naturally defines the space of distributions to be a dual space to
space of all compactly supported form fields.}.

Let us mention here that there is one
unsatisfactory feature of these constructions, namely that for
physical purposes we need much more to incorporate the concept
of the covariant derivative rather than the Lie derivative. There
was some work done in this direction \cite{Hartley, Marsden, Wilson}, but it
is a very basic sketch, rather than a full and satisfactory theory.
It is unclear (in the papers cited) how one can obtain for the
covariant derivative operator the expected and meaningful results
outside the class of smooth tensor fields.

\subsection{Colombeau algebra in the geometric approach}

\subsubsection{Scalar \emph{special} Colombeau algebra}

For arbitrary general manifold it is easy to find a covariant formulation of the \emph{special}
Colombeau algebra. At the end of the day you obtain a
space of $\epsilon$-sequences of functions on the general manifold $M$.
For the non-canonical case the definitions are
similar to the non-geometric formulations, the basic difference
is that Lie derivative plays here the role of the $\mathbb{R}^{n}$
partial derivative. Thus the
definition of the Colombeau algebra will be again:
\begin{equation}
\mathcal{G}(\Omega)=\mathcal{E}_{M}(\Omega)/\mathcal{N}(\Omega).
\end{equation}
$\mathcal{E}_{M}(\Omega)$ (moderate functions) are defined as
algebra of functions~ $\Omega \times (0,1]\to \mathbb{R}$, such that
are smooth on $\Omega$ (this is usually called $\mathcal{E}
(\Omega)$) and for any $K\subset\subset\Omega$ we insist that
\begin{eqnarray}
\forall k\in\mathbb{N}^{n}_{0},~\exists p\in\mathbb{N}~ \hbox{
such that } ~\forall \xi_{1}...\xi_{k}~ \hbox{ which are smooth
vector fields, }\nonumber\\  ~\sup_{x\in
K}|L_{\xi_{1}}...L_{\xi_{k}}f_{\epsilon}(x)|\leq O(\epsilon^{-p})
~\hbox{ as } ~\epsilon\to 0.~~~~~~~~~~~~~~~~~~~~~~~~~~~~~~\label{Def.}
\end{eqnarray}
$\mathcal{N}(\Omega)$ (negligible functions) we define exactly in
the same analogy to the non-geometric formulation\footnote{This
version is due to \cite{Schwarzschild}. There are also different
definitions: in \cite{geometry} the authors use instead of ``for
every number of Lie derivatives along all the possible smooth vector
fields'' the expression ``for every linear differential operator'',
but they prove that these definitions are equivalent. This is also
equivalent to the statement that in any chart holds:
$\Phi\in\mathcal{E}_{M}(\mathbb{R}^{n})$ (see
 \cite{geometry}).}:
\begin{eqnarray}
\forall K\subset\subset\Omega,~\forall k\in\mathbb{N}^{n}_{0},~\forall p\in\mathbb{N},~~\forall \xi_{1}...\xi_{k}~ \hbox{ which are smooth
vector fields, } \nonumber\\  ~\sup_{x\in
K}|L_{\xi_{1}}...L_{\xi_{k}}f_{\epsilon}(x)|\leq O(\epsilon^{p})
~\hbox{ as } ~\epsilon\to 0.~~~~~~~~~~~~~~~~~~~~~~~~~~~~~~\label{Def.}
\end{eqnarray}

\subsubsection{Tensor \emph{special} Colombeau algebra and the embedding of distributions}

After one defines the scalar special Colombeau algebra, it is easy to define the generalized Colombeau tensor
algebra as the tensor product of sections of a tensor bundle and Colombeau algebra. This
can be formulated more generally \cite{geometry} in terms of
maps from $M$ to arbitrary manifold. One can define them by changing the absolute value in the definition \eqref{Def.} to the expression ``any
Riemann measure on the target space''. Then you get the algebra \cite{geometry, generalized, Vickers}:
\begin{equation}
\mathbf{\Gamma_{C}}(X,Y)=\mathbf{\Gamma}_{M}(X,Y)/\mathbf{N}(X,Y).
\end{equation}
The tensor fields are represented when the target space is taken to
be the $TM$ manifold\footnote{This is equivalent to
~$\mathcal{G}(X)\otimes\Gamma(X,E)$ tensor valued Colombeau
generalized functions \cite{generalized}.}. It is clear that any
embedding of distributions into such algebra will be non-canonical
from various reasons. First, it is non-canonical even on
$\mathbb{R}^{n}$. Another, second reason is that this embedding will
necessarily depend on some preferred class of charts on $M$. The
embedding one defines as \cite{geometry}:

We pick an atlas, and take a smooth partition of
unity subordinate to ~$V_{\alpha_{i}}$,~ ($\Theta_{j}$,\\~
$\supp(\Theta)\subseteq V_{\alpha_{j}}$ $j\in\mathbb{N}$) and we
choose for every $j$,~ $\xi_{j}\in D(V_{\alpha_{j}})$, such that
$\xi_{j}=1$~ on ~$\supp(\xi_{j})$.~ Then we can choose in fixed charts
an $A^{\infty}(\mathbb{R}^{n})$ element $\rho$, and the embedding is given by
\begin{equation}
\sum^{\infty}_{j=1}\left(((\xi_{j}(\Theta_{j}\circ\psi^{-1}_{\alpha_{j}})u_{\alpha_{j}})*\rho_{\epsilon})\circ
\psi_{\alpha_{j}}\right)_{\epsilon},
\end{equation}
 where $\psi_{\alpha_{j}}$ is a
coordinate mapping and ~$*$~ is a convolution.

\subsubsection{The equivalence relation}

Now let us define the equivalence relation in analogy to the $\mathbb{R}^{n}$
case. Since in \cite{geometry} the strongest constraint
on the mollifier is taken, one would expect that strong
results will be obtained, but the definition is more
complicated. And in fact, standard results (such as
embedding of smooth function multiplying distribution is equivalent to
product of their embeddings) are not valid here
\cite{geometry}. That is why the stronger concept of
$k$-association is formulated. It states that $U\in\Gamma_{C}$ is $k$ associated to
function $f$, if ~~

\begin{equation}
\lim_{\epsilon\to 0} L_{\xi_{1}}...L_{\xi_{l}}
(U_{\epsilon}-f)\to 0
\end{equation}
 uniformly on compact sets for all ~$l\leq
k$. The cited paper does not contain a precise definition of $k$ equivalence
between two generalized functions, but
it can be easily derived.

\subsubsection{The older formulation of scalar \emph{full} Colombeau algebra}

If we want to get a canonical formulation, we certainly cannot
generalize it straight from the $\mathbb{R}^{n}$ case (the reason is
that the definition of the classes $A^{n}(\mathbb{R}^{n})$ is not
diffeomorphism invariant). However, there is an approach providing
us with a canonical formulation of generalized scalar fields
\cite{global}. The authors define the space $\mathcal{E}(M)$ as a
space of $C^{\infty}(M~\times ~A^{0}(M))$, where $A^{0}(M)$ is the
space of $n$-forms ($n$-dimensional space), such that ~$\int
\omega=1$. Now the authors define a smoothing kernel as $C^{\infty}$
map from
\begin{equation}
M\times I\to A^{0}(M),
\end{equation}
such that it satisfies:
\begin{itemize}
 \item[(i)] $\forall K\subset\subset M~~ \exists\epsilon_{0},~~ C>0~~ \forall p\in K~~ \forall\epsilon\leq\epsilon_{0},~~ \supp~\phi(\epsilon, p)\subseteq B_{\epsilon C}(p)$,
 \item[(ii)] $\forall K\subset\subset M, ~~\forall k,l\in\mathbb{N}_{0}, ~~\forall X_{1},\cdots X_{k},Y_{1}\cdots Y_{l}$~ smooth vector fields,\\ $\sup_{p\in K, q\in M}\left\Vert L_{Y_{1}}\cdots L_{Y_{l}}(L'_{X_{1}}+L_{X_{k}})\cdots (L'_{X_{k}}\cdots L_{X_{k}})\Phi(\epsilon,p)(q)\right\Vert = \\ = O(\epsilon^{-(n+1)})$.
\end{itemize}
Here $L'$ is defined as: ~
\begin{equation}
L'_{X}f(p,q)=L_{X}(p\to f(p,q))=\frac{d}{dt}f((Fl^{x}_{t})(p),q)|_{0}.
\end{equation}
$B_{\epsilon C}$ is a ball centered at $C$ having radius
${\epsilon}$ measured relatively to arbitrary Riemannian metric. Let
us call the class of such smoothing kernels ~$A^{0}(M)$. Then in
\cite{global} classes $A^{m}(M)$ are defined as the set of all
$\Phi\in A^{0}(M)$ such that $\forall f\in C^{\infty}(M)$ and
$\forall K\subset \subset M$ (compact subset) it holds:
\begin{equation}
\sup\left| f(p)-\int_{M}f(q)\Phi(\epsilon,p)(q)\right|=O(\epsilon^{m+1}).
\end{equation}
Moderate and the negligible functions are defined in the following way:\\
$R\in \mathcal{E}(M)$ is moderate if ~~$\forall K\subset\subset
M~~\forall k \in\mathbb{N}_{0} ~~\exists N \in \mathbb{N}~~\forall
X_{1},....X_{k}$ ($X_{1},....X_{k}$ are smooth vector fields) and
~$\forall \Phi\in A^{0}(M)$~ one has:
\begin{equation}
\sup_{p\in
K}~\left\Vert L_{X_{1}}.....L_{X_{k}}(R(\Phi(\epsilon,p),p))\right\Vert=O(\epsilon^{-N}).
\end{equation}
$R\in\mathcal{E}(M)$ is negligible if ~~$\forall K\subset\subset
M,~\forall k,l\in\mathbb{N}_{0}~~\exists m\in\mathbb{N}~~\forall
X_{1},...X_{k}$ ($X_{1},...X_{k}$ are again smooth vector fields)
and ~$\forall\Phi\in A_{m}(M)$~ one has:
\begin{equation}
\sup_{p\in
K}~\left\Vert L_{X_{1}}...L_{X_{k}}(R(\Phi(\epsilon,p),p)\right\Vert=O(\epsilon^{l}).
\end{equation}
Now we can define the Colombeau algebra in the usual way as a factor
algebra of moderate functions over negligible functions. Scalar distributions, defined
as dual to $n$-forms, can be embedded into such algebra in a complete analogy
to the canonical $\mathbb{R}^{n}$ formulation. Also association is in this case
defined in the ``usual'' way (integral with compactly supported
smooth $n$-form field) and has for multiplication the usual properties.
However any attempt to get a straightforward generalization from scalars to tensors
brings immediate problems, since the embedding does not commute
with the action of diffeomorphisms. This problem was finally resolved in
\cite{CanTen}.

\subsubsection{Tensor \emph{full} Colombeau algebra and the embedding of distributions}

The authors of reference \cite{CanTen} realized that
diffeomorphism invariance can be achiev\-ed by adding some background structure
defining how tensors transport from point to point, hence a
transport operator. Colombeau  $(m,n)$ rank tensors are then taken
from the class of smooth maps ~~$C^{\infty}(\omega, q, B)$~~having values in
~$(T^{m}_{n})_{q}M$, ~~ where ~$\omega\in A^{0}(M)$, $q\in M$  ~and~
$B$~ is from the class of compactly supported transport operators.
After defining how Lie derivative acts on such objects and
the concept of the ``core'' of a transport operator, the authors of reference \cite{CanTen} define (in a
slightly complicated analogy to the previous case) the moderate and the
negligible tensor fields. Then by usual factorization they obtain the
canonical version of the generalized tensor fields (for more details
see \cite{CanTen}). The canonical embedding of tensorial
distributions is the following:
The smooth tensorial objects are embedded as

\begin{equation}
\tilde t(p,\omega,B)=\int t(q)B(p,q)\omega(q) dq
\end{equation}
where as expected $\omega\in A^{0}(M)$, $t$ is the smooth tensor field and
$B$ is the transport operator. Then the arbitrary tensorial distribution $s$ is embedded (to
$\tilde s$) by the condition

\begin{equation}
\tilde s(\omega, p, B)\cdot t(p)=\left(s, B(p,.)\cdot
t(p)\otimes\omega(.)\right),
\end{equation}
where on the left side we are contracting the embedded object with a
smooth tensor field $t$, and on the right side we are applying the
given tensorial distribution $s$ in the variable assigned by the
dot. It is shown that this embedding fulfills all
the important properties, such as commuting with the Lie
derivative operator \cite{CanTen}. All the other results related to equivalence relation
(etc.) are obtained in complete analogy to the previous cases.

\subsubsection{The generalized geometry (in \emph{special} Colombeau algebras)}

In \cite{geometry, Connection, generalized} the authors generalized all the basic geometric structures, like connection, covariant derivative, curvature, or geodesics into the geometric formulation of the special Colombeau algebra. That means they defined the whole generalized geometry.

\subsubsection{Why is this somewhat unsatisfying?}

 However in our view, the crucial part is missing. What we would like to see is an intuitive and clear definition of the covariant derivative operator acting on the distributional objects in the canonical Colombeau algebra formulation, on one hand reproducing all the classical results, and on the other hand extending them in the same natural way as in the classical distributional theory with the classical derivative operator. Whether there is any way to achieve this goal by the concept of generalized covariant derivative acting on the generalized tensor fields, as defined in \cite{generalized}, is unclear. Particularly it is not clear whether such generalized geometry can be formulated also within the canonical Colombeau algebra approach. There exist definitions of the covariant derivative operator within the distributional tensorial framework \cite{Hartley, Marsden}. (The reference \cite{Hartley} gives particularly nice application of such distributional tensor theory to signature changing spacetimes.) But these approaches are still ``classical'', in the sense, that they do not fully involve the operation of the tensor product of distributional tensors.
There cannot be any hope of
finding a more appropriate, generalized formulation of classical physics
without finding such a clear and intuitive definition of both, the covariant derivative and the tensor product. All that can be in this situation achieved is
to use these constructions to solve some specific problems within
the area of physics. But as we see, the more ambitious goal can be
very naturally achieved by our own construction, which follows after
this overview.

\subsection{Practical application of the standard results}

Now we will briefly review various applications of
the Coulombeau theory presented before.

\subsubsection{Classical shock waves}
The first application we will mention is the non-general-relativistic one. In
\cite{shock} the authors provide us with weak solutions of
nonlinear partial differential equations (using Colombeau algebra)
representing shock waves. They use a special version of the
Colombeau algebra, and specifically the relation
\begin{equation}
H^{n}H'\approx \frac{1}{n+1}~H',
\end{equation}
(which is related to mollifiers from
$A^{0}(\mathbb{R}^{n})$ class). The more general analysis related to the existence and uniqueness of
weak solutions of nonlinear partial differential
equations can be found in \cite{Oberguggenberger}.

\subsubsection{Black hole ``distributional'' spacetimes}
In general relativity there are many results obtained by the use of Colombeau
algebras. First, we will focus on the distributional Schwarzschild geometry, which is
analysed for example in \cite{Schwarzschild}. The authors of \cite{Schwarzschild} start to
work in Schwarzschild coordinates using the special Colombeau algebra
and $A^{\infty}(\mathbb{R}^{n})$ classes of mollifiers. They obtain
the delta-functional results (as expected) for the Einstein tensor, and
hence also for the stress-energy tensor. But in Schwarz\-schild coordinates there
are serious problems with the embedding of the distributional tensors, since these coordinates do not
contain the 0 point. As a result, if one looks for smooth embeddings, one does
not obtain an inverse element in the Colombeau algebra in the
neighbourhood of 0 for values of $\epsilon$ close to 0. (Although there is no problem, if we
require that the inverse relation should apply only in the sense of
equivalence.) Progress can be made by
turning to Eddington-Finkelstein coordinates \cite{Schwarzschild}. The metric is obtained in Kerr-Schild form, in which one is able to compute $R^a{}_b,
G^a{}_b$ and hence $T^a{}_b$ as delta-functional objects (which is expected). (The (1,1) form of the field equations is used since the metric dependence has a relatively simple form in Eddington-Finkelstein coordinates.) This
result does not depend on the mollifier (see also \cite{Vickers}),
but one misses the analysis of the relation between different embeddings given by the different
coordinate systems \footnote{It seems that
the authors use relations between $R^\mu{}_\nu$ and components of
$g^{\mu\nu}$ obtained in Eddington-Finkelstein coordinates by using algebraic tensor
computations. Then it is not obvious, whether these results can be
obtained by computation in the Colombeau algebra using the $\approx$
relation, since in such case some simple tricks (such as substitution) cannot in general be used.}.
Even in the case of Kerr geometry there is a computation
of $\sqrt{(g')}R^\mu{}_\nu$ (where $g^{ab}=g'^{ab}+fk^{a}k^{b}$)
given by Balasin, but this is
mollifier dependent \cite{Kerr, Vickers}. Here the
coordinate dependence of the results is even more
unclear.

\subsubsection{Aichelburg metric}
There exists an ultrarelativistic weak limit of the Schwarzschild
metric. It is taken in Eddington-Finkelstein coordinates $u=t+r^{*}$,
$v=t-r^{*}$, (where $r^{*}$ is tortoise coordinate) by taking the boost in the weak $v\to c$ limit.
We obtain the ``delta functional'' Aichelburg metric. Reference \cite{geodesic} provides a computation of geodesics in such a geometry. The authors of \cite{geodesic} take the special Colombeau algebra (and take
$A^{0}(\mathbb{R}^{n})$ as their class of mollifiers), and they prove
that geodesics are given by the refracted lines. The results are mollifier
independent. This is again expected. Moreover, what seems to be really
interesting is that there is a continuous metric which is connected with
the Aichelburg metric by a generalized coordinate transformation
\cite{Penrose, Vickers}.

\subsubsection{Conical spacetimes}
The other case we
want to mention are conical spacetimes. One of the papers where the conical spacetimes are analysed
is an old paper of Geroch and Traschen \cite{GerTra}. In \cite{GerTra} it is shown that conical spacetimes can not be analysed through the concept of
$gt$-metric. These are metrics which provide us with a distributional Ricci
tensor in a very naive sense. The multiplication is given just by a simple product
of functions defining the regular distributions. A calculation of the stress energy tensor was given by Clarke, Vickers and Wilson \cite{Clarke}, but this is mollifier dependent (although it is
coordinate independent \cite{Vickers}).

\section{How does our approach relate to current theory?}

\paragraph{General summary}
How does our own approach (to be described in detail in the next section) relate to all what has presently been achieved in the Colombeau theory? We can summarize what we will do in three following points:

\begin{itemize}
\item
We will define tensorial distributional objects, and the basic related operations (especially the \emph{covariant derivative}). The definition directly follows our physical intuition (there is a unique way of constructing it). This generalizes the Schwartz $\mathbb{R}^{n}$ distribution theory.

\item
We will formulate the Colombeau equivalence relation in our approach and obtain all the usual equivalence results from the Colombeau theory.

\item
We will prove that the classical results for the covariant derivative operator (as known within differential geometry) significantly generalize in our approach.
\end{itemize}
The most important point is that our approach is fully based on the Colombeau equivalence relation translated to our language. That means we take and use only this particular feature of the Colombeau theory and completely avoid the Colombe\-au algebra construction.

\paragraph{The advantages compared to usual Colombeau theory}
What are the advantages of this approach comparing to Colombeau theory?
\begin{itemize}
\item
First, by avoiding the algebra factorization (which is how Colombeau algebra is constructed) we fulfill the physical intuitiveness condition of the language used.
\item
Second, we can naturally and easily generalize the concept of the covariant derivative in our formalism, which has not been completely satisfactorily achieved by the Colombeau algebra approach\footnote{As previously mentioned, the current literature lists some ambiguous attempts to incorporate the covariant derivative, but no satisfactory theory.}. This must be taken as an absolutely necessary condition that any generalization of a fundamental physical language must fulfill.
\item
It is specifically worth discussing the third advantage:  Why is the classical approach so focused on Colombeau
algebras? The answer is simple: We want to get an algebra of
$C^{\infty}(M)$ functions as a subalgebra of our algebra. (This is
why we need to factorize by negligible functions, and we need to get
the largest space where they form an ideal, which is the space of
moderate functions.) But there is one strange thing: all our efforts
are aimed at reaching the goal of getting a more rich space than is provided by the space
of smooth functions. But there is no way of getting a larger
differential subalgebra than the algebra of smooth functions, as
is shown by Schwartz impossibility result. That is why we use only
the equivalence relation instead of straight equality. But then the
question remains: Why one should still prefer smooth
functions? Is not the key part of all the theory the equivalence
relation? So why is it that we are not satisfied with the way the equivalence relation recovers multiplication of the smooth objects (we require something stronger), but we are satisfied with the way it recovers multiplication
within the larger class? Unlike the Colombeau algebra based approach, we are simply taking seriously the idea that one should treat all the objects in an equal way. This means we do not see any reason to try to achieve ``something more'' with smooth objects than we do with objects outside this class. And the fact that we treat all the objects in the same way provides the third advantage of our theory; it makes the theory much more natural than the Colombeau algebra approach.
\item
The fourth advantage is that it naturally works with the much more general (and for the language of distributions natural) concept of piecewise smooth manifolds, so the generalization of the physics language into such conceptual framework will give it much higher symmetry.
\item
The fifth and last, ``small'' advantage comes from the fact that by avoiding the Colombeau algebra construction we automatically remove the problem of how to canonically embed arbitrary tensorial distributions into the algebra. But one has to acknowledge that this problem was already solved also within the Colombeau theory approach \cite{CanTen}.
\end{itemize}

\paragraph{The disadvantages compared to usual Colombeau theory}
What are the disadvantages of this approach compared to Colombeau theory? A conservative person might be not satisfied with the fact that we do not have a smooth tensor algebra as a subalgebra of our algebra. This means that the classical smooth tensorial fields have to be considered to be solutions of equivalence relations only (as opposed to the classical, ``stronger'' view to take them as solutions of the equations). But in my view, this is not a problem at all. The
equivalence relations contain all the classical smooth tensorial solutions
(equivalently smooth tensor fields), for the smooth initial value
problem. So for the smooth initial value problem the equivalence relations
reduce to classical equations. We will suggest how to extend the initial
value problem for larger classes of distributions and
show that it has unique solutions. It is true that the equivalence relations
might have also many other (generally non-linear) solutions in the
space constructed,\footnote{By ``solution'' we mean here any
object fulfilling the particular equivalence relations.} but the
situation that there exist many physically meaningless solutions is
for physicists certainly nothing new. Our previous considerations
suggest that we shall look only for distributional solutions (they
are the ones having physical relevance) where the solution is
provided to be unique (up to initial values obviously), if it exists.
(This will also recover the common, but also many ``less'' common
physical results \cite{Vickers}.)

\section{New approach}\label{NewApp}

\subsection{The basic concepts/definitions}\label{basicconcepts}

Before saying anything about distributional
tensor fields, we have to define the basic concepts which will be
used in all the following mathematical constructions. This task is
dealt with in this section, so this is the part crucially important for
understanding all the subsequent theory. An attentive reader, having read through the introduction and abstract will understand
why we are particularly interested in defining these concepts.

\subsubsection{Definition of (M,$\mathcal{A}$)}

\theoremstyle{definition}
\newtheorem{defn}{Definition}[section]
\begin{defn}
 By piecewise smooth function we mean a function from an open set
$\Omega_{1}(\subseteq\mathbb{R}^{4})\to\mathbb{R}^{m}$ such, that there exists
an open set $\Omega_{2}$ (in the usual ``open ball'' topology on
$\mathbb{R}^{4}$) on which this function is smooth and
$\Omega_{2}=\Omega_{1}\setminus \Omega'$ (where $\Omega'$ has a Lebesgue measure 0).
\end{defn}

Take $M$ as a 4D paracompact\footnote{We will use specifically
4-dimensional manifolds, but one can immediately generalize all the following
constructions for $n$-dimensional manifolds.}, Hausdorff locally
Euclidean space, on which there exists a smooth atlas
$\mathcal{\mathcal{\mathcal{\mathcal{\mathcal{S}}}}}$. Hence
$(M,\mathcal{\mathcal{\mathcal{\mathcal{S}}}})$ is a smooth
manifold. Now take a ordered couple $(M,\mathcal{A})$, where
$\mathcal{A}$ is the maximal atlas, where all the maps are
connected by piecewise smooth transformations such that:
\begin{itemize}
\item
the transformations and their inverses have on every compact subset of $\mathbb{R}^{4}$ all the first derivatives (on the domains
where they exist) bounded\\ (hence Jacobians, inverse
Jacobians are on every compact set bounded),
\item
it contains at least one maximal smooth subatlas $\mathcal{\mathcal{\mathcal{\mathcal{S}}}}\subseteq\mathcal{A}$,\\ (coordinate transformations between maps are smooth
there).
\end{itemize}
\medskip
\paragraph{Notation.} The following notation will be used:
\begin{itemize}
\item
By the letter $\mathcal{\mathcal{\mathcal{\mathcal{\mathcal{S}}}}}$ we will always mean some maximal smooth subatlas of $\mathcal{A}$.
\item
Every subset of $M$ on
which there exists a chart from our atlas $\mathcal{A}$, we call $\Omega_{Ch}$.
An arbitrary chart on $\Omega_{Ch}$ from our atlas $\mathcal{A}$ is denoted
$Ch(\Omega_{Ch})$.
\end{itemize}

\paragraph{Notation.}
 Take some set $\Omega_{Ch}$. Take some open subset of that set
$\Omega'\subset\Omega_{Ch}$. Then $ Ch(\Omega_{Ch})_{|\Omega'}$  is defined simply as
$Ch(\Omega')$, which is obtained from $Ch(\Omega_{Ch})$ by limiting the domain to $\Omega'$.

\subsubsection{Definition of ``continuous to the maximal possible degree''}
\theoremstyle{definition}
\newtheorem{defn3.001}[defn]{Definition}
\begin{defn3.001}
We call a function on $M$ continuous to the maximal possible degree, if on arbitrary $\Omega$ of Lebesgue measure 0 it is only in such cases:\footnote{The expression
``Lebesgue measure 0 set'' will have in this chapter extended meaning. It refers to such subsets of a general manifold $M$ that they have in arbitrary chart Lebesgue measure 0.}

\begin{itemize}
\item[a)]
either undefined,
\item[b)]
or defined and discontinuous,
\end{itemize}
 in which there does not exist a way of turning it into a function continuous on $\Omega$ by
\begin{itemize}
\item
in the case of $a)$ extending its domain by the $\Omega$ set,
\item
in the case of $b)$ re-defining it on $\Omega$.
\end{itemize}
\end{defn3.001}

\subsubsection{Jacobians and algebraic operations with Jacobians} Now it is
obvious that since transformations between maps do not have to be
everywhere once differentiable, the Jacobian and inverse Jacobian may always be undefined on a set having Lebesgue measure 0. Now if we understand product in
the sense of a limit, then the relation
\begin{equation}
J^{\mu}_{\alpha}~(J^{-1})^{\alpha}_{\nu}=\delta^{\mu}_{\nu},\label{maxcon}
\end{equation}
for example, might hold even at the points where both Jacobian and inverse
Jacobian are undefined. This generally means the following: any algebraic operation with tensor fields is understood in such way, that in every chart it gives sets of functions continuous to a maximal possible degree. From this follows that the matrix product (\ref{maxcon})
must be, for $\mu=\nu$, equal to 1 and, for $\mu\neq\nu$, ~0.

\subsubsection{Tensor fields on $M$} We understand the tensor field on
$M$ to be an object which is:
\begin{itemize}
\item
Defined relative to the 1-dif\-ferent\-iable subatlas of $\mathcal{A}$ everywhere except for a set
having Lebesgue measure 0 (this set is a function of the given 1-dif\-ferentia\-ble subatlas).
\item
In every chart from $\mathcal{A}$ it is given by
functions continuous to a maximal possible degree.
\item
It transforms $\forall\Omega_{Ch}$ between charts
$Ch_{1}(\Omega_{Ch}), ~Ch_{2}(\Omega_{Ch})\in\mathcal{A}$ in the tensorial way

\begin{equation}
T^{\mu...}_{\nu...~Ch_{2}}=J^{\mu}_{\alpha}~(J^{-1})^{\beta}_{\nu}\cdots
~T^{\alpha...}_{\beta...~Ch_{1}}~~~~~~\hbox{a.e.},\footnote{This expression means ``almost everywhere'', that is, ``everywhere apart from a set having Lebesgue measure 0''.}
\end{equation}

 where $J^{\mu}_{\alpha}$ is the Jacobian of the coordinate transformation from $Ch_{1}(\Omega_{Ch})$ to $Ch_{2}(\Omega_{Ch})$, and $T^{\mu...}_{\nu...~Ch_{1}}, T^{\mu...}_{\nu...~Ch_{2}}$ are tensor field components in charts $Ch_{1}$, $Ch_{2}$.
As we already mentioned: If $T^{\mu...}_{\nu...~Ch_{1}}$ is at some
given point undefined, so in some chart $Ch_{1}(\Omega_{Ch})$ the tensor components
do not have a defined limit,
this limit can still exist in $Ch_{2}(\Omega_{Ch})$, since Jacobians
and inverse Jacobians of the transformation from
$Ch_{1}(\Omega_{Ch})$ to $Ch_{2}(\Omega_{Ch})$ might be undefined at
that point as well. This limit then defines
$T^{\mu...}_{\nu...~Ch_{2}}$ at the given point.
\end{itemize}

\subsubsection{Important classes of test objects}

\paragraph{Notation.} The following notation will be used:
\begin{itemize}
\item
We denote by $C^{P}(M)$ the class of 4-form fields on $M$, such that they are
compactly supported and their support lies within some $\Omega_{Ch}$. For such 4-form fields we will generally use the symbol ~$\omega$.
\item
For the scalar density related to $\omega\in C^{P}(M)$ we use always the symbol ~$\omega'$.
\item
By $C^{P}(\Omega_{Ch})$ we mean a subclass of $C^{P}(M)$, given by 4-form fields having support inside $\Omega_{Ch}$. Note that only the
$C^{P}(\Omega_{Ch})$ subclasses form linear spaces.
\item
Take such maximal atlas $\mathcal{\mathcal{\tilde S}}$, ~($\exists
\mathcal{\mathcal{\mathcal{\mathcal{S}}}}\subset
\mathcal{\mathcal{\tilde S}}\subset\mathcal{A}$),~ that there exist
4-forms from $C^{P}(M)$, such that they are given in this atlas by
everywhere smooth scalar density ~$\omega'$. (``Maximal'' here means
that these 4-forms have in every chart outside this atlas non-smooth
scalar densities.) By $C^{P}_{S (\mathcal{\mathcal{\tilde S}})}(M)$
we mean a class of all such elements from $C^{P}(M)$, that they have
everywhere smooth scalar density in ~$\mathcal{\tilde S}$.
\item
The letter $\mathcal{\mathcal{\tilde S}}$ will from now on be reserved for maximal atlases defining $C^{P}_{S (\mathcal{\mathcal{\tilde S}})}(M)$ classes.
\item
$C^{P}_{S}(M)$ is defined as: $C^{P}_{S}(M)\equiv \cup_{\mathcal{\mathcal{\tilde S}}}C^{P}_{S (\mathcal{\mathcal{\tilde S}})}(M)$.
\item
$C^{P}_{S (\mathcal{\mathcal{\tilde S}})}(\Omega_{Ch})$ ~(or $C^{P}_{S}(\Omega_{Ch})$)
means $C^{P}_{S (\mathcal{\mathcal{\tilde S}})}(M)$  ~(or $C^{P}_{S}(M)$) element having
support inside the given $\Omega_{Ch}$.
\end{itemize}

\subsubsection{Topology on $C^{P}_{S (\mathcal{\mathcal{\tilde S}})}(\Omega_{Ch})$}
Consider the following topology on each ~$C^{P}_{S (\mathcal{\tilde
S)}}(\Omega_{Ch})$: ~A sequence from ~$\omega_{n}\in$\\~ $C^{P}_{S
(\mathcal{\mathcal{\tilde S}})}(\Omega_{Ch})$ converges to an
element ~$\omega$~ from that set if all the supports of
~$\omega_{n}$~ lie in a single compact set, and in any chart
$Ch(\Omega_{Ch})\in \mathcal{\mathcal{\tilde S}}$, for arbitrary
$k$ it is true that $\frac{\partial^{k}\omega'_{n}(x^{i})}{\partial
x^{l_{1}}..\partial x^{l_{k}}}$ converges uniformly to
$\frac{\partial^{k}\omega'(x^{i})}{\partial x^{l_{1}}..\partial
x^{l_{k}}}$.

\subsection{Scalars}

This section deals with the definition of scalar distributions as the easiest
particular example of a generalized tensor field. The explanatory reasons are the main ones
why we deal with scalars separately, instead of taking more
``logical'', straightforward way to tensor fields of arbitrary rank.

\subsubsection{Definition of $D'(M)$, and hence of \emph{linear generalized scalar fields}}

\theoremstyle{definition}
\newtheorem{defn4.99}{Definition}[section]
\begin{defn4.99}
We say that $B$, being a function that maps some subclass of $C^{P}(M)$ to $\mathbb{R}^{n}$, is linear, if the following holds: Take such $\omega_{1}$ and $\omega_{2}$ from the domain of $B$, that they belong to the same class $C^{P}(\Omega_{Ch})$.
Whenever the domain of $B$ contains also their linear combination
$\lambda_{1}\omega_{1}+\lambda_{2}\omega_{2}$, where ~$\lambda_{1},\lambda_{2}\in\mathbb{R}$,~ then ~$B(\lambda_{1}\omega_{1}+\lambda_{2}\omega_{2})=\lambda_{1}B(\omega_{1})+\lambda_{2}B(\omega_{2})$.\label{linear}
\end{defn4.99}

\theoremstyle{definition}
\newtheorem{defn4.999}[defn4.99]{Definition}

\begin{defn4.999}
 Now take the space of linear maps $F\to\mathbb{R}$, where $F$ is such set that $\exists\mathcal{\mathcal{\tilde S}}$, such that
$C^{P}_{S (\mathcal{\mathcal{\tilde S}})}(M)\subseteq F\subseteq
C^{P}(M)$. These linear maps are also required to be for every
$\Omega_{Ch}$ on $C^{P}_{S \mathcal{(\tilde S)}}(\Omega_{Ch})$
continuous, (relative to the topology taken in section
\ref{basicconcepts}). Call set of such maps $D'(M)$, or in words,
the set of \emph{linear generalized scalar fields}.
\end{defn4.999}

%\bigskip

\subsubsection{Important subclasses of $D'(M)$}

\paragraph{Notation.} The following notation will be used:
\begin{itemize}
\item
Now take a subset of $D'(M)$ given by regular distributions
defined as integrals of piecewise continuous
functions (everywhere on $C^{P}(M)$, where it converges). We denote it by $D'_{E}(M)$.\footnote{Actually, it holds that if and only if the function is integrable in every chart on every compact set in $\mathbb{R}^n$, then this function defines a regular $D'(M)$ distribution and is defined at least on the whole $C^{P}_{S}(M)$ class.}
\item
Take such subset of $D'_{E}(M)$ that there  $\exists \mathcal{\mathcal{\mathcal{\mathcal{S}}}}$
in which the function under the integral is smooth. Call this class
$D'_{S}(M)$.
\item
 Now~ take ~subsets ~of ~$D'(M)$~ such ~that~ they~ have
some ~common ~set\\
$\cup_{n}C^{P}_{S (\mathcal{\mathcal{\tilde S}}_{n})}(M)$
belonging to their domains and are ~~ $\forall\Omega_{Ch}$,
~$\forall n$~ continuous on $C^{P}_{S (\mathcal{\mathcal{\tilde
S}}_{n})}(\Omega_{Ch})$. Denote such subsets by
$D'_{(\cup_{n}\mathcal{\tilde S}_{n})}(M)$. Obviously ~$T\in
D'_{(\cup_{n}\mathcal{\mathcal{\tilde S}}_{n})}(M)$ means
~$T\in\cap_{n}D'_{(\mathcal{\mathcal{\tilde S}}_{n})}(M)$.
\item
By $D'_{(\cup_{n}\mathcal{\mathcal{\tilde S}}_{n} o)}(M)$ we mean
objects such that they belong to $D'_{(\cup_{n}\mathcal{\mathcal{\tilde S}}_{n})}(M)$ and their full domain is given as ~$\cup_{n}C^{P}_{S (\mathcal{\mathcal{\tilde S}}_{n})}(M)$.
\item
If we use
the notation $D'_{E (\cup_{n}\mathcal{\mathcal{\tilde S}}_{n} o)}(M)$, we mean objects defined by integrals of piecewise continuous functions, with their domain being the class $\cup_{n}C^{P}_{S (\mathcal{\mathcal{\tilde S}}_{n})}(M)$.
\item
By using $D'_{S(\cup_{n}\mathcal{\mathcal{\tilde S}}_{n} o)}(M)$ we automatically mean subclass of $D'_{E(\cup_{n}\mathcal{\mathcal{\tilde S}}_{n} o)}(M)$, such that it is given by
an integral of a smooth function in some smooth subatlas
$\mathcal{\mathcal{\mathcal{\mathcal{S}}}}\subseteq\cup_{n}\mathcal{\mathcal{\tilde S}}_{n} $.
\end{itemize}

\subsubsection{$D'_{A}(M)$, hence \emph{generalized scalar fields}}

\paragraph{Notation.}
 Let us for any arbitrary set $D'_{(\mathcal{\mathcal{\tilde S}})}(M)$\footnote{The union is here trivial, it just means one element $\mathcal{\mathcal{\tilde S}}$.} construct, by the use of pointwise multiplication of its elements,
an algebra. Another way to describe the algebra is that it is a set of multivariable arbitrary degree polynomials, where different variables represent different elements of $D'_{(\mathcal{\mathcal{\tilde S}})}(M)$. Call it $D'_{(\mathcal{\mathcal{\tilde S}}) A}(M)$.

By pointwise multiplication of linear generalized scalar fields ~~~$B_{1},~~B_{2}~~\in \\ D'_{(\mathcal{\mathcal{\tilde S}})}(M)$~ we mean a mapping from $\omega$ into product of the images (real numbers) of the $B_{1},B_{2}$ mappings:
~$(B_{1}\cdot B_{2})(\omega)\equiv B_{1}(\omega)\cdot B_{2}(\omega)$.~ The domain on which the
product (and the linear combination as well) is defined is an
intersection of domains of $B_{1}$ and $B_{2}$ (trivially always nonempty,
containing $C^{P}_{S (\mathcal{\mathcal{\tilde S}})}(M)$ at least).
Note also that the resulting arbitrary element of $D'_{(\mathcal{\mathcal{\tilde S}}) A}(M)$ has in general all the properties defining
$D'(M)$ objects, except that of being necessarily linear.
\theoremstyle{definition}
\newtheorem{defn5}[defn4.99]{Definition}
\begin{defn5}
 The set of objects obtained by the union~ $\cup_{\mathcal{\mathcal{\tilde S}}}D'_{(\mathcal{\mathcal{\tilde S}}) A}(M)$~ we denote $D'_{A}(M)$, and call
them \emph{generalized scalar fields} (GSF).
\end{defn5}

\subsubsection{Topology on $D'_{(\mathcal{\mathcal{\tilde S}} o)}(M)$}

If we take objects from $D'_{(\mathcal{\mathcal{\tilde S}} o)}(M)$ and we take a weak
($\sigma -$) topology on that
set, we know that any object is a limit of some sequence from $D'_{S
(\mathcal{\mathcal{\tilde S}} o)}(M)$ objects from that set (they form a dense subset of
that set). That is known from the classical theory. Such a
space is complete.

\subsection{Generalized tensor fields}\label{GTFsection}

\emph{This section is of crucial importance.} It provides us with
definitions of all the basic objects we are interested in, the
generalized tensor fields and all their subclasses of special
importance as well.

\subsubsection{The class $D'^{m}_{n}(M)$ of \emph{linear generalized tensor fields}}

First let us clearly state how to interpret the $J^{\mu}_{\nu}$ Jacobian in all the following definitions.~ It is a matrix of
piecewise smooth functions  $\Omega_{1}\setminus \Omega_{2}\to\mathbb{R}$,~ $\Omega_{1}$ being
a open subset of $\mathbb{R}^{4}$ and $\Omega_{2}$ having Lebesgue measure 0. Let it represent transformations from $Ch_{1}(\Omega_{Ch})$ to $Ch_{2}(\Omega_{Ch})$.
We can map the Jacobian by the inverse of the $Ch_{1}(\Omega_{Ch})$ coordinate mapping
to $\Omega_{Ch}$ and it will become a matrix of functions
$\Omega_{Ch}\setminus \Omega'\to\mathbb{R}$, ~~$\Omega'$ having Lebesgue measure 0.
~The object $J^{\mu}_{\nu}\cdot\omega$~ is then understood as a matrix of 4-forms from $C^{P}(\Omega_{Ch})$, which also means that outside $\Omega_{Ch}$ we trivially define them to be 0.

\theoremstyle{definition}
\newtheorem{defn6}{Definition}[section]
\begin{defn6}
Take some set $F$, ~$F_{1}\subseteq F\subseteq F_{2}$, where $F_{1}$
us such set that ~~$\exists\mathcal{\mathcal{\tilde S}}$~ and
~~$\exists \mathcal{\mathcal{\mathcal{\mathcal{S}}}}\subseteq
\mathcal{\mathcal{\tilde S}}$~~
($\mathcal{\mathcal{\mathcal{\mathcal{\mathcal{S}}}}}$ is some
maximal smooth atlas) defining it as:
\begin{equation*}
F_{1}\equiv \cup_{\Omega_{Ch}}\left\{
\big(\omega,Ch(\Omega_{Ch})\big):~\omega\in C^{P}_{S
(\mathcal{\tilde S})}(\Omega_{Ch}),~
Ch(\Omega_{Ch})\in\mathcal{\mathcal{\mathcal{\mathcal{\mathcal{S}}}}}\right\}.
\end{equation*}
$F_{2}$ is defined as:
\begin{equation*}
F_{2}\equiv \cup_{\Omega_{Ch}}\left\{
\big(\omega,Ch(\Omega_{Ch})\big):~\omega\in C^{P}(\Omega_{Ch}),~
Ch(\Omega_{Ch})\in\mathcal{A})\right\}.
\end{equation*}
By a ~$D'^{m}_{n}(M)$~ object, the \emph{linear generalized tensor
field}, we mean a linear mapping from ~$F\to\mathbb{R}^{4^{m+n}}$
for which the following holds: \footnote{We could also choose for
our basic objects maps taking ordered couples from
~$C^{P}(\Omega_{Ch})\times Ch(\Omega'_{Ch})$,
~($\Omega_{Ch}\neq\Omega'_{Ch}$). The linearity condition then
automatically determines their values, since for $\omega\in
C^{P}(\Omega_{Ch})$, whenever it holds that ~$\Omega'_{Ch}\cap
\hbox{supp}(\omega)=\{0\}$, they must automatically give ~0~ for any
chart argument. Hence these two definitions are trivially connected
and choice between them is just purely formal (only a matter of
``taste'').}

\begin{itemize}
\item
$\forall\Omega_{Ch}$~ it is ~$\forall Ch_{k}(\Omega_{Ch})\in
\mathcal{\mathcal{\mathcal{\mathcal{S}}}}\subset\mathcal{\mathcal{\tilde
S}}$~ continuous on the class $C^{P}_{S (\mathcal{\tilde
S})}(\Omega_{Ch})$. (Both $\mathcal{S}$, ~$\mathcal{\tilde S}$ are from the definition of  $F_{1}$.)
\item
This map also ~~$\forall\Omega_{Ch}$ ~~transforms between two charts
from its domain,\\~ $Ch_{1}(\Omega_{Ch})$, ~$Ch_{2}(\Omega_{Ch})$, as:

\begin{equation*}
 T^{\mu...\gamma}_{\nu...\delta}(Ch_{1},\omega)=T^{\alpha...\lambda}_{\beta...\rho}\left(Ch_{2},J^{\mu}_{\alpha}...J^{\gamma}_{\lambda}(J^{-1})^{\beta}_{\nu}....(J^{-1})^{\rho}_{\delta}\cdot\omega\right)~.
\end{equation*}

\item
The following consistency condition holds: If
~$\Omega'_{Ch}\subset\Omega_{Ch}$, then
~$T^{\mu...\alpha}_{\nu...\delta}$~  gives on ~$\omega~\times
~Ch(\Omega_{Ch})_{|\Omega'_{Ch}}$, ~$\omega\in
C^{P}(\Omega'_{Ch})$ the same results\footnote{By the ``same
results'' we mean that they are defined on the same domains, and by
the same values.} as on ~$\omega~\times ~Ch(\Omega_{Ch})$.
\end{itemize}

\end{defn6}

We can formally extend this notation
also for the case $m=n=0$. This means scalars, exactly as
defined before. So from now on $m$, $n$ take also the value 0, which means the theory in the following sections holds also for the scalar objects.

\subsubsection{Important subclasses of $D'^{m}_{n}(M)$}

\paragraph{Notation.} The following notation will be used:
\begin{itemize}
\item
By a complete analogy to scalars we define classes
~$D'^{m}_{n E}(M)$:
On arbitrary $\Omega_{Ch}$, being fixed in
arbitrary chart ~$Ch_{1}(\Omega_{Ch})\in\mathcal{A}$ we can express it in
another arbitrary chart ~$Ch_{2}(\Omega_{Ch})\in\mathcal{A}$, as an integral from a
multi-index matrix of piecewise continuous
functions on such subset of ~$C^{P}(M)$, on which the integral is
convergent\footnote{Actually we will use the
expression ``multi-index matrix'' also later in the text and it just means
specifically ordered set of functions.}.
\item
Analogously the class $D'^{m}_{n S}(M)\subset D'^{m}_{n E}(M)$ ~is defined by objects which can, for some maximal
smooth atlas $\mathcal{\mathcal{\mathcal{\mathcal{\mathcal{S}}}}}$, in arbitrary charts \footnote{The first chart is
an argument of this generalized tensor field and the second chart is
the one in which we express the given integral.}
~$Ch_{1}(\Omega_{Ch})$, ~$Ch_{2}(\Omega_{Ch})\in\mathcal{\mathcal{\mathcal{\mathcal{\mathcal{S}}}}}$, ~be expressed
by an integral from a multi-index matrix of smooth functions.
\item
$D'^{m}_{n (\cup_{l}At(\mathcal{\mathcal{\tilde S}}_{l}))}(M)$~ means a class of objects being for every $\Omega_{Ch}$ in every $Ch(\Omega_{Ch})\in At(\mathcal{\mathcal{\tilde S}}_{l})$ continuous on $C^{P}_{S(\mathcal{\mathcal{\tilde S}}_{l})}(\Omega_{Ch})$. (Here $At(\mathcal{\tilde S})$ stands for a map from atlases $\mathcal{\tilde S}$ to some subatlases of $\mathcal{A}$.)
\item
$D'^{m}_{n (\cup_{l}At(\mathcal{\mathcal{\tilde S}}_{l})o)}(M)$ means a class of objects from $D'^{m}_{n (\cup_{l}At(\mathcal{\mathcal{\tilde S}}_{l}))}(M)$ having as their domain the union
\[\cup_{\Omega_{Ch}}\cup_{l}\left\{ (\omega,Ch(\Omega_{Ch})):~\omega\in C^{P}_{S (\mathcal{\mathcal{\tilde S}}_{l})}(\Omega_{Ch}),~
Ch(\Omega_{Ch})\in At(\mathcal{\mathcal{\tilde S}}_{l})\right\}.\]
\item
If we have classes $D'^{m}_{n (\cup_{l}At(\mathcal{\mathcal{\tilde S}}_{l}))}(M)$ and $D'^{m}_{n (\cup_{l}At(\mathcal{\mathcal{\tilde S}}_{l})o)}(M)$ where $At(\mathcal{\mathcal{\tilde S}}_{l})=\mathcal{\mathcal{\mathcal{\mathcal{S}}}}_{l}\subset\mathcal{\mathcal{\tilde S}}_{l}$, we use the simple notation ~$D'^{m}_{n (\cup_{l} \mathcal{\mathcal{\mathcal{\mathcal{S}}}}_{l})}(M)$, $D'^{m}_{n (\cup_{l} \mathcal{\mathcal{\mathcal{\mathcal{S}}}}_{l}o)}(M)$.\footnote{ We have to realize that the subatlas $\mathcal{\mathcal{\mathcal{\mathcal{S}}}}_{n}$ specifies
completely the atlas $\mathcal{\mathcal{\tilde S}}_{n}$, since taking forms smooth in $\mathcal{\mathcal{\mathcal{\mathcal{S}}}}_{n}$ determines automatically the whole
set of charts in which they are still smooth. This fact contributes to the simplicity of this notation.}
\end{itemize}

\subsubsection{Definition of $D'^{m}_{n A}(M)$, hence \emph{generalized tensor
fields}}

\theoremstyle{definition}
\newtheorem{defn8}[defn6]{Definition}
\begin{defn8}
 Now define~$D'^{m}_{n (\mathcal{\mathcal{\mathcal{\mathcal{S}}}}) A}(M)$ to be the algebra constructed from the objects ~$D'^{m}_{n (\mathcal{\mathcal{\mathcal{\mathcal{S}}}})}(M)$ by the tensor product,
exactly in analogy to the case of scalars (this reduces for scalars to the
product already defined). The
object, being a result of the tensor product, is again a mapping ~$V\to\mathbb{R}^{4^{m+n}}$, defined in every chart by componentwise multiplication.
Now denote by $D'^{m}_{n A}(M)$ a set given as $\cup_{\mathcal{\mathcal{\mathcal{\mathcal{S}}}}} D'^{m}_{n (\mathcal{\mathcal{\mathcal{\mathcal{S}}}}) A}(M)$, meaning a union of all possible $D'^{m}_{n (\mathcal{\mathcal{\mathcal{\mathcal{S}}}}) A}(M)$. Call the objects belonging to this set \emph{the generalized tensor fields (GTF)}.
\end{defn8}

\paragraph{Notation.}
Furthermore let us use the same procedure as in the previous
definition, just instead of constructing the algebras from
the classes $D'^{m}_{n (\mathcal{\mathcal{\mathcal{\mathcal{S}}}})}(M)$, we now construct them only from the classes ~~$D'^{m}_{n
(\cup_{l}At(\mathcal{\mathcal{\tilde S}}_{l}))}(M)\cap D'^{m}_{n (\mathcal{\mathcal{\mathcal{\mathcal{S}}}})}(M)$,
~~(again by tensor product). For the union of such algebras we use the notation
~$D'^{m}_{n (\cup_{l}At(\mathcal{\mathcal{\tilde S}}_{l})) A}(M)$.

\subsubsection{Definition of $\Gamma-$objects, their classes and algebras}

\paragraph{Notation.}
Now let us define the generalized space of objects $\Gamma^{l}(M)$. (For example the Christoffel symbol would fall into this class.) These objects are defined exactly in the same way as $D'^{m}_{n}(M)$~ ($m+n=l$)~ objects, we just do not require that they transform between charts in the tensorial way, (second point in the definition of generalized tensor fields).

Note the following:
\begin{itemize}
\item
The definition of $\Gamma^{l}(M)$ includes also the case
$m=0$. Now we see, that the scalars can be
taken as subclass of $\Gamma^{0}(M)$, given by objects that are constants with respect to the chart argument.
\item
Note also that for a general $\Gamma^{m}(M)$ object
there is no meaningful differentiation between ``upper'' and
``lower'' indices, but we will still use formally the $T^{\mu...}_{\nu...}$ notation (for all cases).
\end{itemize}

\paragraph{Notation.}
In the same way, (by just not putting requirements on the transformation properties), we can generalize the classes
\begin{itemize}
\item
 $D'^{m}_{n (\cup_{l}At(\mathcal{\mathcal{\tilde S}}_{l}))}(M)$ ~to~ $\Gamma^{m+n}_{(\cup_{l}At(\mathcal{\mathcal{\tilde S}}_{l}))}(M)$,
\item
~$D'^{m}_{n (\cup_{l}At(\mathcal{\mathcal{\tilde S}}_{l})o)}(M)$ ~to~ $\Gamma^{m+n}_{(\cup_{l}At(\mathcal{\mathcal{\tilde S}}_{l})o)}(M)$,
\item
 ~$D'^{m}_{n (\cup_{l}\mathcal{\mathcal{\mathcal{\mathcal{S}}}}_{l})}(M)$ ~to~ $\Gamma^{m+n}_{(\cup_{l}\mathcal{\mathcal{\mathcal{\mathcal{S}}}}_{l})}(M)$,
\item
 ~$D'^{m}_{n (\cup_{l} \mathcal{\mathcal{\mathcal{\mathcal{S}}}}_{l}o)}(M)$ ~to~  $\Gamma^{m+n}_{(\cup_{l}\mathcal{\mathcal{\mathcal{\mathcal{S}}}}_{l}o)}(M)$,
\item
~$D'^{m}_{n E}(M)$ ~to~ $\Gamma^{m+n}_{E}(M)$,
\item
~$D'^{m}_{n (\cup_{l}At(\mathcal{\mathcal{\tilde S}}_{l})) A}(M)$ ~to~ $\Gamma^{m+n}_{(\cup_{l}At(\mathcal{\mathcal{\tilde S}}_{l})) A}(M)$, ~and
\item
$D'^{m}_{n A}(M)$~ to ~$\Gamma^{m+n}_{A}(M)$.
\end{itemize}
\medskip

It is obvious that all the latter classes contain all the former classes as their subclasses, (this is a result of what we called a ``generalization'').

Note that when we fix $\Gamma^{m}_{E}(M)$ objects in arbitrary chart from $\mathcal{A}$, they must be expressed by integrals from multi-index matrix of functions integrable on every compact set. In the case of the $D'^{m}_{n E}(M)$ subclass it can be required in only one chart, since the transformation properties together with boundedness of Jacobians and inverse Jacobians, provide that it must hold in any other chart from $\mathcal{A}$.
The specific subclass of $\Gamma^{m}_{E}(M)$
is ~$\Gamma^{m}_{S}(M)$, which is a subclass of distributions given in any chart from $\mathcal{A}$, (being an argument
of the given $\Gamma-$ object), by integrals from multi-index matrix of smooth
functions (when we express the integrals in the same chart, as the one taken as
the argument). ~$\Gamma^{m}_{S (\cup_{n}\mathcal{\mathcal{\mathcal{\mathcal{S}}}}_{n} o)}(M)$ stands again for ~$\Gamma^{m}_{S}(M)$~ objects with domain limited to

\[\cup_{\Omega_{Ch}}\cup_{n}\left\{ (\omega,Ch(\Omega_{Ch})):~\omega\in C^{P}_{S (\mathcal{\mathcal{\tilde S}}_{n})}(\Omega_{Ch}),~
Ch(\Omega_{Ch})\in\mathcal{\mathcal{\mathcal{\mathcal{\mathcal{S}}}}}_{n}\right\},\]
 where $\mathcal{\mathcal{\tilde S}}_{n}$ is given by the condition $\mathcal{\mathcal{\mathcal{\mathcal{S}}}}_{n}\subset\mathcal{\mathcal{\tilde S}}_{n}$.

\paragraph{Notation.}
Take some arbitrary elements
~$T^{\mu...}_{\nu...}\in\Gamma^{m}_{E}(M)$, ~$\omega\in
C^{P}(\Omega_{Ch})$,~ and ~$Ch_{k}(\Omega_{Ch})\in\mathcal{A}$.
~The $T^{\mu...}_{\nu...}(Ch_{k},\omega)$ can be always expressed as
$\int_{\Omega_{Ch}}T^{\mu...}_{\nu...}(Ch_{k})\cdot\omega$. ~Here
$T^{\mu...}_{\nu...}(Ch_{k})$ appearing under the integral denotes
some multi-index matrix of functions continuous to a maximal possible degree
on $\Omega_{Ch}$. For $T^{\mu...}_{\nu...}\in D'^{m}_{n E}(M)$ the
$T^{\mu...}_{\nu...}(Ch_{k})$ multi-index matrix components can be obtained from a tensor
field by:
\begin{itemize}
\item
expressing the tensor field components in $Ch_{k}(\Omega_{Ch})$ on some subset of $\mathbb{R}^{4}$,
\item
mapping the tensor field components to $\Omega_{Ch}$ by the inverse of $Ch_{k}(\Omega_{Ch})$.
\end{itemize}
Furthermore $\omega'(Ch_{k})$ will denote the 4-form scalar density in the
chart $Ch_{k}(\Omega_{Ch})$.

\subsubsection{Topology on $\Gamma^{m}_{(\cup_{n}\mathcal{\mathcal{\mathcal{\mathcal{S}}}}_{n}o)}(M)$}

If we take the class of ~$\Gamma^{m}_{(\cup_{n}
\mathcal{\mathcal{\mathcal{\mathcal{S}}}}_{n} o)}(M)$, and we impose
on this class the weak (point or $\sigma-$) topology, then the
subclass of ~$\Gamma^{m}_{(\cup_{n}
\mathcal{\mathcal{\mathcal{\mathcal{S}}}}_{n} o)}(M)$~ defined as
~$\Gamma^{m}_{S (\cup_{n}
\mathcal{\mathcal{\mathcal{\mathcal{S}}}}_{n} o)}(M)$ is dense in
~$\Gamma^{m}_{(\cup_{n}
\mathcal{\mathcal{\mathcal{\mathcal{S}}}}_{n} o)}(M)$. The same
holds for ~$D'^{m}_{n (\cup_{l}
\mathcal{\mathcal{\mathcal{\mathcal{S}}}}_{l} o)}(M)$ and $D'^{m}_{n
S (\cup_{l}\mathcal{\mathcal{\mathcal{\mathcal{S}}}}_{l} o)}(M)$.

\subsubsection{Definition of contraction}

\theoremstyle{definition}
\newtheorem{defn85}[defn6]{Definition}
\begin{defn85}
We define the contraction of a ~$\Gamma^{m}_{A}(M)$ object in the expected way: It is a map that transforms the object ~$T^{...\mu...}_{...\nu...}\in\Gamma^{m}_{A}(M)$~ to~ the object $T^{...\mu...}_{...\mu...}\in\Gamma^{m-2}_{A}(M)$.
\end{defn85}

Now contraction is a mapping ~$D'^{m}_{n}(M)\to
D'^{m-1}_{n-1}(M)$ and ~$\Gamma^{m}(M)\to\Gamma^{m-2}(M)$, but it is
not in general the mapping~ $D'^{m}_{n A}(M)\to D'^{m-1}_{n-1
A}(M)$, only~  $D'^{m}_{n A}(M)\to\Gamma^{m+n-2}_{A}(M)$.

\subsubsection{Interpretation of physical quantities}

The interpretation of
physical observables as ``amounts'' of quantities on the open sets
is dependent on our notion of volume. So how shall we get the notion
of volume in the context of our language? First, by volume we mean a
volume of an open set. But we will consider only open sets belonging
to some $\Omega_{Ch}$. So take some ~$\Omega_{Ch}$~ and some
arbitrary ~$\Omega'\subset\Omega_{Ch}$.~ Let us now assume that we
have a metric tensor from ~$D'^{m}_{n E}(M)$. This induces a
(volume) 4-form. Multiply this 4-form by a noncontinuous function
$\chi_{\Omega'}$ defined to be 1 inside ~$\Omega'$~ and everywhere
else 0. Call it $\omega_{\Omega'}$. Then by volume of an open set
$\Omega'$ we understand: ~$\int\omega_{\Omega'}$. Also
$\omega_{\Omega'}$ is object from ~$C^{P}(M)$ (particularly from
$C^{P}(\Omega_{Ch})$). The ``amounts'' of physical quantities on
$\Omega'$ we obtain, when the ~$D'^{m}_{n A}(M)$ objects act on
$\omega_{\Omega'}$.

\subsection{The relation of equivalence ($\approx$)}

This section now provides us with the fundamental concept of the theory,
the concept of equivalence of generalized tensor fields. Most of the first
part is devoted to fundamental definitions, the beginning of the second part
deals with the basic, important theorems, which just
generalize some of the basic Colombeau theory results to the
tensor product of generalized tensor fields. It adds several
important conjectures as well. The first part ends with the subsection
``some additional definitions'' and the second part with the subsection ``some additional theory''. They both deal with much
less central theoretical results, but they serve very well to put
light on what equivalence of generalized tensor fields means ``physically''.

\subsubsection{The necessary concepts to define the equivalence relation}

\paragraph{Notation.}
Take some subatlas of our atlas, this will be a
maximal subatlas of charts, which are maps to the whole of
$\mathbb{R}^{4}$. Such maps exist on each set $\Omega_{Ch}$ and they
will be denoted as $Ch'(\Omega_{Ch})$. We say that a chart $Ch'(\Omega_{Ch})$ is centered at the point $q\in\Omega_{Ch}$, if this point is mapped by this chart to 0 (in
$\mathbb{R}^{4}$). We will use the
notation $Ch'(q,\Omega_{Ch})$.

\paragraph{Notation.}
Take some $\Omega_{Ch}$,~$q\in\Omega_{Ch}$~ and
~$Ch'(q,\Omega_{Ch})\in\mathcal{\tilde S}$. The set of 4-forms
~$\omega_{\epsilon}\in A^{n}(\mathcal{\tilde
S},Ch'(q,\Omega_{Ch}))$ is defined in such way that
~$\omega_{\epsilon}\in C^{P}_{S (\mathcal{\mathcal{\tilde
S}})}(\Omega_{Ch})$ belongs to this class if:
\begin{itemize}
\item[a)] in the given $Ch'(q,\Omega'_{Ch})$, $\forall\epsilon$~ it holds that:

~~~~~$\int_{Ch'(\Omega'_{Ch})}(\prod_{i}x^{k_{i}}_{i})~\omega'_{\epsilon}(x)~d^{4}x=\delta_{k
0},~~ \sum_{i} k_{i}=k$,~~ $k\leq n$, ~~$n\in\mathbb{N}$,

\item[b)] the dependence on $\epsilon$ is in $Ch'(q,\Omega'_{Ch})$ given as $\epsilon^{-4}\omega'(\frac{x}{\epsilon})$.
\end{itemize}

\paragraph{Notation.}
Take an arbitrary ~$q$, ~~$\Omega_{Ch}$ ~($q\in\Omega_{Ch}$),
~~$Ch'(q,\Omega_{Ch})\in\mathcal{\mathcal{\tilde S}}$ ~and some natural number $n$. For
any ~$\omega_{\epsilon}\in A^{n}(Ch'(q,\Omega_{Ch}),\mathcal{\mathcal{\tilde S}})$ ~we can,
relatively to ~$Ch'(\Omega_{Ch})$, ~define a continuous set of maps
(depending on the parameter $y$)
\begin{equation}
A^{n}(Ch'(q,\Omega_{Ch}),\mathcal{\mathcal{\tilde S}})\to
C^{P}_{S (\mathcal{\mathcal{\tilde S}})}(\Omega_{Ch}),\label{basic}
\end{equation}
such that they are,  on ~$\Omega_{Ch}$ ~and in ~$Ch'(\Omega_{Ch})$, given as
~~$\omega'(\frac{x}{\epsilon}){\epsilon^{-4}}\to\omega'(\frac{y-x}{\epsilon})\epsilon^{-4}$.
~(To remind the reader $\omega'$ is the density expressing
 $\omega$ in this chart.) This gives us (depending on the parameter ~$y\in\mathbb{R}^{4}$) various $C^{P}_{S (\mathcal{\mathcal{\tilde S}})}(\Omega_{Ch})$ objects, such that they are in the fixed $Ch'(\Omega_{Ch})$ expressed by ~$\omega'(\frac{x-y}{\epsilon})~\epsilon^{-4}dx^{1}\bigwedge ...\bigwedge dx^{4}$. Denote these 4-form fields by $\tilde\omega_{\epsilon}(y)$.

\paragraph{Notation.}
Now, take any ~$T^{\mu...}_{\nu...}\in\Gamma^{m}_{ (At(\mathcal{\mathcal{\tilde S}}))A}(M)$. By applying it in an arbitrary fixed chart  ~~$Ch_{k}(\Omega_{Ch})\in At(\mathcal{\mathcal{\tilde S}},\Omega_{Ch})$~~
 on the 4-form field ~~$\tilde\omega_{\epsilon}(y)$,~~ obtained from ~~$\omega_{\epsilon}\in$ $A^{n}(Ch'(q,\Omega_{Ch}),\mathcal{\tilde
 S})$~ through the map (\ref{basic}), we get a function ~$\mathbb{R}^{4}\to\mathbb{R}^{4^{m+n}}$. As a consequence, the resulting function depends on the following objects:
~$T^{\mu...}_{\nu...}\left(\in\Gamma^{m}_{A}(M)\right),~\omega_{\epsilon}\left(\in
A^{n}(q,\Omega_{Ch},Ch'(\Omega_{Ch}))\right)$ ~and
$~Ch_{k}(\Omega_{Ch})$.~ We denote it by:
~~~~~$F'^{\mu...}_{\nu...}\big(T^{\mu...}_{\nu...}, \mathcal{\mathcal{\tilde S}},
~\Omega_{Ch}, Ch'(q,\Omega_{Ch}),n,\tilde\omega_{\epsilon}(y),
Ch_{k}(\Omega_{Ch})\big)$.

\subsubsection{Definition of the equivalence relation}

\theoremstyle{definition}
\newtheorem{defn11}{Definition}[section]
\begin{defn11}
 ~~$B^{\mu...}_{\nu...},T^{\mu...}_{\nu...}\in \Gamma^{m}_{A}(M)$
~are called equivalent ~($B^{\mu...}_{\nu...}\approx
T^{\mu...}_{\nu...}$), if:
\begin{itemize}
\item
 they belong to the same classes
~$\Gamma^{m}_{(At(\mathcal{\mathcal{\tilde S}}))A}(M)$,
\item
~$\forall\Omega_{Ch},~~~\forall q~~(q\in\Omega_{Ch}),~~~\forall
Ch'(q,\Omega_{Ch})\in\mathcal{\mathcal{\tilde S}}$ ~~~~(~such ~that ~~$B^{\mu...}_{\nu...}$,
~$T^{\mu...}_{\nu...}~\in\\ ~~\Gamma^{m}_{ (At(\mathcal{\mathcal{\tilde S}}))A}(M)$ ),~
~~$\forall Ch(\Omega_{Ch})~\in ~At(\mathcal{\mathcal{\tilde S}},\Omega_{Ch})$~~ $\exists n$,~
~~~such ~~that ~~~ $\forall \omega_{\epsilon}~\in \\
A^{n}(Ch'(q,\Omega_{Ch}),\mathcal{\mathcal{\tilde S}})$ ~~and for any compactly supported,
smooth function\\ $\mathbb{R}^{4}\to\mathbb{R}$, ~~$\phi$,~~ it holds:
\begin{eqnarray}
\lim_{\epsilon\to 0}
~\int_{\mathbb{R}^{4}}\bigg\{F'^{\mu...}_{\nu...}\big(B^{\mu...}_{\nu...}, q,
\Omega'_{Ch}, Ch'(\Omega'_{Ch}),n,\tilde\omega_{\epsilon}(y),
Ch(\Omega_{Ch})\big) \nonumber~~~~~~~~~~~~~~~~~~~~~~~~~~~~~~~\\
- ~F'^{\mu...}_{\nu...}\big(T^{\mu...}_{\nu...}, q, \Omega'_{Ch},
Ch'(\Omega'_{Ch}),n,\tilde\omega_{\epsilon}(y), Ch(\Omega_{Ch})\big)\bigg\} \cdot~\phi(y)
~d^{4}y~=~0.~~~~~~~
\end{eqnarray}
\end{itemize}
\end{defn11}

Note that for
$B^{\mu...}_{\nu...},C^{\mu...}_{\nu...},D^{\mu...}_{\nu...},T^{\mu...}_{\nu...}$
having the same domains and being from the same
$\Gamma^{n}_{(At(\mathcal{\tilde S}))}(M)$ classes, it trivially
follows that: $T^{\mu...}_{\nu...}\approx B^{\mu...}_{\nu...}$,~
$C^{\mu...}_{\nu...}\approx D^{\mu...}_{\nu...}$~ implies
~$\lambda_{1}T^{\mu...}_{\nu...}+\lambda_{2}C^{\mu...}_{\nu...}\approx\lambda_{1}B^{\mu...}_{\nu...}+\lambda_{2}D^{\mu...}_{\nu...}$~
for ~$\lambda_{1}, \lambda_{2}\in\mathbb{R}$.

\paragraph{Notation.}
Now since we have defined an equivalence relation, it divides the objects ~$\Gamma^{m}_{A}(M)$ naturally into equivalence classes. The set of such equivalence classes will be denoted as ~$\tilde\Gamma^{m}_{A}(M)$.
Later we may also use sets of more limited classes of
equivalence ~$\tilde D'^{m}_{n A}(M)$, ~$\tilde D'^{m}_{n E A}(M)$
(etc.), which contains equivalence classes (only) of the objects
belonging to ~$D'^{m}_{n A}(M)$, ~$D'^{m}_{n E A}(M)$
(etc.).

\paragraph{Notation.}
In some of the following theorems, (also for example in the definition of the covariant derivative), we will use some convenient notation:
Take some object $B^{\mu...}_{\nu...}\in\Gamma^{m}_{E}(M)$. The
expression ~$T^{\mu...}_{\nu...}(B^{\alpha...}_{\beta...}\omega)$ will be understood in the following way: Take ~$Ch_{k}(\Omega_{Ch}) \times \omega$ ~~($\omega\in C^{P}(\Omega_{Ch})$) from the domain of $T^{\mu...}_{\nu...}$. ~Then ~$B^{\alpha...}_{\beta...}(Ch_{k})\cdot\omega$~ is a multi-index matrix of $C^{P}(\Omega_{Ch})$ objects. This means that outside $\Omega_{Ch}$ set they are defined to be trivially 0. We substitute this multi-index matrix of $C^{P}(\Omega_{Ch})$ objects
to ~$T^{\mu...}_{\nu...}$, with the chart ~$Ch_{k}(\Omega_{Ch})$ taken as the argument.

\subsubsection{Relation to Colombeau equivalence}

A careful reader now understands the relation between our concept of equivalence and the Colombeau equivalence relation. It is simple: The previous definition just translates the Colombeau equivalence relation (see \cite{Amulti}) into our language and the equivalence classes will naturally preserve all the features of the Colombeau equivalence classes (this will be proven in the following theorems).

\subsubsection{Some additional definitions (concepts of associated field and
$\Lambda$ class)}

We define the concept (of association) to bring
some insight to what our concepts mean in the most simple (but most
important and useful) cases. It enables us to see better
the relation between the calculus we defined (concerning
equivalence) and the classical tensor calculus. It brings us also
better understanding of what equivalence means in
terms of physics (at least in the simple cases). It just means that the quantities might differ on the
large scales, but take the same small scale limit (for the small
scales they approach each other).

\theoremstyle{definition}
\newtheorem{defn12}[defn11]{Definition}
\begin{defn12}\label{DefAssoc}
 Take $T^{\mu...}_{\nu...}\in\Gamma^{m}_{A}(M)$. Assume that:
\begin{itemize}
 \item [a)]
~~$\forall \mathcal{\mathcal{\mathcal{\mathcal{S}}}}$,~~~ such~~ that ~~~$T^{\mu...}_{\nu...}~\in~\Gamma^{m}_{( \mathcal{\mathcal{\mathcal{\mathcal{S}}}})A}(M)$, ~~~$\forall\Omega_{Ch}$~~~ and ~~~$\forall Ch_{k}(\Omega_{Ch})~\in ~\mathcal{\mathcal{\mathcal{\mathcal{S}}}}$ \\  $\exists~\Omega_{Ch}~\setminus ~\Omega'(Ch_{k})$,~~~(the set ~$\Omega'(Ch_{k})$~ being ~0~ in
any Lebesgue measure),
~~such ~that ~~$\forall q~\in~\Omega_{Ch}~\setminus ~\Omega'(Ch_{k})$, ~~~$\forall Ch'(q,\Omega_{Ch})~\in ~\mathcal{\mathcal{\mathcal{\mathcal{S}}}}~\subset~\mathcal{\mathcal{\tilde S}}$  ~~~~$\exists n$,\\~ such ~that  ~$\forall \omega_{\epsilon}\in A^{n}(Ch'(q,\Omega_{Ch}),\mathcal{\mathcal{\tilde S}})$
\begin{equation}
\exists~\lim_{\epsilon\to 0}
T^{\mu...}_{\nu...}(Ch_{k},\omega_{\epsilon}).\label{limit}
\end{equation}
\item [b)]
The limit (\ref{limit}) is
~~$\forall Ch'(q,\Omega_{Ch})~\in ~\mathcal{\mathcal{\mathcal{\mathcal{S}}}}~\subset~\mathcal{\mathcal{\tilde S}}$,~ ~$\forall\omega_{\epsilon}~\in ~A^{n}(Ch'(q,\Omega_{Ch}),\mathcal{\mathcal{\tilde S}})$ ~the same.
\end{itemize}
 If both $a)$ and $b)$ hold, then the object defined by the limit (\ref{limit})
 is a mapping:
\begin{equation}
Ch(\Omega_{Ch})(\in\mathcal{\mathcal{\mathcal{\mathcal{\mathcal{S}}}}}))\times\Omega_{Ch}\setminus \Omega'(Ch)\to\mathbb{R}^{4^{m+n}}.
\end{equation}
 We call this map the field associated to ~$T^{\mu...}_{\nu...}\in\Gamma^{m}_{A}(M)$,
 and we use the expression
 ~$A_{s}(T^{\mu...}_{\nu...})$. (It necessarily fulfills the same consistency conditions for ~$\Omega^{1}_{Ch}\subset\Omega^{2}_{Ch}$ as the $\Gamma^{m}_{A}(M)$ objects.)

\end{defn12}

\theoremstyle{definition}
\newtheorem{deff12}[defn11]{Definition}
\begin{deff12}
Denote by ~$\Lambda\subset\Gamma^{m}_{E (\cup_{n}At(\mathcal{\mathcal{\tilde S}}_{n}) o)}(M)$ a class of objects, such that
each ~$T^{\mu...}_{\nu...}\in\Lambda$~~~ can be
~~~$\forall\Omega_{Ch}$,~~~ $\forall n$,~~~ $\forall\omega\in C^{P}_{S (\mathcal{\mathcal{\tilde S}}_{n})}(\Omega_{Ch})$, ~~~$\forall \mathcal{\mathcal{\mathcal{\mathcal{S}}}}\subset\mathcal{\mathcal{\tilde S}}_{n}~\cap ~At(\mathcal{\mathcal{\tilde S}}_{n})$,\\~~ $\forall Ch_{k}(\Omega_{Ch})\in\mathcal{\mathcal{\mathcal{\mathcal{\mathcal{S}}}}}$~ expressed as a map
\begin{equation}
(\omega,Ch_{k})\to\int_{\Omega_{Ch}}T^{\mu...}_{\nu...}(Ch_{k})\cdot\omega,
\end{equation}
where for $T^{\mu...}_{\nu...}(Ch_{k})$ holds the following: In each
chart from $\mathcal{\mathcal{\mathcal{\mathcal{\mathcal{S}}}}}$ for
every point ~$\mathbf{z_{0}}$, ~where $T^{\mu...}_{\nu...}(Ch_{k})$
is continuous ~$\exists~\delta>0,~\exists~ K^{\mu}_{\nu}>0$, such
that ~$\forall\epsilon$~ $(0\leq\epsilon\leq\delta))$ and for
arbitrary unit vector $\mathbf{n}$~(in the Euclidean metric on
$\mathbb{R}^{4}$)~~
\begin{equation}
T^{\mu...}_{\nu...}(Ch_{k},\mathbf{z_{0}})-K^{\mu...}_{\nu....}\epsilon\leq
T^{\mu...}_{\nu...}(Ch_{k},\mathbf{z_{0}}+\mathbf{n}\epsilon)\leq
T^{\mu...}_{\nu...}(Ch_{k},\mathbf{z_{0}})+K^{\mu...}_{\nu...}\epsilon.
\end{equation}
\label{Lambda}
\end{deff12}

\paragraph{Notation.}
Take from (\ref{Lambda}) arbitrary, fixed ~$T^{\mu...}_{\nu...}$,
~$\Omega_{Ch}$ ~and ~$Ch_{k}(\Omega_{Ch})$. ~By the notation ~$\tilde\Omega(Ch_{k})\subset\Omega_{Ch}$ we denote a set (having Lebesgue measure 0) on
which is  ~$T^{\mu...}_{\nu...}(Ch_{k})$ discontinuous.

\subsubsection{Reproduction of the basic results by the equivalence relation}

\newtheorem{uniqueness}{Theorem}[section]

\begin{uniqueness}
 Any class ~$\tilde\Gamma^{m}_{(\cup_{n}At(\mathcal{\mathcal{\tilde S}}_{n}) o)A}(M)$ contains maximally one linear element.
\end{uniqueness}

\begin{proof}
We need to prove that there do not exist such two elements of
~$\Gamma^{m}(M)$, which are equivalent. Take two elements $B$ and $T$
from the class $\Gamma^{m}_{(\cup_{n}At(\mathcal{\mathcal{\tilde S}}_{n})o)A}(M)$  (both with the given domains and continuity). Take
arbitrary ~$\Omega_{Ch}$, arbitrary ~$\mathcal{\mathcal{\tilde S}}$ from their domains,
and arbitrary ~$Ch'(\Omega_{Ch})\in\mathcal{\mathcal{\tilde S}}$. Map all the
~$C^{P}_{S (\mathcal{\mathcal{\tilde S}})}(\Omega_{Ch})$ objects to smooth, compact
supported functions on $\mathbb{R}^{4}$ through this fixed chart
mapping. Now both $B$ and $T$ give, in fixed but arbitrary
~$Ch(\Omega_{Ch})\in At(\mathcal{\mathcal{\tilde S}})$ linear, continuous maps on the compactly supported smooth
functions. (The only difference from Colombeau distributions is that
it is in general a map to $\mathbb{R}^{m}$, so the difference is only ``cosmetic''.)

Now after applying this construction, our
concept of equivalence reduces for every
~$Ch(\Omega_{Ch})\in At(\mathcal{\mathcal{\tilde S}})$ to Colombeau equivalence from
\cite{Amulti}. The same results must hold. One of the
results says that there are no two distributions being
equivalent. All the parameters are fixed but arbitrary and
all the 4-forms from domains of $B$ and $T$ can be mapped to the
~$\mathbb{R}^{4}$ functions for some proper fixing of ~$\Omega_{Ch}$
and $\mathcal{\mathcal{\tilde S}}$. Furthermore, the $C^{P}_{S}(M)$ 4-forms are arguments of $B$ and $T$ only in the
charts, in which $B$ and $T$ were compared as maps on the spaces of
~$\mathbb{R}^{4}$ functions. So this ``arbitrary
chart fixing'' covers all their domain. As a result $B$ and $T$ must be identical and that is what needed to be proven.
\end{proof}

\newtheorem{extension}[uniqueness]{Theorem}
\begin{extension}
 Any class of equivalence ~$\tilde\Gamma^{m}_{E A}(M)$ contains maximally one linear element.
\end{extension}

\begin{proof}
 First notice that the elements of ~$\Gamma^{m}_{E A}(M)$ are continuous and
 defined on every ~$C^{P}_{S (\mathcal{\mathcal{\tilde S}})}(M)$ in every
chart from $\mathcal{A}$, so they are required to be compared in any arbitrary chart from $\mathcal{A}$. By taking this into account, we can repeat the previous proof. There is one additional trivial fact one has to notice:
the ~$C^{P}_{S}(M)$ domain also uniquely determines how the ~$\Gamma^{m}_{E}(M)$ element acts
outside ~$C^{P}_{S}(M)$. So if $B,~T\in \Gamma^{m}_{E A}(M)$ give the same map on ~$C^{P}_{S}(M)$, they give the same map everywhere.
\end{proof}

\newtheorem{multiplication}[uniqueness]{Theorem}

\begin{multiplication}\label{multiplication}
The following statements hold:
\begin{itemize}
\item[a)] Take ~$T^{\mu...\alpha}_{\nu...\beta}\in\Gamma^{a}_{E A}(M)$ ~such that ~$\forall\Omega_{Ch}$,
~$\forall Ch_{k}(\Omega_{Ch})\in\mathcal{A}$ ~and ~$\forall\omega\in C^{P}(\Omega_{Ch})$ ~~$T^{\mu...\alpha}_{\nu...\beta}$~ is defined as a map
\begin{equation}
(Ch_{k},\omega)\to\int_{\Omega_{Ch}}
T^{\mu...}_{1~\nu...}(Ch_{k})~\omega~...~\int_{\Omega_{Ch}} T^{\alpha...}_{N \beta...}(Ch_{k})~\omega~.\label{Teq}
\end{equation}
Then the class of equivalence ~$\tilde\Gamma^{a}_{E A}(M)$,
to which ~$T^{\mu...}_{\nu...}$~ belongs,
contains a linear element defined (on arbitrary ~$\Omega_{Ch}$) as the map: ~~$\forall\omega~\in ~C^{P}(\Omega_{Ch})$, ~~$\forall Ch_{k}(\Omega_{Ch})~\in ~\mathcal{A}$,
\begin{eqnarray}
(Ch_{k},\omega)~\to~\int_{\Omega_{Ch}}
T^{\mu...}_{1~\nu...}(Ch_{k})...T^{\alpha...}_{N \beta...}(Ch_{k})~\omega~,
\end{eqnarray}
if and only if ~~$\forall\Omega_{Ch}$,~~ $\forall Ch_{k}(\Omega_{Ch})~\in ~\mathcal{A}$ ~~~ $\exists ~Ch_{l}(\Omega_{Ch})~\in ~\mathcal{A}$, ~~such that
\begin{eqnarray}
\int_{Ch_{l}(\Omega_{Ch})_{|\Omega'}} T^{\mu...}_{1~\nu...}(Ch_{k})...T^{...\alpha}_{N...\beta}(Ch_{k}) ~d^{4}x~~~~~~~~~~~~~~~~~~~~~~~~~
\end{eqnarray}
converges on every
compact set $\Omega'\subset\Omega_{Ch}$.

The same statement holds,
if we take instead of ~$\Gamma^{m}_{E A}(M)$ its subclass ~$D'^{a}_{b
E A}(M)$ and instead of the equivalence class ~$\tilde\Gamma^{m}_{E
A}(M)$, the equivalence class ~$\tilde D'^{a}_{b E A}(M)$.

The same statement also holds if we take instead of ~$\Gamma^{m}_{E
A}(M)$ and ~$D'^{m}_{n E A}(M)$ classes, the classes ~$\Gamma^{m}_{E
A (\cup_{l}At(\mathcal{\mathcal{\tilde S}}_{l}) o)}(M)$ and
~$D'^{m}_{n E A (\cup_{l}At(\mathcal{\mathcal{\tilde
S}}_{l})o)}(M)$, ~(with the exception that the given convergence
property shall be considered only for charts from ~
$\cup_{l}At(\mathcal{\tilde S}_{l})$ ).
\item[b)] For any distribution ~~$A^{\alpha...}_{\beta...}~\in ~\Gamma^{a}_{ S
(\cup_{n} \mathcal{\mathcal{\mathcal{\mathcal{S}}}}_{n}~ o)}(M)$,~~ and an element
~~$T^{\mu...}_{\nu...}\in\\
\Gamma^{m}_{(\cup_{n} \mathcal{\mathcal{\mathcal{\mathcal{S}}}}_{n})}(M)$, ~~we have
that ~~$A^{\alpha...}_{\beta...}T^{\mu...}_{\nu...}$~~ is equivalent
to an element of\\ ~~$\Gamma^{m+a}_{(\cup_{n} \mathcal{\mathcal{\mathcal{\mathcal{S}}}}_{n} ~o)}(M)$, ~~(and~
for ~subclasses ~~$D'^{a}_{b
(\cup_{n}\mathcal{\mathcal{\mathcal{\mathcal{S}}}}_{n})}(M)~\subset~\Gamma^{a+b}_{(\cup_{n}\mathcal{\mathcal{\mathcal{\mathcal{S}}}}_{n})}(M)$ ~~and~~\\
$D'^{k}_{l S( \cup_{n}\mathcal{\mathcal{\mathcal{\mathcal{S}}}}_{n}
o)}(M)~\in~\Gamma^{k+l}_{S (\cup_{n}
\mathcal{\mathcal{\mathcal{\mathcal{S}}}}_{n} ~o)}(M)$ ~~it ~is
~equivalent ~to ~an ~element ~of ~~$D'^{k+a}_{l+b
(\mathcal{\mathcal{\mathcal{\mathcal{S}}}} ~o)}(M)$).~ The element
is on its domain given as the mapping
\begin{equation}
(\omega,Ch_{k})\to T^{\mu...}_{\nu...}(A^{\alpha...}_{\beta...}\omega).
\end{equation}
\item[c)] For ~any ~tensor ~distribution ~~$A^{\alpha...}_{\beta...}~\in
~\Gamma^{a}_{S}(M)$  ~~~and ~an ~element
~~$T^{\mu...}_{\nu...}~\in\\
\Gamma^{m}_{(\cup_{n} \mathcal{\mathcal{\mathcal{\mathcal{S}}}}_{n} ~o)}(M)$,
~we ~have ~that ~~$A^{\alpha...}_{\beta...}T^{\mu...}_{\nu...}$~~ is
equivalent ~to ~an ~element ~of ~~\\$\Gamma^{m+a}_{(\cup_{n} \mathcal{\mathcal{\mathcal{\mathcal{S}}}}_{n}
~o)}(M)$. ~The element is on its domain given as mapping
\begin{equation}
(\omega,Ch_{k})\to T^{\mu...}_{\nu...}(A^{\alpha...}_{\beta...}\omega).
\end{equation}
\end{itemize}
\end{multiplication}

\begin{proof}
\begin{itemize}
\item[a)]Use exactly the same construction as in the previous proof. For arbitrary $\Omega_{Ch}$ and arbitrary $Ch_{k}(\Omega_{Ch})\in\mathcal{A}$,
we see that ~$T^{\mu...}_{\nu...}$ is for every $\omega\in C^{P}_{S (\mathcal{\mathcal{\tilde S}})}(\Omega_{Ch})$ given by (\ref{Teq}) (it is
continuous in arbitrary chart on every ~$C^{P}_{S (\mathcal{\mathcal{\tilde S}})}(\Omega_{Ch})$).
 We can express the map (\ref{Teq}) in some chart
$Ch_{l}(\Omega_{Ch})$ as ~
\begin{eqnarray}
(Ch_{k},\omega')~~~~~~~~~~~~~~~~~~~~~~~~~~~~~~~~~~~~~~~~~~~~~~~~~~~~~~~~~~~~~~~~~~~~~~~~~~~~~~~~~~~~~~~~~~~~~~~~\nonumber\\~
\to~\int_{Ch_{l}(\Omega_{Ch})}
T^{\mu...}_{1~\nu...}(Ch_{k})~\omega'~d^{4}x...\int_{Ch_{l}(\Omega_{Ch})}T^{\alpha...}_{N \beta...}(Ch_{k})~\omega'~d^{4}x.~~~~\label{Teq2}
\end{eqnarray}

Then it is a result of Colombeau
theory that if
\begin{equation}
(Ch_{k},\omega')~\to~\int_{Ch_{l}(\Omega_{Ch})}
T^{\mu...}_{1~\nu...}(Ch_{k})...T^{...\alpha}_{N...\beta}(Ch_{k})~\omega'~d^{4}x\label{Teq3}
\end{equation}
is defined as
a linear mapping on compactly supported, smooth ~$\mathbb{R}^{4}$ ~functions ~$\omega'$ ~(in our case
they are related by $Ch_{l}(\Omega_{Ch})$ to given
~$C^{P}_{S}(M)$ objects), it is equivalent to (\ref{Teq2}). Now
everything was fixed, but arbitrary, so the result is
proven. From this proof we also see that the simple
transformation properties of the ~$D'^{m}_{n}(M)$
objects\footnote{We include also the scalar objects here.} are fulfilled by the map (\ref{Teq3}) if the objects multiplied are from ~$D'^{m}_{n
A}(M)$. So the second result can be proven immediately. The last two
results concerning the classes with limited domains trivially follow from the previous proof.
\item[b)] is proven completely in the same way, we just have to understand
that because of the ``limited'' domain of the ~$D'^{a}_{b
S(\cup_{n}\mathcal{\mathcal{\mathcal{\mathcal{S}}}}_{n} o)}(M)$ objects, we can effectively use the concept
of smoothness in this case.
\item[c)] is just the same as b), the only difference is that the domain
of the product is limited because of the ``second'' term in the
product.
\end{itemize}
\end{proof}

Note that this means that tensor product gives, on appropriate
subclasses of ~$D'^{m}_{n E A}(M)$, the mapping ~$\tilde D'^{a}_{b E A}(M)\times\tilde D'^{m}_{n E A}\to \tilde D'^{a+m}_{b+n E A}(M)$. ~It also means that this procedure gives, on
appropriate subclasses of $\Gamma^{m}_{E A}(M)$, the mapping ~$\tilde \Gamma^{a}_{E A}(M)\times\tilde
\Gamma^{m}_{E A}\to \tilde \Gamma^{a+m}_{E A}(M)$.
The disappointing fact is that this cannot be extended to ~$D'^{m}_{n
A}(M)$.

\newtheorem{Comm2}[uniqueness]{Theorem}
\begin{Comm2}
Take ~$T^{\mu...}_{\nu...}\in\Gamma^{m}_{(\cup_{n}
\mathcal{\mathcal{\mathcal{\mathcal{S}}}}_{n}~ o)A}(M)$,~
$B^{\mu...}_{\nu...}\in\Gamma^{m}_{(\cup_{n}
\mathcal{\mathcal{\mathcal{\mathcal{S}}}}_{n}~ o)}(M)$ and
~$L^{\alpha...}_{\beta...}\in\Gamma^{n}_{S (\cup_{l}
\mathcal{\mathcal{\mathcal{\mathcal{S}}}}_{l}~ o)}(M)$. Then
~$T^{\mu...}_{\nu...}\approx B^{\mu...}_{\nu...}$ implies
\footnote{It is obvious that we can extend the definition domains
either of $T^{\mu...}_{\nu...}$ and $B^{\mu...}_{\nu...}$, or of
$L^{\mu...}_{\nu...}$.} ~$(L\otimes
T)^{\alpha...\mu...}_{\beta...\nu...}\approx (L\otimes
B)^{\alpha...\mu...}_{\beta...\nu...}$.
\end{Comm2}

\begin{proof}
Use the same method as previously. It trivially follows from the
results of Colombeau theory (especially from the theorem saying that
if a Colombeau algebra object is equivalent to a distribution, then
after multiplying each of them by a smooth distribution, they remain
equivalent).
\end{proof}

\newtheorem{Contraction}[uniqueness]{Theorem}

\begin{Contraction}
 Contraction (of $\mu$ and $\nu$ index) is always, for such objects
~$T^{...\mu...}_{...\nu...}\in D'^{m}_{n E A}(M)$ that they are
equivalent to some linear element, a map to some element of the
equivalence class from ~~~$\tilde\Gamma^{m+n-2}_{E A}(M)$.~~~ The
~equivalence ~class ~from \\ $\tilde\Gamma^{m+n-2}_{E A}(M)$ is
such, that it contains (exactly) one element from ~$D'^{m-1}_{n-1
E}(M)$ and this element is defined as the map:
~~$\forall\Omega_{Ch}$, ~~$\forall Ch_{k}(\Omega_{Ch})~\in
~\mathcal{A}$, ~~$\forall\omega~\in ~C^{P}(\Omega_{Ch})$,
\begin{eqnarray}
(\omega,Ch_{k})\to\int_{\Omega_{Ch}}
T^{...\alpha...}_{...\alpha...}(Ch_{k})~\omega.~~~~~~~~~~~~
\end{eqnarray}
\end{Contraction}

\begin{proof}
The proof trivially follows from the fact that
contraction commutes with the relation of equivalence (this trivially follows from our previous note about
addition and equivalence).
\end{proof}

\subsubsection{Some interesting conjectures}

\theoremstyle{definition}
\newtheorem{ADD1}{Conjecture}[section]
\begin{ADD1}
Tensor product gives these two maps:
\begin{itemize}
\item
~$\tilde D'^{a}_{b E A}(M)\times\tilde D'^{m}_{n E A}(M)\to
\tilde D'^{a+m}_{b+n E A}(M)$,
\item
~$\tilde \Gamma^{a}_{E A}(M)\times\tilde \Gamma^{b}_{E
A}(M)\to \tilde \Gamma^{a+b}_{E A}(M)$.
\end{itemize}
\end{ADD1}

\theoremstyle{definition}
\newtheorem{ADD2}[ADD1]{Conjecture}
\begin{ADD2}
Take some
~$B^{\mu...}_{\nu...}\in\Gamma^{a}_{(\cup_{n}At(\mathcal{\tilde
S}_{n}) ~o)}(M)$. Take an element
~$T^{\mu...}_{\nu...}\in\Gamma^{b}_{E}(M)$, such that
~$\forall\mathcal{\mathcal{\tilde
S}}\subseteq\cup_{n}\mathcal{\mathcal{\tilde S}}_{n}$,
~$\forall\Omega_{Ch}$, ~$\forall Ch_{k}(\Omega_{Ch})\in
At(\mathcal{\mathcal{\tilde S}})$ it holds that
 ~$\forall\omega\in C^{P}_{S (\mathcal{\mathcal{\tilde S}})}(M)$ ~the elements of the multi-index matrix ~$T^{\mu...}_{\nu...}(Ch_{k})\cdot\omega$~~ remain to be
 from the class ~$C^{P}_{S (\mathcal{\mathcal{\tilde S}})}(M)\subset\cup_{n}C^{P}_{S
(\mathcal{\mathcal{\tilde S}}_{n})}(M)$. Then it holds that
~$B^{\alpha...}_{\beta...}T^{\mu...}_{\nu...}$ is equivalent to an
element of ~$\Gamma^{a+b}_{(\cup_{n}At(\mathcal{\mathcal{\tilde
S}}_{n}) ~o)}(M)$. (For subclasses ~$D'^{m}_{n
(\cup_{l}At(\mathcal{\mathcal{\tilde S}}_{l}) ~o)}(M)$ and
~$D'^{a}_{b E}(M)$ it is equivalent to an element ~$D'^{m+a}_{n+b
(\cup_{l}At(\mathcal{\tilde S}_{l}) ~o)}(M)$.) The element is on its
domain given as mapping
\begin{equation}
(\omega, Ch_{k})\to
B^{\mu...}_{\nu...}(T^{\alpha...}_{\beta...}\omega).
\end{equation}
\end{ADD2}

\subsubsection{Some additional theory}

\newtheorem{association1}{Theorem}[section]
\begin{association1}\label{Association1}
Any arbitrary ~$T^{\mu...}_{\nu...}\in\Lambda$ (as defined by \ref{Lambda}) defines an
~$A_{s}(T^{\mu...}_{\nu...})$ object on $M$. Take any arbitrary ~$\mathcal{\mathcal{\tilde S}}$~ from the domain of $T^{\mu...}_{\nu...}$ ,~any arbitrary ~$\mathcal{\mathcal{\mathcal{\mathcal{S}}}}\subset At(\mathcal{\mathcal{\tilde S}})\cap\mathcal{\mathcal{\tilde S}}$,~any arbitrary $\Omega_{Ch}$ and any arbitrary $Ch_{k}(\Omega_{Ch})\in\mathcal{\mathcal{\mathcal{\mathcal{\mathcal{S}}}}}$. Then for $\Omega_{Ch}\setminus\tilde\Omega(Ch_{k})$ it holds that multi-index matrix of functions $T^{\mu...}_{\nu...}(Ch_{k})$ can be obtained from the tensor components of $A_{s}(T^{\mu...}_{\nu...})$ in $Ch_{k}(\Omega_{Ch})$ by the inverse mapping to $Ch_{k}(\Omega_{Ch})$.
\end{association1}

\begin{proof}
 For ~$\forall\Omega_{Ch}$,~ take fixed but arbitrary ~$ Ch_{k}(\Omega_{Ch})\in\mathcal{\mathcal{\mathcal{\mathcal{\mathcal{S}}}}}$ and take
~$T^{\mu...}_{\nu...}(Ch_{k})$. Then ~$\forall q\in\Omega_{Ch}$, ~$\forall Ch'(\Omega_{Ch},q)\in\mathcal{\mathcal{\mathcal{\mathcal{\mathcal{S}}}}}$,  and ~$\forall ~\omega_{\epsilon}\in A^{n}(\mathcal{\mathcal{\tilde S}}_{l}, Ch'(q,\Omega_{Ch}))$, ~we see that  ~$\omega'(Ch')$~ is a
delta-sequence. That means we just have to show, that on the set where $T^{\mu...}_{\nu...}(Ch_{k})$ is continuous in the \ref{Lambda} sense, the delta-sequencies
give the value of this multi-index matrix. So write the
integral:
\begin{equation*}
\int_{Ch'(\Omega_{Ch})}
T^{\mu...}_{\nu...}\left[(Ch_{k}) (\mathbf{x})\right]~\frac{1}{\epsilon^{4}}~~\omega'\left(\frac{\mathbf{x}}{\epsilon}\right)~d^{4}x.
\end{equation*}
By substitution ~$\mathbf{x}=\epsilon.\mathbf{z}$ we obtain:

\begin{equation*}
\int_{Ch'(\Omega_{Ch})} T^{\mu...}_{\nu...}\left[(Ch_{k})~(\epsilon.\mathbf{
z})\right]~\omega'(\mathbf{z})~d^{4}z.
\end{equation*}
 But from the properties of
~$T^{\mu...}_{\nu...}\left[(Ch_{k})(\mathbf{x})\right]$ it follows that

\begin{eqnarray*}
\int_{Ch'(\Omega_{Ch})}(T^{\mu...}_{\nu...}\left[(Ch_{k})~(\mathbf
z_{0})-K^{\mu...}_{\nu...}\epsilon)\right]~\omega'(\mathbf{z})~d^{4}z~~~~~~~~~~~~~\\
\leq\int_{Ch'(\Omega_{Ch})}T^{\mu...}_{\nu...}\left[(Ch_{k})(\mathbf
z_{0}+\mathbf{n}\epsilon)\right]~\omega'(\mathbf{z})~d^{4}z~~~~~~~~\\
\leq\int_{Ch'(\Omega_{Ch})}(T^{\mu...}_{\nu...}\left[(Ch_{k})(\mathbf
z_{0})\right]+K^{\mu...}_{\nu...}\epsilon)~\omega'(\mathbf{z})~d^{4}z,
\end{eqnarray*}
 for some $\epsilon$ small
enough.

 But we are taking the limit ~$\epsilon\to 0$ which,
considering the fact that ~$\omega(\mathbf{x})$ are normed to 1,
means that the integral must give
~$T^{\mu...}_{\nu...}\left[(Ch_{k})(\mathbf{z_{0}})\right]$. The set, where it is not
continuous in the sense of \ref{Lambda}, has Lebesgue measure 0. That means the
~$\Omega_{Ch}$ part, which is mapped to this set has Lebesgue measure
0. But then the values of the multi-index matrix in the given chart
at this arbitrary, but fixed point give us an associated field (and are independent on delta sequence obviously).
\end{proof}

\newtheorem{association}[association1]{Theorem}

\begin{association}
 The field associated to
a ~$T^{\mu...}_{\nu...}\in \Lambda\cap D'^{m}_{n E}(M)$, ~transforms for each
~$\Omega_{Ch}$, ~for every pair of charts from its
domain, ~$Ch_{1}(\Omega_{Ch}), ~Ch_{2}(\Omega_{Ch})$ ~~on some ~$M/(\tilde\Omega(Ch_{1})\cup \tilde\Omega(Ch_{2}))$, ~ as an ordinary tensor field with
piecewise smooth transformations\footnote{Of course, some transformations in a
generalized sense might be defined also on the $\tilde\Omega(Ch_{1})\cup \tilde\Omega(Ch_{2})$ set.}.
\end{association}

\begin{proof}
All this immediately follows from what was done in
the previous proof, and from the fact that union of sets with Lebesgue measure 0 has Lebesgue measure 0.
\end{proof}

Note, that if there exists such point that for the object ~$\Gamma^{m}_{A}(M)$ we have

\begin{equation*}
\lim_{\epsilon\to 0}
T^{\mu...}_{\nu...}(\omega_{\epsilon})=\pm\infty
\end{equation*}
at that
point, then the field associated to this object, can be associated to
another object, which is nonequivalent to this object. This means that the same ~field can be
associated to mutually non-equivalent elements of ~$\Gamma^{m}_{A}(M)$. This is
explicitly shown and proven by the next example.

\newtheorem{dfunction}[association1]{Theorem}

\begin{dfunction}
 Take $\delta(q,Ch_{k}(\Omega_{Ch}))\in D'_{(\mathcal{\mathcal{\mathcal{\mathcal{S}}}} o)}(M)$ being defined as mapping from each 4-form ~$C^{P}_{S (\mathcal{\mathcal{\tilde S}})}(M)~(\mathcal{\mathcal{\mathcal{\mathcal{S}}}}\subseteq\mathcal{\mathcal{\tilde S}})$ ~to the
 value of this form's density at the point $q$ in the chart ~$Ch_{k}(\Omega_{Ch})\in\mathcal{\mathcal{\mathcal{\mathcal{\mathcal{S}}}}}$,~ ($q\in\Omega_{Ch}$). Then any
 power ~$n\in\mathbb{N}_{+}$ of $\delta(q,Ch_{k}(\Omega_{Ch}))$ is associated to the function being defined on the domain $M\setminus\{q\}$ and everywhere 0. Note that this function is associated to any power ($n\in\mathbb{N}_{+}$) (being a nonzero natural number) of $\delta(q)$, but different powers of $\delta(q)$ are mutually nonequivalent\footnote{It is hard to find in our theory a more ``natural'' definition generalizing the concept of delta function from $\mathbb{R}^{n}$. But there is still another natural generalization: it is an object from ~$\Gamma^{0}_{(\cup_{n}\mathcal{\mathcal{\mathcal{\mathcal{S}}}}_{n}o)}(M)$, defined as:~
$\delta(Ch_{k}(\Omega_{Ch}),q,\omega)=\omega'\left[(Ch_{k})(\tilde
q)\right]$,
~$Ch_{k}(\Omega_{Ch})\in\mathcal{\mathcal{\mathcal{\mathcal{\mathcal{S}}}}}_{n}$,
~$\omega\in C^{P}_{S (\mathcal{\tilde S}_{n})}(\Omega_{Ch})$
~($\mathcal{\mathcal{\mathcal{\mathcal{S}}}}_{n}\subset\mathcal{\mathcal{\tilde
S}}_{n}$), ~$\tilde q$ is image of $q$ given by the chart mapping
$Ch_{k}(\Omega_{Ch})$. So it gives value of the density $\omega'$ in
the chart $Ch_{k}(\Omega_{Ch})$, at the chart image of the point
$q$.}.\label{dfunction}
\end{dfunction}

\begin{proof}
Contracting powers of
~$\delta(q,Ch_{k}(\Omega_{Ch}))$ with a sequence of 4-forms from
arbitrary ~$A^{n}(Ch'(q',\Omega'_{Ch}),\mathcal{\mathcal{\tilde S}}), ~(q'\neq q)$ (they
have a support converging to another point than $q$) will give
0. For $q=q'$ (\ref{limit}) gives
\[\lim_{\epsilon\to 0}~\epsilon^{-4}~\omega'\left[(Ch_{k})(0)\right]=\pm\infty~.\]
 Now explore the equivalence between different powers of ~$\delta(Ch_{k},q)$.\\ $\delta^{n}(q,Ch_{k}(\Omega_{Ch}))$ applied to
~$\omega_{\epsilon}(x)\in A^{n}(Ch'_{k}(q,\Omega_{Ch}),\mathcal{\mathcal{\tilde S}})$ will
lead to the expression
~$\epsilon^{-4n}\omega^{n}(\frac{x}{\epsilon})$. Then if we
want to compute

\begin{equation*}
\lim_{\epsilon\to
0}\int_{Ch'(\Omega_{Ch})}\left(\frac{1}{\epsilon^{4n}}~\omega^{n}\left(\frac{\mathbf{x}}{\epsilon}\right)-\frac{1}{\epsilon^{4m}}~\omega^{m}\left(\frac{\mathbf{x}}{\epsilon}\right)\right)~\Phi(\mathbf{x})~d^{4}x
\end{equation*}
it leads to

\begin{equation*}
 \lim_{\epsilon\to
0}~\frac{1}{\epsilon^{4m-4}}~\Phi(0)~\int_{Ch'(\Omega_{Ch})}\left(\omega^{n}(\mathbf{x})-\frac{1}{\epsilon^{4(n-m)}}~\omega^{m}(\mathbf{x})\right)~d^{4}x
\end{equation*}

 which is for $n\neq m, ~~n,m\in\mathbb{N}_{+}$ clearly
divergent, hence nonzero.
\end{proof}

Note that despite of the fact that within our algebras we, naturally, have all the $n\in\mathbb{N_+}$ powers of the delta
distribution, they are for $n>1$, ~unfortunately, ~\emph{not} equivalent to any distribution.\newline

\newtheorem{nsmooth}[association1]{Theorem}

\begin{nsmooth}
 We see that the map $A_{s}$ is linear (in the sense analogous to \ref{linear}), and for arbitrary number of
~$g^{\mu...}_{\nu...},...,h^{\mu...}_{\nu...}\in\Lambda\cap\Gamma^{m}_{E (\cup_{n}At(\mathcal{\mathcal{\tilde S}}_{n})o)}(M)$ one has:
Take $\forall\Omega_{Ch}$, ~$\forall n$, ~$\forall \mathcal{\mathcal{\mathcal{\mathcal{S}}}}\subset\mathcal{\mathcal{\tilde S}}_{n}\cap At(\mathcal{\mathcal{\tilde S}})$, ~$\forall Ch_{k}(\Omega_{Ch})\in\mathcal{\mathcal{\mathcal{\mathcal{\mathcal{S}}}}}$,~ $\cup_{i}\tilde\Omega_{i}(Ch_{k})$~ to be the union of all $\tilde\Omega_{i}(Ch_{k})$~ related to the objects $g^{\mu...}_{\nu...},...,h^{\mu...}_{\nu...}$. Then
\[A_{s}(g^{\alpha...}_{\beta...}\otimes...\otimes
h^{\mu...}_{\nu...})=A_{s}(g^{\alpha...}_{\beta...})\otimes...\otimes
A_{s}(h^{\mu...}_{\nu...})\]
on $\Omega_{Ch}\setminus \cup_{i}\tilde\Omega_{i}(Ch_{k})$. ~Here the
first term is a product between $\Lambda$ objects and the second is the
classical tensor product.
\end{nsmooth}

\begin{proof}
It is trivially connected with previous proofs:
Note that from the definition (\ref{DefAssoc}) for appropriate 4-form fields $\omega_{\epsilon}$ we have
\begin{equation}
A_{s}(g^{\alpha...}_{\beta...}\otimes...\otimes
h^{\mu...}_{\nu...})(Ch_{k})=\lim_{\epsilon\to 0}~ g^{\alpha...}_{\beta...}(Ch_{k},\omega_{\epsilon})...
h^{\mu...}_{\nu...}(Ch_{k},\omega_{\epsilon}).
\end{equation}

But for the objects $g^{\alpha...}_{\beta...},...,h^{\alpha...}_{\beta...}\in\Lambda$, with respect to the theorem (\ref{Association1}) necessarily
\begin{equation}
\lim_{\epsilon\to 0}~ g^{\alpha...}_{\beta...}(Ch_{k},\omega_{\epsilon})...
h^{\mu...}_{\nu...}(Ch_{k},\omega_{\epsilon})=A_{s}(g^{\alpha...}_{\beta...})\otimes...\otimes
A_{s}(h^{\mu...}_{\nu...}).
\end{equation}
This proves the theorem.
\end{proof}

This means that the result of tensor multiplication of elements from $\Lambda$ (it has product of two scalars as a subcase) is always equivalent to some element from $\Lambda$.
This is a result closely related to the theorem (\ref{multiplication}). It tells us that multiplication is a mapping between equivalence classes formed of more constrained classes as those mentioned in the part $a)$ of the theorem (\ref{multiplication}).

\theoremstyle{definition}
\newtheorem{conj1}{Conjecture}[section]
\begin{conj1}
 If ~$T^{\mu...}_{\nu...}\in D'^{m}_{n (\mathcal{\mathcal{\mathcal{\mathcal{S}}}} o)}(M)$ has an
associated field, then it transforms on its domains as a
tensor field.
\end{conj1}

\theoremstyle{definition}
\newtheorem{conj2}[conj1]{Conjecture}
\begin{conj2}
 If ~$T^{\mu...}_{\nu...}\in \Gamma^{m}_{(\mathcal{\mathcal{\mathcal{\mathcal{S}}}} o)}(M)$ has an
associated field and ~$L^{\alpha...}_{\beta...}\in\Gamma^{n}_{S}(M)$,
then ~$A_{s}(T^{\mu...}_{\nu...}\otimes L^{\alpha...}_{\beta...})=A_{s}(T^{\mu...}_{\nu...})\otimes A_{s}(L^{\alpha...}_{\beta...})$. ~~(The $\otimes$ sign has again slightly different meaning on the different sides of
the equation).
\end{conj2}

\theoremstyle{definition}
\newtheorem{conj3}[conj1]{Conjecture}
\begin{conj3}
Take
~$C^{\mu...}_{\nu...},D^{\mu...}_{\nu...},F^{\alpha...}_{\beta...},B^{\alpha...}_{\beta...}\in\Gamma^{m}_{A}(M)$,
~such that they belong to the same classes
~$\Gamma^{m}_{(At(\mathcal{\mathcal{\tilde S}}))A}(M)$. Also assume
that ~$\forall\mathcal{\mathcal{\tilde S}}$, such that
~$C^{\mu...}_{\nu...},D^{\mu...}_{\nu...},F^{\alpha...}_{\beta...},\\
B^{\alpha...}_{\beta...}\in\Gamma^{m}_{(At(\mathcal{\mathcal{\tilde
S}}))A}(M)$, ~and ~$\forall
\mathcal{\mathcal{\mathcal{\mathcal{S}}}}\subseteq\mathcal{\mathcal{\tilde
S}}$,~ each of the elements
~$C^{\mu...}_{\nu...},D^{\mu...}_{\nu...},F^{\alpha...}_{\beta...},B^{\alpha...}_{\beta...}$
has for every ~$At(\mathcal{\mathcal{\tilde S}})$ associated fields
defined on the whole $M$. Then ~$F^{\alpha...}_{\beta...}\approx
B^{\alpha...}_{\beta...}$ and ~$C^{\mu...}_{\nu...}\approx
D^{\mu...}_{\nu...}$ implies ~$(C\otimes
F)^{\mu...\alpha...}_{\nu...\beta...}\approx (D\otimes
B)^{\mu...\alpha...}_{\nu...\beta...}$.
\end{conj3}

\subsection{Covariant derivative}

~~~~~The last missing fundamental concept is the
covariant derivative operator on generalized tensor fields (GTF). This operator is necessary to formulate an appropriate language for physics and generalize
physical laws. Such an operator must obviously reproduce our concept
of the covariant derivative on the smooth tensor fields (through the
given association relation to the smooth manifold). This is provided
in the following section. The beginning of the first part is again devoted to
fundamental definitions. The beginning of the second part gives us fundamental
theorems, again just generalizing Colombeau results for our case. After these theorems we, (similarly to previous section),
formulate conjectures representing the very important and natural
extensions of our results (bringing a lot of new significance to
our results). The last subsection in the second part being again called ``some additional theory''
brings (analogously to previous section) just physical insight to our abstract calculus and is of lower mathematical
importance.

\subsubsection{Definition of  $\partial$-derivative and connection coefficients}

\theoremstyle{definition}
\newtheorem{defn15}{Definition}[section]
\begin{defn15}
 We define a map, called the $\partial$-derivative, given by smooth vector
field $U^{i}$ (smooth in the atlas $\mathcal{\mathcal{\mathcal{\mathcal{\mathcal{S}}}}}$)~ as a mapping
$\Gamma^{m}_{(\mathcal{\mathcal{\mathcal{\mathcal{S}}}})}(M)\to \Gamma^{m}_{(\mathcal{\mathcal{\mathcal{\mathcal{S}}}})}(M)$,
given on its domain ~$\forall\Omega_{Ch}$~ and ~$Ch_{k}(\Omega_{Ch})$ as:
\begin{equation*}
T^{\mu...}_{\nu...~,(U)}(Ch_{k},\omega)\equiv
-T^{\mu...}_{\nu...}(Ch_{k},(U^{\alpha}\omega)_{,\alpha})~~~~~~\omega\in
C^{P}(\Omega_{Ch}).
\end{equation*}
Here
~$(U^{\alpha}\omega)_{,\alpha}$ is understood in the following way: We
express ~$U^{\alpha}$ in the chart ~~$Ch_{k}(\Omega_{Ch})$ ~~and take the derivatives
~~$\left(U^{\alpha}(Ch_{k})~\omega'(Ch_{k})\right)_{,\alpha}$~~ in the same chart\footnote{We will further express that the derivative is taken in $Ch_{k}(\Omega_{Ch})$ by using the notation $(U^{\alpha}\omega'(Ch_{k}))_{, [(Ch_{k})\alpha]}$.~} ~$Ch_{k}(\Omega_{Ch})$.~ They give us some
function in $Ch_{k}(\Omega_{Ch})$, which can be (in this chart) taken as expression for density of some object from
~$C^{P}(\Omega_{Ch})$. This means we trivially extend it to $M$ by taking it to be 0
everywhere outside ~$\Omega_{Ch}$.~ This is the object used as an argument in $T^{\mu...}_{\nu...}$.
\end{defn15}

To make a consistency check: This ~$T^{\mu...}_{\nu...,(U)}$~ is an
object which is defined at least on the domain ~$C^{P}_{S
(\mathcal{\tilde S})}(M)$~ for
~$\mathcal{\mathcal{\mathcal{\mathcal{S}}}}\subseteq\mathcal{\mathcal{\tilde
S}}$~ ($\mathcal{\mathcal{\mathcal{\mathcal{\mathcal{S}}}}}$
related to $U^{i}$ in the sense that $U^{i}$ is smooth in
$\mathcal{\mathcal{\mathcal{\mathcal{\mathcal{S}}}}}$),~ and is
continuous on the same domain. This means it belongs to the class
~$\Gamma^{n}(M)$.~ To show this take some arbitrary ~$\Omega_{Ch}$
and some arbitrary chart
~$Ch_{k}(\Omega_{Ch})\in\mathcal{\mathcal{\mathcal{\mathcal{\mathcal{S}}}}}$.
We see that within ~$Ch_{k}(\Omega_{Ch})$ the expression
~$(U^{\alpha}\omega'),_{\alpha}$ is smooth and describes 4-forms,
which are compactly supported, with their support being subset of
$\Omega_{Ch}$. Hence they are from the domain of
~$T^{\mu...}_{\nu...}$ in every chart from
$\mathcal{\mathcal{\mathcal{\mathcal{\mathcal{S}}}}}$. In any
arbitrary chart from
$\mathcal{\mathcal{\mathcal{\mathcal{\mathcal{S}}}}}$ we
trivially observe, (from the theory of distributions), that if
~$\omega_{n}\to\omega$, than ~$T^{\mu...}_{\nu...,
(U)}(\omega_{n})\to T^{\mu...}_{\nu...,(U)}(\omega)$. ~It means that
$T^{\mu...}_{\nu...,(U)}(\omega)$ is continuous.

\theoremstyle{definition}
\newtheorem{defn16}[defn15]{Definition}
\begin{defn16}
 Now, by generalized connection we denote an
object from ~$\Gamma^{3}(M)$ such that:
\begin{itemize}
\item
The set
\[\cup_{\Omega_{Ch}}\cup_{\mathcal{\mathcal{\tilde S}}}\bigg\{\big(\omega,Ch(\Omega_{Ch})\big):\omega\in C^{P}_{S(\mathcal{\mathcal{\tilde S}})}(\Omega_{Ch}), ~~Ch(\Omega_{Ch})\in\mathcal{A}\bigg\}\]
 belongs to its domain.
\item
It is ~$\forall\mathcal{\mathcal{\tilde S}}$,~ $\forall\Omega_{Ch}$,~ $\forall Ch_{k}(\Omega_{Ch})\in\mathcal{A}$~
continuous ~on ~$C^{P}_{S(\mathcal{\mathcal{\tilde S}})}(\Omega_{Ch})$ ~with ~$Ch_{k}(\Omega_{Ch})$ ~taken as its argument.
\item
It transforms as:
\begin{equation}
\Gamma^{\alpha}_{\beta\gamma}(Ch_{2},\omega)=
\Gamma^{\mu}_{\nu\delta}(Ch_{1},((J^{-1})^{\nu}_{\beta}(J^{-1})^{\delta}_{\gamma}J^{\alpha}_{\mu}-J^{\alpha}_{m}(J^{-1})^{m}_{\beta,\gamma})~\omega).
\end{equation}
\end{itemize}
\end{defn16}

\subsubsection{Definition of covariant derivative}

\theoremstyle{definition}
\newtheorem{defn17}[defn15]{Definition}
\begin{defn17}
 By a covariant derivative ~ (on $\Omega_{Ch}$) in the
direction of a vector field ~$U^{i}(M)$, smooth with respect to atlas $\mathcal{\mathcal{\mathcal{\mathcal{\mathcal{S}}}}}$, of an object
~$T^{\mu...}_{\nu...}\in \Gamma^{m}_{(\mathcal{\mathcal{\mathcal{\mathcal{S}}}})}(\Omega_{Ch})$,~ we mean:
\begin{equation}
D_{C (U)}T^{\mu...}_{\nu...}(\omega)\equiv
T^{\mu...}_{\nu...~,(U)}(\omega)+\Gamma^{\mu}_{\alpha\rho}(\omega)T^{\rho...}_{\nu...}(U^{\alpha}\omega)-\Gamma^{\alpha}_{\nu\rho}(\omega)T^{\mu...}_{\alpha...}(U^{\rho}\omega).\label{cov}
\end{equation}

\end{defn17}

\medskip

This definition (\ref{cov}) automatically defines covariant derivative everywhere on
~$\Gamma^{m}_{(\mathcal{\mathcal{\mathcal{\mathcal{S}}}})}(M)$. This can be easily observed: The
$\partial$-derivative still gives us an object from ~$\Gamma^{m}_{(\mathcal{\mathcal{\mathcal{\mathcal{S}}}})}(M)$, and the
second term containing generalized connection is from
~$\Gamma^{m}_{(\mathcal{\mathcal{\mathcal{\mathcal{S}}}})}(M)$ trivially too.

\theoremstyle{definition}
\newtheorem{defn177}[defn15]{Definition}
\begin{defn177}
Furthermore, let us
extend the definition of covariant derivative to the class ~$\Gamma^{m}_{
(\mathcal{\mathcal{\mathcal{\mathcal{S}}}})A}(M)$ just by stating that on every nonlinear object (note that every such object is constructed by tensor product of linear objects) it is defined by the Leibniz rule. (This means
it is a standard derivative operator, since it is trivially linear
as well.)
\end{defn177}

\subsubsection{The $S'_n$ class}

\paragraph{Notation.}
Take as ~$\mathcal{D}_{n}$ some $n$-times continuously
differentiable subatlas of $\mathcal{A}$. Then the class $S'_{n}$ related to the atlas ~$\mathcal{D}_{n}$~
is formed by objects ~$T^{\mu...}_{\nu...}\in\Gamma^{a}_{E
(\cup_{m}At(\mathcal{\mathcal{\tilde S}}_{m})o)A}(M)$, such that:
\begin{itemize}
\item
$\cup_{m}\mathcal{\mathcal{\tilde S}}_{m}=\mathcal{D}_{n}$~ and
~$\forall m,~~At(\mathcal{\mathcal{\tilde
S}}_{m})=\mathcal{D}_{n}$,
\item
it is given  $\forall\Omega_{Ch}$, ~$\forall
Ch_{k}(\Omega_{Ch})\in\mathcal{D}_{n}$, ~$\forall m$, ~$\forall\omega\in
C^{P}_{S (\mathcal{\mathcal{\tilde S}}_{m})}(\Omega_{Ch})$ as a map
\[(\omega, Ch_{k})\to\int_{\Omega_{Ch}}T^{\mu...}_{\nu...}(Ch_{k})~\omega.~~~~~~~\]

 Here $T^{\mu...}_{\nu...}(Ch_{k})$ is a multi-index matrix of $n$-times continuously differentiable functions (if it is being expressed in any arbitrary chart from $\mathcal{D}_{n}$).
\end{itemize}

\subsubsection{Basic equivalence relations related to differentiation and some of the
interesting conjectures}

\newtheorem{lcovder1}{Theorem}[section]

\begin{lcovder1}
The following statements hold:
\begin{itemize}
\item[a)] Take a vector field ~$U^{i}$, which is smooth at
 ~$\mathcal{\mathcal{\mathcal{\mathcal{S}}}}\subset \mathcal{D}_{n+1}\subset \mathcal{D}_{n}$ for $n\geq 1$, ~(formally including also ~$n=n+1=\infty$, hence $\mathcal{\mathcal{\mathcal{\mathcal{\mathcal{S}}}}}$), plus a generalized
 connection
 ~$\Gamma^{\mu}_{\nu\alpha}\in\Gamma^{3}_{E}(M)$,~ and
~$T^{\mu...}_{\nu...}\in D'^{a}_{b}(M)\cap S'_{n}$. ~ $S'_{n}$ is
here related to the given ~$\mathcal{D}_{n}$. Then it follows that:
$D_{C(U)}T^{\mu...}_{\nu...}$ is an object from $S'_{n+1}$ ~(being
related to the given $\mathcal{D}_{n+1}$). Moreover, the equivalence
class ~$\tilde S'_{n+1}$ of the image contains exactly one linear
element given for every chart from its domain as integral from some
multi-index matrix of piecewise continuous functions. Particularly
this element is for $\forall\Omega_{Ch}$ and arbitrary such
$\omega\in C^{P}(\Omega_{Ch})$, ~$Ch_{k}(\Omega_{Ch})$ that are from
its domain a map:
\begin{equation}
(Ch_{k},\omega)\to\int
U^{\alpha}T^{\mu...}_{\nu...;\alpha}(Ch_{k})~\omega.
\end{equation}
 ~(Here ``$;$'' means the
classical covariant derivative related to the ``classical'' connection, components of which are in $Ch_{k}$ given by
~$\Gamma^{\alpha}_{\beta\delta}(Ch_{k})$, and the tensor field, components of which are in $Ch_{k}$ given by $T^{\mu...}_{\nu...}(Ch_{k})$.)
\item[b)] Take $U^{i}$ being smooth in $\mathcal{\mathcal{\mathcal{\mathcal{\mathcal{S}}}}}$ and
~$\Gamma^{\mu}_{\nu\alpha}\in\Gamma^{3}_{S}(M)$. Then the following holds:
The covariant derivative is a map  ~$D'^{m}_{n (\mathcal{\mathcal{\mathcal{\mathcal{S}}}}
o)}(M)\to\tilde\Gamma^{m+n}_{ (\mathcal{\mathcal{\mathcal{\mathcal{S}}}} o)A}(M)$~ and the classes
~~$\tilde\Gamma^{m+n}_{ (\mathcal{\mathcal{\mathcal{\mathcal{S}}}} o)A}(M)$ ~of the image contain (exactly)
one element of $D'^{m}_{n (\mathcal{\mathcal{\mathcal{\mathcal{S}}}} o)}(M)$.
\end{itemize}\label{lcovder1}
\end{lcovder1}

\begin{proof}

\begin{itemize}
\item[a)]Covariant derivative is, on its domain, given as
\begin{equation*}
(\omega, Ch_{k})\to T^{\mu...}_{\nu...~,(U)}(\omega)+\Gamma^{\mu}_{\alpha\rho}(\omega)T^{\rho...}_{\nu...}(U^{\alpha}\omega)-\Gamma^{\alpha}_{\nu\rho}(\omega)T^{\mu...}_{\alpha...}(U^{\rho}\omega).
\end{equation*}

Take the first term ~$T^{\mu...}_{\nu...,(U)}$. Express it on
$\Omega_{Ch}$ in ~$Ch_{k}(\Omega_{Ch})\in  \mathcal{D}_{n}$ as the
map:
\[\omega\to\int_{Ch_{k}(\Omega_{Ch})}T^{\mu....}_{\nu...}(Ch_{k})~\omega'(Ch_{k})~d^{4}x.\]
 ~So analogously:
\[T^{\mu...}_{\nu...,(U)}(Ch_{k},\omega)=
-\int_{Ch_{k}(\Omega_{Ch})}T^{\mu...}_{\nu...}(Ch_{k})\left(U^{\alpha}(Ch_{k})~\omega'(Ch_{k})\right)_{,[(Ch_{k})\alpha]}~d^{4}x.\]

Here $\omega'$ is in arbitrary chart from ~$\mathcal{D}_{n}$,
being~ $n$-times continuously differentiable, (the domain is limited
to such objects by the second covariant derivative term, which is
added to the $\partial$-derivative), and such that the expression
~$(U^{i}(Ch_{k})~\omega'(Ch_{k}))_{,[(Ch_{k})i]}$ is in any
$Ch_{k}(\Omega_{Ch})\in\mathcal{D}_{n}$~ $n$-times continuously
differentiable.

Now by using integration by parts,
(since all the objects under the integral are at least continuously
differentiable ($n\geq 1$), it can safely be used), and considering the compactness of
support we obtain:

\begin{eqnarray}
T^{\mu...}_{\nu...,(U)}(Ch_{k},\omega)~~~~~~~~~~~~~~~~~~~~~~~~~~~~~~~~~~~~~~~~~~~~~~~~~~~~~~~~~~~~~~~~~~~~~~~~~~~~~~~~~\nonumber\\
=-\int_{Ch_{k}(\Omega_{Ch})}T^{\mu...}_{\nu...}(Ch_{k})\left(U^{\alpha}(Ch_{k})~\omega'(Ch_{k})\right)_{,[(Ch_{k})\alpha]}~d^{4}x\nonumber\\ =\int_{Ch_{k}(\Omega_{Ch})}T^{\mu...}_{\nu...}(Ch_{k})_{,[(Ch_{k})\alpha]}U^{\alpha}(Ch_{k})~\omega'(Ch_{k})~d^{4}x.~~~
\end{eqnarray}
Then it holds that:
\begin{eqnarray}
T^{\mu...}_{\nu...,(U)}(Ch_{m}(\Omega_{Ch}),\omega)~~~~~~~~~~~~~~~~~~~~~~~~~~~~~~~~~~~~~~~~~~~~~~~~~~~~~~~~~~~~~~~~~~~~~~~~~\nonumber\\
=\int_{Ch_{m}(\Omega_{Ch})}T^{\mu...}_{\nu...}(Ch_{m})_{,[(Ch_{m})\alpha]}(U^{\alpha}(Ch_{m})~\omega'(Ch_{m})~d^{4}x~~~~~~~~~~~~~~~~~~~~\nonumber\\
=\int_{Ch_{k}(\Omega_{Ch})}(J^{\mu}_{\beta}(J^{-1})
^{\delta}_{\nu}....T^{\delta...}_{\beta...}(Ch_{k}))_{,[(Ch_{k})\alpha]}U^{\alpha}(Ch_{k})~\omega'(Ch_{k})~d^{4}x.~~~~~
\end{eqnarray}
We see that this is defined and continuous in any arbitrary chart
from $\mathcal{D}_{n+1}$ for every $C^{P}_{S
(\mathcal{\mathcal{\tilde S}})}(M)$, such that $\exists
\mathcal{\mathcal{\mathcal{\mathcal{S}}}}\subset\mathcal{\mathcal{\tilde
S}}$ and $\mathcal{\mathcal{\mathcal{\mathcal{S}}}}\subset
\mathcal{D}_{n+1}$.

Now we see that the second term in the covariant
derivative expression is equivalent to the map (with the same domain):
\begin{equation*}
(Ch_{k},\omega)\to\int(\Gamma^{\mu}_{\alpha\rho}U^{\alpha}T^{\rho...}_{\nu...}-\Gamma^{\alpha}_{\nu\rho}U^{\rho}T^{\mu...}_{\alpha...})(Ch_{k})~\omega,
\end{equation*}

(see
theorem \ref{multiplication}), and between
charts the objects appearing
inside the integral transform exactly as their classical analogues. This must hold, since $T^{\mu...}_{\nu...}$ is
everywhere continuous (in every chart considered), hence on every
compact set bounded, so the given object is well defined.
This means that when we fix this object in chart ~$Ch_{m}(\Omega_{Ch})$, and express it through the chart ~$Ch_{k}(\Omega_{Ch})$ and
Jacobians (with the integral expressed at chart
~$Ch_{k}(\Omega_{Ch})$ ), as in the previous case, we discover
(exactly as in the classical case), that the resulting object under
the integral transforms as some object ~$D'^{m}_{n E}(M)$, with the classical expression for the
covariant derivative of a tensor field appearing under the integral.
\item[b)]The resulting object is defined particularly only on ~$C^{P}_{S
(\mathcal{\mathcal{\tilde S}})}(M)$ ~($\mathcal{\mathcal{\mathcal{\mathcal{S}}}}\subset\mathcal{\mathcal{\tilde S}}$). We have to realize that
~$T^{\mu...}_{\nu...}$ can be written as a ($N\to\infty$) weak limit
(in every chart from $\mathcal{\mathcal{\mathcal{\mathcal{\mathcal{S}}}}}$) of ~$T^{\mu...}_{\nu... N}\in D'^{m}_{n S
(\mathcal{\mathcal{\mathcal{\mathcal{S}}}} o)}(M)$. It is an immediate result of previous constructions and
Colombeau theory, that

\begin{equation}
-\Gamma^{\alpha}_{\nu\rho}(.)~T^{\mu...}_{\alpha...}(U^{\rho}.)\approx -T^{\mu...}_{\alpha...}(\Gamma^{\alpha}_{\nu\rho}U^{\rho}.)
\label{object}
\end{equation}

Now take ~$\forall\Omega_{Ch}$~ both, $T^{\mu...}_{\nu...,(U)}$ and (\ref{object}) fixed in arbitrary ~$Ch_{k}(\Omega_{Ch})\in\mathcal{\mathcal{\mathcal{\mathcal{\mathcal{S}}}}}$.~ Write both of those objects as limits of integrals of some
sequence of ``smooth'' objects in ~$Ch_{k}(\Omega_{Ch})$.~ But now we
can again use for ~$T^{\mu...}_{\nu...,(U)}$~ an integration per parts and from the ``old'' tensorial relations; we get the
``tensorial'' transformation properties under the limit. This means
that the resulting object, which is a limit of those objects
transforms in the way the ~$D'^{m}_{n}(M)$ objects transform.
\end{itemize}
\end{proof}

\newtheorem {meantime}[lcovder1]{Theorem}
\begin{meantime}

Part $a)$ of the theorem (\ref{lcovder1}) can be also formulated
through a generalized concept of covariant derivative, where we do
not require the $U^{i}$ vector field to be smooth at some
$\mathcal{\mathcal{\mathcal{\mathcal{\mathcal{S}}}}}$, but it
is enough if it is $n+1$ differentiable in $\mathcal{D}_{n+1}$.

\end{meantime}

\begin{proof}
We just have to follow our proof and realize that the only reason why we used
smoothness of $U^{i}$ in $\mathcal{\mathcal{\mathcal{\mathcal{\mathcal{S}}}}}$ was that it is required by our
definition of covariant derivative (for another good reasons related to different cases).
\end{proof}

This statement has a crucial importance, since it shows that not only
all the classical calculus of smooth tensor fields with all the
basic operations is contained in our language (if we take the
equivalence instead of equality being the crucial part of our
theory), but it can be even extended to arbitrary objects from $S'_{n}$.
(If the covariant derivative is obtained through connection from the class
$\Gamma^{3}_{E}(M)$.) In other words it is more general than the
classical tensor calculus.

We can now think
about conjectures extending our results in a very important way:

\theoremstyle{definition}
\newtheorem{cc1}{Conjecture}[section]
\begin{cc1}
 Take an arbitrary piecewise smooth, and on every compact set bounded \footnote{To be exact, the expression ``covariant derivative'' is used in this and the following conjecture in a more general way, since we do not put on $U^{i}$ the condition of being smooth in some subatlas $\mathcal{\mathcal{\mathcal{\mathcal{\mathcal{S}}}}}$.} vector field ~$U^{i}$. Take also ~$\Gamma^{\mu}_{\nu\alpha}\in\Gamma^{3}_{E}(M)$ and such
~$T^{\mu...}_{\nu...}\in D'^{m}_{n E}(M)$, that ~$\forall\Omega_{Ch}$,~ $\forall Ch_{k}(\Omega_{Ch})\in\mathcal{A}$ ~$\exists Ch_{l}(\Omega_{Ch})\in\mathcal{A}$,~ in which\footnote{This means we are trivially integrating~ $T^{\mu...}_{\nu...}\Gamma^{\alpha}_{\beta\delta}$~on compact sets within subset of $\mathbb{R}^{4}$ given as image of the given chart mapping. }
\[\int_{Ch_{l|\Omega'}(\Omega_{Ch})}
T^{\mu...}_{\nu...}(Ch_{k})\Gamma^{\alpha}_{\beta\delta}(Ch_{k})~d^{4}x\]
converges on every compact set~ $\Omega'\subset\Omega_{Ch}$. ~Then the following holds:
The covariant derivative (along $U^{i}$) maps this object to an
element of some equivalence class from ~$\tilde\Gamma^{m+n}_{ A }(M)$. This class contains (exactly) one
element from ~$D'^{m}_{n}(M)$.
\end{cc1}

\theoremstyle{definition}
\newtheorem{cc2}[cc1]{Conjecture}
\begin{cc2}
Take ~$U^{i}$ being a piecewise smooth vector field, and
~$\Gamma^{\mu}_{\nu\alpha}\in\Gamma^{3}_{S}(M)$. Then the following
holds: The covariant derivative along this vector field is a map:
~$D'^{m}_{n (\cup_{l}\mathcal{\mathcal{\mathcal{\mathcal{S}}}}_{l}
~o)}(M)\to\tilde\Gamma^{m+n}_{
(\cup_{l}\mathcal{\mathcal{\mathcal{\mathcal{S}}}}_{l}~o)A}(M)$, and
the classes ~$\tilde\Gamma^{m+n}_{
(\cup_{l}\mathcal{\mathcal{\mathcal{\mathcal{S}}}}_{l}~o)A}(M)$ of
the image contain (exactly) one element of ~$D'^{m}_{n}(M)$.
\end{cc2}

\newtheorem{Comm}[lcovder1]{Theorem}
\begin{Comm}
For ~$U^{i}$ being a smooth tensor field in $\mathcal{\mathcal{\mathcal{\mathcal{\mathcal{S}}}}}$ with the connection taken
from ~$\Gamma^{3}_{S}(M)$, ~$T^{\mu...}_{\nu...}\in\Gamma^{m}_{ (\mathcal{\mathcal{\mathcal{\mathcal{S}}}}
o)A}(M)$ and ~$B^{\mu...}_{\nu...}\in\Gamma^{m}_{(\mathcal{\mathcal{\mathcal{\mathcal{S}}}} o)}(M)$, it holds that
~$T^{\mu...}_{\nu...}\approx B^{\mu...}_{\nu...}$~ implies
~$D^{n}_{C(U)}T^{\mu...}_{\nu...}\approx D^{n}_{C(U)}B^{\mu...}_{\nu...}$~ for arbitrary natural number
$n$.

\end{Comm}

\begin{proof}
Pick an arbitrary ~$\Omega_{Ch}$ and an arbitrary fixed chart
~$Ch'(\Omega_{Ch})\in\mathcal{\mathcal{\mathcal{\mathcal{\mathcal{S}}}}}$. Such a chart maps all the 4-forms from
the domain of ~$\Gamma^{m}_{ (\mathcal{\mathcal{\mathcal{\mathcal{S}}}}o)A}(M)$ objects to smooth compact
supported functions (given by densities expressed in that chart). The objects
~$T^{\mu...}_{\nu...}((U^{\alpha}\omega)_{,\alpha})$ and
~$B^{\mu...}_{\nu...}((U^{\alpha}\omega)_{,\alpha})$  are taken as objects of the
Colombeau algebra (the connection, fixed in that chart,
is also an  object of the Colombeau algebra) and are equivalent to
~$U^{\alpha}(\omega)T^{\mu...}_{\nu...,\alpha}(\omega)$ and
~$U^{\alpha}(\omega)B^{\mu...}_{\nu...,\alpha}(\omega)$. Here the
derivative means the ''distributional derivative'' as used in the
Colombeau theory (fulfilling the Leibniz rule) and ~$U^{\alpha}(\omega)$
is simply a ~$D'^{m}_{n S (\mathcal{\mathcal{\mathcal{\mathcal{S}}}} o)}(M)$ object with the given vector
field appearing under the integral. But in the Colombeau
theory one knows that if some object is equivalent to a distributional
object, then their derivatives of arbitrary degree are also equivalent.
It also holds that if any arbitrary object is equivalent to a
distributional object, then they remain equivalent after being multiplied by arbitrary smooth
distribution. In the fixed chart we have (still in the Colombeau theory sense),

\begin{equation*}
T^{\mu...}_{\nu...}\approx B^{\mu...}_{\nu...}.
\end{equation*}

 But since their
$\partial$-derivatives were, in the fixed chart, obtained only by the distributional derivatives and multiplication by a smooth function, also their
$\partial$-derivatives must remain equivalent. The same holds about the second
covariant derivative term (containing connection). So the objects
from classical Colombeau theory, (classical theory just trivially
extended to what we call multi-index matrices of functions),
obtained by the chart mapping of the covariant derivatives of ~$T^{\mu...}_{\nu...}$~ and ~$B^{\mu...}_{\nu...}$,~ are equivalent in the
sense of the Colombeau theory. But the ~$\Omega_{Ch}$ set was
arbitrary and also the chart was an arbitrary chart from the
domain of ~$T^{\mu...}_{\nu...},~ B^{\mu...}_{\nu...}$.~ So ~$T^{\mu...}_{\nu...}$~ and ~$B^{\mu...}_{\nu...}$~ are equivalent with respect to our definition.
\end{proof}

We can try to extend this statement to a conjecture:

\theoremstyle{definition}
\newtheorem{ADD3}[cc1]{Conjecture}
\begin{ADD3}
Take ~$U^{i}$ being piecewise smooth tensor field, take connection
from the class ~$\Gamma^{3}_{S}(M)$,
~$T^{\mu...}_{\nu...}\in\Gamma^{m}_{
(\cup_{l}\mathcal{\mathcal{\mathcal{\mathcal{S}}}}_{l}~o)A}(M)$, and
~$B^{\mu...}_{\nu...}\in\Gamma^{m}_{(\cup_{l}\mathcal{\mathcal{\mathcal{\mathcal{S}}}}_{l}~o)}(M)$.
Then it holds that ~$T^{\mu...}_{\nu...}\approx
B^{\mu...}_{\nu...}$~ implies
~$D^{n}_{C(U)}T^{\mu...}_{\nu...}\approx
D^{n}_{C(U)}B^{\mu...}_{\nu...}$~ for arbitrary natural number $n$,
if such covariant derivative exists.
\end{ADD3}

This conjecture in fact means that if we have connection from the
class $\Gamma^{3}_{S}(M)$, then the covariant derivative is a map
from such element of the class ~$\tilde\Gamma^{m}_{
(\cup_{l}\mathcal{\mathcal{\mathcal{\mathcal{S}}}}_{l}~o)A}(M)$,
that it contains some linear element, to ~$\tilde\Gamma^{m}_{A}(M)$.
Note that we can also try to prove an extended version of the
conjecture, taking the same statement and just extending the classes
~$\Gamma^{m}_{
(\cup_{l}\mathcal{\mathcal{\mathcal{\mathcal{S}}}}_{l}~o)A}(M)$,
~$\Gamma^{m}_{
(\cup_{l}\mathcal{\mathcal{\mathcal{\mathcal{S}}}}_{l}~o)}(M)$~ to
the classes ~$\Gamma^{m}_{
(\cup_{l}\mathcal{\mathcal{\mathcal{\mathcal{S}}}}_{l})A}(M)$,
~$\Gamma^{m}_{
(\cup_{l}\mathcal{\mathcal{\mathcal{\mathcal{S}}}}_{l})}(M)$.

\subsubsection{Some additional theory}

\newtheorem{lcovder}[lcovder1]{Theorem}

\begin{lcovder}
Take ~$U^{i}$ to be vector field smooth in some
~$\mathcal{\mathcal{\mathcal{\mathcal{S}}}}\subset
\mathcal{D}_{n}$,~ with $\Gamma^{\mu}_{\nu\alpha}\in\Lambda$~ and
~$T^{\mu...}_{\nu...}\in S'_{n}\cap D'^{a}_{b}(M)$ ($n\geq 1$,
$S'_{n}$ is related to $\mathcal{D}_{n}$).~ Then $D_{C (U)}
T^{\mu...}_{\nu...}$ has an associated field which is on
~$M\setminus\tilde\Omega(Ch)$~ the classical covariant derivative of
~$A_{s}(T^{\mu...}_{\nu...})$. ($\tilde\Omega(Ch)$ is a set on which is
~$\Gamma^{\mu}_{\nu\alpha}(Ch)$~ continuous and is of 0 Lebesgue
measure.) It is defined on the whole $\mathcal{D}_{n+1}$
~($\mathcal{\mathcal{\mathcal{\mathcal{S}}}}\subset
\mathcal{D}_{n+1}$). That means association and covariant
differentiation in this case commute.\label{T244}
\end{lcovder}

\begin{proof}
Just take the definition of the classical covariant derivative, and
define the linear mapping given ~$\forall\Omega_{Ch}$~ and arbitrary
~$Ch_{k}(\Omega_{Ch})\in \mathcal{D}_{n}$~ as
\begin{equation}
(\omega, Ch_{k})\to\int_{\Omega_{Ch}}U^{\nu
}(Ch_{k})\left[A_{s}(T^{\mu...}_{\nu...})_{;\mu}\right](Ch_{k})~\omega.\label{AsC}
\end{equation}

This is an internally consistent definition, since
$\left[A_{s}(T^{\mu...}_{\nu...})_{;\mu}\right](Ch_{k})$ is defined everywhere
apart of a set having L measure 0. Now from our previous results
follows that everywhere outside ~$\tilde\Omega(Ch)$
\[\left[A_{s}(T^{\mu...}_{\nu...})_{;\mu}\right](Ch_{k})=T^{\mu...}_{\nu...
~;\mu}(Ch_{k}),\] and so the linear mapping (\ref{AsC}) is
equivalent to the object ~~$D_{C(U)}T^{\mu...}_{\nu...}$. ~Then~~\\ $U^{\nu
}A_{s}(T^{\mu...}_{\nu...})_{;\mu}=A_{s}(D_{C(U)}T^{\mu...}_{\nu...})$~everywhere
outside the set $\tilde\Omega(Ch)$.
\end{proof}

\newtheorem{meantime2}[lcovder1]{Theorem}

\begin{meantime2}
An extended analogy of theorem (\ref{T244}), can be proven, if we
use generalized concept of covariant derivative, without assuming
that vector field is smooth in some
$\mathcal{\mathcal{\mathcal{\mathcal{\mathcal{S}}}}}$, but only
$n+1$ continuously differentiable within $\mathcal{D}_{n+1}$.
\end{meantime2}

\begin{proof}
Exactly the same as before.
\end{proof}

This means that the aim to define a
concept of covariant derivative, ``lifted'' from the smooth manifold
and smooth tensor algebra to GTF in sense of association, has been achieved. It
completes the required connection with the old tensor
calculus.

\theoremstyle{definition}
\newtheorem{conj4}[cc1]{Conjecture}
\begin{conj4}
 Take ~$U^{i}$ smooth
in ~$\mathcal{\mathcal{\mathcal{\mathcal{\mathcal{S}}}}}$, ~~$T^{\mu...}_{\nu...}\in D'^{m}_{n (\mathcal{\mathcal{\mathcal{\mathcal{S}}}} o)}(M)$, ~such that ~$\exists ~A_{s}(T^{\mu...}_{\nu...})$,
~$\Gamma^{\mu}_{\nu\alpha}\in\Gamma^{3}_{S}(M)$. Then
~$\exists ~A_{s}(D_{C(U)}T^{\mu...}_{\nu...})$ and holds that \[A_{s}(D_{C(U)}T^{\mu...}_{\nu...})=U^{\alpha}A_{s}(T^{\mu...}_{\nu...})_{;\alpha}.\]
(Hence, similarly to $\otimes$, the covariant derivative
operator commutes with association for some significant number of objects.)
\end{conj4}

Note that for
every class ~$\Gamma^{m}_{At(\mathcal{\mathcal{\mathcal{\mathcal{S}}}})}(M)$ we can easily define the operator of Lie derivative along arbitrary in $\mathcal{\mathcal{\mathcal{\mathcal{\mathcal{S}}}}}$ smooth vector field $V$ (not along the generalized vector field, but even in the case of covariant derivative we did not prove
anything about larger classes of vector fields than smooth vector
fields). Lie derivative can be defined as
~$(L_{V}T,\omega)\equiv (T,L_{V}\omega)$. This is because the Lie derivative preserves $n$-forms and also preserves the
properties of such ~$C^{P}_{(\mathcal{\mathcal{\tilde S}})}(M)$ classes, for which it holds that ~$\mathcal{\mathcal{\mathcal{\mathcal{S}}}}\subset\mathcal{\mathcal{\tilde S}}$.

\subsection{Basic discussion of previous results and open questions}

We have constructed the algebra of GTFs, being able to incorporate
the concept of covariant derivative, (with the given conditions on
vector fields and connection), for a set of algebras constructed
from specific distributional objects. The use of these ideas in
physics is meaningful where the operations of tensor product and
covariant derivative give a map from appropriate subclass of
~$D'^{m}_{n}(M)$ class to the elements of ~$\tilde\Gamma^{m}(M)$
containing a ~$D'^{m}_{n}(M)$ element. This is always guaranteed to
work between appropriate subclasses of piecewise continuous
distributional objects, but a given physical equivalence might
specify a larger set of objects for which these operations provide
such mapping. Note that the whole problem lies in the multiplication
of distributions outside the ~$D'^{m}_{n E}(M)$ class. (For instance
it can be easily seen that square of ~$\delta(Ch_{k},q)$ as
introduced before is not equivalent to any distribution.) This is
because the product does not have to be necessarily equivalent to a
distributional object. Even worse, in case it is not equivalent to a
distributional object, the product is not necessarily a mapping
between equivalence classes of the given algebra elements ($\tilde
D'^{m}_{n A}(M)$). The same holds about contraction.

But even in such cases there can be a further hope. For example we
can abandon the requirement that certain quantities must be linear,
(for example the connection), and only some results of their
multiplication are really physical (meaning linear). Then it is a
question whether they should be constructed (constructed from the
linear objects as for example metric connection from the metric
tensor) through the exact equality or only through the equivalence.
If we take only the weaker (equivalence) condition, then there is a
vast number of objects we can choose, and many other important
questions can be posed. Even in the case that the mathematical
operations do depend on particular representatives of the
equivalence classes, there is no necessity to give up; in such
situation it might be an interesting question if there are any
specific ``paths'' which can be used to solve the physical
equivalence relations. The other point is that if these operations do depend
on the class members, then we can reverse this process. It means
that for example in the case of multiplication of two delta
functions we can find their nonlinear equivalents first and then
take their square, thus obtaining possibly an object belonging to an
equivalence class of a distribution.

As I mentioned in the introduction, these are not attempts to deal
with physical problems in a random, ad hoc way. Rather I want to
give the following interpretation to what is happening: The
differential equations in physics should be changed into
equivalence relations. For that reason they have plentiful solutions in the
given algebra\footnote{After one for example proves that covariant
derivative is a well defined operator on all GTF elements, then it
is the whole GTF algebra.}. (By a ``solution'' one here means any
object fulfilling the given relations.) One obtains much ``more''
solutions than in the case of classical partial differential
equations (but all the smooth distributions representing
``classical'' solutions of the ``classical'' initial value problem
are there), but what is under question is the possibility to
formulate the initial value problem for larger classes of objects
than $D'^{m}_{n E}(M)$ and $D'^{m}_{n (\mathcal{\mathcal{\mathcal{\mathcal{S}}}}o)}(M)$ (see the next
section). Moreover, if this is possible then there remains another
question about the physical meaning of those solutions. It means
that even in the case we get nonlinear objects as solutions of some
general initial value problem formulation, this does not have to be
necessarily something surprising; the case where physical laws are
solved also by physically meaningless solutions is nothing new. The
set of objects where we can typically search for physically
meaningful solutions is defined by most of the distributional
mappings (that is why classical calculus is so successful), but they
do not have to be necessarily the only ones.

\subsection{Some notes on the initial value problem within the partial
differential equivalence relations ($\approx$) on $D'^{m}_{n}(M)$}

In this section we will suggest how to complete the mathematical structure
developed and will get some idea how a physical problem can be
formulated in our language. It is again divided
into what we call ``basic ideas'' and ``some additional ideas''. The
first part is of a considerable importance, the second part is less
important, it just gives a suggestion how to recover the classical
geometric concept of geodesics in our theory.

\subsubsection{The basic ideas}

\paragraph{The approach giving the definition of the initial value problem}

Take a hypersurface (this can be obviously generalized to any
submanifold of lower dimension) ~$N\subset M$, which is such that it
gives in some subatlas
~$\mathcal{\mathcal{A}}_{N}\subset\mathcal{A}$~ a piecewise
smooth submanifold.  In the same time $\mathcal{\mathcal{A}}_{N}$
is such atlas that ~$\exists
\mathcal{\mathcal{\mathcal{\mathcal{S}}}}$~
~$\mathcal{\mathcal{\mathcal{\mathcal{S}}}}\subseteq
\mathcal{\mathcal{A}}_{N}$.

If we
consider space of 3-form fields living on $N$
(we give up on the
idea of relating them to 4-forms on $M$), we get two types of
important maps:

\begin{itemize}

\item
Take such ~$D'^{m}_{n E}(M)$ objects, that they have in every chart
from ~$\mathcal{\mathcal{A}}_{N}$ associated (tensor) fields
defined everywhere on $N$, apart from a set having 3 dimensional
Lebesgue measure equal to 0. Such objects can be, in every smooth
subatlas
~$\mathcal{\mathcal{\mathcal{\mathcal{S}}}}\subset\mathcal{\mathcal{A}}_{N}$,
mapped to the class ~$D'^{m}_{n E}(N)$ by embedding their associated
tensor fields\footnote{As previously noted, the given
~$T^{\mu}_{\nu}\in D'^{m}_{n E}(M)$ we can define in every chart by
$(Ch_{k},\omega)\to\int_{\Omega_{Ch}}
\left[A_{s}(T^{\mu}_{\nu})\right](Ch_{k})~\omega$.} into $N$. This
defines a tensor field $T^{\mu...}_{\nu...}$ living on a piecewise
smooth manifold $N$. Furthermore if
~$\mathcal{\mathcal{\mathcal{\mathcal{A}}}}_{3D}$~ is some
largest piecewise smooth atlas on $N$, the tensor field
~$T^{\mu...}_{\nu...}$~ defines on its domain a map
~$\forall\Omega_{Ch}\subset N$,~ $ Ch_{k}(\Omega_{Ch})\in
\mathcal{\mathcal{A}}_{3D}$ ~and ~$\omega\in C^{P}(\Omega_{Ch})$
\begin{equation}
(\omega,
Ch_{k})\to\int_{\Omega_{Ch}}T^{\mu...}_{\nu...}(Ch_{k})~\omega .
\end{equation}
 This is object from the class ~$D'^{m}_{n E}(N)$. What remains to be proven is that in
any smooth subatlas we map the same ~$D'^{m}_{n E}(M)$ object to
~$D'^{m}_{n}(N)$, otherwise this formulation is meaningless.

\item
The other case is a map ~$D'^{m}_{n (
\mathcal{\mathcal{\mathcal{\mathcal{S}}}} o)}(M)\to D'^{m}_{n ~(
\mathcal{\mathcal{\mathcal{\mathcal{S}}}} o)}(N)$ ($
\mathcal{\mathcal{\mathcal{\mathcal{S}}}}\subset\mathcal{\mathcal{A}}_{N}$),~
defined in a simple way: The objects from ~$D'^{m}_{n S (
\mathcal{\mathcal{\mathcal{\mathcal{S}}}} o)}(M)$ are mapped as
associated smooth tensor fields (in the previous sense). After this
step is taken, one maps the rest of distributional objects from
~$D'^{m}_{n ( \mathcal{\mathcal{\mathcal{\mathcal{S}}}} o)}(M)$
by using the fact that they are weak limits of smooth distributions
~$D'^{m}_{n S ( \mathcal{\mathcal{\mathcal{\mathcal{S}}}}
o)}(M)$. (This is coordinate independent for arbitrary tensor
distributions.) So in the case of objects outside the class
~$D'^{m}_{n S ( \mathcal{\mathcal{\mathcal{\mathcal{S}}}}
o)}(M)$ we embed the smooth distributions first, and take the limit
afterwards (exchanging the order of operations). The basic
conjecture is that if this limit exists on $M$, it will exist on $N$
(in the weak topology), by using the embedded smooth distributions.
\end{itemize}

Now we can say that initial value conditions of, (for example),
second-order partial~ differential ~equations ~are ~given ~by ~two
~distributional ~objects ~from ~~$D'^{m}_{n ( \mathcal{\mathcal{\mathcal{\mathcal{S}}}} o)}(N_{1})$,
~$D'^{m}_{n (\mathcal{\mathcal{\mathcal{\mathcal{S}}}} o)}(N_{2})$~ (~$D'^{m}_{n E}(N_{1})$, ~$D'^{m}_{n
E}(N_{2})$~)~ on two hypersurfaces ~$N_{1}, N_{2}$~ (not
intersecting each other). The solution is a distributional object
from the same class, which fulfills the ~$\approx$ equation and is
mapped (by the maps introduced in this section) to these two
distributional objects.

\paragraph{Useful conjecture related to our approach}

Note, that we can possibly (if the limit commutes) extend this
``initial value'' approach through the ~$D'^{m}_{n S}(M)$ class to
all the weak topology limits of the sequences formed by the objects from this class.
This means extension to the class of objects belonging to
~$D'^{m}_{n}(M)$, such that for any chart from $\mathcal{A}$ they have the full
~$C^{P}(M)$ domain and ~$D'^{m}_{n S}(M)$ is
dense in this class.

\subsubsection{Some additional ideas}

\paragraph{``Null geodesic solution'' conjecture}

Let us conjecture the following:

\theoremstyle{definition}
\newtheorem{conj7}{Conjecture}[section]
\begin{conj7}\label{conj7}
Pick some atlas $\mathcal{\mathcal{\mathcal{\mathcal{\mathcal{S}}}}}$. Pick some ~$g_{\mu\nu}\in D^{0}_{2 S (\mathcal{\mathcal{\mathcal{\mathcal{S}}}}
o)}(M)$, such that it has  $A_{s}(g_{\mu\nu})$, being a
Lorentzian signature metric tensor field. Take some $\Omega_{Ch}$
and two spacelike hypersurfaces $H_{1},~H_{2}$, ~$H_{1}\cap
H_{2}=\{0\}$,
~$H_{1}\cap\Omega_{Ch}\neq\{0\}$,~$H_{2}\cap\Omega_{Ch}\neq\{0\}$. ~Furthermore ~$H_{1}, H_{2}$~ are such that there exist two points ~$q_{1}\in
H_{1}\cap\Omega_{Ch}$, ~$q_{2}\in H_{2}\cap\Omega_{Ch}$ separated by
a null curve geodesics (relatively to ~$A_{s}(g_{\mu\nu})$), and the geodesics lies within
~$\Omega_{Ch}$. Construct such chart ~$Ch_{k}(\Omega_{Ch})\in \mathcal{\mathcal{\mathcal{\mathcal{\mathcal{S}}}}}$,
that both of the hypersurfaces are hypersurfaces
(they are smooth manifolds relatively to $\mathcal{\mathcal{\mathcal{\mathcal{\mathcal{S}}}}}$) given by $u=const.$
condition ($u$ is one of the coordinates) and the given geodesics is representing $u$-coordinate curve.

Take classical free field equation with equivalence:
~$\Box_{g}\Phi\approx 0$~ (with $g^{\mu\nu}$ as previously defined).
Then look for the distributional solution of this equation with the
initial value conditions being ~$\delta\big(Ch_{1k}(\Omega_{Ch}\cap
H_{1}),q_{1}\big)\in D'^{m}_{n (\mathcal{\mathcal{\mathcal{\mathcal{S}}}} o)}(H_{1})$~ on the first
hypersurface and ~$\delta\big(Ch_{2k}(\Omega_{Ch}\cap H_{2}),q_{2}\big)\in
D'^{m}_{n(\mathcal{\mathcal{\mathcal{\mathcal{S}}}} o)}(H_{2}) $~ on the second hypersurface. (Here
~$Ch_{1k}(\Omega_{Ch}\cap H_{1}),~Ch_{2k}(\Omega_{Ch}\cap H_{2})$
are coordinate charts, which are the same on the intersection of the
given hypersurface and ~$\Omega_{Ch}$ as the original coordinates
without $u$.) Then the solution of this initial value problem is a
mapping ~$\Phi$, which is such, that it can be in chart
~$Ch_{k}(\Omega_{Ch})$~ expressed as
\[\omega~\to~\int
du\int\prod_{i}dx^{i}\left(\delta(x^{i}(q_{1}))\omega'(Ch_{k})(u,x^{j})\right)\]
(where
$x^{i}(q_{1})$ is image of $q_{1}$ in chart mapping
~$Ch_{k}(\Omega_{Ch})$ ).
\end{conj7}

We can formulate similar conjectures for timelike and spacelike
geodesics, we just have to:
\begin{itemize}
\item
instead of point separation by null curve, consider the
separation by timelike or spacelike curve,
\item
instead of the ``massless'' equation we have to solve
the ~$(\Box_{g}\pm m^{2})\Phi\approx 0$~ equation ($m$ being
arbitrary nonzero real number). Here $\pm$ depends on the
signature we use and on whether we look for timelike or spacelike
geodesics.
\end{itemize}

The rest of the conditions are unchanged (see \ref{conj7}). Some
insight to our conjectures can be brought by calculating the
massless case for flat Minkowski space, using modified cartesian
coordinates ($u=x-ct,x,y,z$). We get the expected results.\newline

\section{What previous results can be recovered, and how?}

As was already mentioned, our approach is in some sense a generalization of
Colombeau approach from \cite{multi}, which is equivalent to
canonical $\mathbb{R}^{n}$ approach. So for $\Omega_{Ch}$, after
we pick some ~ $Ch'(\Omega_{Ch})\in\mathcal{\mathcal{\mathcal{\mathcal{\mathcal{S}}}}}$, ~(which determines the classes
$A^{n}(M)$ related to this chart), and by considering
only the objects ~$C^{P}_{S (\mathcal{\mathcal{\tilde S}})}(\Omega_{Ch})$ ~($\mathcal{\mathcal{\mathcal{\mathcal{S}}}}\subset\mathcal{\mathcal{\tilde S}}$),~ (hence considering the ~$D'^{m}_{n (\mathcal{\mathcal{\mathcal{\mathcal{S}}}} o)}(M)$ class only), we
obtain from our construction the mathematical language used in
\cite{multi}. But all the basic equivalence relations from Colombeau
approach have been generalized first to the class ~$D'^{m}_{n  (\mathcal{\mathcal{\mathcal{\mathcal{S}}}} o)A}(M)$ and also to appropriate subclasses of the ~$D'^{m}_{n E A}(M)$
class.

\subsection{Generalization of some particular statements}

Now there are certain statements in $\mathbb{R}^{n}$, where one has
to check whether they are just a result of this specific
reduction, or not. A good example is a statement
\begin{equation}
H^{n}~\delta\approx\frac{1}{n+1}~\delta~.
\end{equation}
($H$ is Heaviside distribution.) What we have to do is to interpret
the symbols inside this equation geometrically. This is a
$\mathbb{R}^{1}$ relation. $H$ is understood as a ~$D'_{(\mathcal{\mathcal{\mathcal{\mathcal{S}}}} o)}(M)$
element and defined on the manifold (one dimensional, so the
geometry would be quite trivial) by integral given by a function (on
$M$) obtained by the inverse coordinate mapping substituted to $H$. Now take some fixed chart ~$Ch_{k}(\mathbb{R}^{1})$.
~The derivative is a covariant derivative along the smooth vector
field ~$U$,~ which is constant and unit in the fixed coordinates ~$Ch_{k}(\mathbb{R}^{1})$.
Then $\delta$ can be reinterpreted as ~$\delta(Ch_{k},q)$, where $q$ is
the 0 point in the chart ~$Ch_{k}(\mathbb{R}^{1})$.~ Then the relation can be generalized,
since it is obvious that (see the covariant derivative section)
~$D_{C(U)}H=\delta(q,Ch_{k})$ and so
\begin{eqnarray}
D_{C(U)}H=D_{C(U)}L(H_{f})~~~~~~~~~~~~~~~~~~~~~~~~~~~~~~~~~~~~~~~~~~~~~~~~~~~~~~~~~~~~~~~~~~~~~~~~~~~~~~~~~~~~~~~~~~\nonumber\\
=D_{C(U)}L(H_{f}^{n+1})\approx D_{C(U)}(H^{n+1})=(n+1)~H^{n}~D_{C(U)}H.~~~~
\end{eqnarray}
By $L$ we mean here a regular distribution defined by the function in the
brackets, hence an object
from $D'_{E}(M)$. ~(To be precise and to avoid confusion in the notation, we
used for the Heaviside function the symbol $H_{f}$, while for the Heaviside distribution the usual symbol $H$.) This is nice, but rather trivial
illustration.

This can be generalized to more nontrivial cases.  Take the flat
$\mathbb{R}^{n}$ topological manifold. Fix such chart ~$Ch_{k}(\mathbb{R}^{n})$~ covering the
whole manifold, that we can express Heaviside distribution in this chart through
~$H_{f}(x_{1})\in D'_{(\mathcal{\mathcal{\mathcal{\mathcal{S}}}} o)}(M)$.~ This means the hypersurface where
~$H_{f}$ is ~discontinuous is given in ~$Ch_{k}(\mathbb{R}^{n})$~ as $x_{1}=0$. Now the derivative will be a covariant derivative taken along a smooth vector
field being perpendicular (relatively to the flat space metric\footnote{Note
that since it is flat space it makes sense to speak about a
perpendicular vector field, since we can uniquely transport vectors
to the hypersurface.}) to the hypersurface on which is $H_{f}$ discontinuous. We easily see that
the covariant derivative of $H$ along such vector field gives a distribution (call it
~$\tilde\delta\in D'_{(\mathcal{\mathcal{\mathcal{\mathcal{S}}}} o)}(M)$), which is in the chart $Ch_{k}(\mathbb{R}^{n})$ expressed
as
\begin{equation}
\tilde\delta(\phi)=\delta_{x_{1}}\left(\int\phi(x_{1},x_{2}\dots x_{n})
~dx_{2}\dots dx_{n}\right)~.
\end{equation}
This distribution reminds us in some sense the
``geodesic'' distribution from the previous part. Then the
following holds:
\begin{equation}
H^{n}~\tilde\delta\approx\frac{1}{n+1}~\tilde\delta~.
\end{equation}
This generalized form of our previous statement can be used for computations with Heaviside functional metrics (computation
of connection in fixed coordinates).

\subsection{Relation to practical computations}

These considerations (for example) imply that the result from
canonical Co\-lom\-beau $\mathbb{R}^{n}$ theory derived by
\cite{multi} can be derived in our formalism as well. This is also
true for the geodesic computation in curved space geometry from
\cite{geodesic}. The results derived in special Colombeau algebras
(in geometrically nontrivial cases) are more complicated, since in such cases the strongest, ~$A^{\infty}(\mathbb{R}^{n})$,~ version
of the theory is used. This version is not contained in our chart representations.
(This is because we are using only ~$A^{n}(\mathbb{R}^{n})$ with finite $n$.)
It is clear that all the equivalence relations from our theory must hold in
such stronger formulation (since obviously ~$A^{\infty}(M)\subset
A^{n}(M)~~\forall n\in\mathbb{N}$) and the uniqueness of
distribution solution must hold as well. This means that at this
stage there seems to be no obstacle to reformulate our theory by
using ~$A^{\infty}(Ch_{k},q,\Omega_{Ch})$ classes (and taking
elements from ~$D'_{(\mathcal{\mathcal{\tilde S}} o)}(M)$ at least), if necessary. But
it is unclear whether one can transfer all the calculations using
~$A^{\infty}(Ch_{k},q,\Omega_{Ch})$ classes to our weaker
formulation.

The strong formulation was used also in the Schwarzschild case
\cite{Schwarzschild}, but there is a problem. The fact that the authors of \cite{Schwarzschild}
regularize various functions piece by piece does not have to be necessarily a
problem in Colombeau theory\footnote{Although the authors use (in the
first part, not being necessarily connected with the results) quite
problematic embeddings.}. But, as already mentioned, the problem lies in the use of formula for
~$R^{~~\nu}_{\mu}$, originally derived within smooth tensor field algebra. If we want to derive in the Schwarzschild case Ricci tensor
straight from its definition, we cannot avoid multiplications of
delta function by a non-smooth function. This is in Colombeau
theory deeply non-trivial.

In the cases of Kerr's geometry and conical spacetimes theory this problem appears as well. As a
consequence of this fact, calculations are mollifier dependent, not being (in the strict sense)
results of our theory anymore. On the other hand there is
no reasonable mathematical theory in which these calculations make sense.
This means that a better understanding of these results will be
necessary. By better understanding of these results provided by our theory we mean
their derivation by a net of equivalence relations, by taking
some intermediate quantities to be nonlinear. So the results should follow from
the principle that the equivalence relations are the
fundamental part of all the mathematical formulation of physics.

\section{Conclusions}

The main objectives of this work were to build foundations of a
mathematical language reproducing the old language of
smooth tensor calculus and extending it at the same time. The reasons for these objectives were given at the
beginning of this chapter. This work is a first step to such theory,
but it already achieves its  basic goals. That means we consider these results as useful independently of how successful future work on the topic will turn out to be. On the other hand, the territory it opens for further exploration is in my opinion large and significant. It offers a large area of possibilities for future work.

Just to summarize: the result of our work is a theory based purely
on equivalence relations instead of equalities, using a well defined concept
of generalized tensor field and the covariant derivative operator. This operator
is well defined at least on the proper subclass of generalized
tensor fields. We also defined (using some conjectures) the initial value problem for partial differential
equivalence relations. Our theory naturally relates to many results beyond
the classical smooth tensor calculus, already derived.

~
\pagestyle{plain}
\chapter*{Conclusions}\label{C:con}
\addcontentsline{toc}{chapter}{Conclusions}

As we mentioned in the introduction, this thesis is based on three separate research projects. 

The first project was related to pseudo-Finsler extensions of the general theory of relativity, and represented an attempt to find a natural geometric framework for possible high energy Lorentz violations. The reason why one was interested in such construction is the question of whether it is possible to find a weaker interpretation of Einstein's equivalence principle consistent with Lorentz symmetry violations. The result obtained was, when mathematical simplicity was taken as the guidance principle, unfortunately a ``no-go'' theorem, at least for significant number of cases. The cases particularly affected were the bi-metric theories, but the analogue model based on bi-refringent crystal optics indicated that the problem might affect much larger class of theories. The problem lies in the fact that, unlike what one would naturally expect, introducing Lorentzian signature puts very tight constraints on Finsler geometry, at least if one wants to keep some of the basic geometric concepts well defined and meaningful. This ``no-go'' result we consider to be disappointing, (which is the case of most ``no-go'' results), but certainly very useful. 

The second project focused on the highly damped quasi-normal modes of different black hole spacetimes. The method of approximation by analytically solvable potentials was used to estimate the highly damped modes for the Sch\-warz\-schild and the Schwarzschild-de Sitter (S-dS) black holes. The first served more as a consistency check (since the asymptotic formula for the highly damped QNMs of the Sch\-warz\-schild black hole is well known). But for the Schwarzschild-de Sitter black hole a lot of new information was extracted from those models, especially the link between rational ratios of horizon surface gravities to the periodic behaviour of the QNMs. Also, when periodic, in general the highly damped modes do not form only one equi-spaced family as in the case of Sch\-warz\-schild black hole, but split into multiple families. Strikingly, as we discovered, the same patterns can be observed in a complementary set of analytic estimations for the highly damped modes, those approximations being obtained by monodromy techniques. This holds for all types of black hole spacetime so far analysed by those techniques. That means we were able to significantly generalize our theorems about the highly damped mode behaviour to all the presently known analytic results. Our results might be interesting also from the viewpoint of the black hole thermodynamics, as the asymptotic QNM behaviour is suspected to be linked to the black hole area spectrum \cite{Hod, Maggiore}.

The third project dealt with the problem of multiplication of tensorial distributions. Despite the fact that lot has been done in the field in the past \cite{II, Vickers}, the full generalization of the covariant derivative operator was, for example, not yet achieved. On the other hand practical results confirm a need for such generalization (see \cite{Vickers}). We built an alternative construction, which fully operates with the Colombeau equivalence relation, but technically avoids Colombeau algebra construction. It generalizes the concept of covariant derivative to tensorial distributions and operates on the much more general, but for the language of distributions natural, piecewise smooth manifolds. We are convinced that such language might offer conceptual extension of general relativity and might have possibly interesting consequences for quantum gravity as well.  

\pagestyle{headings}
\appendix

\chapter{Bi-refringent crystals}

%---------------------------------------------------

%---------------------------------------------------------------------------------------------------------
%---------------------------------------------------------------------------------------------------------
\label{A:BW}

\section{Basic characteristics of the crystal media}
%---------------------------------------------------------------------------------------------------------
The basic optics reference we shall use is Born and  Wolf,
\emph{Principles of Optics} \cite{BW}. In particular we shall focus
on Chapter XV, ``Optics of crystals'', pages 790--818. See
especially pages 796--798 and pages 808--811. Specific page,
chapter, and section references below are to  the 7th (expanded)
edition, 1999/2003.

The theory of bi-refringent crystal optics is formulated in the preferred inertial system, the inertial system of the crystal. The optical medium of the crystal is characterized by permeability and permittivity.
Permeability $\mu$ is taken to be a scalar, permittivity
$\epsilon_{ij}$ a ``spatial'', (relative to the inertial system of the crystal), $3\times3$ tensor. (This is an \emph{excellent}
approximation for all known optically active media.) By going to the
principal axes we can, without loss of generality, take
$\epsilon_{ij}$ to be diagonal
\begin{equation}
\epsilon_{ij} = \left[ \begin{array}{ccc} \ep_x & 0 & 0 \\ 0 & \ep_y
& 0 \\ 0 & 0& \ep_z \end{array} \right].
\end{equation}
This fixes the relativistic inertial coordinate system in a unique way.

We furthermore define ``principal velocities''
\begin{equation}
v_x = {c\over\sqrt{\mu \ep_x}}; \qquad v_y = {c\over\sqrt{\mu
\ep_y}}; \qquad v_z = {c\over\sqrt{\mu \ep_z}}.
\end{equation}
Note (this is a tricky point that has the potential to cause
confusion) that $v_x$ is \emph{not} the velocity of light in the $x$
direction --- since $\ep_x$ (and so $v_x$) is related to the
properties of the electric field in the $x$ direction, the principal
velocity $v_x$  is instead the velocity of a light wave whose
electric field is pointing in the $x$ direction. That is, for light
waves propagating in the $y$-$z$ plane, \emph{one} of the
polarizations will propagate with speed $v_x$.

%---------------------------------------------------------------------------------------------------------
\section{Group velocity and ray equation}
%---------------------------------------------------------------------------------------------------------

The group velocity, $v_g$, in the framework used by Born and Wolf,
is identical to the  ``ray velocity'', and is controlled by the
so-called ``ray equation". See   (15.2.29), page 797. To set some
conventions, $\hat\n$ will always denote a unit vector in physical
space --- a unit with respect to the usual Euclidean norm, while
$\n$ is a generic position in physical 3-space. In contrast, $\hat
\k$ will be reserved for a unit wave-vector in the dual
``wave-vector space''.

Born and Wolf exhibit the ray equation in a form equivalent (Born
and Wolf use $\t$ where we use $\n$) to:
\begin{equation}
{\hat n_x^2\over 1/v_g^2 - 1/v_x^2} +  {\hat n_y^2\over 1/v_g^2 -
1/v_y^2} +  {\hat n_z^2\over1/ v_g^2 - 1/v_z^2} = 0.
\end{equation}
Here the group velocity (ray velocity) is defined by looking at the
energy flux and
\begin{equation}
\v_g = v_g \; \hat \n.
\end{equation}
We can rewrite this as
\begin{equation}
{\hat n_x^2 v_x^2\over v_g^2 - v_x^2} +  {\hat n_y^2 v_y^2 \over
v_g^2 - v_y^2} +  {\hat n_z^2 v_z^2\over v_g^2 - v_z^2} = 0.
\end{equation}
This form of the ray equation encounters awkward ``division by
zero'' problems when one looks along the principal axes, so it is
advisable to eliminate the denominators by multiplying through by
the common factor $ (v_g^2 - v_x^2) ( v_g^2 - v_y^2) ( v_g^2 -
v_z^2)$, thereby obtaining:
\begin{eqnarray}
&& \hat n_x^2 v_x^2 ( v_g^2 - v_y^2) ( v_g^2 - v_z^2) + \hat n_y^2
v_y^2 ( v_g^2 - v_z^2) ( v_g^2 - v_x^2) \qquad
\nonumber\\
&& \quad + \hat n_z^2 v_z^2 ( v_g^2 - v_x^2) ( v_g^2 - v_y^2) = 0.
\qquad
\end{eqnarray}
It is this form of the ray equation that, (because it is much better
behaved), we shall use as our starting point. Now this is clearly a
quartic in $v_g$, and by regrouping it we can write
\begin{eqnarray}
&&v_g^4 \left[  \hat n_x^2 v_x^2  +  \hat n_y^2 v_y^2  + \hat n_x^2
v_z^2  \right]
\nonumber\\
&& - v_g^2 \left[ \hat n_x^2 v_x^2 (v_y^2+v_z^2) +\hat n_y^2 v_y^2
(v_z^2+v_x^2) +\hat n_z^2 v_z^2 (v_x^2+v_y^2) \right] \nonumber
\\
 && \qquad%\qquad\qquad
+ \left[ v_x^2 v_y^2v_z^2 \right] = 0.
\end{eqnarray}
Equivalently
\begin{eqnarray}
v_g^4 \left[ \hat n_x^2 v_y^{-2} v_z^{-2}  + \hat n_y^2 v_z^{-2}
v_x^{-2}  + \hat n_z^2 v_x^{-2} v_y^{-2}  \right]~~~~~~~~~~~~~~~~~~~~~~~~~~~~~~~~~~~~~~~~~~~~~~~~~~~~~~~~~ \nonumber
\\
- ~v_g^2 \big[ \hat n_x^2  (v_y^{-2}+v_z^{-2}) +\hat n_y^2
(v_z^{-2}+v_x^{-2})+\hat n_z^2  (v_x^{-2}+v_y^{-2}) \big]~~~~~~~~~~~~~~~~~~~~~\nonumber\\
 +1 = 0.~~~~~~~~~~~~~~~
\end{eqnarray}
Now define two quadratics (in terms of the three direction cosines
$\hat n_i$)
\begin{equation}
\label{E:bar-q0} \bar q_0(\hat\n,\hat\n) = \left[ \hat n_x^2
v_y^{-2} v_z^{-2}  + \hat n_y^2 v_z^{-2} v_x^{-2}  + \hat n_z^2
v_x^{-2} v_y^{-2}  \right] ,
\end{equation}
\begin{eqnarray}
\label{E:bar-q2} \bar q_2(\hat\n,\hat\n) &=& {1\over2}  \big[ \hat
n_x^2  (v_y^{-2}+v_z^{-2}) +\hat n_y^2 (v_z^{-2}+v_x^{-2}) +\hat
n_z^2  (v_x^{-2}+v_y^{-2}) \big] ,
\end{eqnarray}
then
\begin{equation}
v_g^2(\hat\n) = {\bar q_2(\hat\n,\hat\n) \pm \sqrt{ \bar
q_2(\hat\n,\hat\n)^2 - \bar q_0(\hat\n,\hat\n)}\over \bar
q_0(\hat\n,\hat\n)}.
\end{equation}
But, since $\hat\n$ is a unit vector, we could equally well rewrite
this as
\begin{equation}
v_g^2(\hat\n) = {\bar q_2(\hat\n,\hat\n) \pm \sqrt{ \bar
q_2(\hat\n,\hat\n)^2 - \bar q_0(\hat\n,\hat\n) \; (\hat\n \cdot
\hat\n) }\over \bar q_0(\hat\n,\hat\n)}.
\end{equation}
In this form both numerator and denominator are manifestly
homogeneous and quadratic in the components of $\hat\n$, so for any
3-vector $\n$ (now \emph{not} necessarily of unit norm) we can take
the further step of writing
\begin{equation}
v_g^2(\n) = {\bar q_2(\n,\n) \pm \sqrt{ \bar q_2(\n,\n)^2 - \bar
q_0(\n,\n)\;(\n \cdot \n) }\over \bar q_0(\n,\n)}.
\end{equation}
The function $v_g(\n)$ so defined is homogeneous of degree zero in
the components of $\n$:
\begin{equation}
v_g( \kappa\, \n ) = v_g(\n) = v_g(\hat\n).
\end{equation}
The homogeneous degree zero property should remind one of the
relevant feature exhibited by the Finsler metric. It is also useful
to note that
\begin{equation}
{1\over v_g(\n)^2} =  {\bar q_2(\n,\n) \mp \sqrt{ \bar q_2(\n,\n)^2
- \bar q_0(\n,\n)\;(\n\cdot \n) }\over (\n\cdot\n)}.
\end{equation}

%---------------------------------------------------------------------------------------------------------
\section{Phase velocity and Fresnel equation}
%---------------------------------------------------------------------------------------------------------

In contrast, the phase velocity, in the framework used by Born and
Wolf, is controlled by the so-called ``equation of wave normals",
also known as the ``Fresnel equation''. See equation (15.2.24), page
796. The relevant computations are similar to, but not quite
identical to,  those for the group velocity.

Let us consider a plane wave $\exp(i[\k \cdot \x - \omega t])$ and
define the phase velocity by
\begin{equation}
\v_p = v_p \; \hat \k = {\omega\over k} \; \hat \k,
\end{equation}
then the Fresnel equation is equivalent (Born and Wolf use $\s$
where we use $\hat\k$)  to
\begin{equation}
{\hat k_x^2\over v_p^2 - v_x^2} +  {\hat k_y^2\over v_p^2 - v_y^2} +
{\hat k_z^2\over v_p^2 - v_z^2} = 0.
\end{equation}
This form of the equation exhibits ``division by zero'' issues if
you try to look along the principal axes, so it is for many purposes
better to multiply through by the common factor $ ( v_p^2 - v_x^2) (
v_p^2 - v_y^2) ( v_p^2 - v_z^2)$ thereby obtaining the equivalent of
their equation (15.3.1) on page 806:
\begin{eqnarray}
\hat k_x^2 ( v_p^2 - v_y^2) ( v_p^2 - v_z^2) + \hat k_y^2 ( v_p^2
- v_z^2) ( v_p^2 - v_x^2)~~~~~~~~~~~~~~~~~~~
\nonumber\\
+ ~\hat k_z^2 ( v_p^2 - v_x^2) ( v_p^2 - v_y^2) = 0.
\end{eqnarray}
This is clearly a quartic in $v_p$ and by regrouping it, and using
$\hat\k \cdot \hat\k=1$, we can write
\begin{eqnarray}
v_p^4 - v_p^2 \left[ \hat k_x^2 (v_y^2+v_z^2) +\hat k_y^2
(v_z^2+v_x^2) +\hat k_z^2 (v_x^2+v_y^2) \right]~~~~~~~~~~~~~~~~~~
\nonumber\\
+ \left[ \hat k_x^2 v_y^2v_z^2 +\hat k_y^2 v_z^2v_x^2 +\hat k_z^2
v_x^2v_y^2 \right] = 0.
\end{eqnarray}
Let us now define two quadratics (in terms of the direction cosines
$\hat k_i$)
\begin{equation}
\label{E:q2} q_2(\hat\k,\hat\k) = {1\over2} \left[ \hat k_x^2
(v_y^2+v_z^2) +\hat k_y^2 (v_z^2+v_x^2) +\hat k_z^2 (v_x^2+v_y^2)
\right],
\end{equation}
and
\begin{equation}
\label{E:q0} q_0(\hat\k,\hat\k) =  \left[ \hat k_x^2 v_y^2v_z^2
+\hat k_y^2 v_z^2v_x^2 +\hat k_z^2 v_x^2v_y^2 \right],
\end{equation}
so as a function of direction the phase velocity is
\begin{equation}
v_p^2(\hat\k) = q_2(\hat\k,\hat\k) \pm\sqrt{ q_2(\hat\k,\hat\k)^2 -
q_0(\hat\k,\hat\k) }.
\end{equation}
This is very similar to the equations obtained for the ray velocity.
In fact, we can naturally extend this formula to arbitrary
wave-vector $\k$ by writing
\begin{equation}
v_p^2(\k) = {q_2(\k,\k) \pm\sqrt{ q_2(\k,\k)^2 - q_0(\k,\k) \;
(\k\cdot \k) }\over (\k\cdot\k)}.
\end{equation}
This expression is now homogeneous of order zero in $\k$, so that
 \begin{equation}
v_p(\kappa  \, \k) = v_p(\k) = v_p(\hat \k).
\end{equation}
Again, we begin to see a hint of Finsler structure emerging.

\section{Connecting the ray and the wave vectors}
%----------------------------------------------------------------------------------------------------------

Connecting the ray-vector $\hat \n$ and the wave-vector $\hat \k$ in
birefringent optics is rather tricky --- for instance, Born and Wolf
provide a rather turgid discussion on page 798 --- see section
15.2.2, equations (34)--(39).  The key result is
\begin{equation}
{v_g(\n)\; \hat n_i\over v_g(\n)^2-v_i^2} = {v_p(\k) \; \hat
k_i\over v_p(\k)^2-v_i^2},
\end{equation}
which ultimately can be manipulated to calculate $\hat \n$ as a
rather complicated ``explicit'' function of $\hat \k$ --- albeit an
expression that is so complicated that even Born and Wolf do not
explicitly write it down. Unfortunately if it comes to pseudo-Finsler geometry, all the extra technical
machinery provided by Finsler notions of norm and distance do not
serve to simplify the situation. (The fact that phase and group
velocities can be used to define quite distinct, and in some
situations completely unrelated,  effective metrics has also been
noted in the context of acoustics \cite{trieste, bled}.)

\section{Optical axes}
%--------------------------------------------------------------------------------------------------------
\label{A:axes}
%--------------------------------------------------------------------------------------------------------
To find the ray optical axes we (without loss of generality) take
$v_z>v_y>v_z$, and define quantities $\bar \Delta_\pm$ (this is of
course the result of considerable hindsight) by:
\begin{equation}
v_x^2 = {v_y^2 \over 1- v_y^2  \; \bar \Delta_+^2}; \qquad v_z^2 =
{v_y^2 \over 1 + v_y^2 \; \bar \Delta_-^2}.
\end{equation}
Furthermore eliminate $\hat n_y$ by using
\begin{equation}
\hat n_y^2 = 1 - \hat  n_x^2 -\hat  n_z^2,
\end{equation}
then
\begin{eqnarray}
 \bar D = \bar q_2^2 - \bar q_0~~~~~~~~~~~~~~~~~~~~~~~~~~~~~~~~~~~~~~~~~~~~~~~~~~~~~~~~~~~~~~~~~~~~~~~~~~~~~~~~~~~~~~~~~~~~~~~~~~~~~~~~~~
\nonumber\\
=
 {1\over4}
\left[ (\hat n_x \bar\Delta_+ +\hat n_z \bar\Delta_-)^2 -
(\bar\Delta_+^2 + \bar\Delta_-^2 ) \right]~~~~~~~~~~~~~~~~~~~~~~~~~~~~~~~~~~~~~~~~~~~~~~~~
 \nonumber\\
~~~~~~~~~~~~~~~~~~~~~~~~~~~\times
\left[ (\hat n_x \bar\Delta_+ -\hat n_z \bar\Delta_-)^2 -
(\bar\Delta_+^2 + \bar\Delta_-^2 ) \right]~.~~~~~~~~~~~~~~~~~~~~~~~~~~~~~~~~~~~~~~~~~
\end{eqnarray}
Thus the (ray) optical axes are defined by
\begin{equation}
(\hat n_x \bar\Delta_+ \pm \hat n_z \;\bar\Delta_-)^2 =
(\bar\Delta_+^2 + \bar\Delta_-^2 ).
\end{equation}
But thanks to the Cauchy--Schwartz inequality
\begin{equation}
(\hat n_x \bar\Delta_+ \pm \hat n_z \bar\Delta_-)^2 \leq (\hat n_x^2
+ \hat n_z^2)  (\bar\Delta_+^2 + \bar\Delta_-^2 ) \leq
(\bar\Delta_+^2 + \bar\Delta_-^2 ).
\end{equation}
Therefore on the (ray) optical axis we must have $\hat n_y=0$, and
$(\hat n_x^2 + \hat n_z^2)=1$. So (up to irrelevant overall signs)
\begin{equation}
\bar \e_{1,2} = \left( \pm { \bar\Delta_+ \over\sqrt{\bar\Delta_+^2
+ \bar\Delta_-^2 }}; \;\; 0 \;\; ;
 { \bar\Delta_- \over\sqrt{\bar\Delta_+^2 + \bar\Delta_-^2 }}  \right),
\end{equation}
which we can recast in terms of the principal velocities as
\begin{equation}
\bar \e_{1,2} = \left( \pm {
\sqrt{1/v_y^2-1/v_x^2\over1/v_z^2-1/v_x^2}}; \;\; 0 \;\; ;
 { \sqrt{1/v_z^2-1/v_y^2\over1/v_z^2-1/v_x^2}}  \right),
\end{equation}
or
\begin{equation}
\bar \e_{1,2} = \left( \pm { {v_z\over v_y} \sqrt{v_x^2-v_y^2\over
v_x^2-v_z^2}}; \;\; 0 \;\; ;
 { {v_x\over v_y} \sqrt{v_y^2-v_z^2\over v_x^2-v_z^2}}  \right).
\end{equation}
These are the two ray optical axes. (Compare with equation (15.3.21)
on p.~811 of Born and Wolf.)

A similar computation can be carried through for the phase optical
axes. We again take $v_z>v_y>v_z$, and now define
\begin{equation}
v_x^2 = v_y^2 + \Delta_+^2; \qquad  v_z^2 = v_y^2 - \Delta_-^2.
\end{equation}
Eliminate $\hat k_y$ by using
\begin{equation}
\hat k_y^2 = 1 - \hat  k_x^2 -\hat  k_z^2.
\end{equation}
Then
\begin{eqnarray}
D &=& q_2^2-q_4
\\
&=& {1\over4} \left[ (\hat k_x \Delta_+ + \hat k_z \Delta_-)^2 -
(\Delta_+^2+\Delta_-^2) \right] \nonumber
\\
&& \times \left[ (\hat k_x \Delta_+ - \hat k_z \Delta_-)^2 -
(\Delta_+^2+\Delta_-^2) \right]\!\!.
\end{eqnarray}
This tells us that the discriminant factorizes, \emph{always}. The
discriminant vanishes if
\begin{equation}
(\hat k_x \Delta_+ \pm \hat k_z \Delta_-)^2 = \Delta_+^2+\Delta_-^2.
\end{equation}
But by the Cauchy--Schwartz inequality
\begin{equation}
(\hat k_x \Delta_+ \pm \hat k_z \Delta_-)^2 \leq (\hat k_x^2 + \hat
k_z^2) ( \Delta_+^2+\Delta_-^2) \leq  ( \Delta_+^2+\Delta_-^2).
\end{equation}
Thus on the (phase) optical axis we must have $\hat k_y=0$ and $
(\hat k_x^2 + \hat k_z^2) = 1$. The two unique directions (up to
irrelevant overall sign flips) that make the discriminant vanish are
thus
\begin{equation}
\e_{1,2} = \left( \pm {\Delta_+\over\sqrt{ \Delta_+^2+\Delta_-^2}} ;
\;\; 0 \;\; ; {\Delta_-\over\sqrt{ \Delta_+^2+\Delta_-^2}}\right),
\end{equation}
which can be rewritten as
\begin{equation}
\e_{1,2} = \left( \pm { \sqrt{v_x^2-v_y^2\over v_x^2-v_z^2}}; \;\; 0
\;\; ;
 { \sqrt{v_y^2-v_z^2\over v_x^2-v_z^2}}  \right).
\end{equation}
These are the two phase optical axes. (Compare with equation
(15.3.11) on p.~810 of Born and Wolf.)

%----------------------------------------------------------------------------------------------------------

%---------------------------------------------------------------------------------------------------------

\chapter{Special functions: some important formulas}

\section{Trigonometric identities}\label{A:trig}
%-----------------------------------------------------------------------------------------------------------------------------------------

In the body of the thesis we needed to use some slightly unusual
trigonometric identities. They can be derived from standard ones
without too much difficulty but  are sufficiently unusual to be
worth mentioning explicitly:
 \begin{equation}
 \label{E:trig1}
 \tan A \; \tan B = { \cos(A-B) - \cos(A+B)\over\cos(A-B) + \cos(A+B)};
 \end{equation}
 \begin{equation}
  \label{E:trig2}
 \tan\left({ A+B\over2}\right)  \; \tan\left({ A-B\over2}\right) = { \cos B- \cos A\over\cos B + \cos A};
 \end{equation}
 and
 \begin{equation}
  \label{E:trig3}
\cos(A+2B) + \cos A = 2 \cos B \; \cos(A+B).
\end{equation}

%-----------------------------------------------------------------------------------------------------------------------------------------
\section{Gamma function identities and approximations}\label{A:gamma}
%-----------------------------------------------------------------------------------------------------------------------------------------
The key Gamma function identity we need is
\begin{equation}
\label{E:reflection} \Gamma(z)\; \Gamma(1-z) = {\pi\over\sin(\pi
z)}.
\end{equation}
The Stirling approximation for Gamma function gives

\begin{equation}
\Gamma(x)=\sqrt{\frac{2\pi}{x}}\left(\frac{x}{e}\right)^{x}\left(1+O\left(\frac{1}{x}\right)\right);~~~~~Re(x)>0,~~~~~|x|\to\infty.
\end{equation}
We also need the following asymptotic estimate based on the Stirling
approximation
\begin{equation}
\label{E:stirling}
 {\Gamma(x+{1\over2})\over\Gamma(x)} = \sqrt{x} \;  \left[1+ O\left({1\over x}\right)\right];   \qquad\qquad Re(x)>0,~~~~~~|x|\to \infty.
\end{equation}

%-----------------------------------------------------------------------------------------------------------------------------------------
\section{Hypergeometric function identities}\label{A:hyper}
%-----------------------------------------------------------------------------------------------------------------------------------------
The key hypergeometric function identities we need are Bailey's
theorem
\begin{equation}
\label{E:bailey} _2F_1\left(a,1-a,c,{1\over2}\right) =
{\Gamma({c\over2}) \Gamma({c+1\over2}) \over \Gamma({c+a\over2})
\Gamma({c-a+1\over2})},
\end{equation}
which is easily found in many standard references (for example \cite{specialfunctions}), and the
particular differential identity
\begin{equation}
\label{E:differential} {\d \left\{ _2F_1\left(a,b,c,z\right)
\right\} \over \d z} = {c-1\over z} \left[  \;
_2F_1\left(a,b,c-1,z\right) -  \;_2F_1\left(a,b,c,z\right) \right],
\end{equation}
which is easy
enough to verify once it has been presented.

\section{Bessel function identities and expansions}\label{A:Bessel}

The important differential identity which holds for Bessel functions is the following:

\begin{eqnarray}\label{Besselderivative}
\frac{dJ_{\alpha}(x)}{dx}=\frac{1}{2}\big[J_{\alpha-1}(x)-J_{\alpha+1}(x)\big].
\end{eqnarray}
What was also needed was the asymptotic expansion (see for example \cite{specialfunctions}):
\begin{eqnarray}\label{Besselasymptotic}
J_{\alpha}(x)=\sqrt{\frac{2}{\pi x}}\bigg[P(\alpha,x)\cos\bigg(x-\frac{\pi\alpha}{2}-\frac{\pi}{4}\bigg)~~~~~~~~~~~~~~~~~~~~~~~~~~~\nonumber\\
-~Q(\alpha,x)\sin\bigg(x-\frac{\pi\alpha}{2}-\frac{\pi}{4}\bigg)\bigg],
\end{eqnarray}
where
\begin{equation}\label{P}
P(\alpha,x)\equiv\sum_{n=0}^{\infty}(-1)^{n}\frac{\Gamma\left(\frac{1}{2}+\alpha+2n\right)}{(2x)^{2n}(2n)!\Gamma\left(\frac{1}{2}+\alpha-2n\right)}
\end{equation}
and
\begin{equation}\label{Q}
Q(\alpha,x)\equiv\sum_{n=0}^{\infty}(-1)^{n}\frac{\Gamma\left(\frac{1}{2}+\alpha+2n+1\right)}{(2x)^{2n+1}(2n+1)!\Gamma\left(\frac{1}{2}+\alpha-2n-1\right)}.
\end{equation}

%----------------------------------------------------------------------------------------------------------------------------------------
%\clearpage
%----------------------------------------------------------------------------------------------------------------------------------------

%add about connection etc.

\pagestyle{plain}
\chapter*{Table of symbols and terms for Chapter 3}
\addcontentsline{toc}{chapter}{Table of symbols and terms for Chapter 3}

\begin{tabular}{|p{4cm}|p{8.5cm}|p{2cm}|}
\hline
{\bf Symbol:}  &   {\bf Definition/Explanation:}  &  {\bf Page:}   \\
\hline
$D(\mathbb{R}^{n})$ &  compactly supported, smooth functions ($C^{\infty}(\mathbb{R}^{n})$)  &  101   \\
\hline
$D'(\mathbb{R}^{n})$  &  space of distributions on the space of compactly supported, smooth functions  & 102   \\
\hline
$f_{\epsilon}(x_{i})$, ~~$x_{i}\in\mathbb{R}^{n}$  &  ~~~~~~~~~~~~~~~~~~~~~~~~$\frac{1}{\epsilon^{n}}f\left(\frac{x_{i}}{\epsilon}\right)$  & 106  \\
\hline
$\mathcal{E}_{M}(\mathbb{R}^{n})$  &   ~~~~~~~~~~~~~~~Moderate functions   & 105,107,111 \\
\hline
$\mathcal{N}(\mathbb{R}^{n})$   &  ~~~~~~~~~~~~~~~Negligible functions  & 105,107,111 \\
\hline
mollifier,   &  (usually) $\epsilon-$sequence of smooth, compactly  &   105   \\
smoothing kernel &  supported functions with integral normed to 1~ and with support ``stretching'' to ~0~ as ~$\epsilon\to 0$ &  \\
\hline
$C(f)$   &   embedding of distributions into Colombeau algebra by a convolution with a mollifier  &  107 \\
\hline
$M$  &   (in the section \ref{NewApp}) manifold, on which one can define a smooth atlas  & 121 \\
\hline
$\mathcal{A}$   &   maximal piecewise smooth atlas, where the transformation Jacobians are bounded on every compact set    &  121 \\
\hline
$\mathcal{S}$  &   maximal smooth subatlas of $\mathcal{A}$   & 122 \\
\hline
$\Omega_{Ch}$  &   subset of a manifold, such that it can be mapped to $\mathbb{R}^{n}$ by a chart mapping  & 122 \\
\hline
\end{tabular}

\begin{tabular}{|p{3cm}|p{8.5cm}|p{1.5cm}|}
\hline
{\bf Symbol:}  &   {\bf Definition/Explanation:}  &  {\bf Page:}   \\
\hline
$Ch(\Omega_{Ch})$  &  chart from the atlas $\mathcal{A}$ on the set $\Omega_{Ch}$  & 122 \\
\hline
$C^{P}(\Omega_{Ch})$  &   class of compactly supported piecewise smooth 4-forms, (we work from the beginning with a 4D manifold), having their support inside the set $\Omega_{Ch}$  & 124 \\
\hline
$\mathcal{\tilde S}$    &   maximal subatlas of $\mathcal{A}$, such that there exist elements of the class $C^{P}(\Omega_{Ch})$, that have in this subatlas smooth scalar densities   & 124  \\
\hline
$C^{P}_{S(\mathcal{\tilde S})}(\Omega_{Ch})$  &    set of all 4-forms from $C^{P}(\Omega_{Ch})$, that are in $\mathcal{\tilde S}$ given by smooth scalar densities  & 124 \\
\hline
$C^{P}_{S}(\Omega_{Ch})$  &  $\cup_{\mathcal{\tilde S}} ~C^{P}_{S (\mathcal{\tilde S})}(\Omega_{Ch})$ & 124\\
\hline
$C^{P}(M)$  & $\cup_{\Omega_{Ch}}C^{P}(\Omega_{Ch})$ & 124 \\  $C^{P}_{S (\mathcal{\tilde S})}(M)$ &    $\cup_{\Omega_{Ch}}C^{P}_{S (\mathcal{\tilde S})}(\Omega_{Ch})$  &  \\ $C^{P}_{S}(M)$  & $\cup_{\Omega_{Ch}}C^{P}_{S}(\Omega_{Ch})$  &  \\
\hline
$D'(M)$  &   linear generalized scalar fields  & 125 \\
\hline
$D'_{(\mathcal{\tilde S})}(M)$  &  linear generalized scalar fields defined (at least) on the class $C^{P}_{S (\mathcal{\tilde S})}(M)$  & 126 \\
\hline
$D'_{(\mathcal{\tilde S}o)}(M)$  & linear generalized scalar fields defined exclusively on the class $C^{P}_{S (\mathcal{\tilde S})}(M)$  & 126 \\
\hline
$D'_{E}(M)$, $D'_{E (\cup_{n}\mathcal{\tilde S}_{n} o)}(M)$  & distributional analogue of piecewise smooth scalar fields  & 125, 126 \\
\hline
$D'_{S}(M)$,  $D'_{S (\cup_{n}\mathcal{\tilde S}_{n} o)}(M)$  & objects as close as possible to a distributional analogue of smooth scalar fields  & 125, 126 \\
\hline
$D'_{A}(M)$  &  generalized scalar fields  & 126 \\
\hline
$D'_{E A}(M)$, $D'_{S A}(M)$ (etc.)  &  generalized scalar fields constructed from the elements of the classes $D'_{E}(M)$, $D'_{S}(M)$ (etc.)  & 129 \\
\hline
$D'^{m}_{n}(M)$  &  linear generalized tensor fields of rank $(m,n)$ & 127-128 \\
\hline
$At(\mathcal{\tilde S})$  &  specific function, which maps atlases to atlases  & 129  \\
\hline
\end{tabular}

\begin{tabular}{|p{4.7cm}|p{8.5cm}|p{1.5cm}|}
\hline
{\bf Symbol:}  &   {\bf Definition/Explanation:}  &  {\bf Page:}   \\
\hline
$D'^{m}_{n (\cup_{l}At(\mathcal{\tilde S}_{l}))}(M)$  & linear generalized tensor fields of rank $(m,n)$, defined (at least) on all the classes $C^{P}_{S (\mathcal{\tilde S}_{l})}(M)$, labeled by $l$, in the atlases $At(\mathcal{\tilde S}_{l})$ & 129 \\
\hline
$D'^{m}_{n (\cup_{l}At(\mathcal{\tilde S}_{l})o)}(M)$  & linear
generalized tensor fields of rank $(m,n)$, defined exclusively on the classes $C^{P}_{S (\mathcal{S}_{l})}(M)$, labeled by $l$, in the atlases $At(\mathcal{\tilde S}_{l})$ & 128\\  \hline
$D'^{m}_{n E}(M)$  & distributional analogue of rank $(m,n)$ piecewise smooth tensor fields  & 128  \\ \hline
$D'^{m}_{n S}(M)$  &  objects as close as possible to a distributional analogue of rank $(m,n)$ smooth tensor fields & 128 \\ \hline
$D'^{m}_{n (\cup_{l}\mathcal{\tilde S}_{l})}(M)$  & subclass of $D'^{m}_{n (\cup_{l}At(\mathcal{\tilde S}_{l}))}(M)$, such that  $At(\mathcal{\tilde S}_{l})$  has as image of $\mathcal{\tilde S}_{l}$ the atlas $\mathcal{S}_{l}\subseteq\mathcal{\tilde S}_{l}$  & 129 \\
\hline
$D'^{m}_{n A}(M)$  &  generalized tensor fields of rank $(m,n)$  & 129 \\
\hline
$D'^{m}_{n E A}(M)$,  $D'^{m}_{n S A}(M)$ (etc.)  &  generalized tensor fields of rank $(m,n)$ constructed from the elements of the classes $D'^{m}_{n E}(M)$, $D'^{m}_{n S}(M)$  & 129  \\
 \hline
$\Gamma^{m}(M)$,~$\Gamma^{m}_{E}(M)$,~$\Gamma^{m}_{S}(M)$, ~$\Gamma^{m}_{(\cup_{n}At(\mathcal{\tilde S}_{n}))}(M)$, ~$\Gamma^{m}_{(\cup_{n}At(\mathcal{\tilde S}_{n})o)}(M)$, ~~$\Gamma^{m}_{(\cup_{n}\mathcal{\tilde S}_{n})}(M)$ & Gamma objects (generalized from linear generalized tensor fields of rank $(a,b)$, by imposing no condition on the transformation properties)  & 129-130 \\ \hline
$\Gamma^{m}_{A}(M)$  &  sets of algebras constructed from the elements of the class $\Gamma^{m}(M)$ (generalization of $D^{m}_{n A}(M)$) & 130\\
\hline
$\Gamma^{m}_{E A}(M)$, $\Gamma^{m}_{S A}(M)$ (etc.)  &  sets of algebras constructed from the elements of the classes $\Gamma^{m}_{E}(M)$, $\Gamma^{m}_{S}(M)$ (etc.) (generalization from  $D'^{m}_{n E A}(M)$,  $D'^{m}_{n S A}(M)$ ) & 130 \\
\hline
$T^{\mu...}_{\nu...}(Ch_{k})$  &  multi-index matrix obtained from the tensor components of $T^{\mu...}_{\nu...}$ in the chart $Ch_{k}$  & 131 \\
\hline
$Ch'(q,\Omega_{Ch})$  &  chart mapping from $\Omega_{Ch}$ to the whole $\mathbb{R}^{4}$, such that it maps the point $q$ to 0  & 133 \\
\hline
$A^{n}(\mathcal{\tilde S}, Ch'(q,\Omega_{Ch}))$  & specific subclass of the class $C^{P}_{S (\mathcal{\tilde S})}(\Omega_{Ch})$ & 133\\
\hline
\end{tabular}

\begin{tabular}{|p{3.5cm}|p{7.5cm}|p{1.5cm}|}
\hline
{\bf Symbol:}  &   {\bf Definition/Explanation:}  &  {\bf Page:}   \\
\hline
$\omega_{\epsilon}(y)$  & specific one parameter class of elements of the class $C^{P}_{S (\mathcal{\tilde S)}}(M)$  & 133 \\
\hline
$\approx$  &   the equivalence relation  & 133-134 \\
\hline
$\tilde\Gamma^{m}(M)$, $\tilde\Gamma^{m}_{E}(M)$, $\tilde D'^{m}_{n A}(M)$, $\tilde D'^{m}_{n E A}(M)$ (etc.) & sets of equivalence classes constructed from the elements of the classes $\Gamma^{m}(M)$, $\Gamma^{m}_{E}(M)$, $D'^{m}_{n A}(M)$, $D'^{m}_{n E A}(M)$ (etc.) & 134 \\
\hline
$\Lambda$  &   specific subclass of the class $\Gamma^{m}_{E (\cup_{n}At(\mathcal{\tilde S}_{n})o)}(M)$  & 135-136 \\
\hline
$A_{s}(T^{\mu...}_{\nu...})$   & tensor field associated to $T^{\mu...}_{\nu...}\in\Gamma^{m}_{A}(M)$   & 135 \\
\hline
$\tilde\Omega(Ch_{k})$  & set of Lebesgue measure 0, on which a particular multi-index matrix $T^{\mu...}_{\nu...}(Ch_{k})$ is discontinous & 136 \\
\hline
$\partial$-derivative  &  specific operator acting on the elements of the class $\Gamma^{m}_{A}(M)$ and generalizing the operator of the partial derivative (as a part of the covariant derivative definition) & 144-145 \\
\hline
$D_{C(U)}$  & covariant derivative along the vector field $U^{i}$ & 145-146 \\
\hline
$\mathcal{D}_{n}$  &  $n$-times differentiable subatlas of $\mathcal{A}$  &  146\\
\hline
$S'_{n}$ related to the atlas $\mathcal{D}_{n}$ &  ~~~specific subclass of  $\Gamma^{m}_{E (\cup_{n}At(\mathcal{\tilde S}_{n})o)A}(M)$  & 146 \\
\hline
\end{tabular}

\chapter*{Publications and papers}
\addcontentsline{toc}{chapter}{Publications and papers}

\section*{Journal}
\begin{itemize}

\item
J. Skakala ~and ~M. Visser,  ~~\emph{``Generic master equations for quasi-normal frequencies''},~~ Journal of High Energy Physics ~~(JHEP) ~{\bf 1007}:070, ~(2010),  ~~[arXiv: gr-qc/1009.0080]
\item
J. Skakala ~and ~M. Visser,   ~~\emph{``Semi-analytic results for quasi-normal frequencies''},~~ Journal of High Energy Physics ~~(JHEP)  ~{\bf 1008}:061,~ (2010),  ~~[arXiv: gr-qc/1004.2539]
\item
J. Skakala ~and ~M. Visser,   ~~\emph{``Highly-damped quasi-normal frequencies for piecewise Eckart potentials''}, ~~Physical Review D, ~~{\bf 81}:125023, ~(2010), ~~[arXiv: gr-qc/1007.4039]
\item
J. Skakala ~and ~M. Visser,   ~~\emph{``The causal structure of spacetime is a parametrized Randers geometry''}, ~~Classical and Quantum Gravity,~~ {\bf 28}:065007,~(2011), ~~[arXiv: gr-qc/1012.4467]
\item
J. Skakala ~and ~M. Visser,   ~~\emph{``Bi-metric pseudo-Finslerian spacetimes''}, ~~Journal of Geometry and Physics,~~ {\bf 61}:1396-1400, ~(2011), ~~[arXiv: gr-qc/1008.0689]
\item
J. Skakala ~and ~M. Visser,   ~~\emph{``Pseudo-Finslerian spacetimes and multi-refringence''}, ~~International Journal for Modern Physics D, ~~{\bf 19}:1119-1146, ~(2010),  ~~[arXiv: gr-qc/0806.0950]
\end{itemize}

\section*{Conference proceedings}
\begin{itemize}
\item
J. Skakala ~and ~M. Visser,  ~~\emph{``Birefringence in pseudo-Finsler spacetimes''}, ~~NEBXIII ~conference (~Recent Developments in Gravity), ~~Thessalonika (Greece), ~~June 2008,  ~~J.Phys.Conf.Ser.~{\bf 189}:012037, ~(2009),  ~~[arXiv: gr-qc/0806.0950]
\item
J. Skakala ~and ~M. Visser,  ~~\emph{Quasi-normal frequencies: Semi-analytic results for highly damped modes''},  ~~The Spanish relativity meeting - ERE 2010, ~Granada (Spain), September 2010, ~will be published in ~J.Phys.Conf.Ser.  ~~[arXiv: gr-qc/1011.4634]
\item
J. Skakala,  ~~\emph{``New ideas about multiplication of tensorial distributions''}, ~~12-th Marcel Grossman meeting (MG12), ~Paris (France), July 2009, ~soon to be published
\end{itemize}
\section*{E-print only}
\begin{itemize}
\item
J. Skakala,  ~~\emph{``New ideas about multiplication of tensorial distributions''},  ~~[arXiv: gr-qc/0908.0379] 
\item
J. Skakala ~and ~M. Visser,  ~~\emph{``Birkhoff-like theorem for rotating stars in 2+1 dimensions''},  ~~[arXiv: gr-qc/0903.2128]
\end{itemize}

%----------------------------------------------------------------------

%-------------
\chapter*{Some work that did not appear in the thesis}

\subsection*{Jozef Skakala and Matt Visser, ~~The causal structure of spacetime is a parametrized Randers geometry }

\subsubsection*{Abstract:}
There is a by now well-established isomorphism between stationary 4-dimensional spacetimes and 3-dimensional purely spatial Randers geometries - these Randers geometries being a particular case of the more general class of 3-dimensional Finsler geometries. We point out that in stably causal spacetimes, by using the (time-dependent) ADM decomposition, this result can be extended to general non-stationary spacetimes - the causal structure (conformal structure) of the full spacetime is completely encoded in a parameterized (time-dependent) class of Randers spaces, which can then be used to define a Fermat principle, and also to reconstruct the null cones and causal structure. \\
{\bf  published in ~Classical and Quantum Gravity ~28,~ (2011),~ 065007,}\\
{\bf [arXiv: gr-qc/1012.4467].}

\subsection*{Jozef Skakala and Matt Visser, Birkhoff-like theorem for rotating stars in (2+1) dimensions}

\subsubsection*{Abstract:}
Consider a rotating and possibly pulsating "star" in (2+1) dimensions. If the star is axially symmetric, then in the vacuum region surrounding the star, (a region that we assume at most contains a cosmological constant), the Einstein equations imply that under physically plausible conditions the geometry is in fact stationary. Furthermore, the geometry external to the star is then uniquely guaranteed to be the (2+1) dimensional analogue of the Kerr-de Sitter spacetime, the BTZ geometry. This Birkhoff-like theorem is very special to (2+1) dimensions, and fails in (3+1) dimensions. Effectively, this is a "no hair" theorem for (2+1) dimensional axially symmetric stars: the exterior geometry is completely specified by the mass, angular momentum, and cosmological constant. \\
{\bf arXiv: gr-qc/0903.2128.}

\pagestyle{plain}

%%%%%%%%%%%%%%%%%%%%%%%%%%%%%%%%%%%%%%%%%%%%%%%%%%%%%%%

% and of course book style knows about backmatter
% \backmatter caused problems with appendices :-(
% and of course report style doesn't
%%%%%%%%%%%%%%%%%%%%%%%%%%%%%%%%%%%%%%%%%%%%%%%%%%%%%%%

%\bibliographystyle{ieeetr}
\bibliographystyle{acm}
\bibliography{myrefs}

\end{document}